\documentclass[11pt]{article}
\usepackage[margin=2.54cm]{geometry}

\usepackage{graphicx}
\usepackage{multirow}
\usepackage{stmaryrd}
\usepackage{subfiles}
\usepackage{multicol}
\usepackage{pifont}%
\usepackage{caption}
\usepackage{bbm}
\usepackage[title]{appendix}
\usepackage{booktabs}            
\usepackage{mathtools}     

\usepackage{natbib}
\usepackage{amsmath, amssymb, enumerate,amsthm}
\usepackage{amsfonts}
\usepackage{algorithm}
\usepackage{scalerel}
\usepackage{algpseudocode}
\usepackage{setspace}
\usepackage{bigints}            
\usepackage{xcolor}
\usepackage{accents}
\usepackage[colorlinks=true,
            pdfpagemode=UseNone,
            urlcolor=blue,
            citecolor=blue,
            linkcolor=violet,
            bookmarks=true,
            backref=page
            ]{hyperref}
            
\newcommand{\ind}[1]{\mathbbm{1}\left\{#1\right\}}
\newcommand{\re}{\mathbb{R}}
\newcommand{\CC}{\texttt{C++}}
\newcommand{\rd}[1]{\mathbb{R}^{#1}}
\newcommand{\rhatbar}{\overline{ \widehat R} }

\newcommand{\scaled}[1]{\scaleto{#1}{4pt}}
\newcommand{\ranki}{\widehat{R}_i}
\newcommand{\norm}[1]{\left\lVert#1\right\rVert}
\newcommand{\eqd}{\stackrel{\text{d}}{=}}

\newcommand{\cond}{\stackrel{\text{d}}{\rightarrow}}

\newcommand{\conp}{\stackrel{\text{p}}{\rightarrow}}
\newcommand{\E}[2]{{\rm E}_{#1}\left(#2\right)}
\newcommand{\Eee}[2]{{\rm E}_{#1}(#2)}
\newcommand{\Var}[1]{{\rm V}ar\left(#1\right)}
\newcommand{\Varr}[1]{{\rm V}ar(#1)}
\newcommand{\limn}{\lim_{N\rightarrow \infty}}
\newcommand{\iid}{i.i.d.}
\makeatletter
\newcommand{\Biggg}{\bBigg@{4}}
\newcommand{\Vast}{\bBigg@{5}}
\makeatother

\DeclareMathOperator*{\argmax}{ {\rm argmax}}

\DeclareMathOperator*{\INT}{ {\rm INT}}
\DeclareMathOperator*{\argmin}{ {\rm argmin}}

\DeclarePairedDelimiter\floor{\lfloor}{\rfloor}

\newtheorem{theorem}{Theorem}

\newtheorem{ass}{Assumption}
\newtheorem{definition}{Definition}

\providecommand{\keywords}[1]
{
  \small	
    \begin{center}\textbf{\textit{Keywords---}} #1\end{center}
}
\title{Robust multiple change-point detection for multivariate variability using data depth}
 \author{Kelly Ramsay, Shoja'eddin Chenouri}
\date{November 2021}

\begin{document}

\maketitle
\doublespacing
\begin{abstract}
    In this paper, we introduce two robust, nonparametric methods for multiple change-point detection in the variability of a multivariate sequence of observations.  
    We demonstrate that changes in ranks generated from data depth functions can be used to detect changes in the variability of a sequence of multivariate observations. 
    In order to detect more than one change, the first algorithm uses methods similar to that of wild-binary segmentation. 
    The second algorithm estimates change-points by maximizing a penalized version of the classical Kruskal Wallis ANOVA test statistic. 
    We show that this objective function can be maximized via the well-known PELT algorithm. 
    Under mild, nonparametric assumptions both of these algorithms are shown to be consistent for the correct number of change-points and the correct location(s) of the change-point(s). 
    We demonstrate the efficacy of these methods with a simulation study, where we compare our new methods to several competing methods. 
    We show our methods outperform existing methods in this problem setting, and our methods can estimate changes accurately when the data are heavy tailed or skewed. 
\end{abstract}
\section{Introduction}
The manufacturing industry motivated the development of change-point methods, that is, methods for detecting and dating distributional changes in a sequence of observations \citep{Page1954}. 
Change-point methods have since been applied to a much wider variety of research areas including climate change \citep{Reeves2007}, speech recognition \citep{Aminikhanghahi2017} and finance \citep{Wied2012}, among others. 
With respect to a sequence of observations, the terms `structural break' and `change-point' refer to time points in the sequence during which there is a sudden change in the distribution from which the data is being generated. 
Change-point detection can be separated into two settings: `online' and `offline'. 
In the online setting, the data are being received by the analyst one datum at a time, and the goal is to detect a change as soon as possible, without too many false alarms. 
In the offline setting, the analyst has access to the entirety (or at least enough) of the data set, and the goal is to identify if and when changes occurred over the course of observation. 
In this work, we focus on the offline setting, for a summary of nonparametric methods in the online setting see \citep[][]{Chakraborti2019}. 

There are different variants of the offline change-point problem. 
Instead of identifying general changes in distribution, one might only be interested in identifying changes in the mean of the sequence \citep{Chenouri2019, Fryzlewicz2014}, changes in the correlations of the sequence \citep{Galeano2014} or changes in the covariance matrix of the sequence \citep{Chenouri2020DD, Wang2021}. 
One may also be interested in another type of distributional change entirely. 
In this paper, we aim to detect changes in the variability of a sequence of multivariate observations. 

To elaborate, suppose that the analyst suspects that there may exist time point(s) at which there are increases or decreases in the variance of one or more variates, or that there may exist time point(s) at which there is a change in the strength of the relationship between at least two variates. 
One can think of an analyst looking for changes in the variability and/or correlation strength of several financial asset returns. 
We call this type of change-point a multivariate variability change-point. 
For example, changes in a commonly used norm of the covariance matrix would be considered a change in variability.
Another way to interpret changes in variability is through how the data cloud is affected; changes in variability produce changes in the magnitude and/or shape of the data cloud, rather than say, rotations or translations of the data cloud. 
A change in variability is the result of expansions and/or contractions of one ore more parameters of the covariance matrix. 
Note that this type of change-point differs from a change in the covariance matrix; it doesn't include cases where the covariance matrix is multiplied by an orthonormal matrix. 
For example, it does not include the situation where the correlation between two variates switches signs. 

In order to detect changes in the variability of the data, we first transform the sequence of observations to a univariate sequence via data depth ranks, as is done by \cite{Chenouri2020DD} for the at most one change-point problem. 
We then try to detect multiple changes in the mean of this univariate sequence. 
We introduce two methods to do this task, the first of which is a wild binary segmentation type algorithm based on rank CUSUM statistics \citep{Fryzlewicz2014, Chenouri2020DD}. 
The second method is based on finding the set of change-points which maximize a penalized version of the classical Kruskal-Wallis test statistic used in nonparametric ANOVA \citep{Kruskal1952}. 
The implementation of this second method is based on the ``pruned exact linear time'' algorithm \citep{Killick2012}. 
To see the benefits of our proposed methods, we must first review existing methods.


There is a vast literature relating to the change-point problem, going back almost a century \citep{Shewhart1931, Page1954}. 
The literature includes a variety of approaches for both univariate, multivariate, single and multiple change-point detection methods \citep[see the following review papers][and the references therein]{Reeves2007, Aue2013, Aminikhanghahi2017}. 
Much of the literature, especially in the multivariate setting, has focused on the detection of shifts in the mean of the process, e.g., \citep{Truong2018}.

Considerably less attention has been given to shifts in the second order behaviour of a sequence of observations. 
When second order change-points in the multivariate setting have been studied, the bulk of the literature has been concerned with detecting changes in the correlation structure. 
\cite{Galeano2007} proposed a parametric framework for detecting changes in the correlation and variance structure of a multivariate time series, using both a likelihood ratio and a CUSUM statistic approach. 
\cite{Wied2012} proposed a nonparametric approach based on cumulative sums of sample correlation coefficients to detect a single change-point in the correlation structure of bivariate observations. 
This was later extended to multiple change-points \citep{Galeano2014} and further to the multivariate setting \citep{Galeano2017}. 
\cite{Posch2019} has further extended the methods of \citep{Galeano2017} to the high-dimensional setting by first applying dimension reduction techniques. 
One draw-back to the methods of \citep{Galeano2014} is that they assume constant variances and expectations over time. 
Rather recently, a few alternative methods have been proposed, which include methods related to eigenvalues \citep{Bhattacharyya2018}, residuals \citep{Duan2018}, semi-parametric CUSUM statistics \citep{Zhao2017} and kernel methods \citep{Cabrieto2018}.

Literature related to estimating a change-point in the covariance matrix is quite recent, and relatively sparse. 
\cite{Aue2009} take a CUSUM statistic approach similar to that of \cite{Galeano2014}. 
\cite{Kao2018} suggested a CUSUM statistic procedure based on eigenvalues. 
\cite{Chenouri2020DD} considered a CUSUM based on ranks generated by data depth functions for detecting a single change-point.
The high-dimensional setting has been tackled by \cite{Dette2018} and \cite{Wang2021}. 
\cite{Dette2018} considers a two-step procedure based on dimension reduction techniques and a CUSUM statistic. 
\cite{Wang2021} is the only paper, to the best of our knowledge, seeking to identify multiple change-points, rather than a single change-point. 
They compare binary segmentation procedures \citep{Venkatraman1992} and wild binary segmentation procedures \citep{Fryzlewicz2014} based on a CUSUM statistic, under the assumption of sub-Gaussian observations.

\cite{Fryzlewicz2014} developed wild binary segmentation as an improvement on the well-known univariate multiple change-point algorithm binary segmentation \citep{Venkatraman1992}. 
Binary segmentation has been used to extend single change-point algorithms to multiple change-point algorithms in many settings \citep[such as][]{Aue2013, Galeano2014, Galeano2017, Duan2018, Wang2021,Chenouri2019}. 
The extension and study of wild binary segmentation in the multivariate setting, with respect to changes in the covariance structure of a time series has only been done by \cite{Wang2021}. 

In addition to methods where the change-point type is specified, there exists several nonparametric algorithms designed to detect general changes in the distribution of the observations. 
\cite{Matteson2014} studied the e-divisive algorithm, which can detect the location and number of change-points in the distribution of a sequence of multivariate observations. 
Their method is based on distances between characteristic functions and a hierarchical clustering inspired iteration. 
Their methods are extended in \cite{Zhang2017}, where a pruning component is added to an existing, dynamic programming-based change-point algorithm. 
These authors apply this pruning method to the e-divisive algorithm and the kernel change-point methods of \cite{2012Arlot}. 
This group of methods are implemented in the \texttt{ecp R} package \citep{ecp_package}. 
At first the rank-based multiple change-point method of \cite{2011arXiv1107.1971L} may seem similar to our methods, but their procedure requires the number of change-points to be fixed. 
Further, their methods are based on component-wise ranks, which have several known issues, such as a lack of transformation invariance \citep{bickel1965}. 


The change-point literature is therefore lacking methods for specifically detecting multiple changes in the variability of multivariate data. 
Many of the papers discussed focus on the at most one change problem, or are not designed to detect changes in variability, or even changes in the covariance matrix of the data. 
The only directly comparable paper is that of \citep{Wang2021}. 
Even this method is designed for the high dimensional setting; in our simulation study, when the dimension is low to moderate, our method outperforms this method. 
The other comparable methods would be those that detect multiple, general changes in the distribution of the data, such as those of \citep{Zhang2017}. 
We demonstrate that our method is able to outperform these general methods in simulation, when the change-points are all variability change-points. 
This is not surprising; our method sacrifices generality for accuracy. 

In addition, the aforementioned change-point methods do not consider robustness to outlying observations. 
For example, many of the existing methods assume that the data are sub-Gaussian \citep{Dette2018, Wang2021}. 
Furthermore, existing papers often present no simulation results concerning a method's performance under heavy tailed data. 
For example, we show in simulation that the methods of \citep{Matteson2014, Wang2021} do not perform well when the data are heavy tailed. 
By contrast, our theoretical and simulation results show that our method works well in scenarios where the data are heavy tailed. 

The rest of the paper is organized as follows, Section \ref{sec::prel} introduces the data model, data depth and depth-based ranks. 
Section \ref{sec::meth} outlines the proposed change-point detection procedures. 
Section \ref{sec::theo} presents consistency results (with rates) for both of our presented methods.  
Section \ref{sec:sim} presents simulation results, including a discussion of the tuning parameters. 
We test the proposed change-point methods in a variety of scenarios and compare the methods to one another as well as to the methods of \citep{Matteson2014, Zhang2017,  Wang2021}. 
In Section \ref{sec::da} we analyze four European daily stock returns. 
This is the same data set analyzed by \cite{Galeano2017} and we compare our results to theirs. 
\section{Preliminary material}\label{sec::prel}
\subsection{Data depth functions}\label{sec::depth}
Data depth functions, among other things, provide a method of defining quantiles and ranks for multivariate data, which in turn facilitates the extension of univariate methods based on these functions to the multivariate setting and beyond. 
A data depth function $\mathcal{D}(\cdot ;F)\colon \rd{d}\rightarrow\re$ assigns each value in $x\in \rd{d}$ a real number which describes how central $x$ is with respect to some distribution $F$ (over $\rd{d}$). 
Often $F=F_{*,\scaled{N}}$, the empirical distribution of the data and the depth values $\mathcal{D}(x ;F_{*,\scaled{N}})$ describe how central or `deep' $x$ is in the sample. 

Sample ranks based on data depth functions can be calculated as follows. Suppose that $X_1,\dots,X_{\scaled{N}}$ is a random sample and $F_{*,\scaled{N}}$ is the associated empirical cumulative distribution function, then the quantity 
\begin{equation}
     \ranki\coloneqq \#\{X_j\colon \mathcal{D}(X_j;F_{*,\scaled{N}})\leq \mathcal{D}(X_i;F_{*,\scaled{N}})\},\ j\in\{1,\ldots,N\}
     \label{eqn:d_ranks}
\end{equation}
represents the depth-based rank of $X_i$. 
The interpretation of $\ranki$ is slightly different than that of univariate ranks, because here the observations have a high rank when they are deep inside the data cloud, rather than on the extreme end of the data. 
In fact, center outward ranks have been used for detecting differences in univariate variability before \citep{Sieg1960, ansari1960}.

Many definitions of depth functions exist \citep{ Tukey1974,Dyckerhoff1996, Zuo2003, Serfling2006, RAMSAY201951} and so we limit ourselves to three popular ones. 
The first of these is halfspace depth \citep{Tukey1974}, the seminal depth function. 
\begin{definition}[Halfspace depth.] Let $S^{d-1}\coloneqq \{x\in \rd{d} \colon \norm{x}=1\}$ be the set of unit vectors in $\rd{d}$. Define the halfspace depth of a point $x\in \rd{d}$ with respect to some distribution $X\sim F$ as,
\begin{equation}
    \mathcal{D}_H(x;F)\coloneqq\inf_{u\in S^{d-1}} \Pr(X^\top u\leq x^\top u)= \inf_{u\in S^{d-1}} F_u(x) ,
\end{equation}
where $F_u$ is the distribution of $X^\top u$ with $X\sim F$ and $\norm{\cdot}$ represents the Euclidean norm. 
\label{def::hs}
\end{definition}
Halfspace depth is the minimum of the projected mass below the projection of $x$, over all directions. 
Halfspace depth satisfies many desirable properties of a depth function such as affine invariance, consistency, maximality at center and decreasing along rays \citep[see][for more details]{Zuo2000}.
It should be mentioned that halfspace depth is frequently cited as being computationally expensive \citep{Serfling2006}, though recently an algorithm for computing half-space depth in high dimensions has been proposed \citep{Zuo2019}.

Another, less computationally prohibitive depth function is spatial depth \citep{Serfling2002}. 
Let $u\in S^{d-1}$ as defined in Definition \ref{def::hs}. 
Spatial depth is based on spatial quantiles: 
$$\mathcal{Q}(u;F)\coloneqq\min_{y} \E{F}{\norm{X-y}+(X-y)^\top u-\norm{X}-X^\top u},$$
where $\mathrm{E}_F$ represents expectation with respect to a distribution $F$. 
Spatial quantiles are extensions of univariate quantiles. 
Inverting this function at a point $x\in \rd{d}$ gives a measure of outlyingness: $\norm{\mathcal{Q}^{-1}(x;F)}$ \citep{Serfling2002}. 
Let 
$$S(x)\coloneqq \left\{\begin{array}{lr}
\frac{x}{\norm{x}} &  x\neq 0\\
0 &  x=0
\end{array}\right .  $$ 
and then define $$\norm{\mathcal{Q}^{-1}(x;F)}\coloneqq \norm{\E{F}{S(x-X)}}.$$
We can now define spatial depth.
\begin{definition}[Spatial Depth]
Define the spatial depth $\mathcal{D}_S$ of a point $x\in \rd{d}$ with respect to some distribution $F$ as
\begin{equation}
    \mathcal{D}_S(x;F)\coloneqq1-\norm{\mathcal{Q}^{-1}(x;F)}.
\end{equation}
\end{definition}
One of the main weaknesses of spatial depth is that it is only invariant under similarity transformations; not under all affine transformations. 
One way to circumvent this issue is to replace $\norm{x}$ with the generalised norm $\norm{x}_\Sigma\coloneqq\sqrt{x^\top \Sigma^{-1} x}$, where $\Sigma$ is the covariance matrix related to $F$. The depth function based on this norm is known as Mahalanobis depth.
\begin{definition}[Mahalanobis Depth]
Define the Mahalanobis depth $\mathcal{D}_M$ of a point $x\in \rd{d}$ with respect to a distribution $F$ as
\begin{equation}
    \mathcal{D}_M(x;F)\coloneqq\frac{1}{1+\norm{x-\E{F}{X}}^2_\Sigma}.
\end{equation}
\end{definition}
One criticism of Mahalanobis depth is that $\Sigma$ and $\E{F}{X}$ are usually replaced by estimators which are not robust, such as the sample covariance matrix and sample mean, respectively. 
In order for the Mahalanobis depth function to remain robust, it is necessary to use robust estimators of $\Sigma$ and $\E{F}{X}$. 
Examples of such estimators are the re-weighted MCD estimators \citep{Rousseeuw1990}. 
We denote the depth values computed using these MCD estimators by $\mathcal{D}_{M75}$, where the 75\% comes from the fact that we are using the $25\%$ breakdown version of the MCD estimators. 
Lastly, note that sample versions of all four of the depth functions discussed in this section can be obtained by replacing expectations and probabilities with sample means and probabilities based on the empirical distribution, respectively. 
\subsection{The data model, variability changes and their relation to depth ranks}\label{sec::model}
We now describe the change-point model that we will focus on. 
Suppose that $X_1,\dots,X_{N}$ is a sequence of random variables such that $X_{k_{i-1}+1},\dots, X_{k_{i}}$ are a random sample from distribution $F_i$, with, $k_0=0<k_1<\dots<k_\ell<k_{\ell+1}=N$ for some fixed, unknown $\ell$. 
Suppose that $k_i/N\rightarrow \theta_i$ as $N\rightarrow\infty$ for all $i\in \{1,\ldots,\ell\}$. 
Let $\vartheta_i=\theta_i-\sum_{j=0}^{i-1}\theta_j$ be the approximate fraction of the observations coming from $F_i$ and define
$$F_*\coloneqq \vartheta_1 F_1+\vartheta_2 F_2+\vartheta_3 F_3+\dots+\vartheta_\ell F_\ell+\vartheta_{\ell+1} F_{\ell+1}.$$
In this paper, the aim is to estimate $\ell$ and each $k_i$; the correct number of change-points along with their location, given only the sample. 
Let $\Sigma_j$ represent the covariance matrix corresponding to the distribution $F_j$, $\Sigma_*$ represent the covariance matrix corresponding to the distribution $F_*$ and let $F_{*,\scaled{N}}$ denote the empirical distribution invoked by the combined sample $X_1,\dots,X_{N}$.
We further suppose that for any $i=1,\ldots,\ell,$ $F_i$ differs from $F_{i+1}$ only in variability. 

In order for our methods to work, we must ensure that a change in variability is reflected by a change in the mean of the combined sample depth values. 
The relationship between variability and combined sample depth values has already been explored by several other authors \citep{2021arXiv210610173R, li2004}, but we give an intuitive explanation below. 
The reader may also view a short simulation study of the distribution of the depth ranks under different covariance changes in Appendix \ref{sim::rankd}. 

The fact that changes in the variability of the data produce a change in the mean of the data depth values is guaranteed from the construction of depth functions, specifically the maximality at center property combined with the quasi-concavity property. Since we assume that the pre-change and post-change data have the same location, we can also assume that the combined sample depth function will be maximized roughly at that location. 
Additionally, recall that a change in variability is a change in the magnitude and/or shape of the post-change data cloud. 
The change in the magnitude and/or shape of the data cloud will result in the post-change data being, on average, a different distance from the centre. 
Due to the quasi-concavity property, this change in distance will result in the post-change data having higher/lower combined sample depth values, on average.

 \begin{figure}[t]
\begin{minipage}[c]{.32\textwidth} 
\centering%
\includegraphics[width=.9\textwidth,trim={0 00 0 0},clip]{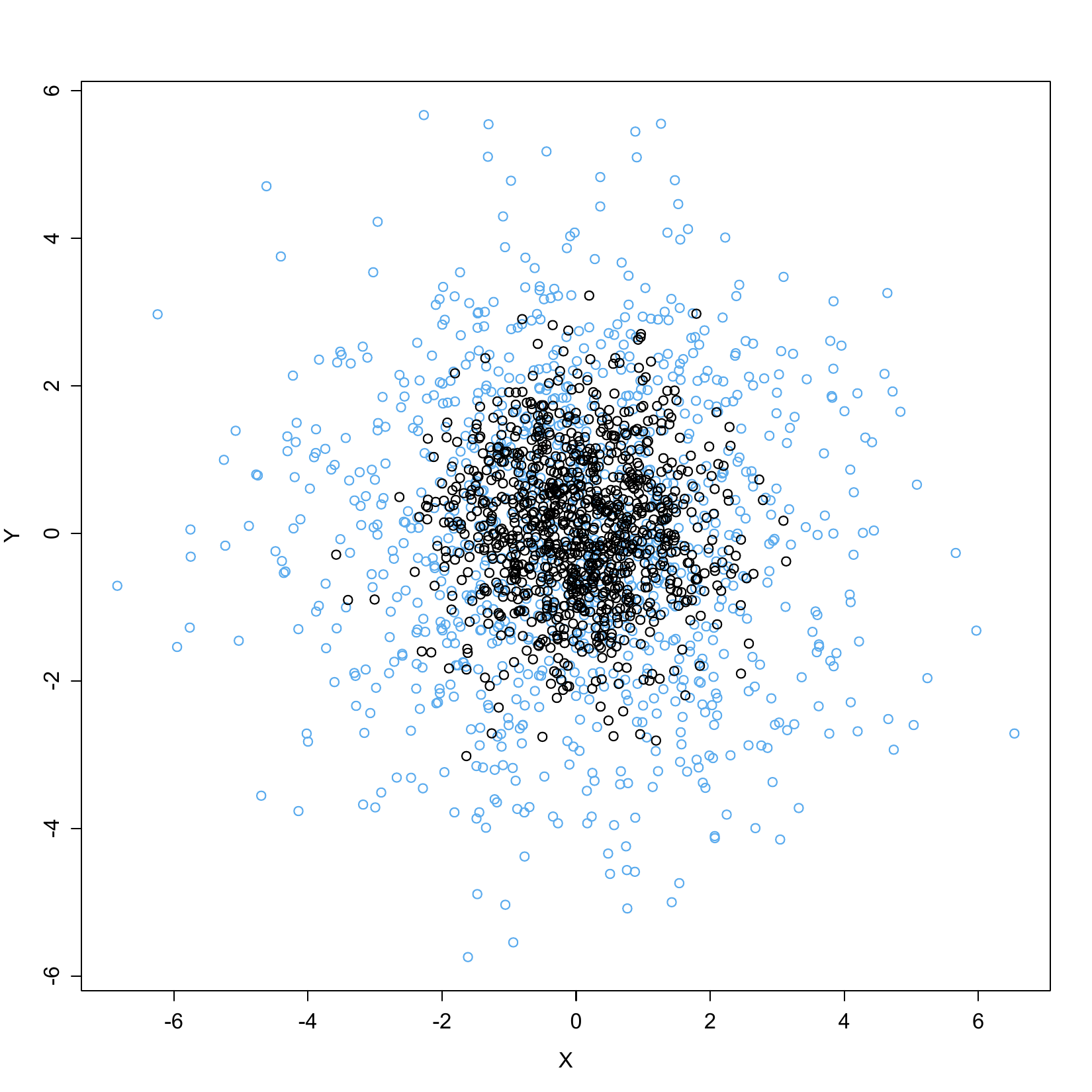}
\caption*{(a)}
\end{minipage}\hfill
\begin{minipage}[c]{.32\textwidth} 
\centering%
\includegraphics[width=.9\textwidth,trim={0 0 0 0},clip]{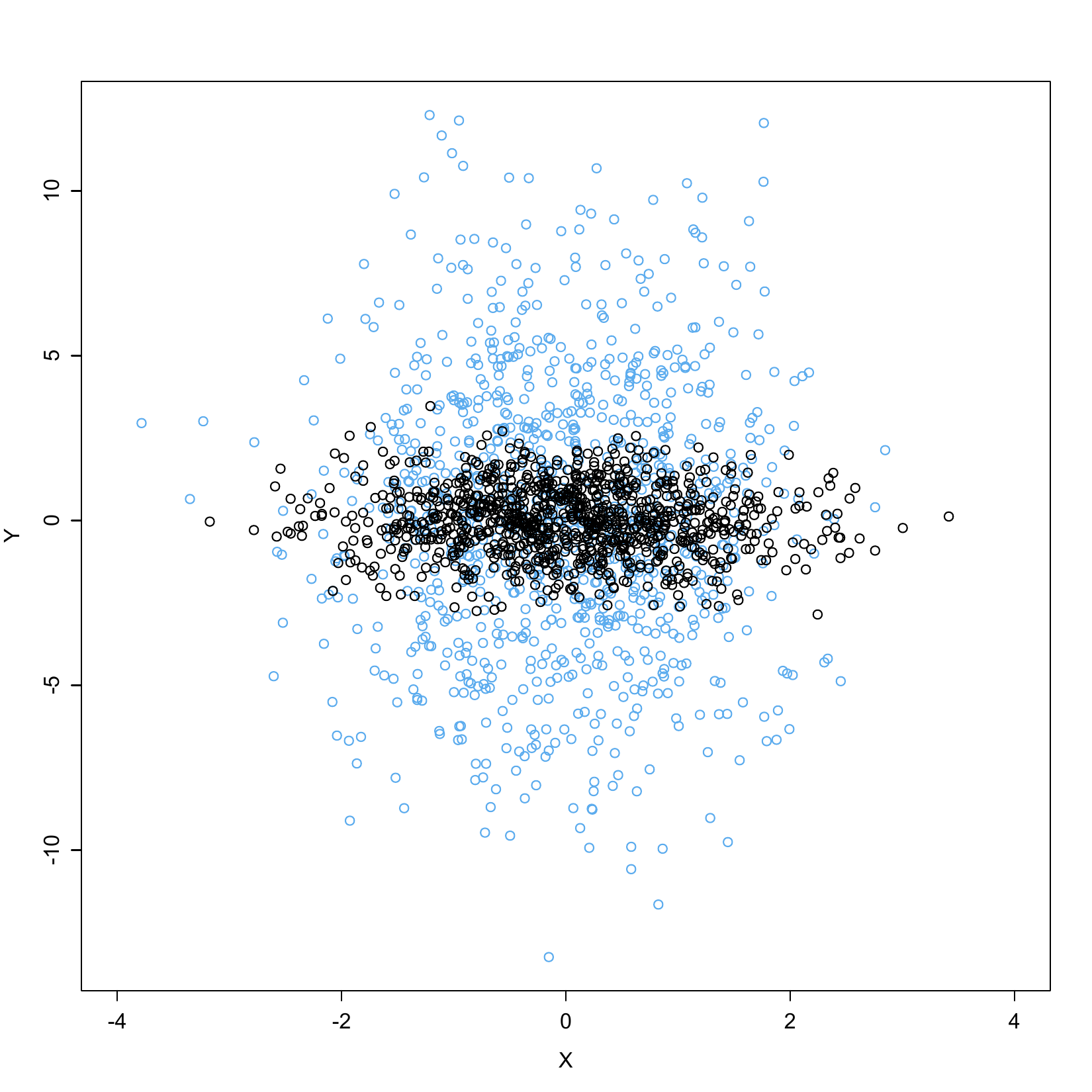}
\caption*{(b)}
\end{minipage}\hfill
\begin{minipage}[c]{.32\textwidth} 
\centering%
\includegraphics[width=.9\textwidth,trim={0 0 0 0},clip]{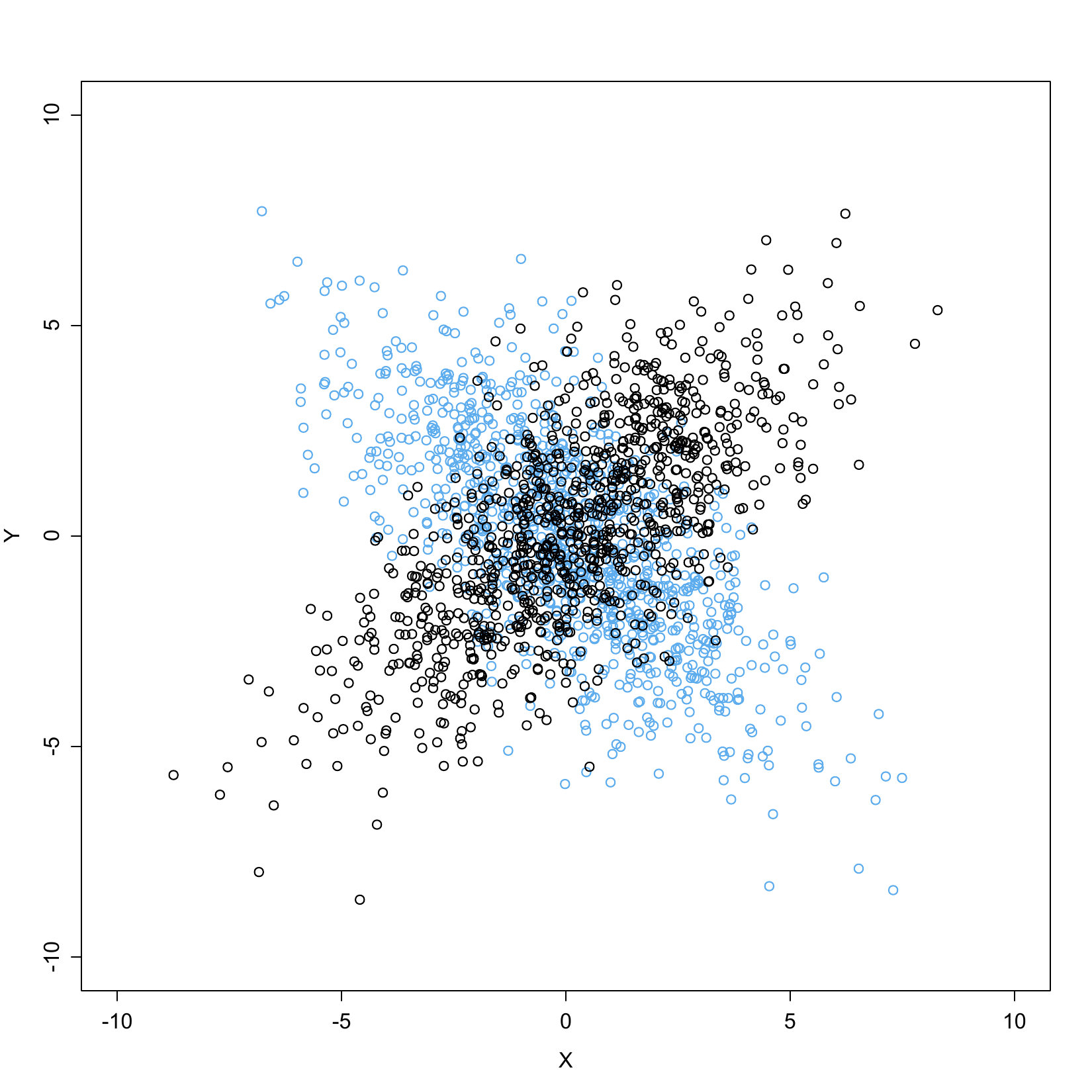}
\caption*{(c)}
\end{minipage}\hfill
\caption{Two samples of 1000 points with different covariance matrices. The differences can be characterised as (a) an expansion difference (b) a sub-matrix expansion difference (c) a sign change. Notice the magnitude or shape of the data cloud changes in panes (a) and (b), but not in pane (c). }%
\label{fig:depth}%
\end{figure}
For example, panel (a) of Figure \ref{fig:depth} shows two bivariate normal samples overlaid. 
The blue sample has an expanded covariance matrix relative to the black sample. 
Notice that both data clouds have the same shape, but the magnitude differs between them. 
It is easily seen that the black points are more central, relative to the shared center of the samples.
Therefore, when we compute the depth values with respect to the combined sample, the black points will have, on average, higher depth values. 
Panel (b) shows again two bivariate normal samples overlaid, but this time the expansion is only in one parameter of the covariance matrix. 
In this case, the shapes of the two data clouds differ, which results in the black points being more central.
This results in the black points having generally higher depth values in the combined sample. 
Panel (c) of Figure \ref{fig:depth} shows again two bivariate normal samples overlaid, except only the sign of the correlation between the two variates differs between the two samples. 
Notice that the data cloud does not change in size or shape, it is simply rotated. 
The combined sample depth values will not change in this case, since one sample is not more central relative to the other.

It is entirely possible that changes in the direction of outlyingness, e.g., a sign change in correlation, are outside of the scope of interesting, or plausible changes in a particular dataset. 
For example, the hypothesis that the spread of the data is increasing or decreasing in a particular direction is considerably different from the hypothesis that the relationship between two or more variates has reversed. 
If there is reason to believe that the only plausible changes in the data process are variability changes, rather than say, general changes, or even general covariance changes, it is beneficial to use the depth-based procedure. 
This is reflected in our simulation study where we compare our method to some change-point methods which make less assumptions about the type of change in the data (see Section \ref{sec:sim}). 
This is intuitive; including more information about the data into the assumptions of the procedure should improve the results of the procedure. 
The downside of course would be missing other types of changes if they are not suspected to be present.

\section{Proposed change-point algorithms}\label{sec::meth}
In this section we describe two multiple change-point algorithms that can be used to detect changes in the variability of multivariate data. 
The first algorithm, takes a local approach, in the sense that the idea is to look at small sections of the data and treat the problem as a single change problem within each small section. 
The second algorithm takes a global approach, such that all of the change-points are simultaneously estimated. 
We will compare the methods in the subsequent sections.

We first restrict ourselves to the `at most one change' setting and review the results of \citep{Chenouri2020DD}. 
\cite{Chenouri2020DD} propose using the following rank CUSUM statistic
$$ Z_{1,\scaled{N}}(m/N) \coloneqq\frac{1}{\sqrt{N}} \sum_{i=1}^{m} \frac{\widehat{R}_{i}-\left(N+1\right) / 2}{\sqrt{\left(N^2-1\right) / 12}},$$
where $\widehat{R}_{i}$ are the ranks described in Section \ref{sec::depth}.

\cite{Chenouri2020DD} show that, when there are no change-points present, $$\sup_{m\in [N]} |Z_{1,\scaled{N}}(m/N)|\cond \sup_{t\in [0,1]}|B(t)|,$$ where $B(t)$ is a standard Brownian Bridge, $[N]$ represents the set $\{1,\dots,N\}$ and $\cond$ refers to convergence in distribution. 
\cite{Chenouri2020DD} also show that under the assumption that their exists a change-point, Assumption \ref{ass:Lipschitz}, Assumption \ref{ass:consDepth} and Assumption \ref{ass:CSthm} given below in Section \ref{sec::theo}, the change-point estimator 
$$\widehat{\theta}=\frac{1}{N}\argmax_{m\in[N]} |Z_{1,\scaled{N}}(m/N)|$$
is weakly consistent and that $$\sup_m |Z_{1,\scaled{N}}(m/N)|\conp \infty,$$
where $\conp$ refers to convergence in probability. 
In the simulation study done in \cite{Chenouri2020DD} this estimator was very robust against skewed or heavy-tailed distributions, especially when compared to the method of \cite{Aue2009}. 
We utilise and extend the aforementioned results of \cite{Chenouri2020DD} to the setting of multiple changes by combining the above CUSUM statistic with either wild binary segmentation or a Kruskal Wallis type statistic. 

\subsection{Wild binary segmentation with a depth rank CUSUM statistic}\label{sec::CUSUM}
Wild binary segmentation, introduced by \cite{Fryzlewicz2014} was originally developed for detecting multiple change-points in the mean of univariate data. 
Seeing as the problem here is essentially to detect changes in the mean of the depth-based ranks, it seems natural to use a similar approach. 
In fact, \cite{Chenouri2019} combined wild binary segmentation with univariate rank statistics with quite favourable results, providing some further motivation for its use with depth-based ranks. 
Let $e,\ s\in [N]$ and $s<e$. 
We define the following rank CUSUM statistic for the set $\{X_{s},\dots, X_{e}\}$ of size $N_{s,e}=e-s+1$
$$ Z_{s,e}(m/N_{s,e}) \coloneqq\frac{1}{\sqrt{N_{s,e}}} \sum_{i=1}^{m} \frac{\widehat{R}_{i,s,e}-\left(N_{s,e}+1\right) / 2}{\sqrt{\left(N_{s,e}^2-1\right) / 12}},$$
where $\widehat{R}_{i,s,e}$ are the linear ranks resulting from ranking the depth values of the observations in the subsample $\{X_{s},\dots, X_{e}\}$, with respect to only the observations in $\{X_{s},\dots, X_{e}\}$. 
More precisely, the depth values are taken with respect to the empirical distribution generated by $\{X_{s},\dots, X_{e}\}$. These ranks range from $1,\dots, N_{s,e}$.

Following the lines of \cite{Fryzlewicz2014} we can now outline our algorithm as follows. First choose $J$ uniformly random intervals and let $$\INT=\{(s_j,e_j)\colon \ j\in [J],\ s_j<e_j,\ s_j,\ e_j\in [N]\}$$ be the set of those intervals. 
After choosing the intervals, the algorithm runs recursively. 
In one run, the algorithm starts with a supplied interval $(s,e)$. 
First, $\INT_{s,e}\subset \INT$ is computed; $\INT_{s,e}$ is the set of intervals  $(s_j,e_j)$ such that $e_j\leq e$ and $s_j\geq s$. 
Then for each interval $(s_j,e_j)\in \INT_{s,e}$, the maximal CUSUM statistic is computed: $$\sup_{s_j\leq m<e_j} |Z_{s_j, e_j}(m/N_{s,e})|.$$
This produces $\binom{e_j-s_j}{2}$ change-point estimates paired with their respective CUSUM statistics. 
The change-point estimate which produces the maximal CUSUM statistic out of all the computed CUSUM statistics is then selected as the candidate change-point
$$(j^*,m^* )_{s,e}=\argmax_{(j,m)\colon (s_j,e_j)\in \INT_{s,e},\ m\in \{s_j,\dots,e_j-1\}} \left|Z_{s_{j},e_{j}}\left(\frac{m-s_{j}+1}{N_{s,e}}\right)\right|.$$ 
If it holds that 
\begin{equation}
    \left|Z_{s_{j^*},e_{j^*}}\left(\frac{m^*-s_{j^*}+1}{N_{s,e}}\right)\right|>T,
    \label{eqn:alg1}
\end{equation}
for some $T$, then the algorithm adds the index to the list of change points. 
Additionally, if \eqref{eqn:alg1} holds then the algorithm calls itself twice, once with the new supplied interval being $(s,m*)$ and once with the new interval being $(m^*+1,e)$. 
If \eqref{eqn:alg1} does not hold then the algorithm stops and returns the set of current change-points. 
Pseudo-code for this algorithm is summarized in Algorithm \ref{alg:wbs}. 
\begin{algorithm}
\setstretch{1.3}
\caption{Rank-Based Wild Binary Segmentation}
\label{alg:wbs}
\begin{algorithmic}
\Procedure{WBS\_Rank}{$e,s,T, \INT$}
\If{$e-s<1$}
\State STOP
\Else
\State $\INT_{s,e}\coloneqq$ intervals $(s_j,\ e_j)\in\INT$ such that $(s_j,\ e_j)\subset (s,e)$
\State $(j^*,m^* )\coloneqq \argmax_{\mathcal{B}} \left|Z_{s_{j},e_{j}}\left(\frac{m-s_{j}+1}{N_{s,e}}\right)\right|$,\\
\qquad\qquad\qquad\qquad\qquad\qquad\qquad with $\mathcal{B}\coloneqq\{ (j,m) \colon (s_j,e_j)\in \INT_{s,e},\ m\in \{s_j,\dots,e_j-1\}\}$
    \If{$\left|Z_{s_{j^*},e_{j^*}}\left(\frac{k^*-s_{j^*}+1}{N_{s,e}}\right)\right|>T$}
    \State    Append $m^*$ to the list of change-points $\widehat{\mathbf{k}}$
    \State    WBS\_Rank($s,m^*,T,\INT$)
     \State   WBS\_Rank($m^*+1,e,T,\INT$)
    \Else 
    \State STOP
\EndIf
\EndIf
\State \Return $\mathbf{\widehat{k}}$
\EndProcedure
\end{algorithmic}
\end{algorithm}
\subsection{KW-PELT: A Kruskal-Wallis change-point algorithm}
As mentioned above, Algorithm \ref{alg:wbs} takes a local approach to the problem, utilising only sections of the data to estimate each change-point. 
Additionally, there is the issue of subjectivity with regard to choosing the number of intervals.
As an alternative, we can instead maximize a single objective function based on the whole data set. 
Recall from Section \ref{sec::model} that a mean change in the depth values is implied by a change in variability.
The Kruskal-Wallis test statistic is used to check for mean differences among multiple groups of univariate data; this value is large for univariate mean differences.
It is very natural to then base the objective function on the Kruskal-Wallis test statistic. To this end, we propose using the following as an estimator of the change-points
\begin{equation}
    \mathbf{\widehat{k}}\coloneqq \argmax_{k_0=0<k_1<\dots<k_\ell<N=k_{\ell+1}} \frac{12}{N(N+1)} \sum_{i=1}^{\ell+1} (k_{i}-k_{i-1}) \rhatbar_{i}^{2}-3(N+1) -\beta_{\scaled{N}} (\ell+1),
    \label{eqn:KW}
\end{equation}
where $\beta_{\scaled{N}}$ is a parameter for which higher values correspond to higher penalization on the number of estimated change-points and $\rhatbar_{i}$ is the mean of the sample depth ranks in group $i$, viz.
$$\rhatbar_i=\frac{1}{k_i-k_{i-1}}\sum_{i=k_{i-1}+1}^{k_i}\ranki\ .$$ 
One can recall that $\widehat{R}_i$ are defined in \eqref{eqn:d_ranks}, or, also in relation to the wild binary segmentation algorithm $\ranki=\widehat{R}_{i,1,N}$. 
Note that the penalization is necessary; without it the solution to this maximization problem is simply choosing every point as a change-point. 
It is apparent that \eqref{eqn:KW} is a difficult maximization problem in the sense that the number of possible solutions is $2^N$.  
However, we can circumvent this issue by applying the pruned exact linear time algorithm \citep{Killick2012}. 
Indeed, rewrite the objective function,  in \eqref{eqn:KW}, by which we denote $\mathbf{G}(N)$, as 
$$\mathbf{G}(N)\coloneqq\sum_{i=1}^{\ell+1}-c(k_{i-1}+1:k_i)-\beta_{\scaled{N}}\ell$$
where 
\begin{equation}
c(s+1:e)=-\frac{12(e-s)}{N(N+1)}\left[\frac{1}{e-s}\sum_{i=s+1}^{e}\ranki-\frac{N+1}{2}\right]^2.
    \label{eqn:cost}
\end{equation}
Letting $k_0=0$ and $k_{\ell+1}=e$, we can write the maximization problem in \eqref{eqn:KW} as
\begin{align*}
   \max_{k_0<k_1<\dots<k_\ell<k_{\ell+1}} \mathbf{G}(e)&=\min_{k_0<k_1<\dots<k_\ell<k_{\ell+1}} \frac{12}{N(N+1)} \sum_{i=1}^{\ell+1} -(k_i-k_{i-1}) \left(\ranki-\frac{N+1}{2}\right)^{2} +\beta_{\scaled{N}}( \ell+1)\\
    &=\min _{s}\left\{\min_{k_0<k_1<\dots<k_\ell<s}\sum_{i=1}^\ell (c(k_{i-1}+1:k_i)+\beta_{\scaled{N}})+c(s+1:e)+\beta_{\scaled{N}} \right\}\\
    &=\min _{s}\left\{-\mathbf{G}(s)+c(s+1:e)+\beta_{\scaled{N}}\right\}.
\end{align*}
It is straightforward to show that \eqref{eqn:KW} satisfies the assumption in \citep{Killick2012} required for PELT to be applicable and so we omit the proof. (One can simply expand the expression out, and make a geometric argument about the number of roots.) 
\begin{algorithm}
\setstretch{1.2}
\caption{KW-PELT}
\label{alg:pelt}
\begin{algorithmic}
\Procedure{KW\_PELT}{$\mathbf{R},\beta$}
\State $N\coloneqq length(\mathbf{R})$
\State $\mathbf{\widehat{k}}(0)=NULL$
\State $\mathcal{N}_0\coloneqq\{0\}$
\State $\mathbf{G}(0)=-\beta$
\For{$k \in 1,\dots,N$}
\State $\mathbf{G}(k)=\min_{s\in \mathcal{N}_k}\{\mathbf{G}(s)+c((s+1):k)+\beta\}$
\State $k^1=\argmin_{s\in \mathcal{N}_k}\{\mathbf{G}(k)+c((s+1):k)+\beta\}$
\State $\mathbf{\widehat{k}}(k)=(\mathbf{\widehat{k}}(k^1),k^1)$
\State $\mathcal{N}_{k+1}\coloneqq\{k\}\cup \{s\in\mathcal{N}_{k}\colon \mathbf{G}(s)+c((s+1):k)\leq \mathbf{G}(k) \}$
\EndFor
\State \Return $\mathbf{\widehat{k}}(N)\backslash\{0\}$
\EndProcedure
\end{algorithmic}
\end{algorithm}
Algorithm \ref{alg:pelt} outlines this procedure, which we call KW-PELT, in pseudo-code. It is simply the PELT algorithm in \citep{Killick2012} applied to the objective function $\mathbf{G}$ in \eqref{eqn:KW}.


We end this section with a remark about computation time. 
Computationally, the limiting factor for both procedures will (in general) be the computation time for the sample depths. 
Consequentially, we expect Algorithm \ref{alg:pelt} to be faster, due to the fact that sample depth functions need only be calculated once rather than once for every sampled interval. 
If $f(N;d)$ is the time it takes to compute the sample depths, then Algorithm \ref{alg:wbs} would take $O(JN\log N+Jf(N;d))$ time as opposed to $O(N\log N+f(N;d))$ time for Algorithm \ref{alg:pelt}. 
It is worth noting that Algorithm \ref{alg:pelt} was implemented partially in \CC 
\ whereas Algorithm \ref{alg:wbs} was implemented completely in \texttt{R} (except for possibly the depth computations, for which existing packages were used) so the empirical times in simulation are not directly comparable. 
This being said, both algorithms ran within minutes on a desktop computer when applied to the data set analyzed in Section \ref{sec::da}.
\section{Consistency of the algorithms}\label{sec::theo}
In this section we provide consistency results for both algorithms under some mild assumptions. 
For $j\in [\ell+1]$, let $Y_j\sim F_j$ and let $$H_j(x)=\Pr(\mathcal{D}(Y_j;F_*)\leq x).$$
The following assumptions are used in the consistency theorems that follow.
\begin{ass}
$H_j(x)$ are Lipschitz continuous with constant $C$, that is
$$|H_j(x)-H_j(y)|\leq C|x-y|,$$
for $x,y\in \re^{d}.$
\label{ass:Lipschitz}
\end{ass}
\begin{ass} It holds that
$$\E{}{\sup_{x\in \rd{d}}|\mathcal{D}(x;F_{*,\scaled{N}})-\mathcal{D}(x;F_*)|}=O(N^{-1/2}).$$
\label{ass:consDepth}
\end{ass}
\begin{ass}
The number of change-points $\ell$ is fixed and the change-points are well spread for all $N$, meaning there is a constant $\Delta$ such that the change-points are separated by at least $ \Delta N$ .
\label{ass:numcp}
\end{ass}
\begin{ass}
For all $j\in [\ell]$, it holds that
$$\int_{\re}\left( \vartheta_j H_j(x)+\vartheta_{j+1} H_{j+1}(x)\right)dH_{j}(x)\neq \frac{\vartheta_{j+1}+\vartheta_{j}}{2}.$$
\label{ass:CSthm}
\end{ass}
\begin{ass}
For the threshold $T$, it holds that $T=o(\sqrt{N}).$
\label{ass:thresh}
\end{ass}
\begin{ass}
Let $p_{ij}=\Pr(\mathcal{D}(Y_i;F_*)>\mathcal{D}(Y_j;F_*))$. Then for any $j\in [\ell+1]$ it holds that 
$$\sum_{i=1}^{\ell+1} \vartheta_i p_{j,i}\neq \frac{1}{2}.$$
\label{ass:KWthm}
\end{ass}
Assumptions \ref{ass:Lipschitz} and \ref{ass:consDepth} are satisfied by most depth functions under absolutely continuous $F$, including those defined in Section \ref{sec::depth} \citep[see][and the references therein]{Liu1999}. 
Assumption \ref{ass:numcp} says that the number of change-points is fixed, and their closeness is not arbitrarily small in $N$. 
Assumptions \ref{ass:CSthm} and \ref{ass:KWthm} are concerned with the type of changes that can be detected, and are related to the discussion in Section \ref{sec::model}. 
Assumption \ref{ass:CSthm} says that the random variables $\mathcal{D}(Y_{j};F_*)$ and $\mathcal{D}(Y_{j+1};F_*)$ are ordered in a probabilistic sense, i.e., $$\Pr(\mathcal{D}(Y_{j};F_*)<\mathcal{D}(Y_{j+1};F_*))\neq 1/2.$$ 
In order for consistency, we must have that the distribution of depth values in one segment is distinguishable from a neighboring segment. 
By distinguishable, we mean that a change in variability implies a probabilistic ordering on the random depth values generated by the observations.
Recall that Section \ref{sec::model} examined this idea. 
Additionally, Assumption \ref{ass:CSthm} implies Assumption \ref{ass:KWthm}; under Assumption \ref{ass:KWthm} for each change-point, we just need two of the segments of $\iid$ observations, not necessarily neighboring, to be distinguishable. 
Assumption \ref{ass:CSthm} says that the distributions of depth values of all neighboring pairs (of segments) must be distinguishable.

Suppose there is a single change-point and that $Y_1\sim \mathcal{N}_d(0,I)$ and that $Y_2\eqd \sqrt{a}Y_1$ with $a>1$. 
Clearly, we have that $\E{F_*}{X}=\mathbf{0}$ and $\Sigma_*=(\vartheta_1+a(1-\vartheta_1))I=\sigma^2_*I$. 
It follows that 
\begin{equation*}
    \norm{Y_1-\E{F_*}{X}}_{\Sigma^{-1}_*}\sim \frac{1}{\sigma^2_*}\chi^2_d \qquad \text{and}\qquad  \norm{Y_2-\E{F_*}{X}}_{\Sigma^{-1}_*}\sim \frac{a}{\sigma^2_*}\chi^2_d,
\end{equation*}
Now, for any $x\in \re$ we have that
$$F_{\chi^2_d}\left(\frac{1}{\sigma^2_*} x\right)<F_{\chi^2_d}\left(\frac{a}{\sigma^2_*} x\right),$$
where $F_{\chi^2_d}$ represents the cumulative distribution function of a $\chi^2_d$ random variable. It follows immediately that $p_{1,2}=1-p_{2,1}\neq \frac{1}{2}$.
Additionally,
$$\E{\sigma^2_*\chi^2_d}{F_{\chi^2_d}\left(\frac{a}{\sigma^2_*} X\right)} > \E{\sigma^2_*\chi^2_d}{F_{\chi^2_d}\left(\frac{1}{\sigma^2_*} X\right)} =\frac{1}{2};$$
both assumptions are satisfied.
Clearly, neither Assumption \ref{ass:CSthm} or Assumption \ref{ass:KWthm} hold if $a=1$.
\begin{theorem} 
Let $C>0,\ 1/2<\phi<1$ be constants independent of $N$.
Let the estimated change-points $\hat{k}_1<\hat{k}_2<\dots<\hat{k}_{\hat{\ell}}$ be as in Algorithm \ref{alg:wbs}. 
Provided Assumptions \ref{ass:Lipschitz}-\ref{ass:thresh} hold, 
and the number of intervals $J_{\scaled{N}}\rightarrow\infty$ as $N\rightarrow\infty$
we have that
$$\Pr\left(\left\{\hat{\ell}=\ell\right\}\cap \left\{\max_{i\in [\ell]}|\hat{k}_i-k_i|\leq C N^\phi \right\}\right)\rightarrow 1\ \text{as } N\rightarrow\infty.$$ 
\label{thm:WBS}
\end{theorem}
Theorem \ref{thm:WBS} states that for large $N$, it is highly probable that the change-point estimates produced by Algorithm \ref{alg:wbs} will be close to the location of the true change-points and that the number of these estimates is equal to the true number of change-points. 
The next theorem gives a similar result for Algorithm \ref{alg:pelt} under a wide range of penalty terms ($O(1)<\beta_{\scaled{N}}<O(N)$). 
\begin{theorem} For $\beta_{\scaled{N}}$ as in \eqref{eqn:KW}, assume that $O(1)<\beta_{\scaled{N}}<O(N)$ and let $\delta>0,\ 1/2<\phi<1$. 
Provided Assumptions \ref{ass:Lipschitz}-\ref{ass:numcp} hold and Assumption \ref{ass:KWthm} holds, for $\widehat{\mathbf{k}}$ and $\hat{\ell}$ as in Algorithm \ref{alg:pelt}, we have that
$$\Pr\left(\left\{\hat{\ell}=\ell\right\}\cap\left\{ \max_{i\in [\ell]}|\hat{k}_i-k_i|\leq \delta N^\phi \right\}\right)\rightarrow 1\ \text{as } N\rightarrow\infty.$$ 
\label{thm:PELT}
\end{theorem}
\section{Simulation study}\label{sec:sim}
In this section we use a simulation study to compare our methodology to existing procedures as well as to investigate different choices of the algorithm parameters $\beta_{\scaled{N}}$ and $T$. 
Specifically, we compare our method to the BSOP and WBSIP algorithms developed by \cite{Wang2021}, as well as to the methods in the \texttt{ecp R} package \citep{ecp_package}. 
It should be noted that the aim of \citep{Wang2021} was to detect change-points in the high dimensional setting, and not necessarily in low or moderate dimensions. 
The WBSIP algorithm performed much better than the BSOP algorithm, and so we only present the results of the WBSIP algorithm. 
Results from the BSOP algorithm can be found in Appendix \ref{app::add_sim}. 
As mentioned above, we also compare to the nonparametric change-point methods in the \texttt{ecp\ R} package \citep{ecp_package}. 
These methods are designed to detect general types of change-points; time points where there was a change in distribution.  
We compared our methods with the \texttt{e.dvisive}, \texttt{e.cp3o\_delta} and \texttt{e.kcp3o} methods. 
The other methods in the package were either too slow to run with our simulation set-up or performed much worse than the chosen methods. 
The best of these three methods in our simulation set-up was by far the \texttt{e.dvisive} method, and so we defer the results of the other two methods to Appendix \ref{app::add_sim}.

The threshold parameter for the WBSIP algorithm was chosen to be 561, which was based on visually assessing the error $\hat{\ell}-\ell$ so that the median was zero in most of the scenarios. 
We also tried choosing the threshold in order to minimize the mean squared error $\hat{\ell}-\ell$ (which would not be known in practice) and the results were similar. 
For Algorithm \ref{alg:wbs} and the WBSIP algorithm we used $100\floor{\log{N}}$ intervals.

The simulation study is limited to evenly spaced change points, from distributions with independent marginals. 
Note that the transformation invariance properties possessed by the depth functions imply the results from similarity transformations of the data would be the same. 
This transformation invariance implies that the study also covers some cases where the marginal distributions of the data are not independent. 
We set the mean of all distributions to be 0.

The simulation study consisted of several scenarios. 
The first scenario is a set of expansions and contractions controlled by the parameter $\sigma^2.$ 
We let $\Sigma_j=\sigma^2_{j}I_d$ for each $F_j$,\ $j\in[\ell+1]. $ 
We set $$\sigma_1^2=1,\ \sigma_2^2=2.5,\ \sigma_3^2=4,\ \sigma_4^2=2.25,\ \sigma_5^2=5,\ \sigma_6^2=1,$$ e.g., for 2 change-points, $\sigma^2$ would vary as follows $1\shortrightarrow 2.5\shortrightarrow 4$. 
The second scenario is another set of expansions and contractions, of which the results were so similar that we defer discussion and results from the second scenario to Appendix \ref{app::add_sim}. 

We simulated data from three different distribution types,  normal, Cauchy and skewed normal with skewness parameter $\gamma=0.1/d$. 
We ran the simulation for values of $d=$2, 3, 5 and 10 under 2, 3 and 5 change-points. 
To see results on zero change-points and one change-point \citep[see][]{Chenouri2020DD}. 
We used sample sizes of $N=$1000, $N=$2500, and $N=$5000, running each scenario 100 times. 
We tested the four depth functions introduced in Section \ref{sec::depth}.

Lastly, we ran several simulations designed to assess the performance of our methods under sparsity and/or high dimensions. 
In these scenarios, $d$ was at most 500 and/or the expansions/contractions were only applied to a submatrix of the covariance matrix. 


\texttt{R} codes to replicate this simulation study, as well as implementations of Algorithm \ref{alg:wbs}, Algorithm \ref{alg:pelt}, the WBSIP algorithm and the BSOP algorithm are available \citep{Ramsay2019}. 

\subsection{Choosing the algorithm parameters}\label{sec::params}
\begin{figure}[t]
\begin{minipage}[c]{.31\textwidth} 
\centering%
\includegraphics[width=\textwidth,trim={0 00 0 0},clip]{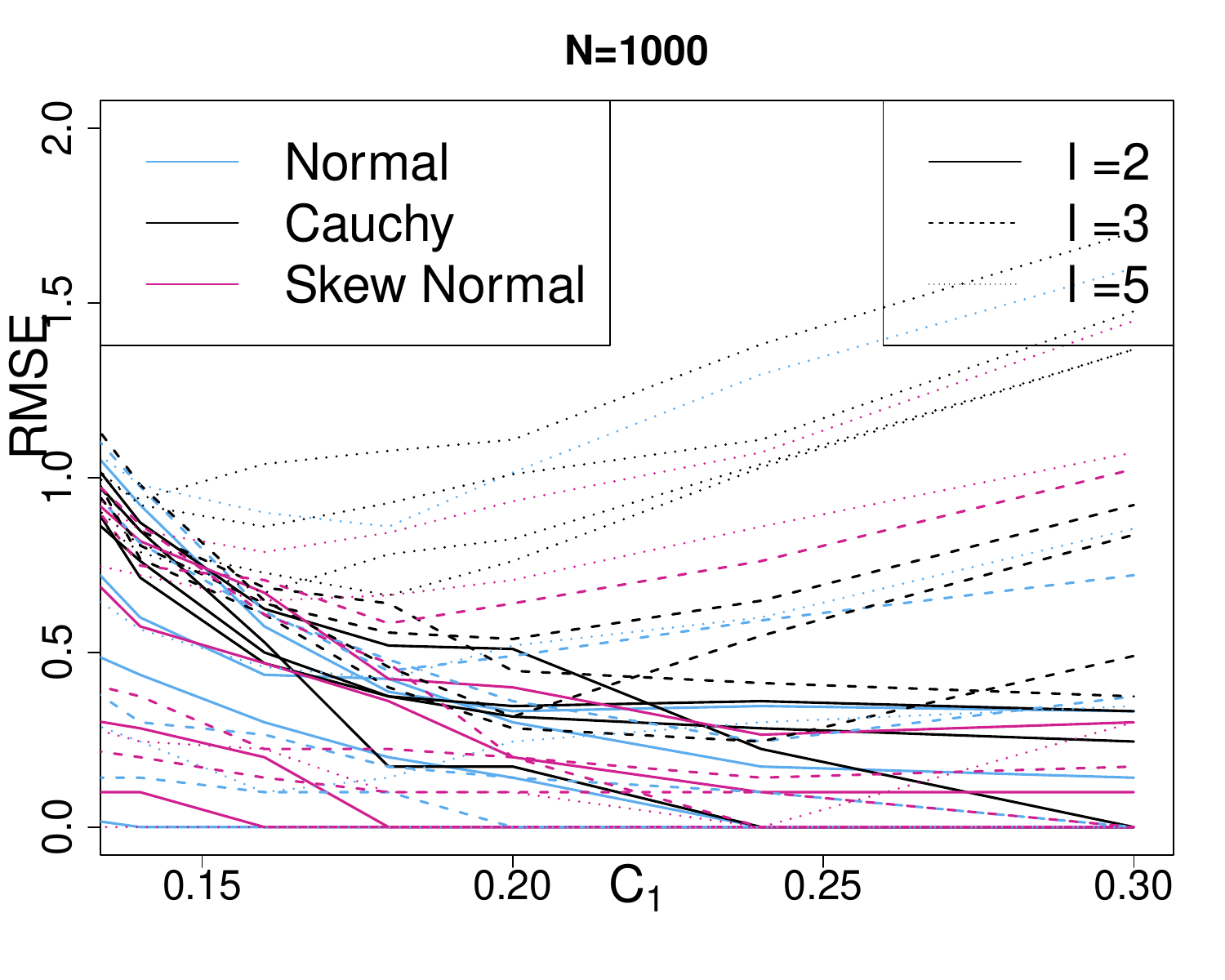}
\end{minipage}
\begin{minipage}[c]{.31\textwidth} 
\centering%
\includegraphics[width=\textwidth,trim={0 0 0 0},clip]{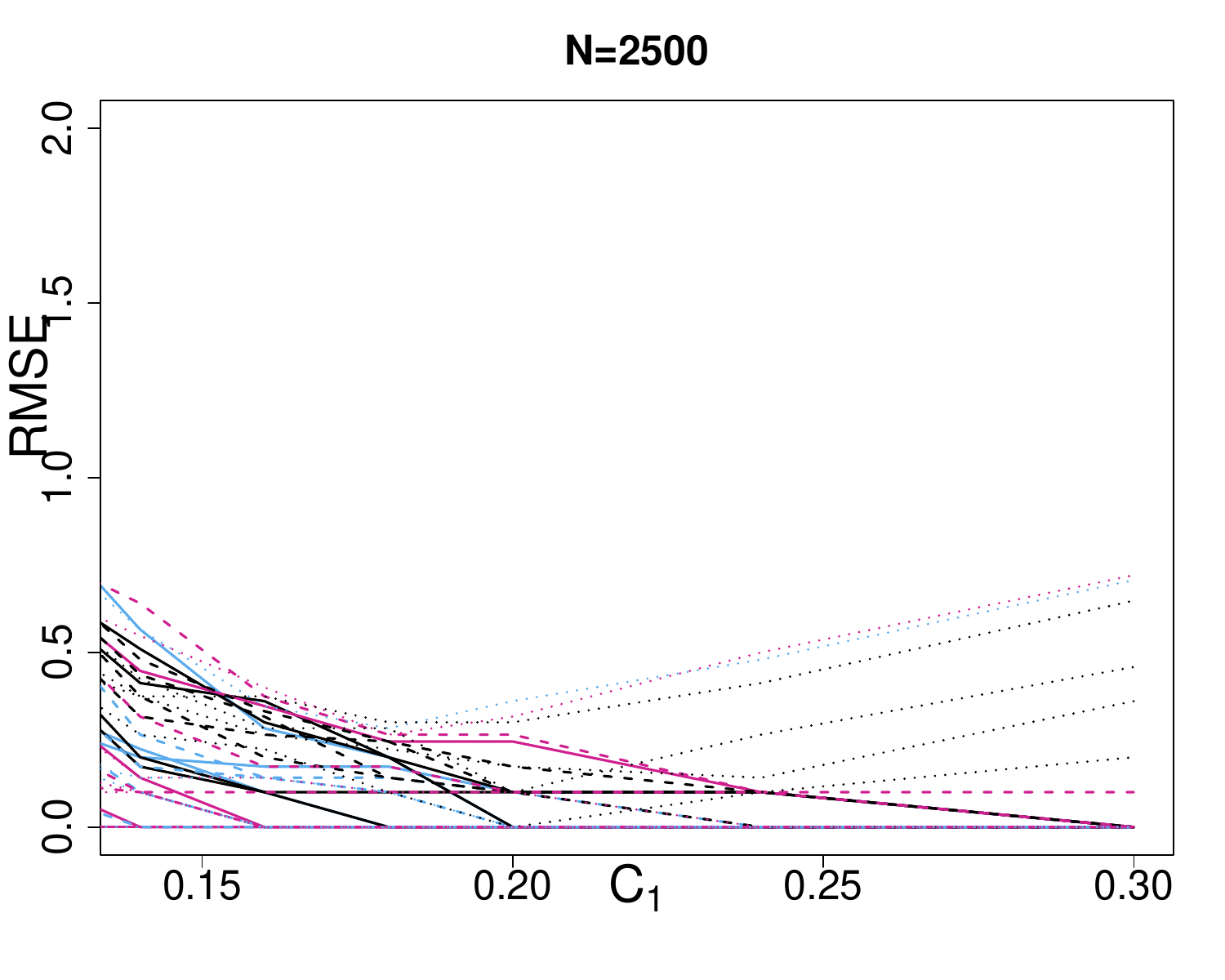}
\end{minipage}
\begin{minipage}[c]{.31\textwidth} 
\centering%
\includegraphics[width=\textwidth,trim={0 00 0 0},clip]{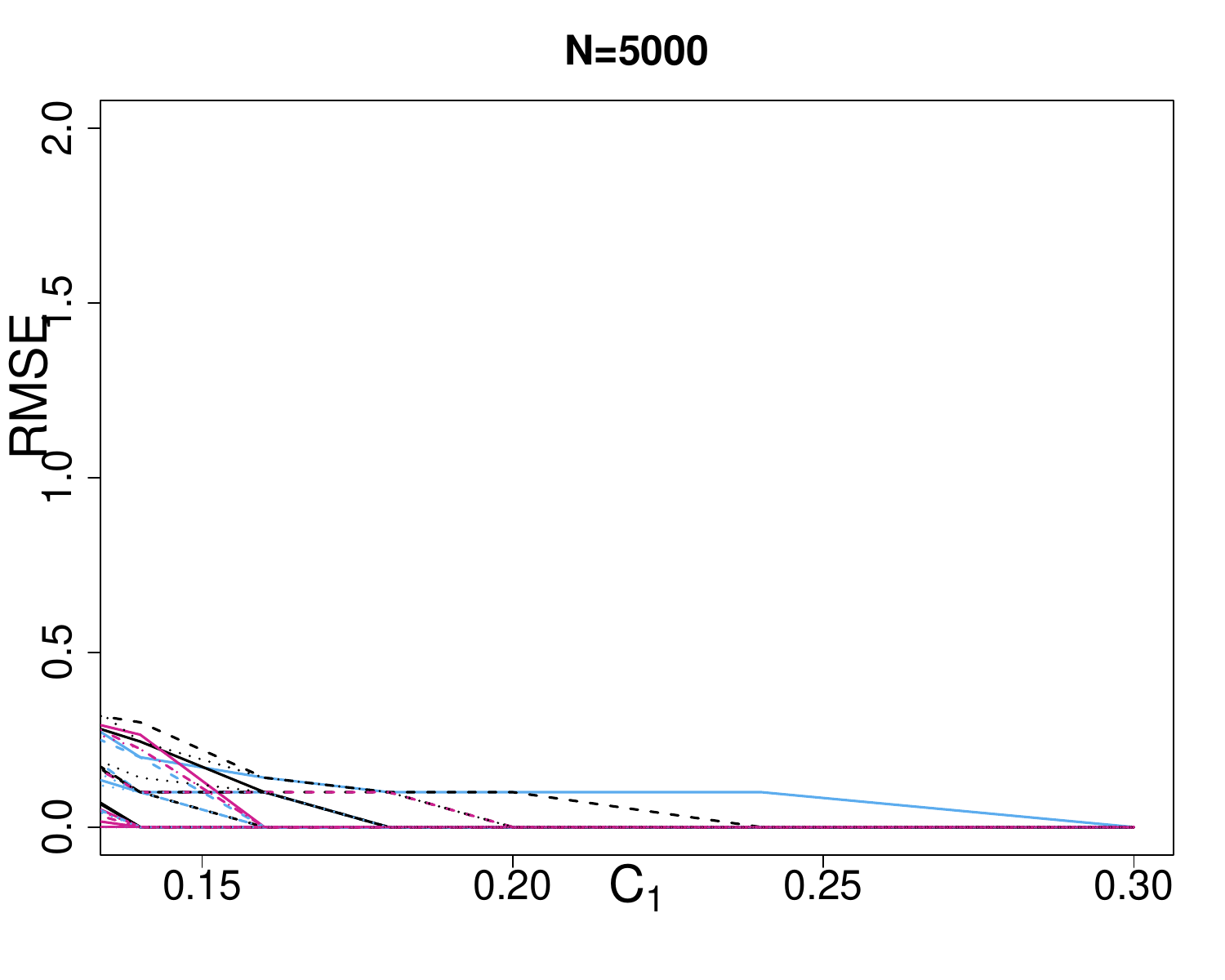}
 \end{minipage}
\caption{Empirical root mean squared error of $\hat{\ell}$ for different values of $C_1$ under spatial depth for all the simulation parameter combinations, under Algorithm \ref{alg:pelt}.}
\label{fig:C_ch}%
\end{figure}

In order to have consistency of the estimates produced by Algorithm \ref{alg:wbs}, the threshold must satisfy $T~=~o(\sqrt{N})$. 
One option is to choose a fixed threshold $T^*$, which will produce a set of change-point estimates and their corresponding CUSUM statistics. 
The final set of change-points could then be chosen by testing each change-point for significance using a Bonferroni correction or Benjamini-Hochberg correction \citep{BH1995} along with the quantiles of $\sup |B(t)|$. 
This would imply a threshold $T\geq T^*$. 
However, it might be that smaller sampled intervals are not large enough for the asymptotic approximation to work well. 
Additionally, one has to choose the significance level, and the threshold $T^*$. 
As a result of these considerations, we suggest a data driven thresholding approach, based on the generalized Schwartz Information Criteria, as done by \cite{Fryzlewicz2014}.

Algorithm \ref{alg:wbs} produces a nested set of models, indexed by the threshold parameter. 
Lowering the threshold can only add new change-points to the model; all previously estimated change-points remain. 
In other words, as the threshold decreases, new change-points are added to the model one at a time. 
It is then easier to re-index the models by the number of estimated change-points $\hat{\ell}$. 
The threshold problem can then be reformulated as a model selection problem.

Suppose we have a univariate sample $Z_1,\dots , Z_{\scaled{N}}$ and the goal is to estimate a change-point in the mean. For this problem, \cite{Fryzlewicz2014} chooses the `best' model by minimizing the following criteria:

\begin{equation}
    \mathcal{G}(\hat{\ell})=\frac{N}{2}\log(\hat{\varsigma}_{\hat{\ell}}^2)+\hat{\ell}\log^\alpha N,
    \label{eqn:sic}
\end{equation}
with $\hat{\varsigma}_{\hat{\ell}}^2$ equal to the average within group squared deviation (a group is an estimated period of constant mean) and $\hat{\ell}$ is still the estimated number of change-points. 
Let $\mu_i,$ for $i \in [N]$, be the within group mean for the group that contains univariate observation $Z_i$. 
Then we can write
$$\hat{\varsigma}_{\hat{\ell}}^2=\frac{1}{N}\sum_{i=1}^N(Z_i-\mu_i)^2.$$ 
Here, $\alpha$ is a parameter such that the larger $\alpha$, the larger the penalty against choosing a model with many change-points.

The only difference for the multivariate, variability problem is that $\hat{\varsigma}$ must be modified. 
We recall from Section \ref{sec::model} that a variability change is equivalent to a change in the mean ranks. 
We can treat the sample ranks produced by the depth functions as a univariate sample, and minimize the within group deviation, amongst the ranks:
$$\hat{\varsigma}_{\hat{\ell}}^2=\frac{1}{N}\sum_{i=1}^N(\ranki-\rhatbar_i)^2.$$
We remark that the use of ranks ensures $\mathcal{G}(\hat{\ell})$ is still robust. 
For a sketch of the proof of consistency, in the sense of Theorem \ref{thm:WBS}, of Algorithm \ref{alg:wbs} paired with this method of thresholding with $\alpha=1$ see Appendix \ref{app::proofs}. 

In order to make a practical recommendation for the parameter $\alpha$, we rely on the simulation study. 
Figure \ref{fig:alpha} shows the empirical root mean squared error of $\hat{\ell}$ under Algorithm \ref{alg:wbs} for a range of $\alpha$ values. Each curve is for a different combination of parameters in the first simulation scenario. 
The depth function used was Mahalanobis depth (for computational ease). 
Figure \ref{fig:alpha} shows that choosing $\alpha$ in the range $(0.75,1.25)$ works well. 
Similar plots under the second simulation scenario can be found in Appendix \ref{app::add_sim}. 

\begin{figure}[t]
\begin{minipage}[c]{.32\textwidth} 
\centering%
\includegraphics[width=\textwidth,trim={0 0 0 0},clip]{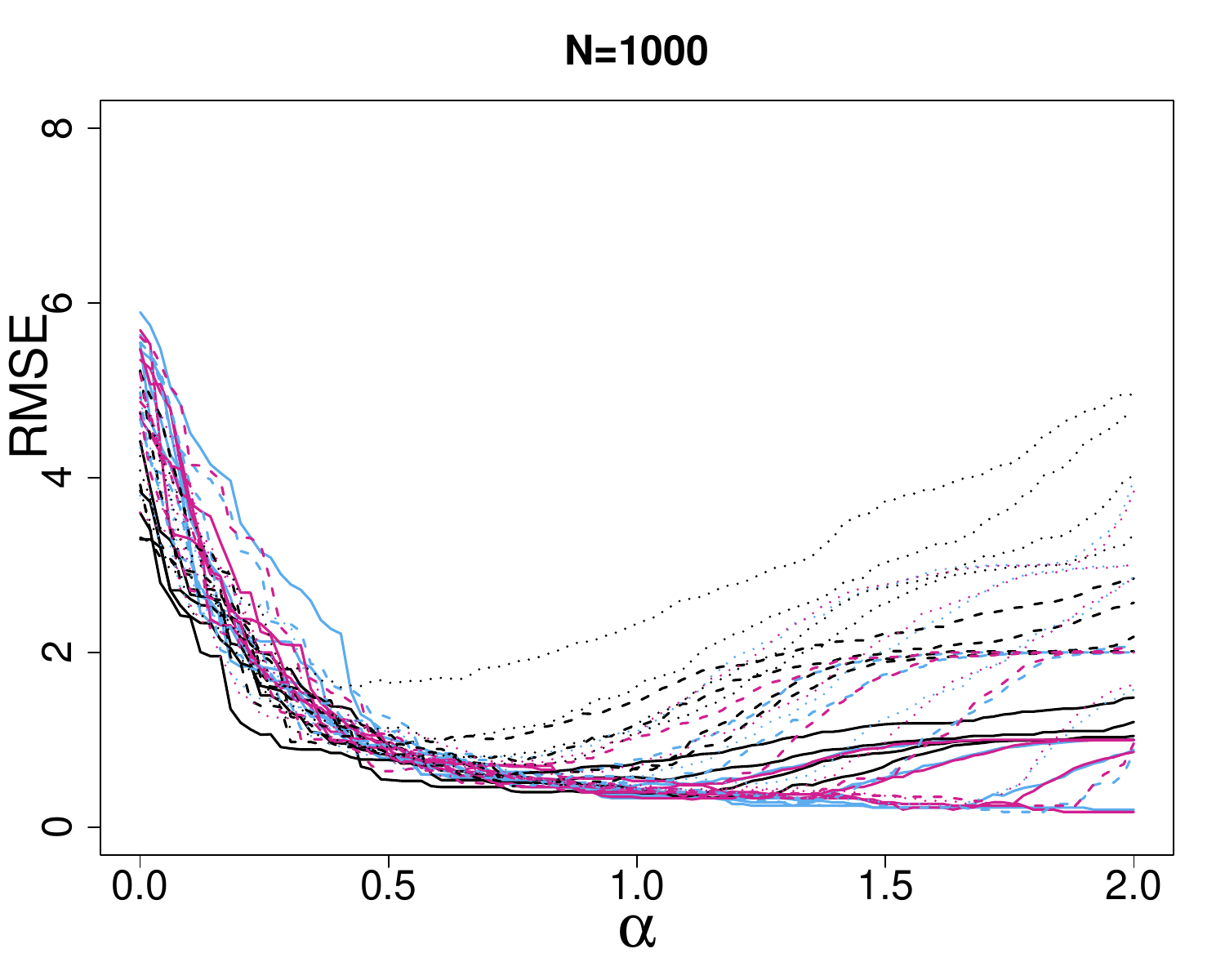}
\end{minipage}
\begin{minipage}[c]{.32\textwidth} 
\centering%
\includegraphics[width=\textwidth,trim={0 0 0 0},clip]{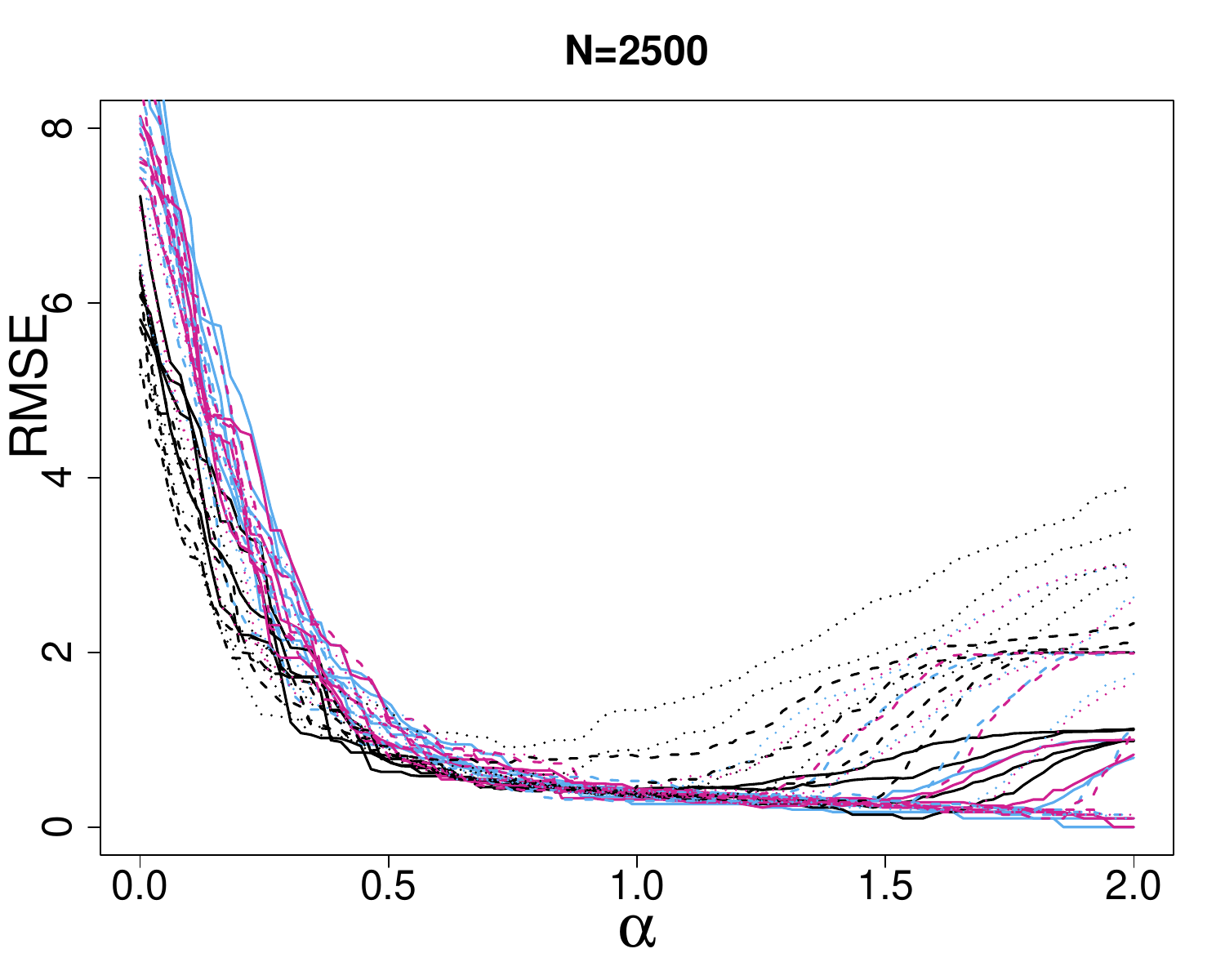}
\end{minipage}
\begin{minipage}[c]{.32\textwidth} 
\centering%
\includegraphics[width=\textwidth,trim={0 0 0 0},clip]{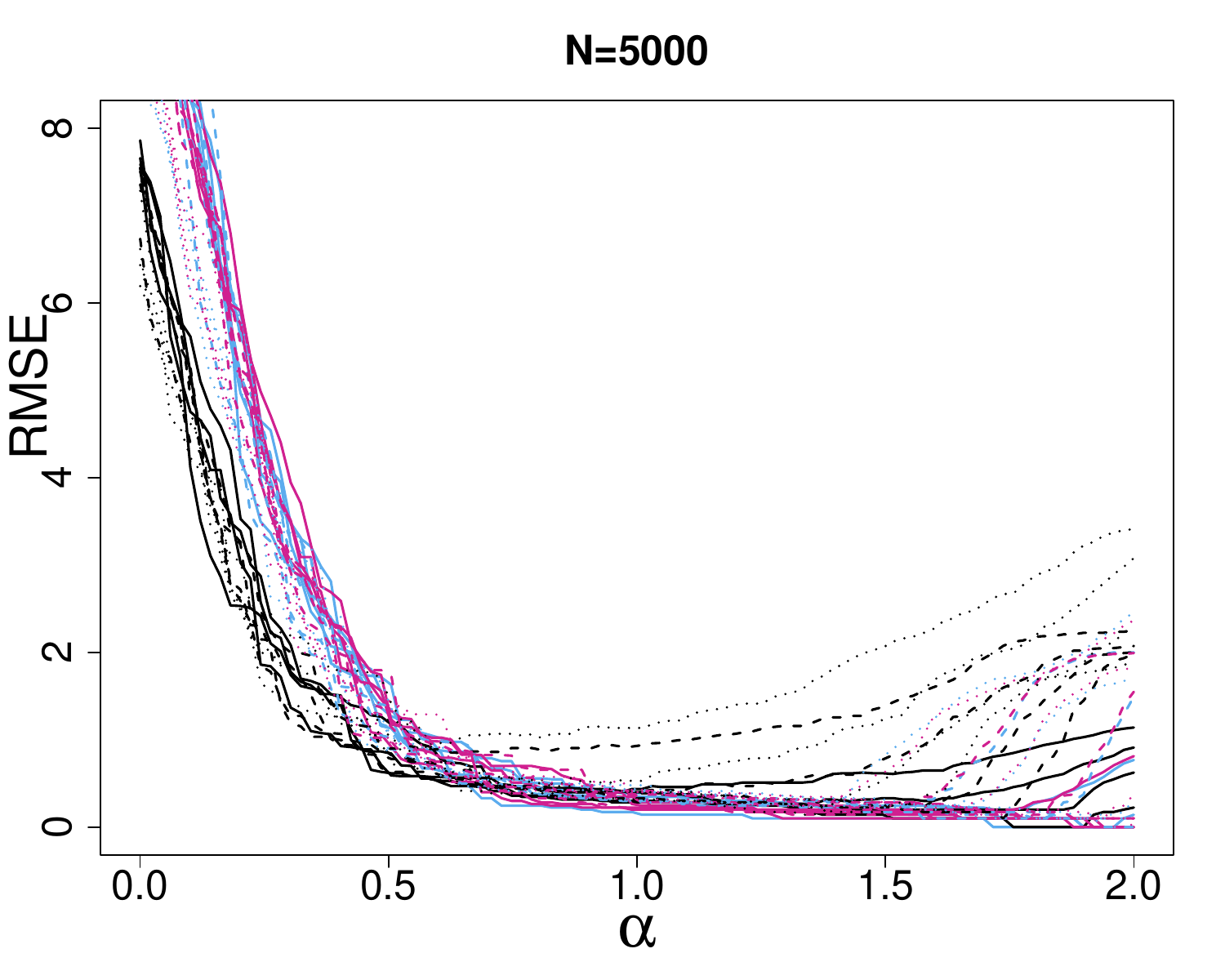}
\end{minipage}
%
\caption{Each curve shows the empirical root mean squared error of $\hat{\ell}$ for different values of $\alpha$ under Algorithm \ref{alg:wbs} paired with Mahalanobis depth, for a given simulation parameter combination, following the legend of Figure \ref{fig:C_ch}. Mahalanobis depth is used rather than spatial depth because of computational efficiency.}
\label{fig:alpha}%
\end{figure}
For consistency of Algorithm \ref{alg:pelt} to hold, the penalty term should satisfy $O(1)<\beta_{\scaled{N}}<O(N)$; this gives a wide range of choices for the penalty term. 
In practice what penalty term should be used? 
The results of the simulation study suggested using a penalty term of the form $\beta_{\scaled{N}}=C_1\sqrt{N}+C_2$. 
Figure \ref{fig:C_ch} plots the empirical root mean squared errors of $\hat{\ell}$ under spatial depth for different values of $C_1$ and $N$, with $C_2$ fixed at 3.74. 
Each curve represents one combination of the parameters in the first simulation scenario described above. 
The same graphs for other depth functions or other simulation scenarios can be seen in Appendix \ref{app::add_sim}. 
Notice that the curves are not shifting laterally as $N$ increases, meaning that an increase in $N$ is sufficiently captured by the $\sqrt{N}$ term in the penalty. 
Additionally, Figure \ref{fig:C_ch} also shows a flattening of the RMSE curves with increased $N$, which is expected from the consistency theorem. 
Based on low root mean squared error in simulation, we recommend to fix $C_2=3.74$ and run Algorithm \ref{alg:pelt} for a grid of penalties defined by $C_1\in (0.15,0.25)$. 
One can then choose the set of change-points according to a model selection criteria or by visual inspection. 
In the simulation study, we fix $C_1$ at 0.18. 

It should also be noted that a non-linear penalty could be applied, as discussed in \cite{Killick2012}. 
Some non-linear penalties were tested in the simulation study, such as $\log{\ell}$, but the results were not as good as when using a linear penalty. 
Our investigation into non-linear penalties was fairly limited, as such, more investigation into non-linear penalties could be done in the future. 

\subsection{Analysing and comparing the algorithm performance}
\begin{figure}[t]
\begin{minipage}[c]{.49\textwidth} 
\centering
\includegraphics[width=\textwidth,trim={0 00 0 0},clip]{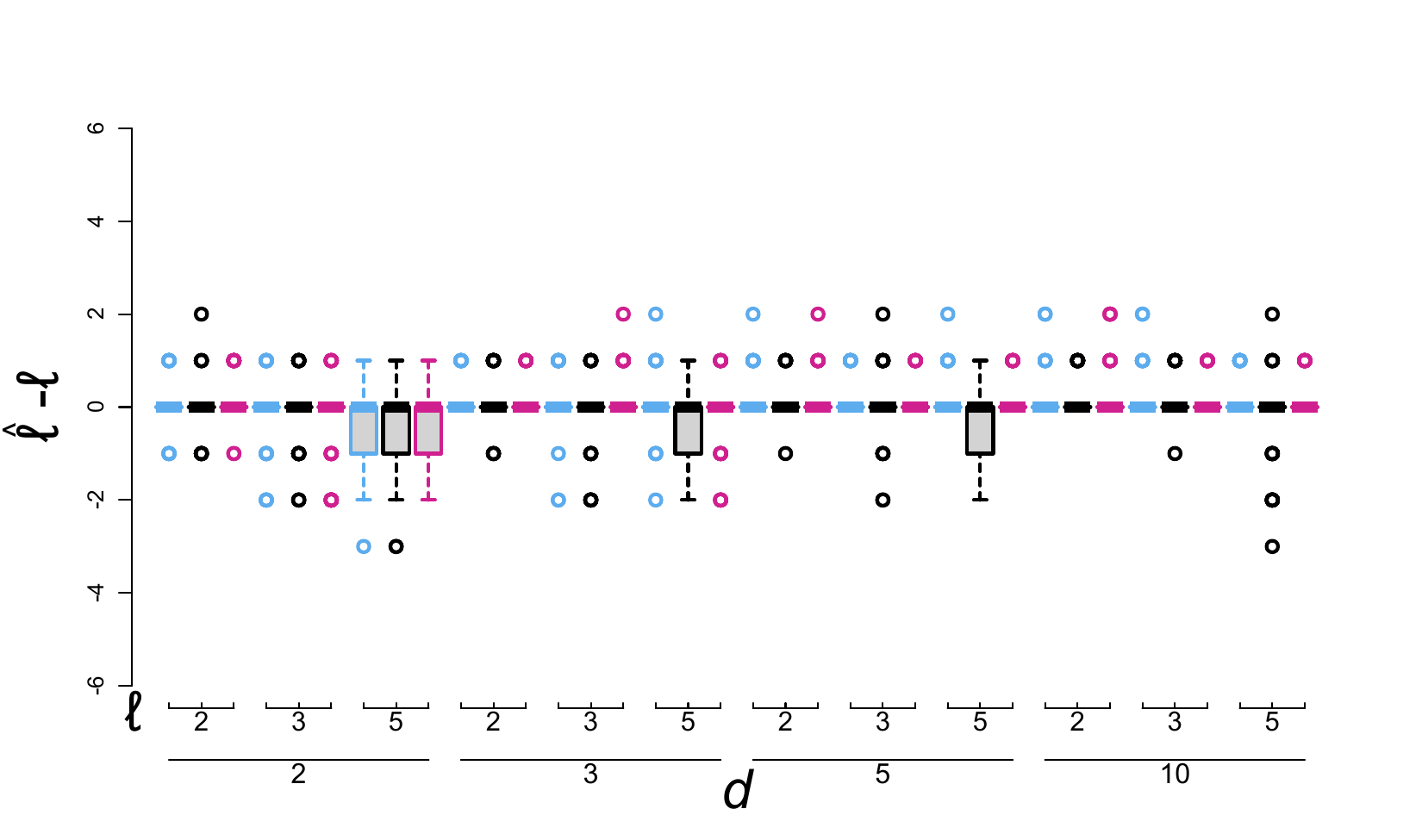}
\caption*{(a) WBS algorithm}
\end{minipage}\hfill
\begin{minipage}[c]{.49\textwidth} 
\centering%
\includegraphics[width=\textwidth,trim={0 0 0 0},clip]{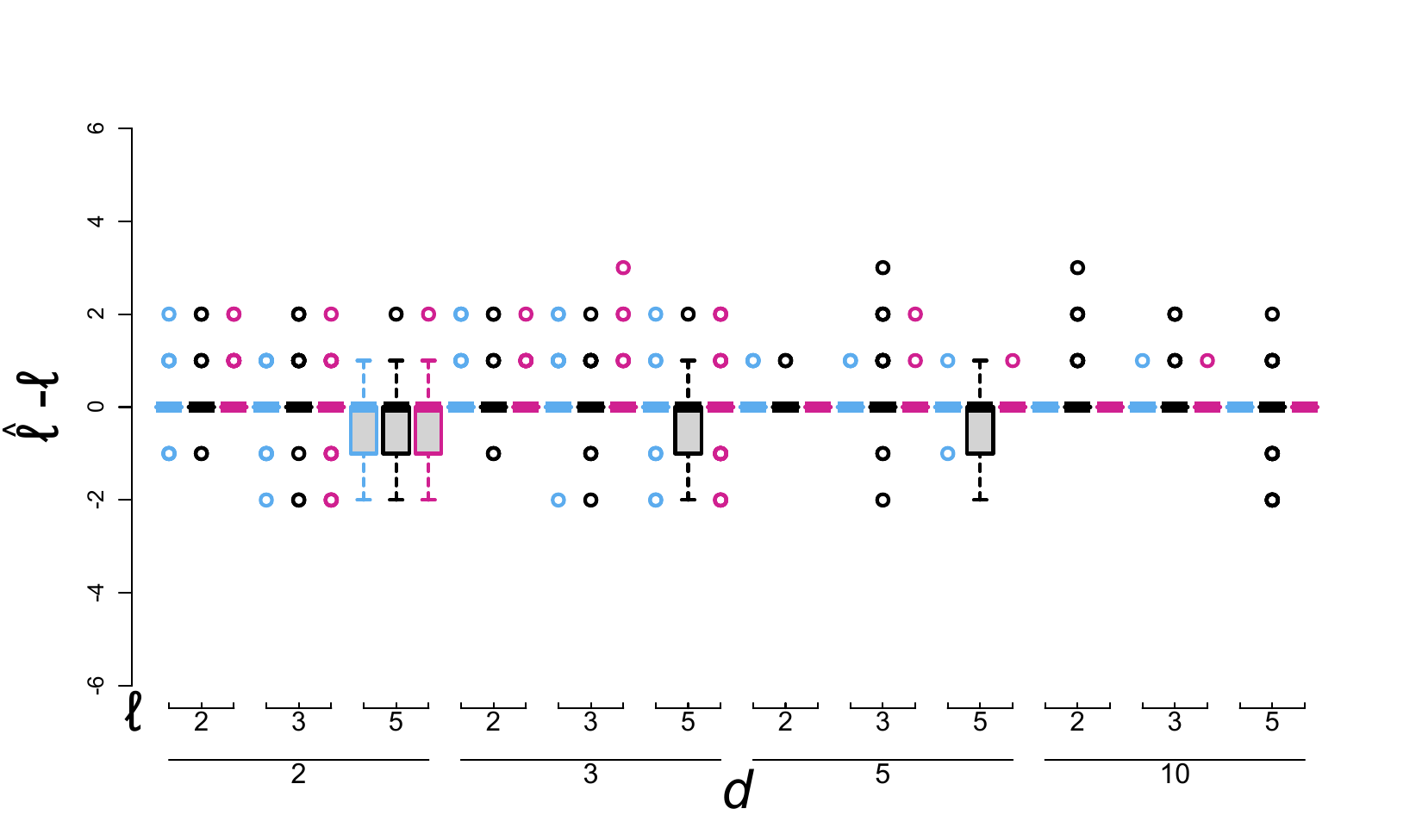}
\caption*{(b) KW-PELT algorithm}
\end{minipage}\hfill\newline
\begin{minipage}[c]{.49\textwidth} 
\centering%
\includegraphics[width=\textwidth,trim={0 00 0 0},clip]{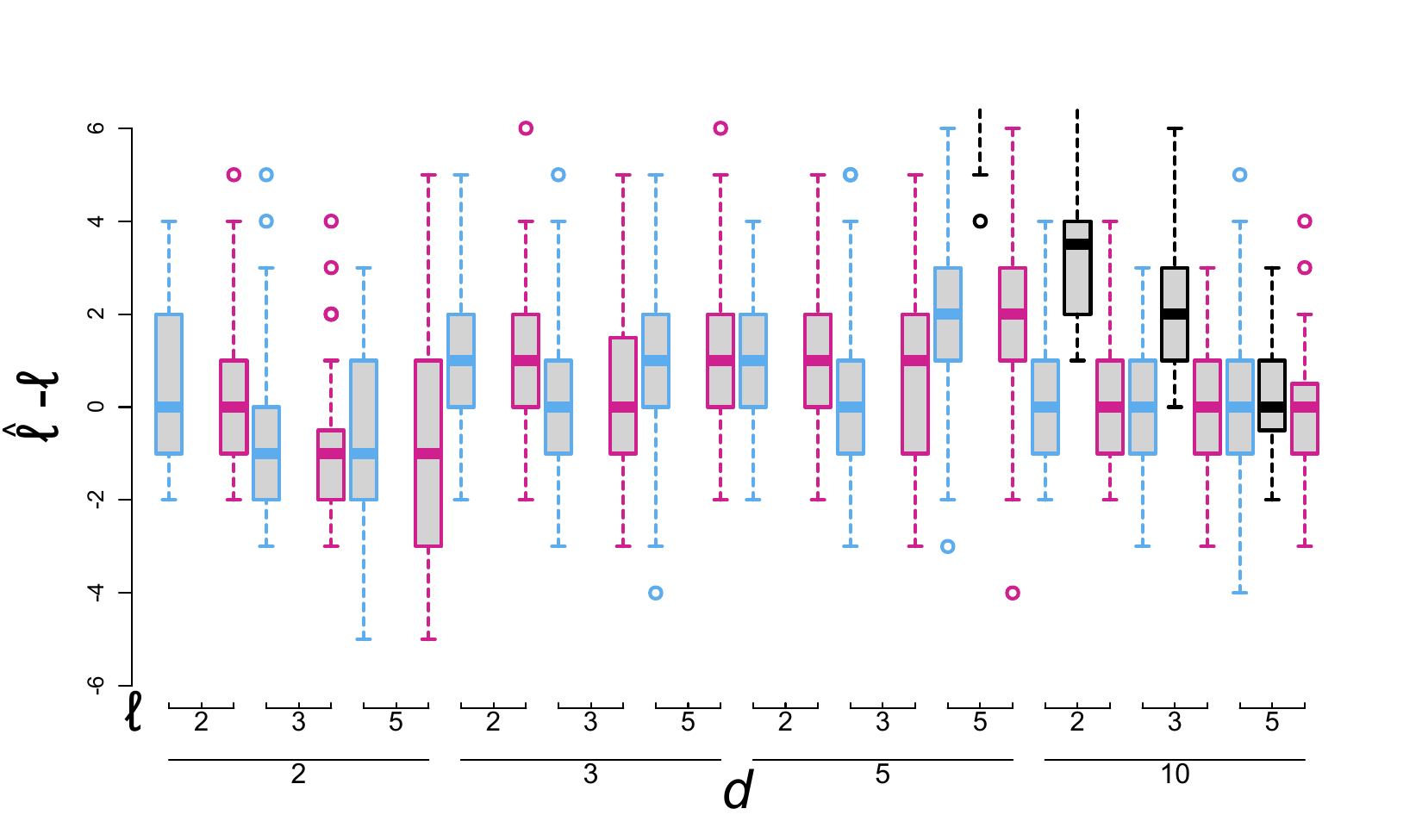}
\caption*{(c) WBSIP algorithm}
\end{minipage}\hfill
\begin{minipage}[c]{.49\textwidth} 
\centering%
\includegraphics[width=\textwidth,trim={0 00 0 0},clip]{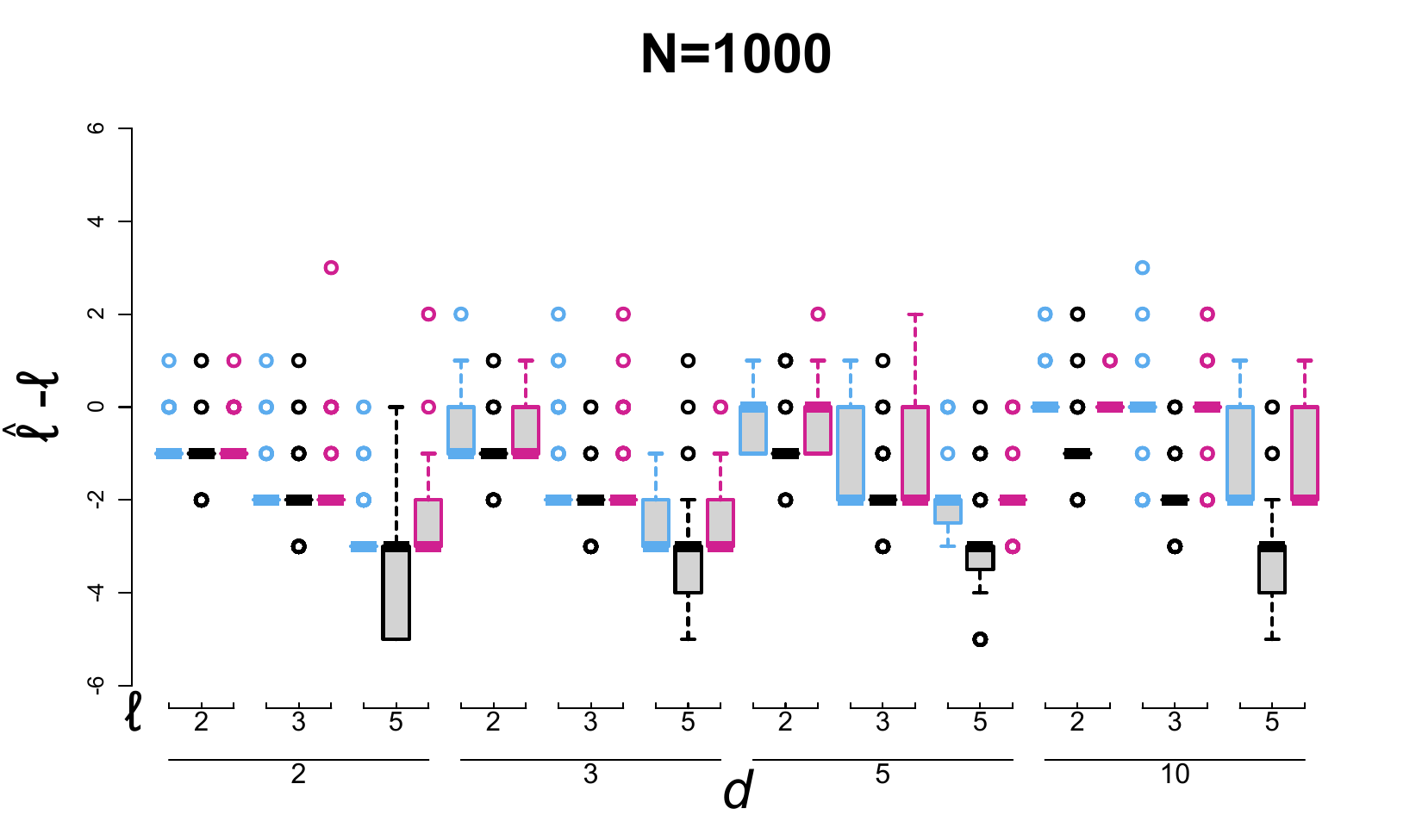}
\caption*{(c) e-divisive algorithm}
\end{minipage}\hfill
\caption{Boxplots of $\hat{\ell}-\ell$ for the different algorithms with $\alpha=0.9$, $C_1=0.18$ and $C_2=3.74$. 
Each boxplot represents the values of $\hat{\ell}-\ell$ for a particular simulation parameter combination. 
Here, the color of the boxplot represents the distribution. 
The top underlying number represents the number of true change-points and the bottom number represents the dimension. The colors follow the legend of Figure \ref{fig:C_ch}.}%
\label{fig:sim1}%
\end{figure}
\begin{figure}[t]
\begin{minipage}[c]{.49\textwidth} 
\centering%
\includegraphics[width=.9\textwidth,trim={0 00 0 0},clip]{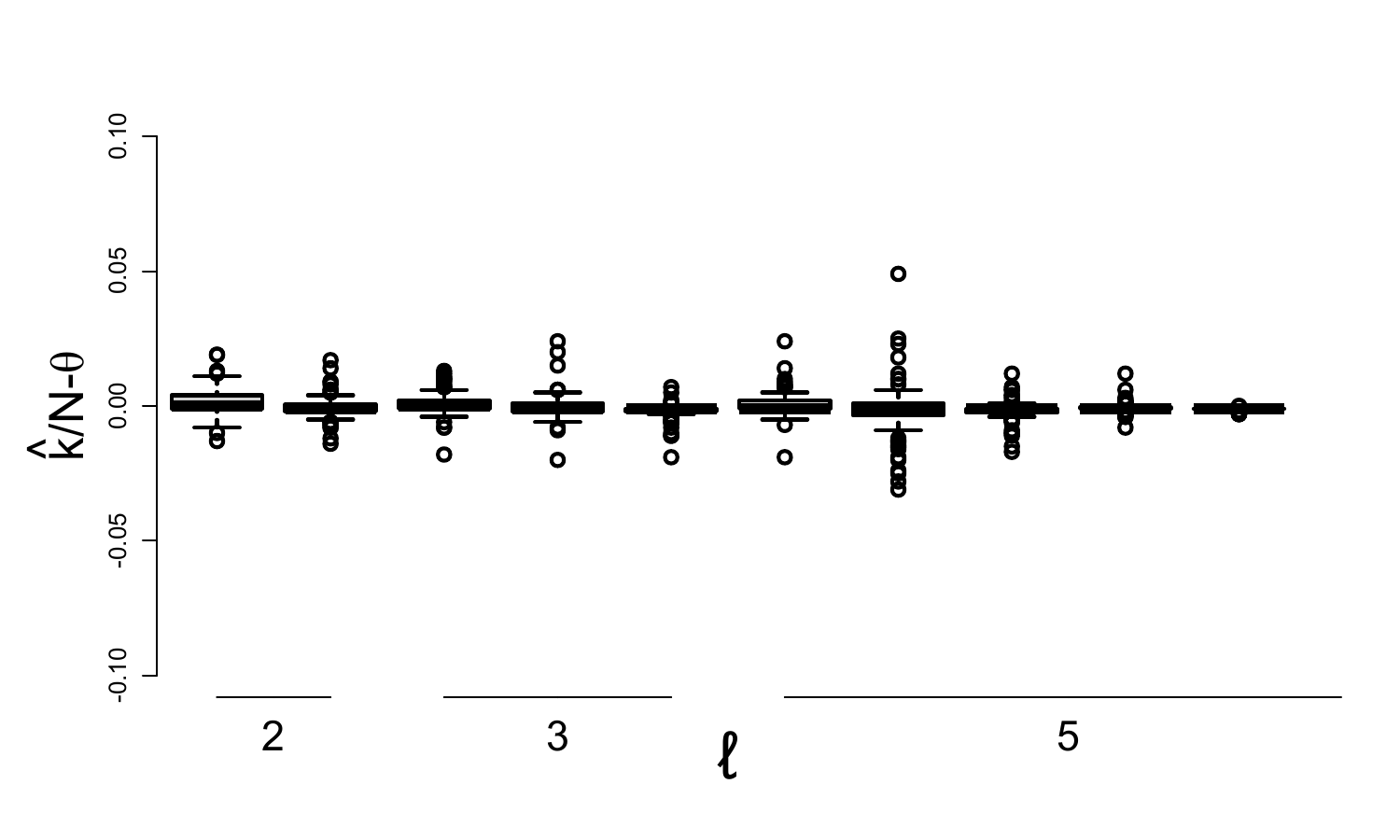}
\caption*{(a) WBS algorithm}
\end{minipage}\hfill
\begin{minipage}[c]{.49\textwidth} 
\centering%
\includegraphics[width=.9\textwidth,trim={0 0 0 0},clip]{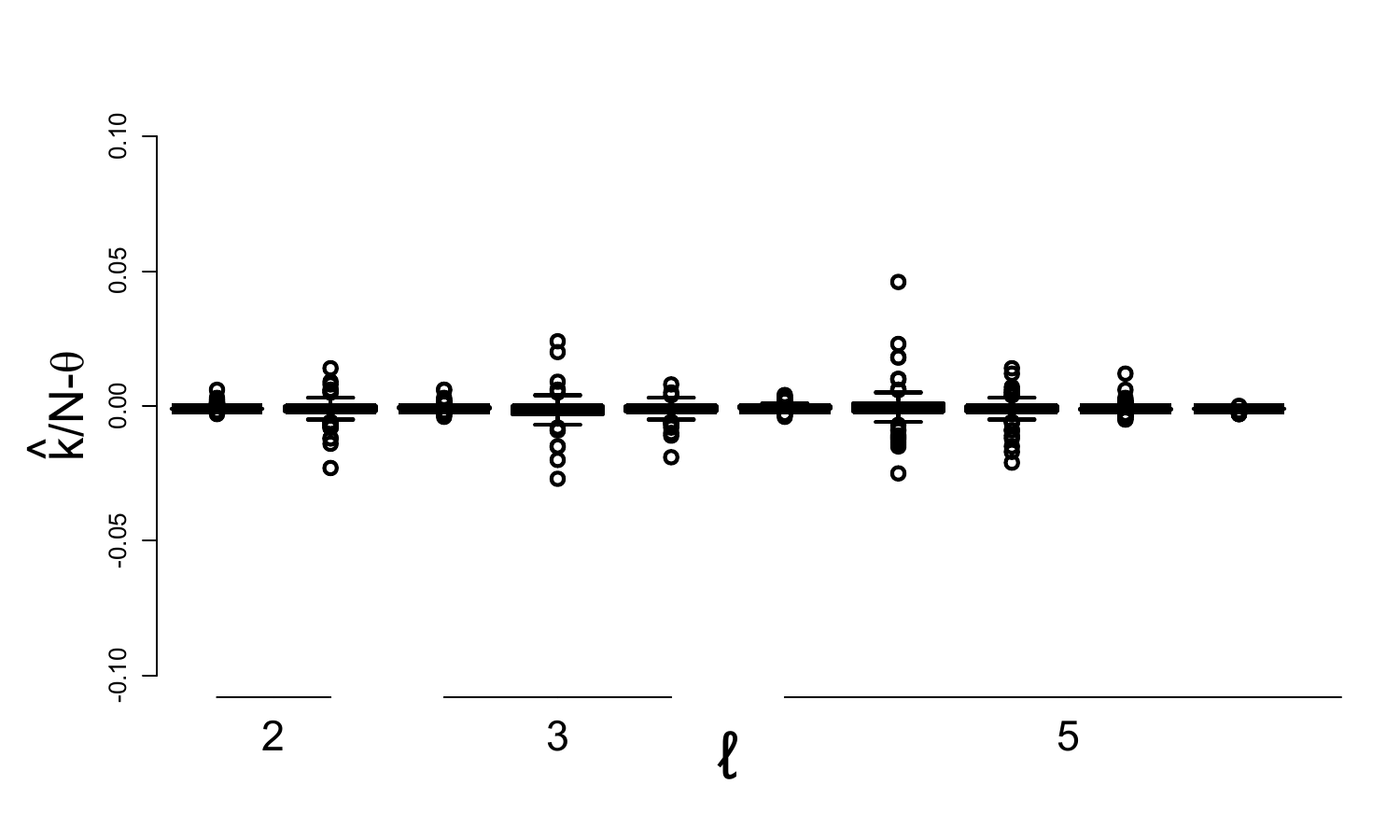}
\caption*{(a) KW-PELT algorithm}
\end{minipage}\hfill\newline
\begin{minipage}[c]{.49\textwidth} 
\centering%
\includegraphics[width=.9\textwidth,trim={0 00 0 0},clip]{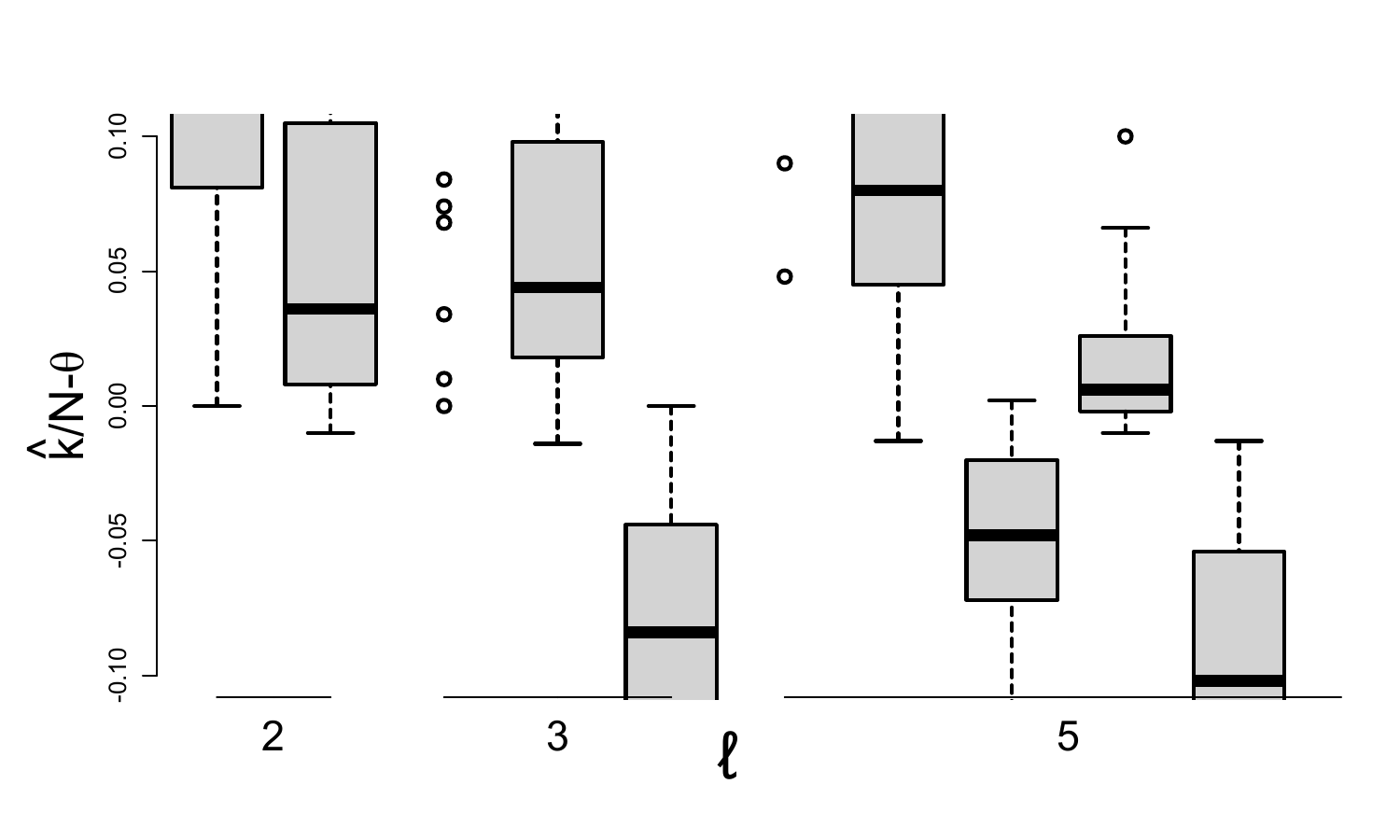}
\caption*{(a) WBSIP algorithm}
\end{minipage}\hfill
\begin{minipage}[c]{.49\textwidth} 
\centering%
\includegraphics[width=.9\textwidth,trim={0 00 0 0},clip]{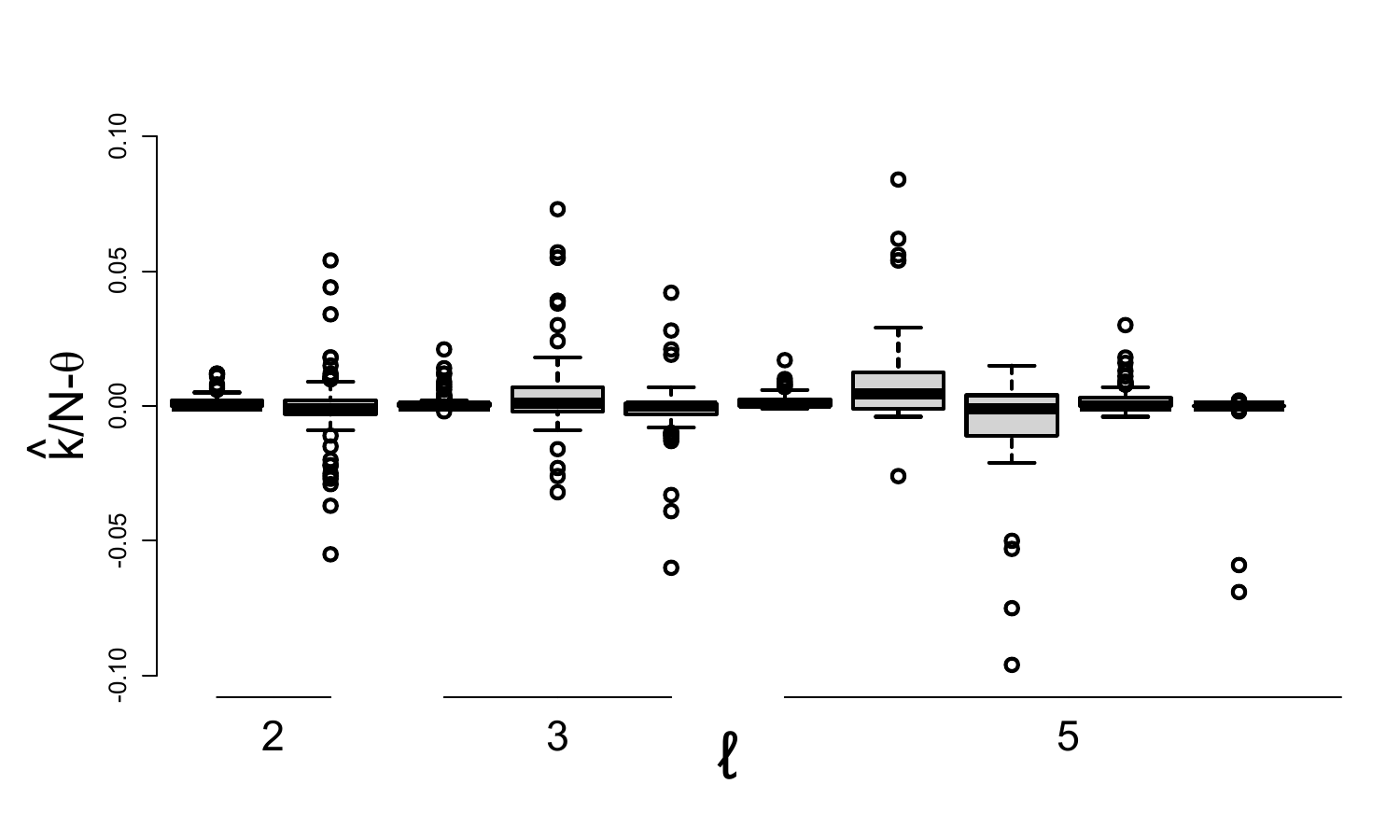}
\caption*{(a) e-divisive algorithm}
\end{minipage}\hfill
\caption{Boxplots of $\widehat{k}/N-\theta$ for the different algorithms with $\alpha=0.9$, $C_1=0.18$ and $C_2=3.74$. 
The distribution of the data was made up of independent normal marginals with $d=10$. 
Each boxplot represents the ability to estimate a particular change-point for a fixed number of true change-points. 
For example, the first two boxplots are the empirical distributions of $\widehat{k}_1/N-\theta_1$ and $\widehat{k}_2/N-\theta_2$ when there were two true change-points. 
The underlying numbers represent the number of true change-points in that simulation parameter combination. 
Each boxplot represents the ability to estimate a single change-point in a run.}%
\label{fig:sim2}%
\end{figure}

We only present the results of the simulation when the algorithms were paired with spatial depth, but one can view the results of the algorithm performance under the other depth functions in Appendix \ref{app::add_sim}. 
Spatial depth was the best performing depth function when taking into account computational speed and estimation accuracy. 
Half-space depth performed similarly to spatial depth, but computationally it was much slower. 
Half-space depth is affine-invariant whereas spatial depth is only similarity-invariant, therefore if affine invariance is desired in the analysis the analyst should use the half-space depth function. 
Modified Mahalanobis depth and Mahalanobis depth both performed slightly worse than the other two depth functions if the data was Gaussian, but performed considerable worse when the distribution of the data was Cauchy, which could be attributed to the robustness considerations of Mahalanobis depth discussed in Section \ref{sec::depth}.

\begin{figure}[t!]
\begin{minipage}[c]{.49\textwidth} 
\centering%
\includegraphics[width=.8\textwidth]{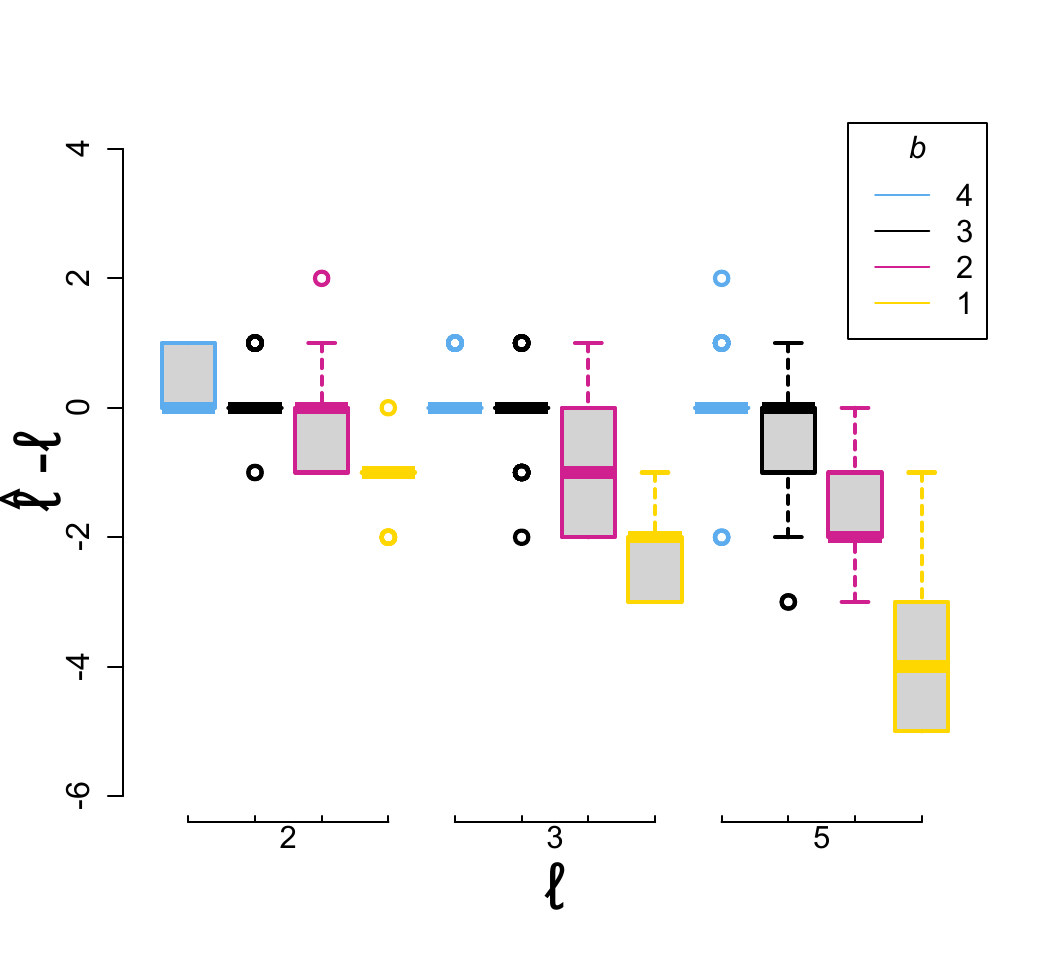}
\caption*{(a) WBS algorithm}
\end{minipage}
\begin{minipage}[c]{.49\textwidth} 
\centering%
\includegraphics[width=.8\textwidth]{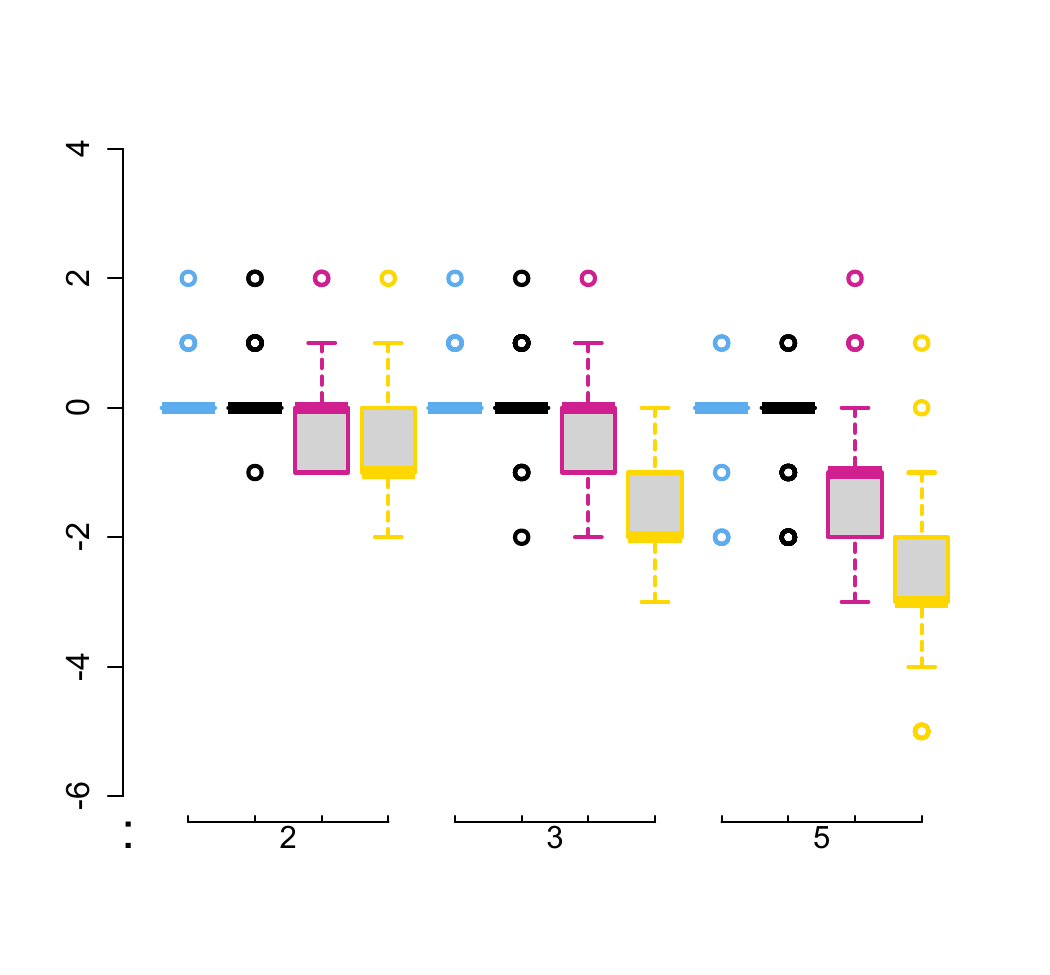}
\caption*{(b) KW-PELT algorithm}
\end{minipage}
\caption{Boxplots of $\hat{\ell}-\ell$ for the third simulation scenario under spatial depth, where the colours indicate $b$, the size of the submatrix to which the expansion/contraction was applied. In other words, the submatrix in which an expansion or contraction was applied had dimension $b\times b$. The underlying numbers represent the number of true change-points in the simulation.}%
\label{fig:sub}
\end{figure}
Figure \ref{fig:sim1} shows boxplots of $\hat{\ell}-\ell$ under Algorithm \ref{alg:wbs}, Algorithm \ref{alg:pelt}, the WBSIP algorithm and the \texttt{e.divisive} algorithm, for the first simulation scenario with $N=1000$. 
Recall that $\hat{\ell}-\ell$ is the estimated number of change-points minus the actual number in the simulation run. 
Each boxplot represents a different combination of simulation parameters, e.g., the first boxplot represents the (empirical) distribution of $\hat{\ell}-\ell$ with simulated two-dimensional Gaussian data that had 2 change-points. 
The empirical distributions are computed over the 100 replications of each combination of simulation parameters.

Figure \ref{fig:sim1} shows that both Algorithm \ref{alg:wbs} and Algorithm \ref{alg:pelt} estimate the number of change-points more accurately than the WBSIP algorithm and the e-divisive algorithm. 
One reason for this is that our method is less general in terms of the types of changes it can detect when compared to the other two algorithms; our method trades generality for accuracy. 
A second reason for these results is robustness. 
When the data is Cauchy, neither of the competing algorithms perform very well. 
Neither of these methods are designed to handle outliers or heavy-tailed data. 
For example, when the data have Cauchy marginals, the assumptions for consistency of the WBSIP procedure are violated. 
One should also recall that WBSIP was designed for high dimensional data which is not the main focus of this simulation study. 
In Figure \ref{fig:sim1} it is easily seen that as the dimension increases WBSIP performs better. 
Speaking of dimension, both Algorithms \ref{alg:wbs} and \ref{alg:pelt} were very insensitive to the dimension and the number of change-points. 

In terms of accuracy when the change-point was detected, both algorithms performed very well. 
Figures \ref{fig:sim2} shows boxplots of $\widehat{k}/N-\theta$ for Algorithm \ref{alg:wbs}, Algorithm \ref{alg:pelt}, WBSIP and the e-divisive algorithm when $d$ was 10 and the distribution was normally distributed. 
Boxplots under other simulation parameters were similar\footnote{They were of course worse when the data came from a Cauchy distribution.} and can be seen in the Appendix \ref{app::add_sim}. 
Generally, the estimates were at most about 5\% off of the true break fraction, with the majority of biases being in the 1\% range. 

This corresponds to 10 time units away when $N=1000$; $\hat{k}$ was typically within 10 time units of the true change-point $k$. 
Again our methods were insensitive to the dimension and the number of change-points. 
When compared with the WBSIP algorithm and e-divisive algorithm, we see that both Algorithm \ref{alg:wbs} and Algorithm \ref{alg:pelt} appear to estimate the location of the change-points more accurately. 
We can then conclude that our algorithms outperforms the WBSIP algorithm and e-divisive algorithm when the data only contains changes in variability. 

\begin{figure}[t]
\begin{minipage}[c]{.5\textwidth} 
\centering%
\includegraphics[width=.9\textwidth,trim={0 00 0 0},clip]{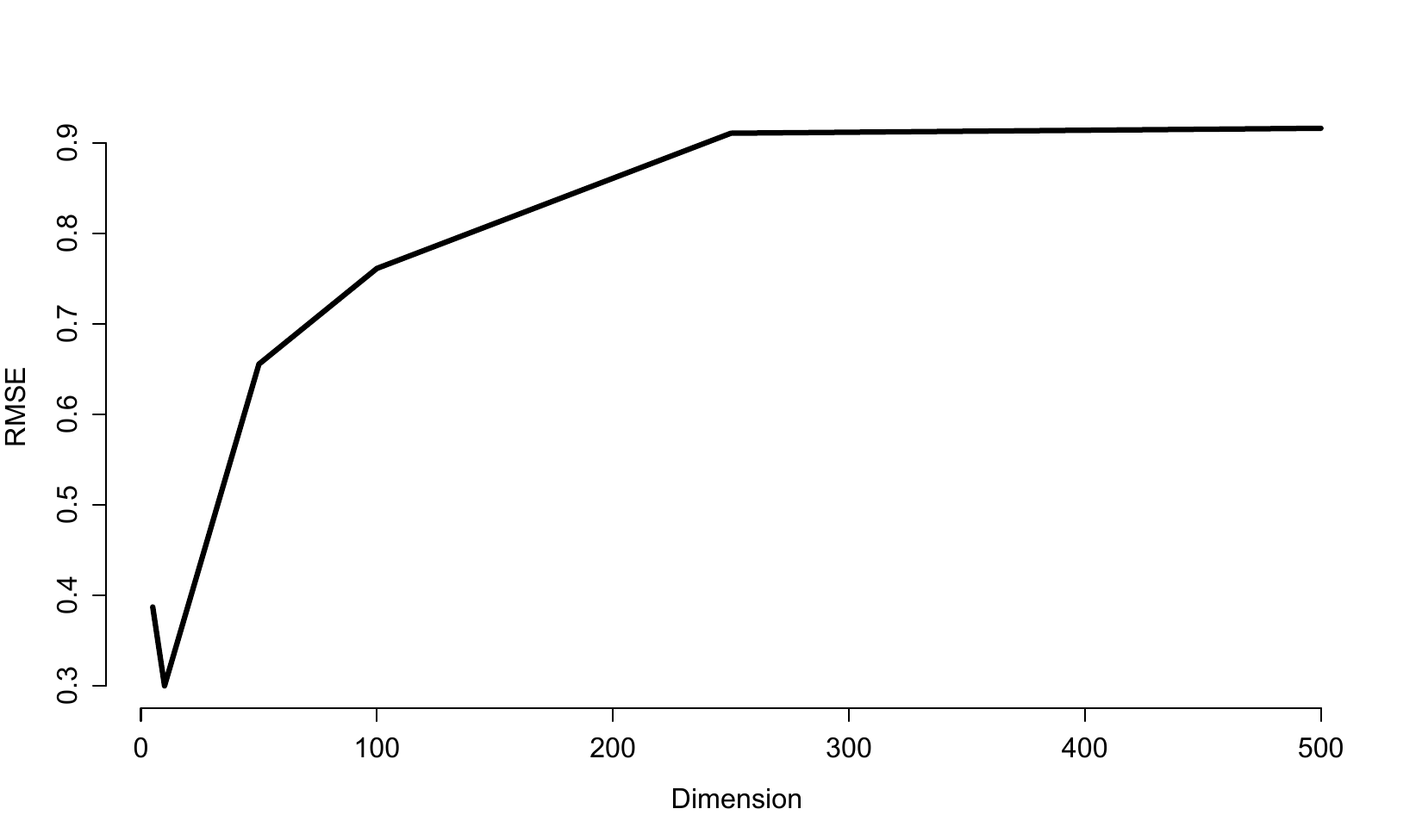}
\caption*{(a)}
\end{minipage}\hfill
\begin{minipage}[c]{.5\textwidth} 
\centering%
\includegraphics[width=.9\textwidth,trim={0 0 0 0},clip]{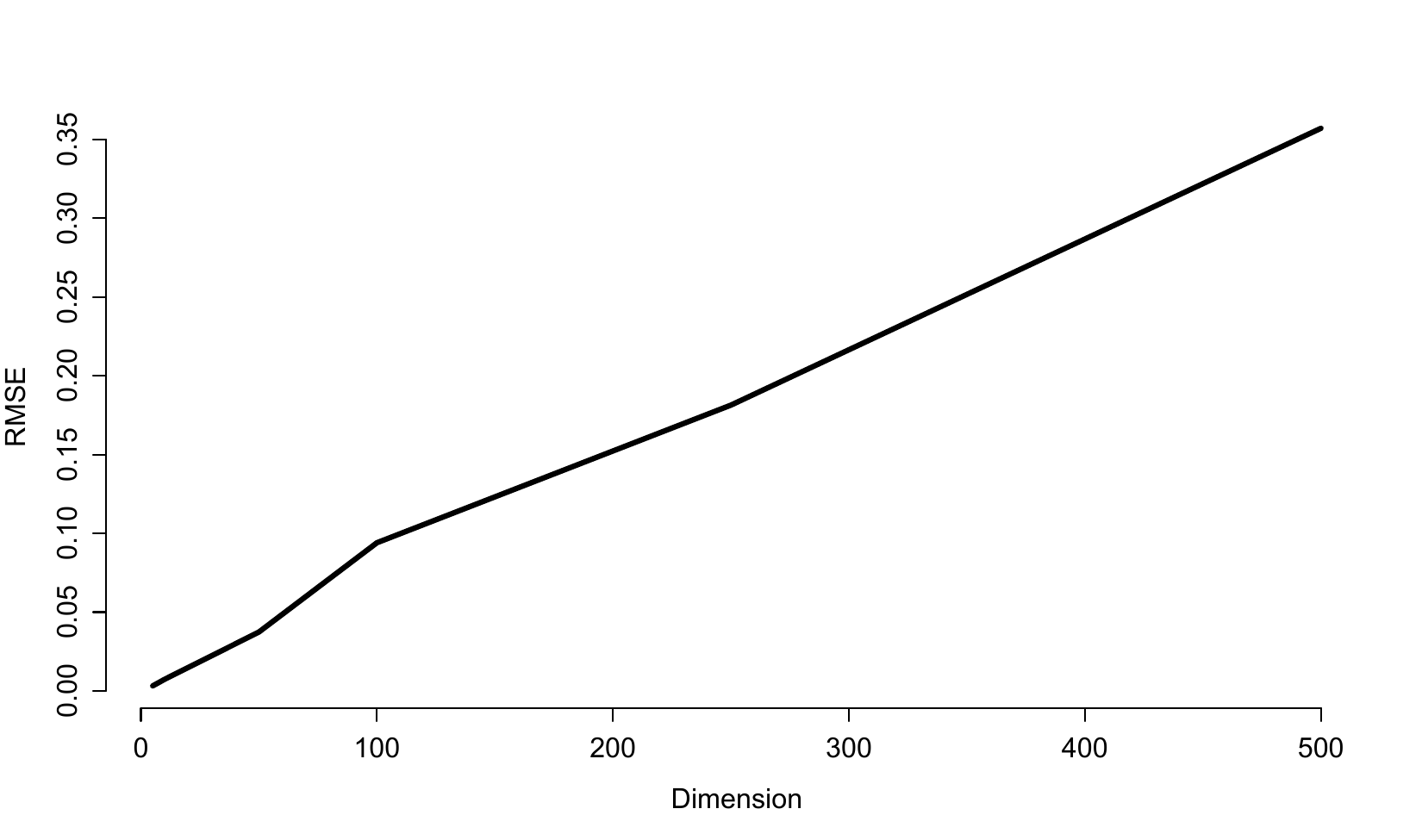}
\caption*{(b)}
\end{minipage}\hfill
\caption{(a) Empirical root mean square error of $\hat{\ell}$ as the dimension increases, but the size of the change remains fixed. (b) Empirical root mean square error of $\hat{k}$ as the dimension increases, but the size of the change remains fixed. The change-points were estimated using the KW-PELT algorithm paired with spatial depth. }%
\label{fig:HD}%
\end{figure}

In a different simulation scenario, we fixed $d=5$ and applied the expansions and contractions of the first scenario to a $b\times b$ submatrix of the covariance matrix. 
Figure \ref{fig:sub} shows boxplots of $\hat{\ell}-\ell$ for both algorithms, under this simulation scenario. 
Figure \ref{fig:sub} shows the results for spatial depth, with $N=1000.$ 
The plots for the other depth functions can be seen in Appendix \ref{app::add_sim}. 
The colour of the boxplot in Figure \ref{fig:sub} represents $b$, the size of the submatrix to which the expansion/contraction was applied, i.e., the submatrix in which an expansion or contraction was applied had dimension $b\times b$. 
The underlying numbers represent the number of change-points in that particular simulation scenario. 
We see that as $b$ decreases, the ability to detect the changes decreases. 
This is expected, since a smaller change should be more difficult to detect. 
We can also see that here, the KW-PELT algorithm performs better than the WBS algorithm. 
One remedy for detecting changes in relatively low dimensions might be to subsample dimensions of the data and run the procedure on each of the subsampled dimensions. 
We leave that for future work.

Lastly, we ran simulations designed to assess the performance of our methods in high dimensions. 
We use the KW-PELT algorithm with spatial depth, due to both its performance and the fact that spatial depth can be computed quickly in high dimensions. 
We simulated normal data, with one change-point and with two change-points, at $N=1000$ for $d=50$ and $d=500$. 
The KW-PELT algorithm estimated both the number of change-points and the location of the change-points were detected with 100\% accuracy. 
This is not surprising since we might view an expansion of a very large matrix as a very large change in variability. 
For example, the trace of the expanded matrix is increasing as the dimension is increased, and so the signal is increasing with the dimension under an expansion-type change. 
The story changes if the data is high dimensional and the data is sparse, i.e., the change only occurs in a submatrix of the covariance matrix which has a fixed dimension. 
We ran another simulation where there is was change point, and the change only occurred in a $5\times 5$ submatrix. 
Figure \ref{fig:HD} shows that as we increase the dimension, the algorithm has a more difficult time estimating the change-point accurately. 
This suggests that when the data is suspected to be very sparse, we may wish to develop a depth function that accounts for sparsity. 

In summary, which both algorithms performed very well relative to competitors in this simulation set-up. 
The results also show that the WBS algorithm and the KW-PELT algorithm are very comparable. 
The KW-PELT algorithm is computationally faster, and can be more accurate under sparsity.
However, its tuning parameter requires some subjectivity. 
Furthermore, both algorithms have the same theoretical rate of convergence. 
We are tempted to recommended the KW-PELT algorithm with the understanding that the performance of the algorithms is very similar. 
In terms of the depth functions, half-space depth and spatial depth performed better than the two Mahalanobis depth variants. 
Since half-space depth takes longer to compute, we would ultimately recommend using spatial depth with either algorithm in practice. 

\section{An application to financial returns}\label{sec::da}
In this section we apply the methodology to four daily stock returns. 
\texttt{R} codes for this analysis can be found on Github at \citep{Ramsay2019}. 
We analyze the same data set analyzed by \cite{Galeano2017} and compare our results to those produced by their method. 
It is expected that algorithms will produce different results, due to the fact that the aim of \cite{Galeano2017} was to detect changes in the correlation structure of the returns; not necessarily the covariance matrix. For example, they assume constant variances over time. 
The results should be seen as complementary to those of \cite{Galeano2017}.

\begin{figure}[t]
\begin{minipage}{.64\linewidth} 
\centering%
\begin{tabular}[b]{ccc}
\toprule
Change-point WBS & Change-point KW-PELT & CUSUM value   \\ \midrule
    Jul 18 `07 &Jul 26 `07 & 2.43  \\
    Sep 05 `08 &Sep 25 `08 & 5.49 \\
    Dec 08 `08 &Dec 08 `08 & 2.75  \\
    May 01 `09 & May 19 `09& 2.13  \\
    Aug 25 `09 & ND & 6.36 \\
    ND & Jul 22 `10& -\\
    Jul 25 `11 &Jul 25 `11 & 2.36 
\end{tabular}
\end{minipage}
\begin{minipage}{.35\textwidth} 
\centering
\includegraphics[width=\textwidth]{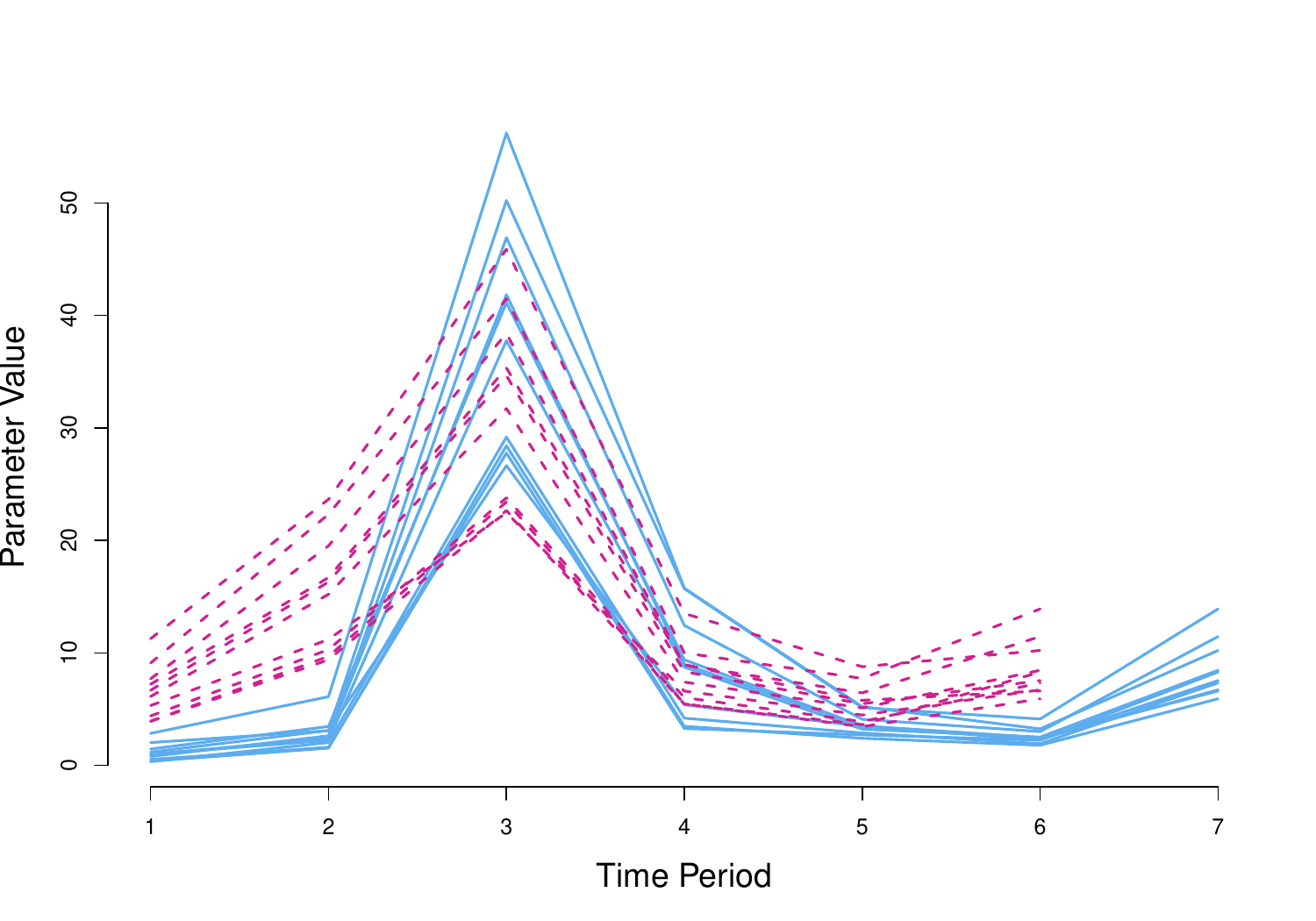}
\end{minipage}
\caption{\textit{Left}: Change-points estimated by Algorithms \ref{alg:wbs} and \ref{alg:pelt}. ND stands for not detected by the Algorithm. Associated CUSUM statistics are also provided. \textit{Right}: Covariance matrix parameters at each interval for both Algorithm \ref{alg:wbs} (pink dashed) and Algorithm \ref{alg:pelt} (blue solid) connected by lines to emphasize the change in the parameter values.}
\label{fig:CovOT}
\end{figure}

It is clear that this data has some serial dependence; it does not fit the independence assumption. 
That being said, we feel that the results still provide some insight into the data. For example, the data appears to admit a weak dependence structure.
As a result of the concentration inequality for rank statistics for $m$-dependent data \citep{Wang2019}, we only need Assumption \ref{ass:consDepth} to hold under $m$-dependence in order for the consistency properties to hold. 
In fact, the consistency of many depth functions is, in part, a result of Glivenko-Cantelli type theorems. 
Seeing as extensions of such theorems exist for $m$-dependent data \citep{Bobkov2010} it is likely possible to extend the results of Section \ref{sec::theo}.
The convergence of depth functions for dependent data is an interesting topic for further research.

We applied the both proposed Algorithms to the raw daily returns.
We ran the WBS algorithm with 700 intervals ($100\floor{\log{N}}$) using all depth functions with $\alpha=0.9$.  
When running Algorithm \ref{alg:pelt}, we used penalty constants $C_1=0.24$ and $C_2=3.74$, these were chosen according to the discussion in Section \ref{sec::params}. 
The results did not vary at all among the different depth functions for the Algorithm \ref{alg:wbs}, and were virtually the same under Algorithm \ref{alg:pelt}, the only difference was that the modified Mahalanobis depth predicted the December 2008 change-point on December 9th rather than on the $8^{th}$.

Table \ref{fig:CovOT} contains the estimated change-points produced by the Algorithms and the associated CUSUM statistic values from Algorithm \ref{alg:wbs}. 
Figure \ref{fig:ret} plots the estimated change-points on the data from both Algorithms. 
Observe that algorithms are also both unaffected by the outliers in the Siemens returns, which can be seen to the left and right of January 2008. 
Some of the change-points have a clear interpretation.
For example, the first change-point (July 18, 2007) signifies the beginning of the global financial crisis and the second (September 05, 2008) is associated with the collapse of Lehman brothers.
In the following months, measures to stem the effects of the crisis may contribute to the next two change-points. 
For example, in early December 2008 the EU agreed to a 200 billion dollar stimulus package. 
The later change-points are associated with the Greek government debt crisis; in July 2011, the Troika approved a second bailout (of the Greek government).

The algorithms reproduced both change-points found by \citep{Galeano2017} (July 18, 2007 and September 05, 2008). 
Changes in correlation could be accompanied by expansions or contractions in the covariance matrix of these returns. 
It is possible that these changes (correlation and covariance) are byproducts of a general increase/decrease in systematic volatility. 
Many financial returns are generally thought to have some systematic/market-wide dependence \citep{zvibodieprofessor2017}. 
Figure \ref{fig:CovOT} shows the estimated pairwise covariances as well as the estimated variances of each stock within each period of `no change'. 
The uniform movement of the parameters indicate contractions and expansions, rather than some other type of change. 
Additionally, we note that all changes under Algorithm \ref{alg:wbs} were significant when the Bonferoni correction was applied to the set of test statistics at the 5\% level of significance. 

\begin{figure}[t]
\begin{minipage}[c]{.5\textwidth} 
\centering%
\includegraphics[width=\textwidth,trim={0 00 0 0},clip]{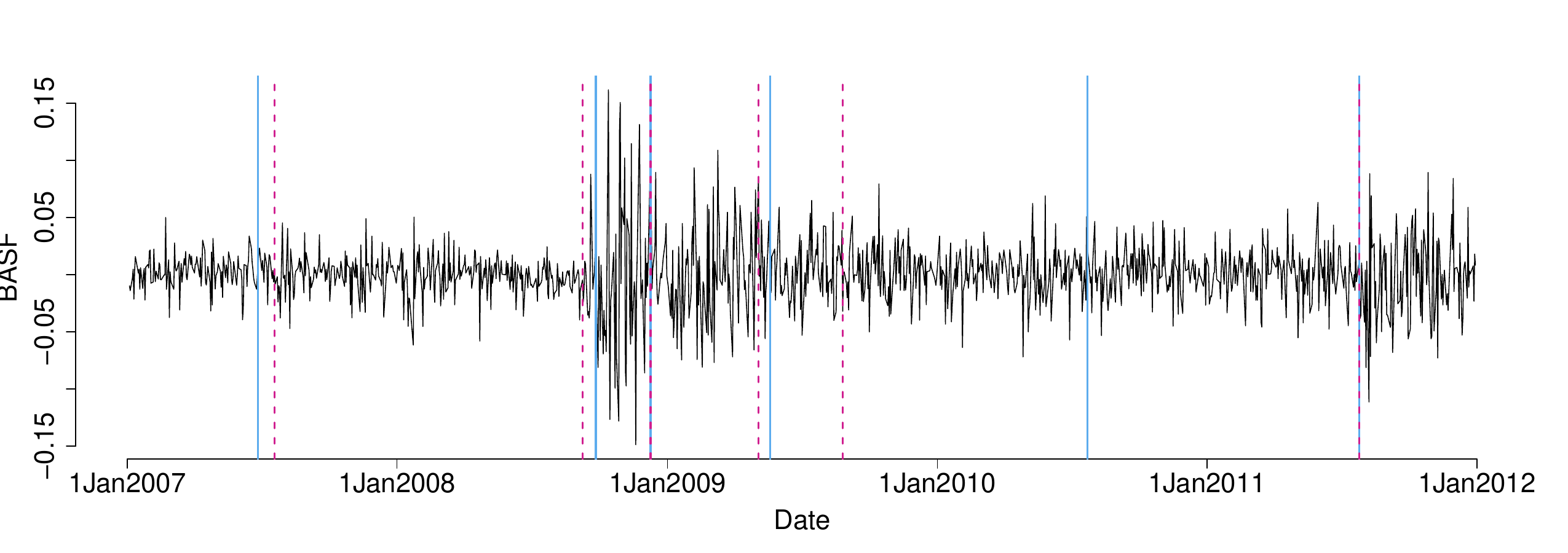}
\end{minipage}\hfill
\begin{minipage}[c]{.5\textwidth} 
\centering%
\includegraphics[width=\textwidth,trim={0 0 0 0},clip]{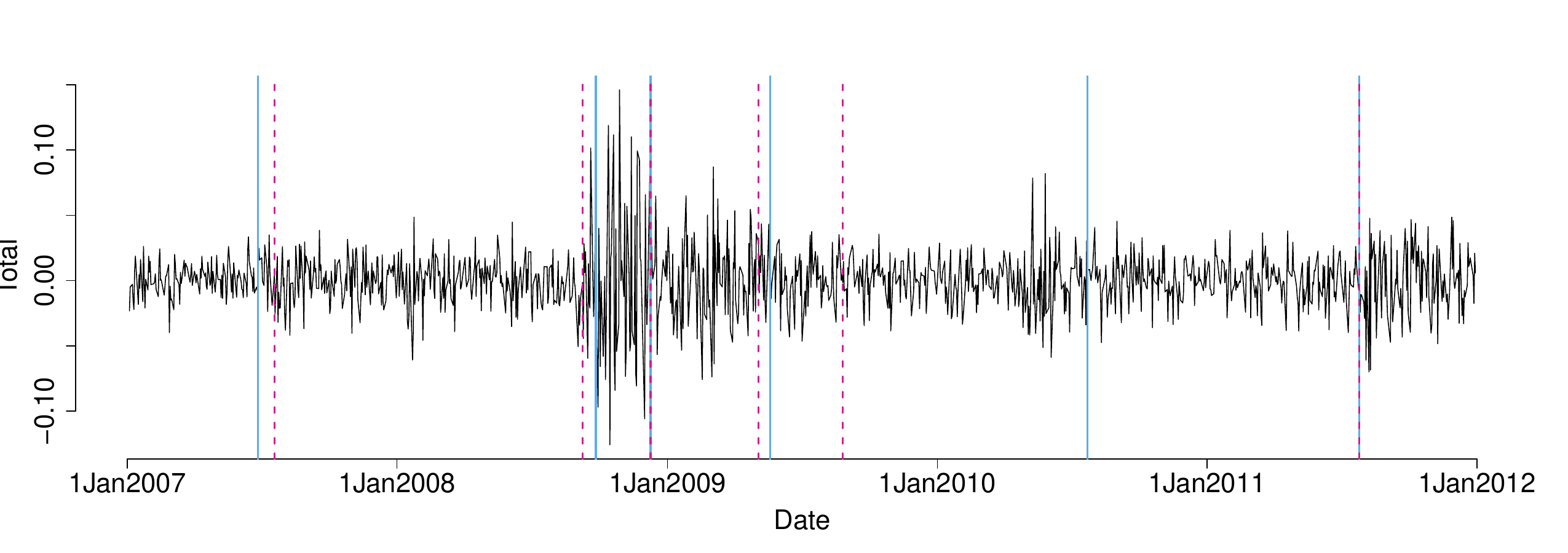}
\end{minipage}\hfill\newline
\begin{minipage}[c]{.5\textwidth} 
\centering%
\includegraphics[width=\textwidth,trim={0 00 0 0},clip]{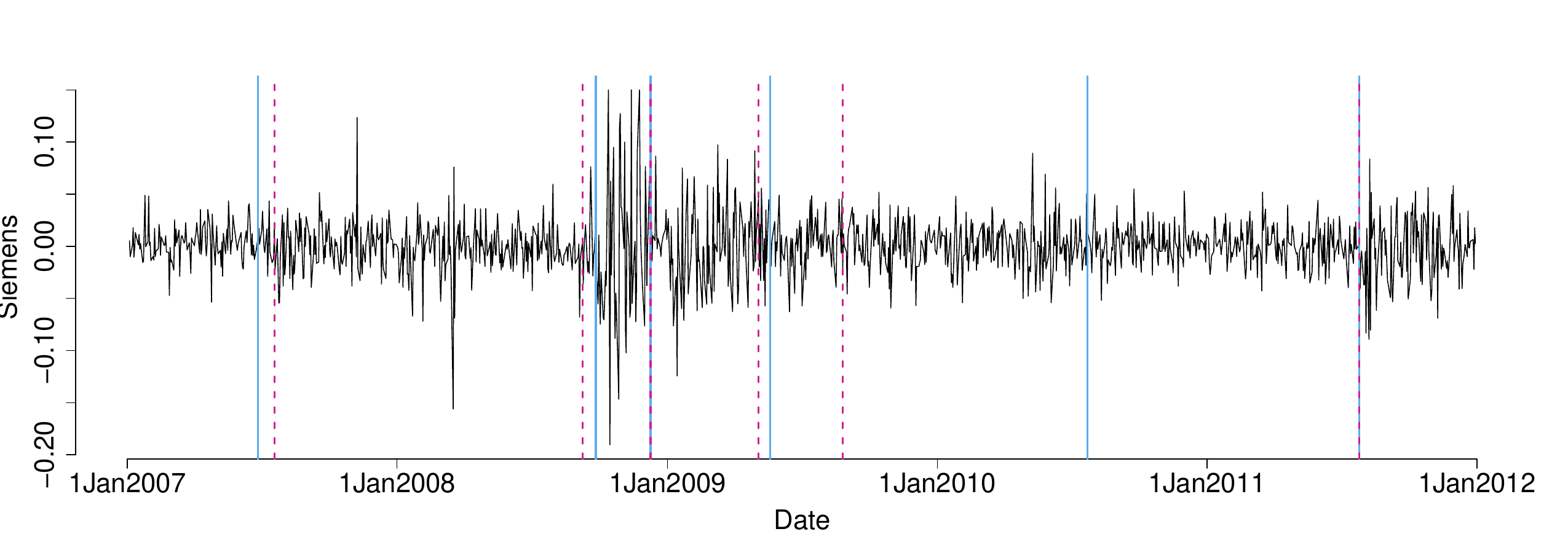}
\end{minipage}\hfill
\begin{minipage}[c]{.5\textwidth} 
\centering%
\includegraphics[width=\textwidth,trim={0 0 0 0},clip]{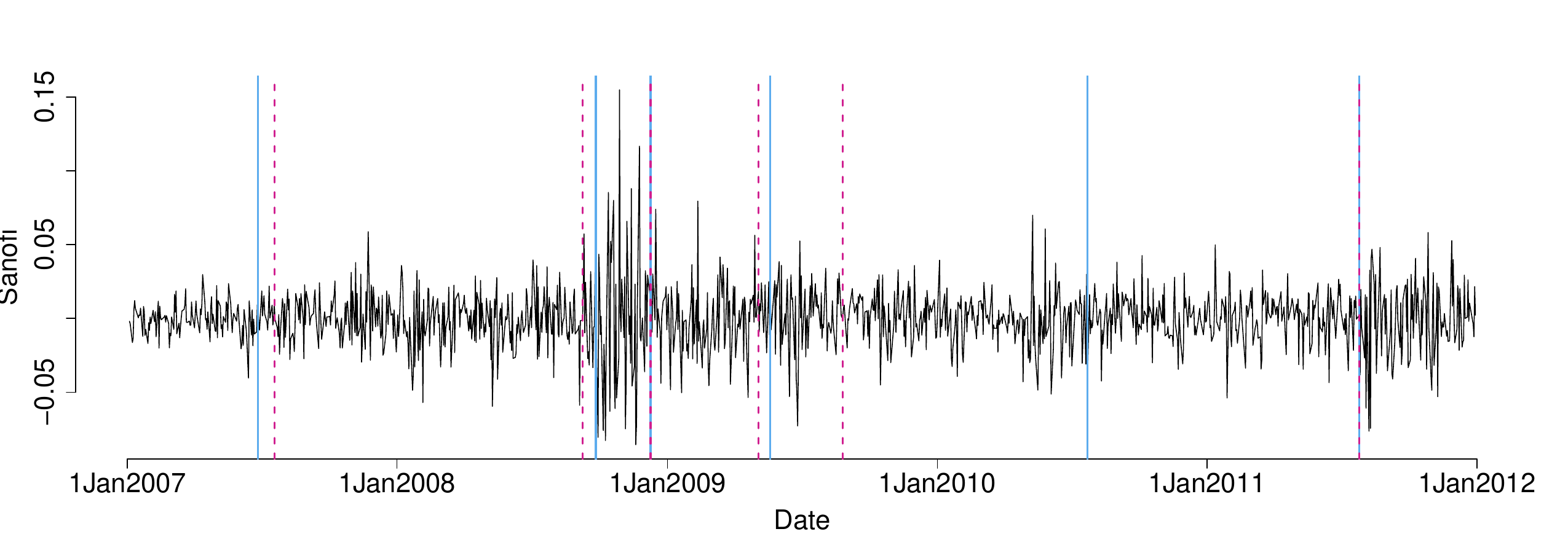}
\end{minipage}\hfill
\caption{Returns with estimated change-points from Algorithm \ref{alg:pelt} marked by solid, blue lines and Algorithm \ref{alg:wbs} marked by dashed pink lines.}%
\label{fig:ret}%
\end{figure}

\bibliography{main} 

\begin{thebibliography}{}

\bibitem[\protect\astroncite{Aminikhanghahi and
  Cook}{2017}]{Aminikhanghahi2017}
Aminikhanghahi, S. and Cook, D. (2017).
\newblock {A survey of methods for time series change point detection}.
\newblock {\em Knowledge and Information Systems}, 51(2):339--367.

\bibitem[\protect\astroncite{Ansari and Bradley}{1960}]{ansari1960}
Ansari, A.~R. and Bradley, R.~A. (1960).
\newblock Rank-sum tests for dispersions.
\newblock {\em The Annals of Mathematical Statistics}, 31(4):1174--1189.

\bibitem[\protect\astroncite{{Arlot} et~al.}{2012}]{2012Arlot}
{Arlot}, S., {Celisse}, A., and {Harchaoui}, Z. (2012).
\newblock {A Kernel Multiple Change-point Algorithm via Model Selection}.
\newblock {\em arXiv e-prints}, page arXiv:1202.3878.

\bibitem[\protect\astroncite{Aue et~al.}{2009}]{Aue2009}
Aue, A., H{\"{o}}rmann, S., Horv{\'{a}}th, L., and Reimherr, M. (2009).
\newblock Break detection in the covariance structure of multivariate time
  series models.
\newblock {\em The Annals of Statistics}, 37(6B):4046--4087.

\bibitem[\protect\astroncite{Aue and Horv{\'{a}}th}{2013}]{Aue2013}
Aue, A. and Horv{\'{a}}th, L. (2013).
\newblock {Structural breaks in time series}.
\newblock {\em Journal of Time Series Analysis}, 34(1):1--16.

\bibitem[\protect\astroncite{Benjamini and Hochberg}{1995}]{BH1995}
Benjamini, Y. and Hochberg, Y. (1995).
\newblock Controlling the false discovery rate: A practical and powerful
  approach to multiple testing.
\newblock {\em Journal of the Royal Statistical Society. Series B
  (Methodological)}, 57(1):289--300.

\bibitem[\protect\astroncite{{Bhattacharyya} and
  {Kasa}}{2018}]{Bhattacharyya2018}
{Bhattacharyya}, M. and {Kasa}, S.~R. (2018).
\newblock A test for detecting structural breakdowns in markets using
  eigenvalue decompositions.
\newblock {\em arXiv e-prints}, page arXiv:1809.07114.

\bibitem[\protect\astroncite{Bickel}{1965}]{bickel1965}
Bickel, P.~J. (1965).
\newblock On some asymptotically nonparametric competitors of hotelling's $t^{2
  1}$.
\newblock {\em Annals of Mathematical Statistics}, 36(1):160--173.

\bibitem[\protect\astroncite{Bobkov and G{\"{o}}tze}{2010}]{Bobkov2010}
Bobkov, S. and G{\"{o}}tze, F. (2010).
\newblock {Concentration of empirical distribution functions with applications
  to non-i.i.d. models}.
\newblock {\em Bernoulli}, 16(4):1385--1414.

\bibitem[\protect\astroncite{Bodie et~al.}{2017}]{zvibodieprofessor2017}
Bodie, Z., Kane, A., Marcus, A., Perrakis, S., and Ryan, P. (2017).
\newblock {\em Investments}.
\newblock McGraw-Hill Education.

\bibitem[\protect\astroncite{Cabrieto et~al.}{2018}]{Cabrieto2018}
Cabrieto, J., Tuerlinckx, F., Kuppens, P., Hunyadi, B., and Ceulemans, E.
  (2018).
\newblock Testing for the presence of correlation changes in a multivariate
  time series: A permutation based approach.
\newblock {\em Scientific Reports}, 8(1):769.

\bibitem[\protect\astroncite{Chakraborti and Graham}{2019}]{Chakraborti2019}
Chakraborti, S. and Graham, M.~A. (2019).
\newblock {Nonparametric (distribution-free) control charts: An updated
  overview and some results}.
\newblock {\em Quality Engineering}, pages 1--22.

\bibitem[\protect\astroncite{Chenouri et~al.}{2019}]{Chenouri2019}
Chenouri, S., Mozaffari, A., and Rice, G. (2019).
\newblock {Multiple change point detection based on standard and wild
  rank-cusum binary segmentation}.

\bibitem[\protect\astroncite{Chenouri et~al.}{2020}]{Chenouri2020DD}
Chenouri, S., Mozaffari, A., and Rice, G. (2020).
\newblock Robust multivariate change point analysis based on data depth.
\newblock {\em Canadian Journal of Statistics}, 48(3):417--446.

\bibitem[\protect\astroncite{{Dette} et~al.}{2018}]{Dette2018}
{Dette}, H., {Pan}, G.~M., and {Yang}, Q. (2018).
\newblock {Estimating a change point in a sequence of very high-dimensional
  covariance matrices}.
\newblock {\em arXiv e-prints}, page arXiv:1807.10797.

\bibitem[\protect\astroncite{Duan and Wied}{2018}]{Duan2018}
Duan, F. and Wied, D. (2018).
\newblock {A residual-based multivariate constant correlation test}.
\newblock {\em Metrika}, 81(6):653--687.

\bibitem[\protect\astroncite{Dyckerhoff et~al.}{1996}]{Dyckerhoff1996}
Dyckerhoff, R., Mosler, K., and Koshevoy, G. (1996).
\newblock Zonoid data depth: Theory and computation.
\newblock In {\em COMPSTAT}, pages 235--240, Heidelberg. Physica-Verlag HD.

\bibitem[\protect\astroncite{Fryzlewicz}{2014}]{Fryzlewicz2014}
Fryzlewicz, P. (2014).
\newblock {Wild Binary Segmentation for Multiple Changepoint Detection}.
\newblock {\em The Annals of Statistics}, 42(6):2243--2281.

\bibitem[\protect\astroncite{Galeano and Pe{\~{n}}a}{2007}]{Galeano2007}
Galeano, P. and Pe{\~{n}}a, D. (2007).
\newblock {Covariance changes detection in multivariate time series}.
\newblock {\em Journal of Statistical Planning and Inference}, 137(1):194--211.

\bibitem[\protect\astroncite{Galeano and Wied}{2014}]{Galeano2014}
Galeano, P. and Wied, D. (2014).
\newblock {Multiple break detection in the correlation structure of random
  variables}.
\newblock {\em Computational Statistics {\&} Data Analysis}, 76:262--282.

\bibitem[\protect\astroncite{Galeano and Wied}{2017}]{Galeano2017}
Galeano, P. and Wied, D. (2017).
\newblock {Dating multiple change points in the correlation matrix}.
\newblock {\em TEST}, 26:331--352.

\bibitem[\protect\astroncite{James and Matteson}{2015}]{ecp_package}
James, N.~A. and Matteson, D.~S. (2015).
\newblock ecp: An r package for nonparametric multiple change point analysis of
  multivariate data.
\newblock {\em Journal of Statistical Software}, 62(7):1–25.

\bibitem[\protect\astroncite{Kao et~al.}{2018}]{Kao2018}
Kao, C., Trapani, L., and Urga, G. (2018).
\newblock {Testing for instability in covariance structures}.
\newblock {\em Bernoulli}, 24(1):740--771.

\bibitem[\protect\astroncite{Killick et~al.}{2012}]{Killick2012}
Killick, R., Fearnhead, P., and Eckley, I.~A. (2012).
\newblock Optimal detection of changepoints with a linear computational cost.
\newblock {\em Journal of the American Statistical Association},
  107(500):1590--1598.

\bibitem[\protect\astroncite{Kruskal}{1952}]{Kruskal1952}
Kruskal, W.~H. (1952).
\newblock A nonparametric test for the several sample problem.
\newblock {\em The Annals of Mathematical Statistics}, 23(4):525--540.

\bibitem[\protect\astroncite{Li and Liu}{2004}]{li2004}
Li, J. and Liu, R.~Y. (2004).
\newblock New nonparametric tests of multivariate locations and scales using
  data depth.
\newblock {\em Statist. Sci.}, 19(4):686--696.

\bibitem[\protect\astroncite{Liu et~al.}{1999}]{Liu1999}
Liu, R.~Y., Parelius, J.~M., and Singh, K. (1999).
\newblock Multivariate analysis by data depth: Descriptive statistics, graphics
  and inference.
\newblock {\em The Annals of Statistics}, 27(3):783--840.

\bibitem[\protect\astroncite{{Lung-Yut-Fong}
  et~al.}{2011}]{2011arXiv1107.1971L}
{Lung-Yut-Fong}, A., {L{\'e}vy-Leduc}, C., and {Capp{\'e}}, O. (2011).
\newblock {Homogeneity and change-point detection tests for multivariate data
  using rank statistics}.
\newblock {\em arXiv e-prints}, page arXiv:1107.1971.

\bibitem[\protect\astroncite{Matteson and James}{2014}]{Matteson2014}
Matteson, D.~S. and James, N.~A. (2014).
\newblock A nonparametric approach for multiple change point analysis of
  multivariate data.
\newblock {\em Journal of the American Statistical Association},
  109(505):334--345.

\bibitem[\protect\astroncite{Page}{1954}]{Page1954}
Page, E.~S. (1954).
\newblock {Continuous Inspection Schemes}.
\newblock {\em Biometrika}, 41(1-2):100--115.

\bibitem[\protect\astroncite{Posch et~al.}{2019}]{Posch2019}
Posch, P.~N., Ullmann, D., and Wied, D. (2019).
\newblock {Detecting structural changes in large portfolios}.
\newblock {\em Empirical Economics}, 56(4):1341--1357.

\bibitem[\protect\astroncite{{Ramsay} and
  {Chenouri}}{2021}]{2021arXiv210610173R}
{Ramsay}, K. and {Chenouri}, S. (2021).
\newblock {Robust nonparametric hypothesis tests for differences in the
  covariance structure of functional data}.
\newblock {\em arXiv e-prints}, page arXiv:2106.10173.

\bibitem[\protect\astroncite{Ramsay et~al.}{2019}]{RAMSAY201951}
Ramsay, K., Durocher, S., and Leblanc, A. (2019).
\newblock Integrated rank-weighted depth.
\newblock {\em Journal of Multivariate Analysis}, 173:51 -- 69.

\bibitem[\protect\astroncite{Ramsay}{2019}]{Ramsay2019}
Ramsay, K.~A. (2019).
\newblock Mvt-wbs-rankcusum.
\newblock \url{https://github.com/12ramsake/MVT-WBS-RankCUSUM}.

\bibitem[\protect\astroncite{Reeves et~al.}{2007}]{Reeves2007}
Reeves, J., Chen, J., Wang, X.~L., Lund, R., Lu, Q.~Q., Reeves, J., Chen, J.,
  Wang, X.~L., Lund, R., and Lu, Q.~Q. (2007).
\newblock A review and comparison of changepoint detection techniques for
  climate data.
\newblock {\em Journal of Applied Meteorology and Climatology}, 46(6):900--915.

\bibitem[\protect\astroncite{Rousseeuw and van Zomeren}{1990}]{Rousseeuw1990}
Rousseeuw, P.~J. and van Zomeren, B.~C. (1990).
\newblock Unmasking multivariate outliers and leverage points.
\newblock {\em Journal of the American Statistical Association},
  85(411):633--639.

\bibitem[\protect\astroncite{Serfling}{2002}]{Serfling2002}
Serfling, R. (2002).
\newblock A depth function and a scale curve based on spatial quantiles.
\newblock In {\em Statistical Data Analysis Based on the $L_1$-Norm and Related
  Methods}, pages 25--38. {Birkh{\"{a}}user}, Basel.

\bibitem[\protect\astroncite{Serfling}{2006}]{Serfling2006}
Serfling, R.~J. (2006).
\newblock Depth functions in nonparametric multivariate inference.
\newblock {\em Data Depth: Robust Multivariate Analysis, Computational
  Geometry, and Applications}, pages 1--16.

\bibitem[\protect\astroncite{Shewhart}{1931}]{Shewhart1931}
Shewhart, W.~A. (1931).
\newblock {\em {Economic control of quality of manufactured product.}}
\newblock Van Nostrand, Oxford, England.

\bibitem[\protect\astroncite{Siegel and Tukey}{1960}]{Sieg1960}
Siegel, S. and Tukey, J.~W. (1960).
\newblock A nonparametric sum of ranks procedure for relative spread in
  unpaired samples.
\newblock {\em Journal of the American Statistical Association},
  55(291):429--445.

\bibitem[\protect\astroncite{Truong et~al.}{2020}]{Truong2018}
Truong, C., Oudre, L., and Vayatis, N. (2020).
\newblock Selective review of offline change point detection methods.
\newblock {\em Signal Processing}, 167:107299.

\bibitem[\protect\astroncite{Tukey}{1974}]{Tukey1974}
Tukey, J.~W. (1974).
\newblock Mathematics and the picturing of data.
\newblock In {\em Proceedings of the International Congress of Mathematicians}.

\bibitem[\protect\astroncite{{Venkatraman, E.}}{1992}]{Venkatraman1992}
{Venkatraman, E.} (1992).
\newblock {\em Consistency Results in Multiple Change-Point Problems}.
\newblock PhD thesis, Stanford University, Department of Statistics.

\bibitem[\protect\astroncite{Wang et~al.}{2021}]{Wang2021}
Wang, D., Yu, Y., and Rinaldo, A. (2021).
\newblock {Optimal covariance change point localization in high dimensions}.
\newblock {\em Bernoulli}, 27(1):554--575.

\bibitem[\protect\astroncite{Wang et~al.}{2019}]{Wang2019}
Wang, Y., Wang, Z., and Zi, X. (2019).
\newblock {Rank-based multiple change-point detection}.
\newblock {\em Communications in Statistics - Theory and Methods}, 0(0):1--17.

\bibitem[\protect\astroncite{Weber}{1980}]{Weber1980}
Weber, N.~C. (1980).
\newblock A martingale approach to central limit theorems for exchangeable
  random variables.
\newblock {\em Journal of Applied Probability}, 17(3):662--673.

\bibitem[\protect\astroncite{Wied et~al.}{2012}]{Wied2012}
Wied, D., Kr{\"{a}}mer, W., and Dehling, H. (2012).
\newblock Testing for a change in correlation at an unknown point in time using
  an extended functional delta method.
\newblock {\em Econometric Theory}, 28(3):570--589.

\bibitem[\protect\astroncite{Zhang et~al.}{2017}]{Zhang2017}
Zhang, W., James, N., and Matteson, D. (2017).
\newblock Pruning and nonparametric multiple change point detection.

\bibitem[\protect\astroncite{Zhao}{2017}]{Zhao2017}
Zhao, Y. (2017).
\newblock {\em {An analysis of the stability in multivariate correlation
  structures}}.
\newblock PhD thesis, Birmingham Business School, Department of Economics.

\bibitem[\protect\astroncite{Zuo}{2003}]{Zuo2003}
Zuo, Y. (2003).
\newblock Projection-based depth functions and associated medians.
\newblock {\em The Annals of Statistics}, 31(5):1460--1490.

\bibitem[\protect\astroncite{Zuo}{2019}]{Zuo2019}
Zuo, Y. (2019).
\newblock {A new approach for the computation of halfspace depth in high
  dimensions}.
\newblock {\em Communications in Statistics - Simulation and Computation},
  48(3):900--921.

\bibitem[\protect\astroncite{Zuo and Serfling}{2000}]{Zuo2000}
Zuo, Y. and Serfling, R. (2000).
\newblock General notions of statistical depth function.
\newblock {\em The Annals of Statistics}, 28(2):461--482.

\end{thebibliography}
\bibliographystyle{apa}
\appendix
\section{Proofs}\label{app::proofs}
\begin{proof}[Proof of Theorem \ref{thm:WBS}]
We first define the following ranks based on the population depth functions
\begin{align*}
    R_{i,s,e}&\coloneqq \#\left\{X_j\colon \mathcal{D}(X_j;F_{*,s,e})\leq \mathcal{D}(X_i;F_{*,s,e}),\ j\in \{s,\ldots ,e\}\right\}, \ i \in \{s,\ldots ,e\}. 
\end{align*}
The distribution $F_{*,s,e}$ is a mixture distribution with weights proportional to the number of observations coming from $F_j$ in the subsample $\{X_{s},\ldots,X_{e}\}$. 
It should be noted that these weights depend on $N$, since they depend on the subsample. 
More specifically, for some interval with length that satisfies $N_{s,e}=O(N)$ we have that 
$F_{*,s,e}\rightarrow \sum_{j=1}^{\ell+1}\widetilde{\vartheta}_j F_j$,  for some $\widetilde{\vartheta}_j\geq 0$, as $N\rightarrow\infty$. 
We also define the quantities
\begin{align*}
    \widetilde{Z}_{s, e}(k/N_{s,e}) &\coloneqq\frac{1}{\sqrt{N_{s, e}}} \sum_{i=1}^{k} \frac{R_{i,s,e}-\left(N_{s, e}+1\right) / 2}{\sqrt{\left(N_{s, e}^{2}-1\right) / 12}}\\
    G_{s, e}(k/N_{s,e}) &\coloneqq\widetilde{Z}_{s, e}(k/N_{s,e})-Z_{s, e}(k/N_{s,e})=\frac{1}{\sqrt{N_{s, e}}} \sum_{i=1}^{k} \frac{R_{i,s,e}-\widehat{R}_{i,s,e}}{\sqrt{\left(N_{s, e}^{2}-1\right) / 12}}. 
\end{align*}
Now, some small fixed $\nu<\Delta$ and for $i\in [\ell]$,  define
\begin{equation}
    D_{\scaled{N},i}=\{\exists\ \mathcal{I}_i=(s_i,e_i)\in \INT\ \colon \nu N<e_i-k_i<\Delta N,\ \nu N<k_i-s_i<\Delta N \}
    \label{eqn::DNi}
\end{equation}
and set
$$D_{\scaled{N}}=\bigcap_{i=1}^{\ell} D_{\scaled{N},i}\ .$$
First, we show that $\Pr(D_{\scaled{N}})\rightarrow 1.$ 
Notice that $D_{\scaled{N},i}$ is the event that there exists some interval which contains $k_i$ and has size that satisfies $2\nu N<N_{s_i,e_i}<2\Delta N$. 
Assumption \ref{ass:numcp} further implies that such an interval does not contain any other true change-points. 
Note that for fixed $i$, the probability that some $\mathcal{I}_i$ as in \eqref{eqn::DNi} is not drawn satisfies
\begin{align*}
\Pr(D_{\scaled{N},i}^c)&\leq\left(1-\frac{(\Delta-\nu)^2 N^2}{\binom{N}{2}}\right)^{J_{\scaled{N}}},
\end{align*}
since there are $2(\Delta-\nu)^2 N^2$ intervals of the desired size. 
We can use this result and sub-additivity of measures to show that,
\begin{align*}
    \limn\Pr(D^c_{\scaled{N}})=\limn\Pr\left(\bigcup_{i=1}^\ell D_{\scaled{N},i}^c\right)\leq\limn\sum_{i=1}^\ell \Pr(D_{\scaled{N},i}^c)\leq\limn\ell\left(1-\frac{(\Delta-\nu)^2 N^2}{\binom{N}{2}}\right)^{J_{\scaled{N}}}= 0.
\end{align*}

Now, define the following event on the appropriate joint probability space of the sample and the execution of Algorithm \ref{alg:wbs}:
\begin{align*}
    A_{\scaled{N}}&=\left\{\max_{s,k,e} |G_{s,e}(k/N_{s,e})|\leq \lambda_{\scaled{N}} \right\}
\end{align*}
where $\lambda_{\scaled{N}}$ is an increasing sequence such that $\lambda_{\scaled{N}}<O(N^{1/2})$. 
First, consider the case where $N_{s,e}=O(N)$. 
By the same reasoning as (A13) on page 439 of \cite{Chenouri2020DD} we have that $$\max_{k,s,e,N_{s,e}=O(N) } |G_{s,e}(k/N_{s, e})|=O_p(1).$$ 
Consider the set of intervals with length bounded above by some fixed constant, i.e.,  $N_{s,e}<C'$. 
It is easily seen from the Markov inequality that
\begin{align*}
\Pr\left(\max_{k,s,e,N_{s,e}<C'}|G_{s,e}(k/N_{s, e})|>\lambda_{\scaled{N}}\right)&\leq\E{}{\max_{s,e,N_{s,e}<C'} \frac{1}{\lambda_{\scaled{N}}\sqrt{N_{s, e}}} \sum_{i=s}^{e-1}\frac{ \left|R_{i,s,e}-\widehat{R}_{i,s,e}\right|}{\sqrt{\left(N_{s, e}^{2}-1\right) / 12}}}\\
&\leq \max_{k,s,e,N_{s,e}<C'}\frac{1}{\lambda_{\scaled{N}}\sqrt{N_{s, e}}} \sum_{i=s}^{e-1}\frac{C'}{\sqrt{\left(N_{s, e}^{2}-1\right) / 12}}\\
&\leq \max_{k,s,e,N_{s,e}<C'}\frac{1}{\lambda_{\scaled{N}}} \sum_{i=s}^{e-1}\E{}{ \left|R_{i,s,e}-\widehat{R}_{i,s,e}\right|}\\
&\leq\frac{1}{\lambda_{\scaled{N}}} (C')^2.
\end{align*}
It then follows that $\limn \Pr(A_{\scaled{N}})=1.$
These results concerning $A_{\scaled{N}}$ and $D_{\scaled{N}}$ allow us to condition on them;
{\footnotesize
\begin{align*}
    \Pr\left(\left\{\hat{\ell}=\ell\right\}\cap \left\{\max _{i \in[\ell]}\left|\hat{k}_{i}-k_{i}\right| \leq C \xi_{N}\right\}\right)&=  \Pr\left(\left\{\hat{\ell}=\ell\right\}\cap \left\{\max _{i \in[\ell]}\left|\hat{k}_{i}-k_{i}\right| \leq C \xi_{N}\right\}\Big|A_{\scaled{N}}\cap D_{\scaled{N}}\right)\Pr(A_{\scaled{N}} \cap D_{\scaled{N}})\\
    &\geq \Pr\left(\left\{\hat{\ell}=\ell\right\}\cap \left\{\max _{i \in[\ell]}\left|\hat{k}_{i}-k_{i}\right| \leq C \xi_{N}\right\}\Big|A_{\scaled{N}}\cap D_{\scaled{N}}\right)(\Pr(A_{\scaled{N}})+\Pr(D_{\scaled{N}})-1),
\end{align*}
}
which means that it suffices to show that for all $0<\epsilon<1$, there exists $n$ such that for all $N>n$
$$\Pr\left(\left\{\hat{\ell}=\ell\right\}\cap \left\{\max _{i \in[\ell]}\left|\hat{k}_{i}-k_{i}\right| \leq C \xi_{N}\right\}\Big|A_{\scaled{N}}\cap D_{\scaled{N}}\right)>1-\epsilon.$$  
We start by analyzing the event $\{\hat{\ell}< \ell\}$. 
The event $\{\hat{\ell}< \ell\}$ implies that there is at least one unidentified change point. 
Suppose that $k_{i^*}$, for some fixed $i^*\in[\ell]$, is an unidentified change-point. 
The interval $\mathcal{I}_{i^*}$ as defined in \eqref{eqn::DNi} satisfies 
\begin{equation}
N_{s_{i^*},e_{i^*}}=O(N).
\label{eqn::int_on}
\end{equation}
Let 
$$\hat{k}=\min\left\{k\colon Z_{s_{i^*},e_{i^*}}(\floor{k/N_{s_{i^*},e_{i^*}}})=\sup_t Z_{s_{i^*},e_{i^*}}(t),\ t\in(0,1),\ k\in[N_{s_{i^*},e_{i^*}}]\right\},$$
and $\theta$ be the true break-fraction in this interval. 

Further, recall that Assumption \ref{ass:thresh} says that the threshold satisfies $T= o(\sqrt{N})$. 
This assumption, combined with page 437 of \cite{Chenouri2020DD} implies that 
$$|Z_{s_{i^*},e_{i^*}}(\widehat{k}/N_{s_{i^*},e_{i^*}})/T|\conp \infty.$$ 
It follows that for any $\delta' <1$ there exists $n'$ such that for $N>n'$ $$\Pr(|Z_{s_{i^*},e_{i^*}}(\widehat{k}/N_{s_{i^*},e_{i^*}})|\geq T)>1-\delta'.$$
Thus, when $N>\max(n,n')$ (where $n$  and $n'$ relate to the above argument)
\begin{equation}
     \Pr(\{\text{change-point } k_{i^*} \text{ is undetected}\}|D_{\scaled{N}}\cap A_{\scaled{N}})\leq \epsilon.
     \label{eqn::od}
\end{equation}
For each true change point that is undetected the above analysis applies, since we have conditioned on $D_{\scaled{N},i}$ occurring for all $i\in [\ell]$. 
Thus,
\begin{align*}
  \limn  \Pr\left(\hat{\ell}<\ell|D_{\scaled{N}}\cap A_{\scaled{N}}\right)&=\limn \Pr\left(\bigcup_{i=1}^\ell \{\text{change-point }k_i\ \text{undetected}\}|D_{\scaled{N}}\cap A_{\scaled{N}}\right)\\
  &\leq \limn \sum_{i=1}^\ell \Pr(\{\text{change-point } k_i \text{ is undetected}\}|D_{\scaled{N}}\cap A_{\scaled{N}})=0,
\end{align*}
where the inequality follows from subadditivity of measures and the last equality follows from \eqref{eqn::od} and the fact that $\ell$ does not depend on $N$. 
Now, it follows directly from the arguments in the proof of Theorem 2.1 of \cite{Chenouri2019} that there exists $n$ such that for $N>n$ we have that 
$$\Pr(\hat{\ell}>\ell|D_{\scaled{N}}\cap A_{\scaled{N}})<\epsilon ',$$ 
for any $\epsilon '>0$. To conclude we have that $$\Pr(\hat{\ell}\neq\ell|D_{\scaled{N}}\cap A_{\scaled{N}})\rightarrow 0\ as\ N\rightarrow \infty.$$

Now, consider the event that $\{ \max_{i\in [\ell]}|\widehat{k}_i-k_i|\leq C\ N^\phi\}$. 
Following the argument in \cite{Chenouri2019}, subadditivity of measures gives that
\begin{align*}
    \limn \Pr\left( \max_{i\in [\ell]}|\widehat{k}_i-k_i|\leq C\ N^\phi\right)&= \limn\left( 1-\Pr\left(\bigcup_{i=1}^\ell\left\{|\widehat{k}_i-k_i|\geq C\ N^\phi\right\}\right)\right)\\
    &\geq \limn\left(1- \sum_{i=1}^\ell\Pr\left( |\widehat{k}_i-k_i|\geq C\ N^\phi\right)\right)\\
    &= 1,
\end{align*}
where the last equality follows from the proof of Theorem \ref{thm:PELT} and the fact that $\ell$ is fixed and the fact that $\frac{1}{2}<\phi<1$. 
In more detail, in the case of one change-point, the objective functions $|Z_{1,\scaled{N}}(t)|$ and $\mathcal{C}(\mathbf{k})$ produce equivalent maximizers. 
This fact, combined with conditioning on the events $A_N$ and $D_N$ and the last two paragraphs of the proof of Theorem \ref{thm:PELT} give that 
$\Pr\left(|\widehat{k}_i-k_i|\leq C\ N^\phi\right)\rightarrow 1$. 
We now have that, for any $ \epsilon''>0$ there exists $n''$ such that for all $N>n''$, we have that 
$$\Pr\left( \max_{i\in [\ell]}|\widehat{k}_i-k_i|\leq C\ N^\phi\right)\geq 1-\epsilon''.$$
Now, both events can be combined with Bonferroni's inequality and we can make the statement that for all $0<\epsilon^*=\ell\epsilon+\epsilon''<1$ there must exist some $n^*=\max(n'',n',n)$, such that for $N>n^*$ we have that 
\begin{align*}
    \Pr\left(\left\{\hat{\ell}=\ell\right\}\cap\left\{ \max_{i\in [\ell]}|\widehat{k_i}-k_i|\leq C\ N^\phi\}\right\}\Big|A_{\scaled{N}}\cap  D_{\scaled{N}}\right)\geq 1-\ell\epsilon+1-\epsilon'-1=1-\epsilon^*.
\end{align*}
Thus, we have that for all $0<\epsilon<1$, there exists $n$ such that for all $N>n$ 
\begin{align*}
      \Pr\left(\left\{\hat{\ell}=\ell\right\}\cap\left\{ \max_{i\in [\ell]}|\widehat{k_i}-k_i|\leq C\ N^\phi\}\right\}\right)&\geq 1-\epsilon. \hfill\qedhere
\end{align*}

\end{proof}
\begin{proof}[Proof of Theorem \ref{thm:PELT}]
Let $C_i$ be fixed positive constants independent of $N$, $|A|$ represent the cardinality of the set $A$, and $$\widetilde{\sigma}^2_{\scaled{N}}\coloneqq\frac{N(N+1)}{12}.$$
Define the set $\mathbf{X}_{\scaled{N}}\coloneqq 2^{[N-1]}\times\{0\}\times\{N\}$; elements of $\mathbf{X}_{\scaled{N}}$ are sets of indices ranging from 0 to $N$, which represent locations of change-points. 
A member of $\mathbf{X}_{\scaled{N}}$ is a set $\mathbf{x}$ that contains 0 and $N$ joined with an element of the power set of $[N-1]$. 
We will represent such an element with $\mathbf{x}=\{x_0,\ldots,x_{p+1}\}$ where $x_0\coloneqq 0 <x_1<\ldots<x_p<x_{p+1}\coloneqq N$. 
$\mathbf{X}_{\scaled{N}}$ forms the space of possible sets of change-points for a fixed $N$. 
We can then write the objective function based on the population depth ranks $\mathcal{T}$ and the objective function based on the sample depth ranks $\widehat{\mathcal{T}}$ as follows
\begin{equation*}
    \widehat{\mathcal{T}}(\mathbf{x})\coloneqq\frac{1}{\widetilde{\sigma}^2_{\scaled{N}}}\sum_{i=1}^{|\mathbf{x}|} (x_{i}-x_{i-1}) \overline{\widehat{R}}_{i}^{2}-3(N+1) -\beta_{\scaled{N}} (|\mathbf{x}|-1)\coloneqq\widehat{\mathcal{C}}(\mathbf{x})-\beta_{\scaled{N}} (|\mathbf{x}|-1)
\end{equation*}
\begin{align}
     \label{eqn:T}
    \mathcal{T}(\mathbf{x})&\coloneqq \frac{1}{\widetilde{\sigma}^2_{\scaled{N}}} \sum_{i=1}^{|\mathbf{x}|} (x_{i}-x_{i-1}) \bar{R}_{i}^{2}-3(N+1) -\beta_{\scaled{N}}
    (|\mathbf{x}|-1)\coloneqq\mathcal{C}(\mathbf{x})-\beta_{\scaled{N}} (|\mathbf{x}|-1),
\end{align}
where $\mathbf{x}_{\scaled{N}}\in \mathbf{X}_{\scaled{N}}$. 
Now, suppose that $\mathbf{x}_{\scaled{N}}\in \mathbf{X}_{\scaled{N}}$ is such for each $x_j\in \mathbf{x}_{\scaled{N}}\backslash 0$, it holds that $x_j-x_{j-1}=O(N)$ and there exists some $k_i>0,\ k_i\in \mathbf{k}$ such that $k_{i-1}\leq x_{j-1}<x_j\leq k_{i}$. 
Colloquially, there are no change-points between neighboring elements of $\mathbf{x}_{\scaled{N}}$. 
Additionally impose that $|\mathbf{x}_{\scaled{N}}|$ is fixed in $N$ . 
It is helpful to note that the elements of $\mathbf{x}_{\scaled{N}}$ depend on $N$, which we omit in the notation for brevity.  
First, we show that $|\widehat{\mathcal{T}}(\mathbf{x}_{\scaled{N}})-\mathcal{T}(\mathbf{x}_{\scaled{N}})|=O_p(1) $. 
To this end, note that for any $j\in [|\mathbf{x}_N|]$ the sequences $R_{x_{j-1}+1},\ldots, R_{x_{j}}$ and $\widehat{R}_{x_{j-1}+1},\ldots, \widehat{R}_{x_{j}}$ are both triangular arrays of exchangeable random variables. 
This form allows us to apply the central limit theorem of \cite{Weber1980}. 
Specifically, it holds that
\begin{align*} 
\frac{\sqrt{(x_i-x_{i-1})}}{\Var{R_{x_i}}}(\overline{R}_{i}-\E{}{R_{x_i}})=O_p(1) \qquad\text{ and }\qquad \frac{\sqrt{(x_i-x_{i-1})}}{\Varr{\widehat{R}_{x_i}}}\left(\overline{\widehat{R}}_{i}-\Eee{}{\widehat{R}_{x_i}}\right)=O_p(1).
\end{align*}
We now relate these to quantities. 
Consider the representation of $\widehat{R}_{i}$
\begin{align}
\label{eqn::rankrep}
    \widehat{R}_{i}&=R_{i}+\sum_{m=1}^N\ind{B_{{i},m}}-\sum_{m=1}^N\ind{A_{{i},m}}\coloneqq R_{i}+\mathcal{E}_i,
\end{align}
where 
\begin{align*} A_{i, j} &=\left\{D\left(X_{j}, F_{*}\right) \leq D\left(X_{i}, F_{*}\right)\right\} \cap\left\{D\left(X_{j}, F_{*,\scaled{N}}\right)>D\left(X_{i}, F_{*,\scaled{N}}\right)\right\} \\ B_{i, j} &=\left\{D\left(X_{j}, F_{*}\right)>D\left(X_{i}, F_{*}\right)\right\} \cap\left\{D\left(X_{j}, F_{*,\scaled{N}}\right) \leq D\left(X_{i}, F_{*,\scaled{N}}\right)\right\}.
\end{align*}
We can use this representation, Assumption \ref{ass:Lipschitz} and Assumption \ref{ass:consDepth} to show that
$$\Eee{}{\mathcal{E}_{x_i}}=\Eee{}{\widehat{R}_{x_i}}-\Eee{}{R_{x_i}}=O(N^{1/2}).$$ 
For more details, see pages 436-437 of \cite{Chenouri2020DD}. 
We next show that 
\begin{equation}
\label{eqn::varsame}
    \Varr{\widehat{R}_{x_i}}/\Varr{R_{x_i}}=O(1)\qquad\text{ and }\qquad \Varr{R_{x_i}}/\widetilde{\sigma}^2_{\scaled{N}}=O(1).
\end{equation}
The right-side identity follows easily from Assumption \ref{ass:numcp}; $\Varr{R_{i}}=O(N^2),$ for any $i\in [N]$. 
Using \eqref{eqn::rankrep}, we can write
\begin{align*}
\Var{\widehat{R}_{x_i}}&=\Var{R_{x_i}+\mathcal{E}_{x_i}}\\
&=\Var{R_{x_i}}+\Var{\mathcal{E}_{x_i}}+2\mathrm{C}ov\left(\mathcal{E}_{x_i},R_{x_i}\right)\\
&\leq \Var{R_{x_i}}+\Var{\mathcal{E}_{x_i}}+2\E{}{|\mathcal{E}_{x_i}-\E{}{\mathcal{E}_{x_i}}|}N\\
&= \Var{R_{x_i}}+\Var{\mathcal{E}_{x_i}}+O(N^{3/2})\\
&=\Var{R_{x_i}}+\E{}{\left(\sum_{m=1}^N\ind{B_{{x_i},m}}-\sum_{m=1}^N\ind{A_{{x_i},m}}\right)^2}+O(N)+O(N^{3/2})\\
&=\Var{R_{x_i}}+\E{}{\sum_{m_1=1}^N\sum_{m_2=1}^N\left[\ind{B_{{x_i},m_1}}-\ind{A_{{x_i},m_1}}\right]\left[\ind{B_{{x_i},m_2}}-\ind{A_{{x_i},m_2}}\right]}+O(N^{3/2})\\
&\leq\Var{R_{x_i}}+\E{}{\sum_{m_1=1}^N\sum_{m_2=1}^N\left[\ind{B_{{x_i},m_1}}+\ind{A_{{x_i},m_1}}\right]}+O(N^{3/2})\\
&\leq \Var{R_{x_i}}+O(N^{3/2}),
\end{align*}
where the fourth line comes from applying equation (A5) of \cite{Chenouri2020DD} and the last line is from the the fact that $\E{}{\ind{B_{i,m}}}=O(N^{-1/2})$ and $\E{}{\ind{A_{i,m}}}=O(N^{-1/2})$ \citep{Chenouri2020DD}. 
Now, 
$$\limn \frac{\Varr{\widehat{R}_{x_i}}}{\Varr{R_{x_i}}}=\limn\frac{\Varr{\widehat{R}_{x_i}}/N^2}{\Varr{R_{x_i}}/N^2}=\limn\frac{\Varr{R_{x_i}}/N^2+o(1)}{\Varr{R_{x_i}}/N^2}=1.$$
It then follows from Slutsky's theorem, continuous mapping theorem and the central limit theorem of \cite{Weber1980} that
\begin{align}
    \widehat{\mathcal{T}}(\mathbf{x}_{\scaled{N}})-\mathcal{T}(\mathbf{x}_{\scaled{N}})&=\frac{1}{\widetilde{\sigma}^2_{\scaled{N}}}\sum_{i=1}^{\ell+2} (x_i-x_{i-1})\left(\overline{\widehat{R}}_{i}^{2}-\overline{R}_{i}^{2}\right)\nonumber\\
    &= \sum_{i=1}^{|\mathbf{x}_{\scaled{N}}|} \left( \frac{\sqrt{(x_i-x_{i-1})}\ \overline{\widehat{R}}_{i}}{\widetilde{\sigma}_{\scaled{N}}}\right)^2-\left( \frac{\sqrt{(x_i-x_{i-1})}\overline{R}_{i}}{\widetilde{\sigma}_{\scaled{N}}}\right)^2\nonumber\\
    &=O_p(1)+ \frac{1}{\widetilde{\sigma}^2_{\scaled{N}}}\sum_{i=1}^{|\mathbf{x}_{\scaled{N}}|} [ (x_i-x_{i-1})\Eee{}{\widehat{R}_{x_i}}^2-\Eee{}{R_{x_i}}^2+\Eee{}{R_{x_i}}\overline{R}_{i}-\Eee{}{\widehat{R}_{x_i}}\overline{\widehat{R}}_{i}]\nonumber\\
    &=O_p(1).\nonumber
\end{align}
This analysis gives the result that 
\begin{equation}
\label{eqn::cost_pop_same}
    \widehat{\mathcal{T}}(\mathbf{x}_{\scaled{N}})-\mathcal{T}(\mathbf{x}_{\scaled{N}})=\widehat{\mathcal{C}}(\mathbf{x}_{\scaled{N}})-\mathcal{C}(\mathbf{x}_{\scaled{N}})=O_p(1).
\end{equation}
Note that if there are some $x_j\in\mathbf{x}_{\scaled{N}} $ such that $x_j-x_{j-1}<C_1$ for some constant $C_1>0$, the above result still holds. 

Next, we want to compare $\widehat{\mathcal{T}}(\widehat{\mathbf{k}})$ and $\mathcal{T}(\mathbf{k})$. 
To this end, we make an argument by contradiction, similar to that of \cite{Wang2021}. 
However, we use the previously discussed exchangeablility results, i.e., \citep{Weber1980} which were not used in their paper. 
Recall, $\widehat{\mathbf{k}}$ is the estimated set of change-points and $\mathbf{k}$ is the true set of change-points. 
We examine the events $\{\hat{\ell}<\ell\}$, $\{\hat{\ell}>\ell\}$ and $\left\{\max_{k\in \mathbf{k}}\min_{\hat{k}\in \widehat{\mathbf{k}}} |\hat{k}-k|\geq\delta N^\phi\right\}$ separately.

Assume $\hat{\ell}<\ell$; by Assumption \ref{ass:numcp}, there is at least one change-point $0<k_{i^*}<N$ such that for any $ j\in [\hat{\ell}]$ it is true that $|k_{i^*}-\widehat{k}_j|\geq \Delta N/2$ with $\Delta$ independent of $N$. 
Now, define $$\mathbf{w}_1=\{k_{i^*}-\Delta N/2,k_{i^*}+\Delta N/2\}\cup \mathbf{k}\backslash k_{i^*}\qquad \text{and}\qquad \mathbf{w}_2=\mathbf{w}_1\cup \widehat{\mathbf{k}}.$$ 
Clearly, $\widehat{\mathcal{C}}(\mathbf{w}_2)\geq \widehat{\mathcal{C}}(\widehat{\mathbf{k}})$ (which is the necessary condition for PELT, recall that $\widehat{\mathcal{C}}$ is the portion of the objective function without the penalty) and so we work with $\widehat{\mathcal{C}}(\mathbf{w}_2)$. 
The goal is to show that following contradiction to the assumption that some $\widehat{\mathbf{k}}$ such that $\hat{\ell}<\ell$ is the maximizer of $\widehat{\mathcal{T}}$. 
To see this, we have 
\begin{align*}
   \mathcal{T}(\mathbf{k})- \widehat{\mathcal{T}}(\widehat{\mathbf{k}})&=\mathcal{C}(\mathbf{k})- \widehat{\mathcal{C}}(\widehat{\mathbf{k}})-O(\beta_{\scaled{N}})\\
   &\geq \mathcal{C}(\mathbf{k})- \widehat{\mathcal{C}}(\mathbf{w}_2)-O(\beta_{\scaled{N}})\\
   &= \mathcal{C}(\mathbf{k})- \mathcal{C}(\mathbf{w}_2)+O_p(1)-O(\beta_{\scaled{N}})\\
   &=\mathcal{C}(\mathbf{k})- \mathcal{C}(\mathbf{w}_1)+O_p(1)-O(\beta_{\scaled{N}})\\
   &=O_p(N)-O(\beta_{\scaled{N}})\conp\infty,
\end{align*}
as $N\rightarrow\infty$, since $\beta_{\scaled{N}}<O(N)$ and we have shown that $\mathcal{C}(\mathbf{w}_2)-\widehat{\mathcal{C}}(\mathbf{w}_2)=O_p(1)$ in \eqref{eqn::cost_pop_same}. 
It remains to show that $$\mathcal{C}(\mathbf{w}_2)=\mathcal{C}(\mathbf{w}_1)+O_p(1)\qquad\text{and}\qquad \mathcal{C}(\mathbf{k})- \mathcal{C}(\mathbf{w}_1)=O_p(N).$$ 
First, we show that $$\mathcal{C}(\mathbf{w}_2)=\mathcal{C}(\mathbf{w}_1)+O_p(1).$$
To this end, letting $w_0=0,\ w_{\ell+\hat{\ell}+2}=N$ and $\mathbf{w}_2=\{w_0,w_1,w_2,\dots,w_{\ell+\hat{\ell}+1},w_{\ell+\hat{\ell}+2}\}$ where $w_m<w_j$ for $m<j$, we can write
\begin{align*}
    \mathcal{C}(\mathbf{w}_1)-\mathcal{C}(\mathbf{w}_2)&=\frac{1}{\widetilde{\sigma}^2_{\scaled{N}}}\sum_{j=1}^{|\mathbf{w_2}|} (w_j-w_{j-1})\left[\bar{R}_j(\mathbf{w}_1)^2-\bar{R}_j(\mathbf{w}_2)^2\right]
\end{align*}
where $$\bar{R}_j(\mathbf{x})=\frac{1}{n_{j,2}(\mathbf{x})-n_{j,1}(\mathbf{x})}\sum_{i=n_{j,1}(\mathbf{x})+1}^{n_{j,2}(\mathbf{x})} R_i,$$
with
$$  n_{j,1}(\mathbf{x})=\argmin_{x\in \mathbf{x}\colon x\leq w_{j-1}}|x-w_{j-1}|,\qquad  n_{j,2}(\mathbf{x})=\argmin_{x\in \mathbf{x}\colon x\geq w_{j}}|x-w_{j}|.$$
In this context, 
$$\bar{R}_j(\mathbf{w}_2)=\frac{1}{(w_j-w_{j-1})}\sum_{m=w_{j-1}+1}^{w_j} R_m \qquad\text{and}\qquad\bar{R}_j(\mathbf{w}_1)=\frac{1}{n_{j,2}(\mathbf{w}_1)-n_{j,1}(\mathbf{w}_1)}\sum_{m=n_{j,1}(\mathbf{w}_1)+1}^{n_{j,2}(\mathbf{w}_1)} R_m\ .$$
To elaborate, ordering the points in $\mathbf{w}_1$ defines $\ell+2$ disjoint groups of ranks and therefore $\ell+2$ group means. 
The value $\bar{R}_j(\mathbf{w}_1)$ is the mean of such a group of ranks which also contains the ranks $\{R_{w_{j-1}},\ldots,R_{w_{j}}\}$.    

Let $j^*$ represent $w_{j^*}=k_{i^*}+\Delta N/2$. Then we have that 
\begin{align}
    \mathcal{C}(\mathbf{w}_1)-\mathcal{C}(\mathbf{w}_2)&=\frac{1}{\widetilde{\sigma}^2_{\scaled{N}}}\sum_{j=1}^{|\mathbf{w}_2|} (w_j-w_{j-1})\left(\bar{R}_j(\mathbf{w}_1)^2-\bar{R}_j(\mathbf{w}_2)^2\right)\nonumber\\ 
        \label{eqn:sum1}
    &=\frac{1}{\widetilde{\sigma}^2_{\scaled{N}}}\sum_{j\in [\ell+\hat{\ell}+1] \backslash {j^*}} (w_j-w_{j-1})\left(\bar{R}_j(\mathbf{w}_1)^2-\bar{R}_j(\mathbf{w}_2)^2\right)\\
    &=O_p(1),
\end{align}
where the second equality is due to the fact that  $\bar{R}_{j^*}(\mathbf{w}_1)^2=\bar{R}_{j^*}(\mathbf{w}_2)^2$ and the last equality follows from the central limit theorem of \cite{Weber1980} and the analysis of $\mathcal{T}(\mathbf{x}_{\scaled{N}})-\widehat{\mathcal{T}}(\mathbf{x}_{\scaled{N}})$. 
To elaborate, note that for any $j\neq j^*$, if $w_j-w_{j-1}=O(N)$ it holds that
{
\begin{align*}
    \frac{(w_j-w_{j-1})}{\widetilde{\sigma}^2_{\scaled{N}}}\left(\bar{R}_j(\mathbf{w}_1)^2-\bar{R}_j(\mathbf{w}_2)^2\right) = O(1)\frac{(w_j-w_{j-1})}{\Var{R_{w_j}}}\left(\bar{R}_j(\mathbf{w}_1)^2-\bar{R}_j(\mathbf{w}_2)^2\right) = O_p(1),
\end{align*}
}
where the first equality follows from \eqref{eqn::varsame} and the second equality comes from a direct application of the central limit theorem of \cite{Weber1980} followed by Slutsky's Lemma and continuous mapping theorem. 
If $w_j-w_{j-1}<C_2$ for some $C_2>0$ then 
\begin{align*}
    \frac{(w_j-w_{j-1})}{\widetilde{\sigma}^2_{\scaled{N}}}\left(\bar{R}_j(\mathbf{w}_1)^2-\bar{R}_j(\mathbf{w}_2)^2\right)=  o(1).
\end{align*}
Now, we want to show that $$ \limn   \mathcal{C}(\mathbf{k})- \mathcal{C}(\mathbf{w}_1)=O_p(N).$$ 
Let $k_{i^*-1}$ and $k_{i^*+1}$ be the true change-points immediately preceding and following $k_{i^*}$ respectively. 
Recall $k_{i^*}$ is the change-point that is at least $\Delta N/2$ points away from any estimated change-point. 
Note that $\mathbf{k}-\mathbf{w}_1=\{k_{i^*}\}$ and  $\mathbf{w}_1-\mathbf{k}=\{k_{i^*}\pm \Delta N/2\}$.
We have
{\small
\begin{align*}
   \frac{N+1}{N}\left(\mathcal{C}(\mathbf{k})-\mathcal{C}(\mathbf{w}_1)\right) 
   &=\frac{12\vartheta_{i^*}N}{N^2}\left[\frac{1}{N\vartheta_{i^*}}\sum_{j=k_{i^* \text{-}1}+1}^{k_{i^*}}R_j\right]^2+\frac{12\vartheta_{i^*+1}N}{N^2}\left[\frac{1}{N\vartheta_{i^*+1}}\sum_{j=k_{i^*}+1}^{k_{i^*+1}}R_j\right]^2\\
    &\indent-\frac{12\Delta N}{N^2}\left[\frac{1}{N\Delta}\sum_{j=k_{i^*}\text{-}\Delta N/2}^{k_{i}+\Delta N/2} R_j\right]^2-\frac{12N(\vartheta_{i^*}-\Delta/2)}{N^2}\left[\frac{1}{N(\vartheta_{i^*}-\Delta/2)}\sum_{j=k_{i^* \text{-}1}+1}^{k_{i^*}-\Delta N/2}R_j\right]^2\\
    &\indent-\frac{12N(\vartheta_{i^*+1}-\Delta/2)}{N^2}\left[\frac{1}{N(\vartheta_{i^*+1}-\Delta/2)}\sum_{j=k_{i^*}+\Delta N/2}^{k_{i^*+1}}R_j\right]^2.
    \end{align*}
    }
For arbitrary $k_m\in \mathbf{k}$ choose $j\in \{k_{m-1}+1,\dots,k_m\}$, then
\begin{align*}
    \E{}{R_j}&=\sum_{j\in [\ell+1]\backslash m }N\vartheta_j p_{m,j}-\frac{N\vartheta_i-1}{2}=N\left[\sum_{j=1}^{\ell+1}\vartheta_j p_{m,j}-\frac{1}{2}\right]\\
    \Var{R_j}&\leq N-1+N(N-1)/2.
\end{align*}
It follows from continuous mapping theorem and \citep{Weber1980} that
\begin{align*}
    \frac{1}{N^2}\left[\frac{1}{N\vartheta_i}\sum_{j=k_{i^*-1}+1}^{k_{i^*}}R_j\right]^2&\conp \left[\sum_{j=1}^{\ell+1}\vartheta_j p_{i^*,j}-\frac{1}{2}\right]^2,\\%
    \frac{1}{N^2}\left[\frac{1}{N(\vartheta_i-\Delta /2)}\sum_{j=k_{i^*-1}+1}^{k_{i^*}-\Delta N/2}R_j\right]^2&\conp \left[\sum_{j=1}^{\ell+1}\vartheta_j p_{i^*,j}-\frac{1}{2}\right]^2,\\%
    \frac{1}{N^2}\left[\frac{1}{N\vartheta_{i+1}}\sum_{j=k_{i^*-1}+1}^{k_{i^*}}R_j\right]^2&\conp \left[\sum_{j=1}^{\ell+1}\vartheta_j p_{i^*+1,j}-\frac{1}{2}\right]^2,\\%
    \frac{1}{N^2}\left[\frac{1}{N(\vartheta_{i+1}-\Delta /2)}\sum_{j=k_{i^*}+\Delta N/2}^{k_{i^*+1}}R_j\right]^2&\conp \left[\sum_{j=1}^{\ell+1}\vartheta_j p_{i^*+1,j}-\frac{1}{2}\right]^2,\\%
    \frac{1}{N^2}\left[\frac{1}{N\Delta}\sum_{j=k_{i^*}-\Delta N/2}^{k_{i^*}+\Delta N/2}R_j\right]^2&\conp \frac{1}{4}\left[\sum_{j=1}^{\ell+1}\vartheta_j p_{i^*,j}-\frac{1}{2}+\sum_{j=1}^{\ell+1}\vartheta_j p_{i^*+1,j}-\frac{1}{2}\right]^2.
\end{align*}
Slutsky's lemma and continuous mapping theorem directly imply that 
{\footnotesize
\begin{align*}\frac{N+1}{N^2}\left(\mathcal{C}(\mathbf{k})-\mathcal{C}(\mathbf{w}_1)\right)&\conp \frac{12\Delta}{4}\Biggg[\left(\sum_{j=1}^{\ell+1}\vartheta_j p_{i^*,j}-\frac{1}{2}\right)^2+\left(\sum_{j=1}^{\ell+1}\vartheta_j p_{i^*+1,j}-\frac{1}{2}\right)^2 -2\left(\sum_{j=1}^{\ell+1}\vartheta_j p_{i^*+1,j}-\frac{1}{2}\right)\left(\sum_{j=1}^{\ell+1}\vartheta_j p_{i^*,j}-\frac{1}{2}\right)\Biggg]\\
&=3\Delta\left[\sum_{j=1}^{\ell+1}\vartheta_j p_{i^*+1,j}-\frac{1}{2}-\sum_{j=1}^{\ell+1}\vartheta_j p_{i^*,j}+\frac{1}{2}\right]^2>0.
\end{align*}
}
We can then conclude that $\mathcal{C}(\mathbf{k})-\mathcal{C}(\mathbf{w}_1)\rightarrow +\infty$ in probability at a rate of $O_p(N)$.
Then, we have that 
$$\mathcal{T}(\mathbf{k})-\widehat{\mathcal{T}}(\widehat{\mathbf{k}})= O_p(N)-\beta_{\scaled{N}} \rightarrow\infty,$$
providing a contradiction to the assumption that $\hat{\ell}<\ell$. 

Now assume that $\hat{\ell}>\ell$. 
It is easy to see that $\widehat{\mathcal{C}}(\widehat{\mathbf{k}})\leq \widehat{\mathcal{C}}(\widehat{\mathbf{k}}\cup \mathbf{k})$. 
Using this fact and a similar analysis as to that of the event $\{\hat{\ell}<\ell\}$, we can write that $$\mathcal{C}(\mathbf{k})-\widehat{\mathcal{C}}(\widehat{\mathbf{k}})\geq\mathcal{C}(\mathbf{k})- \widehat{\mathcal{C}}(\widehat{\mathbf{k}}\cup \mathbf{k})=O_p(1).$$
We then have that 
\begin{align*}
    \mathcal{T}(\mathbf{k})-\widehat{\mathcal{T}}(\widehat{\mathbf{k}})&=\mathcal{C}(\mathbf{k})-\widehat{\mathcal{C}}(\widehat{\mathbf{k}}) +\beta_{\scaled{N}}(\hat{\ell}-\ell)\geq \mathcal{C}(\mathbf{k})-\widehat{\mathcal{C}}(\widehat{\mathbf{k}}\cup \mathbf{k})+\beta_{\scaled{N}}(\hat{\ell}-\ell)=O(\beta_{\scaled{N}})+O_p(1)\rightarrow\infty,
\end{align*}
as $N\rightarrow\infty$.

Lastly, we want to show that $\max_{k\in\mathbf{k}}\min_{\hat{k}\in \widehat{\mathbf{k}}} \frac{1}{N}|\hat{k}-k|\conp 0$. We take the contradiction approach again; consider there exists $k_{i^*}\in \mathbf{k}$ such that $\min_{k\in\mathbf{k}}|\widehat{k}-k_{i^*}| >\delta N^\phi$. 
Define $\mathbf{w}'_1$ in the same way as $\mathbf{w}_1$ but replace $\Delta$ with $\delta$: 
$$\mathbf{w}'_1=\{k_{i^*}-\delta N^\phi/2,k_{i^*}+\delta N^\phi/2\}\cup \mathbf{k}\backslash k_{i^*}\qquad \text{and}\qquad \mathbf{w}'_2=\mathbf{w}'_1\cup \widehat{\mathbf{k}}.$$ 
Similar to the analysis of $\{\hat{\ell}<\ell\}$, we can write 
{\footnotesize
\begin{align*}
   \frac{N+1}{N}\left(\mathcal{C}(\mathbf{k})-\mathcal{C}(\mathbf{w}_1')\right) 
   &=\frac{12\vartheta_{i^*}N}{N^2}\left[\frac{1}{N\vartheta_{i^*}}\sum_{j=k_{i^* \text{-}1}+1}^{k_{i^*}}R_j\right]^2+\frac{12\vartheta_{i^*+1}N}{N^2}\left[\frac{1}{N\vartheta_{i^*+1}}\sum_{j=k_{i^*}+1}^{k_{i^*+1}}R_j\right]^2\\
    &\indent-\frac{12\delta N^\phi}{N^2}\left[\frac{1}{\delta N^\phi}\sum_{j=k_{i^*}\text{-}\delta N^\phi/2}^{k_{i}+\delta N^\phi/2} R_j\right]^2-\frac{12N(\vartheta_{i^*}-\delta N^{\phi-1}/2)}{N^2}\left[\frac{1}{N(\vartheta_{i^*}-\delta N^{\phi-1}/2)}\sum_{j=k_{i^* \text{-}1}+1}^{k_{i^*}-\delta N^\phi/2}R_j\right]^2\\
    &\indent-\frac{12N(\vartheta_{i^*+1}-\delta N^{\phi-1}/2)}{N^2}\left[\frac{1}{N(\vartheta_{i^*+1}-\delta N^{\phi-1}/2)}\sum_{j=k_{i^*}+\delta N^\phi/2}^{k_{i^*+1}}R_j\right]^2.
    \end{align*}
    }
We can let $\mu_{k_i^*-1}=\E{}{R_{k_i^*}}$ and $\mu_{k_i^*}=\E{}{R_{k_i^*+1}}$ be the means of the ranks before and after the change-point $k_{i^*}$. 
Similarly, we can let $\varsigma^2_{k_i^*-1}=\Var{R_{k_i^*}}$ and $\varsigma^2_{k_i^*}=\Var{R_{k_i^*+1}}$ be the variances of the ranks before and after the change-point $k_{i^*}$. 
We also define $b_{N,k_i^*-1}=\frac{12\varsigma^2_{k_i^*-1}}{N^2}=O(1)$ and $b_{N,k_i^*}=\frac{12\varsigma^2_{k_i^*}}{N^2}=O(1)$. 
Now, let $$\tilde{R}_j=\frac{R_j-\E{}{R_{j}}}{\Var{R_j}}.$$
It follows that 
\begin{align*}
    \frac{12\vartheta_{i^*}N}{N^2}\left[\frac{1}{N\vartheta_{i^*}}\sum_{j=k_{i^* \text{-}1}+1}^{k_{i^*}}R_j\right]^2&=b_{N,k_i^*}\left[\sqrt{N\vartheta_{i^*}}\sum_{j=k_{i^* \text{-}1}+1}^{k_{i^*}}\tilde{R}_j+\sqrt{N\vartheta_{i^*}}\frac{\mu_{k_i^*}}{\varsigma_{k_i^*}}\right]^2\\
    &=b_{N,k_i^*}N\vartheta_{i^*}\left(\frac{\mu_{k_i^*}}{\varsigma_{k_i^*}}\right)^2+O_p(N^{1/2}). 
\end{align*}
The last line follows from the central limit theorem of \citep{Weber1980} and the previous paragraph. 
we can produce similar analyses to give 
\begin{align*}
  \frac{N+1}{N}\left(\mathcal{C}(\mathbf{k})-\mathcal{C}(\mathbf{w}_1')\right) 
   &=b_{N,k_i^*-1}N\vartheta_{i^*}\left(\frac{\mu_{k_i^*-1}}{\varsigma_{k_i^*-1}}\right)^2+b_{N,k_i^*}N\vartheta_{i^*+1}\left(\frac{\mu_{k_i^*}}{\varsigma_{k_i^*}}\right)^2-\frac{\delta N^{\phi}}{4}b_{N,k_i^*-1}\left(\frac{\mu_{k_i^*-1}}{\varsigma_{k_i^*-1}}\right)^2\\
   &\indent -\frac{\delta N^{\phi}}{4}b_{N,k_i^*}\left(\frac{\mu_{k_i^*}}{\varsigma_{k_i^*}}\right)^2-b_{N,k_i^*-1}(\vartheta_{i^*}-\delta N^{\phi-1}/2)\left(\frac{\mu_{k_i^*-1}}{\varsigma_{k_i^*-1}}\right)^2\\
    &\indent-b_{N,k_i^*}N(\vartheta_{i^*+1}-\delta N^{\phi-1}/2)\left(\frac{\mu_{k_i^*}}{\varsigma_{k_i^*}}\right)^2+O_p(N^{1/2})\conp\infty,
\end{align*}
where the conclusion follows from the fact that $\phi>1/2$ and the right expression is positive. 
From which it follows that
\begin{align*}
   \mathcal{T}(\mathbf{k})- \widehat{\mathcal{T}}(\widehat{\mathbf{k}})&=\mathcal{C}(\mathbf{k})- \widehat{\mathcal{C}}(\widehat{\mathbf{k}})\geq\mathcal{C}(\mathbf{k})- \widehat{\mathcal{C}}(\mathbf{w}'_2)=\mathcal{C}(\mathbf{k})- \mathcal{C}(\mathbf{w}'_1)+O_p(1)\conp\infty. \hfill\qedhere
\end{align*}
\end{proof}
\begin{proof}[Sketch of consistency of WBS with the Schwartz Criteria]
Let $\rhatbar_i'$ be the mean of the observations in the sample with the same distribution as that of $\widehat{R}_i$.
Recall that from the proof of Theorem \ref{thm:WBS} that $\Pr(|k_i-\hat{k}_i|<N^{\phi}C)\rightarrow 1$, where this result is independent of the thresholding technique. We can then condition on the event that $|k_i-\hat{k}_i|<N^{1/2+r}C$ for all $i\in[\ell]$ and some very small $r>0$. 
Consider some estimated change-point $\hat{k}_j$, it follows that
\begin{align*}
    O(N^2)\frac{1}{N}\sum_{i=\hat{k}_{j-1}}^{\hat{k}_j-1} \left(\frac{\widehat{R}_i-\rhatbar_j}{\Var{\widehat{R}_i}^{1/2}}\right)^2&=O(N^2)\frac{1}{N}\sum_{i=k_{j-1}}^{k_j-1} \left(\frac{\widehat{R}_i-\rhatbar_j}{\Var{\widehat{R}_i}^{1/2}}\right)^2+O(N^{3/2+r}),
\end{align*}
which follows from the fact that $$ O(N^2)\frac{1}{N}\sum_{i=k_j}^{k_j+\floor{N^{1/2+r}C}} \left(\frac{\widehat{R}_i-\rhatbar_j}{\Var{\widehat{R}_i}^{1/2}}\right)^2=O(N)N^{1/2+r}O(1)=O(N^{3/2+r}).$$
We can know work with the sum over the true change-point interval:
\begin{align*}
   \sum_{i=k_{j-1}}^{k_j-1} \left(\frac{\widehat{R}_i-\rhatbar_j}{\Var{\widehat{R}_i}^{1/2}}\right)^2&= \sum_{i=k_{j-1}}^{k_j-1} \left(\frac{\widehat{R}_i-\rhatbar_j'+\rhatbar_j'-\rhatbar_j}{\Var{\widehat{R}_i}^{1/2}}\right)^2\\
    &=\sum_{i=k_{j-1}}^{k_j-1} \left(\frac{\widehat{R}_i-\rhatbar_j'}{\Var{\widehat{R}_i}^{1/2}}\right)^2-(k_j-k_{j-1}) \left(\frac{\rhatbar_j'-\rhatbar_j}{\Var{\widehat{R}_i}^{1/2}}\right)^2\\
    &=\sum_{i=k_{j-1}}^{k_j-1} \left(\frac{\widehat{R}_i-\rhatbar_j'}{\Var{\widehat{R}_i}^{1/2}}\right)^2+O(N)O_p(N^{-1+2r}).
\end{align*}
It follows that 
\begin{equation}
\label{eqn::true_app}
    O(N^2)\frac{1}{N}\sum_{i=\hat{k}_{j-1}}^{\hat{k}_j-1} \left(\frac{\widehat{R}_i-\rhatbar_j}{\Var{\widehat{R}_i}^{1/2}}\right)^2=O(N^2)\frac{1}{N}\sum_{i=\hat{k}_{j-1}}^{\hat{k}_j-1} \left(\frac{\widehat{R}_i-\rhatbar_j}{\Var{\widehat{R}_i}^{1/2}}\right)^2+O(N^{3/2+r})+O_p(N^{1+2r}).
\end{equation}
When $\alpha=1$, we have that 
$$\mathcal{G}(\hat{\ell})=\frac{N}{2}\log(\hat{\varsigma}_{\hat{\ell}}^2)+\hat{\ell}\log N,$$
where one recalls that
$$\hat{\varsigma}_{\hat{\ell}}^2=\frac{1}{N}\sum_{i=1}^N(\ranki-\rhatbar_i)^2.
$$
Minimizing $\mathcal{G}(\hat{\ell})$ is equivalent to minimizing 
$$\hat{\varsigma}_{\hat{\ell}}^N N^{\hat{\ell}}.$$
Assume first that $\hat{\ell}=\ell$. 
\begin{align*}
\hat{\varsigma}_{\hat{\ell}}^2&=\frac{1}{N}\sum_{i=1}^N(\ranki-\rhatbar_i)^2\\
&=O(N^2)\frac{1}{N}\sum_{i=1}^{k_1-1} \left(\frac{\widehat{R}_i-\rhatbar_1}{\Var{\widehat{R}_i}^{1/2}}\right)^2+\ldots+O(N^2)\frac{1}{N}\sum_{i=k_\ell}^{N} \left(\frac{\widehat{R}_i-\rhatbar_\ell}{\Var{\widehat{R}_i}^{1/2}}\right)^2+O_p(N^{3/2+r})\\
&=O(N^2)O_p(N^{-1/2})+O_p(N^{3/2+r})=O_p(N^{3/2+r})
\end{align*}
since $\ell$ is fixed and $$\frac{1}{N}\sum_{i=k_\ell}^{N}  \left(\frac{\widehat{R}_i-\rhatbar_\ell}{\Var{\widehat{R}_i}^{1/2}}\right)^2$$ is asymptotically normal. 
Using \eqref{eqn::true_app}, we have that 
It then follows that 
$$\hat{\varsigma}_{\hat{\ell}}^N N^{\hat{\ell}}=O_p(N^{3N/4+rN/2+\ell}).$$
Suppose that $\hat{\ell}<\ell$, and WLOG assume that $k_1$ is undetected. 
We have that
\begin{align*}
\frac{1}{N}\sum_{i=1}^{k_1+k_2-1}(\ranki-\rhatbar_i)^2&=\frac{1}{N}\sum_{i=1}^{k_1-1}(\ranki-\rhatbar_1'-\rhatbar_i+\rhatbar_{1}')^2+\frac{1}{N}\sum_{i=k_1}^{k_2-1}(\ranki-\rhatbar_{k_1}'-\rhatbar_i+\rhatbar_{k_1}')^2.
\end{align*}
Looking at the left-hand term,
\begin{align*}
\frac{1}{N}\sum_{i=1}^{k_1-1}(\ranki-\rhatbar_1'-\rhatbar_i+\rhatbar_{1}')^2&=\frac{1}{N}\sum_{i=1}^{k_1-1}(\ranki-\rhatbar_1')^2+\frac{k_1-1}{N}(\rhatbar_i-\rhatbar_1')^2+\frac{2}{N}(\rhatbar_i-\rhatbar_1')\sum_{i=1}^{k_1-1}(\ranki-\rhatbar_1')\\
&= O_p(N^{3/2})+O(N^2).
\end{align*}
We then have that 
$$\hat{\varsigma}_{\hat{\ell}}^N N^{\hat{\ell}}>O_p(N^{N})>O_p(N^{3N/4+\ell}),$$
which gives a contradiction. 

Assume that $\hat{\ell}>\ell$. 
Then, it will still hold that $\hat{\varsigma}_{\hat{\ell}}^2\approx O_p(N^{3/2+r})$
for $\hat{\ell} < O(N)$ and of course $ N^\ell< N^{\hat{\ell}}$. 
If $\hat{\ell}=O(N)$ then $\hat{\varsigma}_{\hat{\ell}}^2$ is roughly a sum of $O(N)$ random variables with non-zero mean and we have that $\hat{\varsigma}_{\hat{\ell}}^N N^{\hat{\ell}}\approx N^{2N}>O_p(N^{3N/4+rN/2+\ell})$. 
\end{proof}

\section{Simulation on the rank distributions}\label{sim::rankd}
We show here that more types of changes in the covariance matrix are exhibited by changes in the depth rankings.
Consider two samples each from a 6-dimensional multivariate normal distribution. 
We fix
$$\Sigma_1=\left[\begin{array}{cccccc}
1  & 0.4 & 0.4 & 0  & 0 & 0  \\ 
0.4 & 1  & 0.4 & 0  & 0 & 0  \\ 
0.4 & 0.4 & 1  & 0  & 0 & 0  \\
0  & 0  & 0  & 1  & 0 & 0.4 \\
0  & 0  & 0  & 0  & 1 & 0  \\
0  & 0  & 0  & 0.4 & 0 & 1  \\ 
\end{array}\right]$$
as the covariance matrix of the first sample. 
Additionally, let $\sigma_{d_1,d_2,m}$ be the $(d_1,\ d_2)^{th}$ entry of the covariance matrix for sample $m$ (where $d_1,\ d_2 \in \{1,\dots,6\}$ and $m\in\{1,2\}$). We test four specifications of $\Sigma_2$, the covariance matrix of the second sample, and check for a difference in the distribution of ranks:
\begin{enumerate}
    \item \textbf{Submatrix on the diagonal change:} $\sigma_{d_1,d_2,2}=2\sigma_{d_1,d_2,1}$ for $d_1,d_2>3$ and $\sigma_{d_1,d_2,2}=\sigma_{d_1,d_2,1}$ otherwise.
    \item \textbf{Submatrix off the diagonal change:} $\sigma_{6,4,2}=\sigma_{4,6,2}=2\sigma_{6,4,1}$ and $\sigma_{d_1,d_2,2}=\sigma_{d_1,d_2,1}$ otherwise.
    \item \textbf{Mixed change scenario:} $\sigma_{6,4,2}=\sigma_{4,6,2}=-\sigma_{6,4,1}$, $\sigma_{4,4,2}=0.2\sigma_{4,4,1}$, $\sigma_{d_1,d_2,2}=2\sigma_{d_1,d_2,1}$ for $d_1,d_2\leq 3,\ d_1\neq d_2$ and $\sigma_{d_1,d_2,2}=\sigma_{d_1,d_2,1}$ otherwise.
    \item \textbf{Offsetting Expansion and Contraction:} $\sigma_{4,4,2}=0.5\sigma_{4,4,1}$,  $\sigma_{6,6,2}=2\sigma_{6,6,1}$ and $\sigma_{d_1,d_2,2}=\sigma_{d_1,d_2,1}$ otherwise.
\end{enumerate}
\begin{figure}[t]
\begin{minipage}[c]{.5\textwidth} 
\centering%
\includegraphics[width=\textwidth,trim={0 00 0 0},clip]{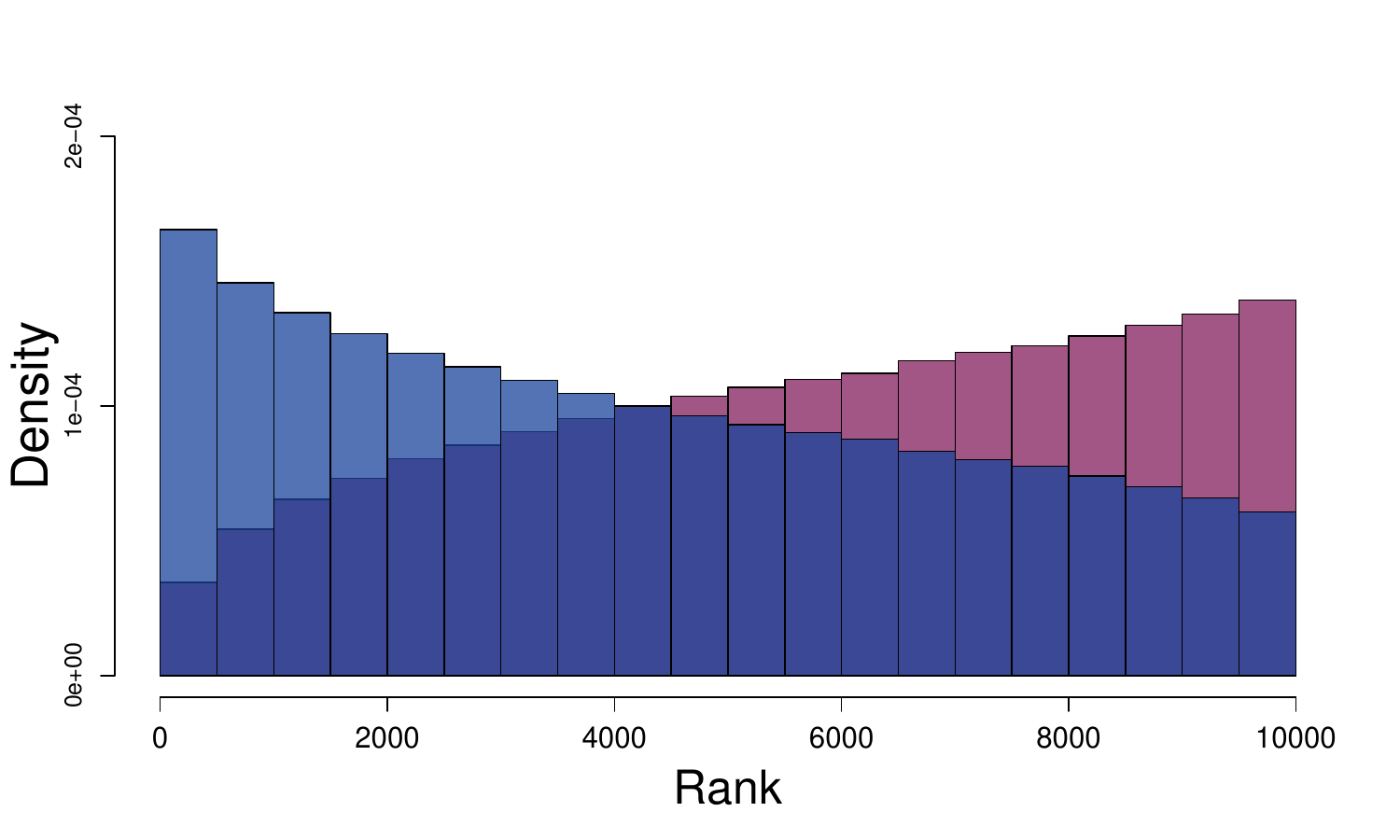}
\caption*{(a)}
\end{minipage}\hfill
\begin{minipage}[c]{.5\textwidth} 
\centering%
\includegraphics[width=\textwidth,trim={0 0 0 0},clip]{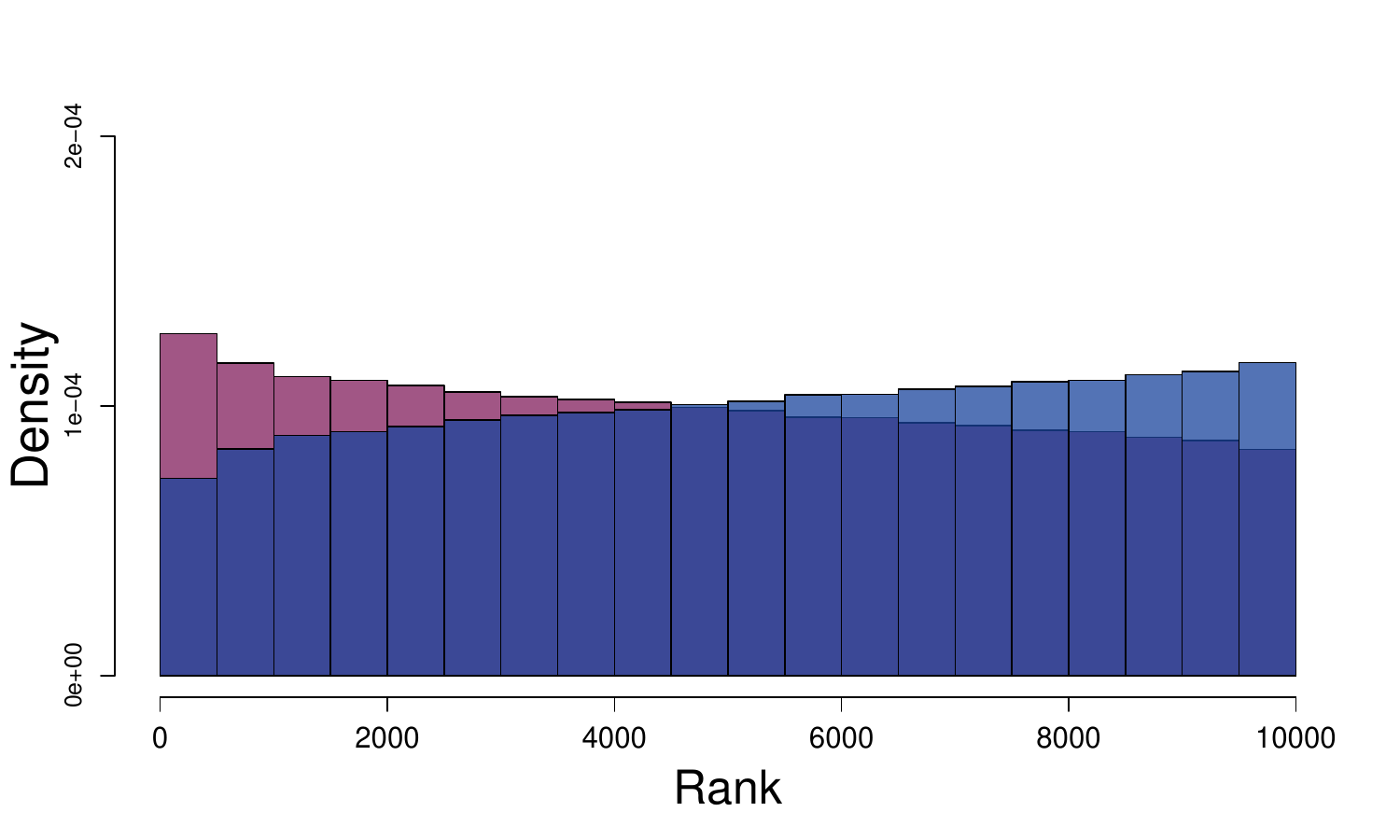}
\caption*{(b)}
\end{minipage}\hfill\newline
\begin{minipage}[c]{.5\textwidth} 
\centering%
\includegraphics[width=\textwidth,trim={0 00 0 0},clip]{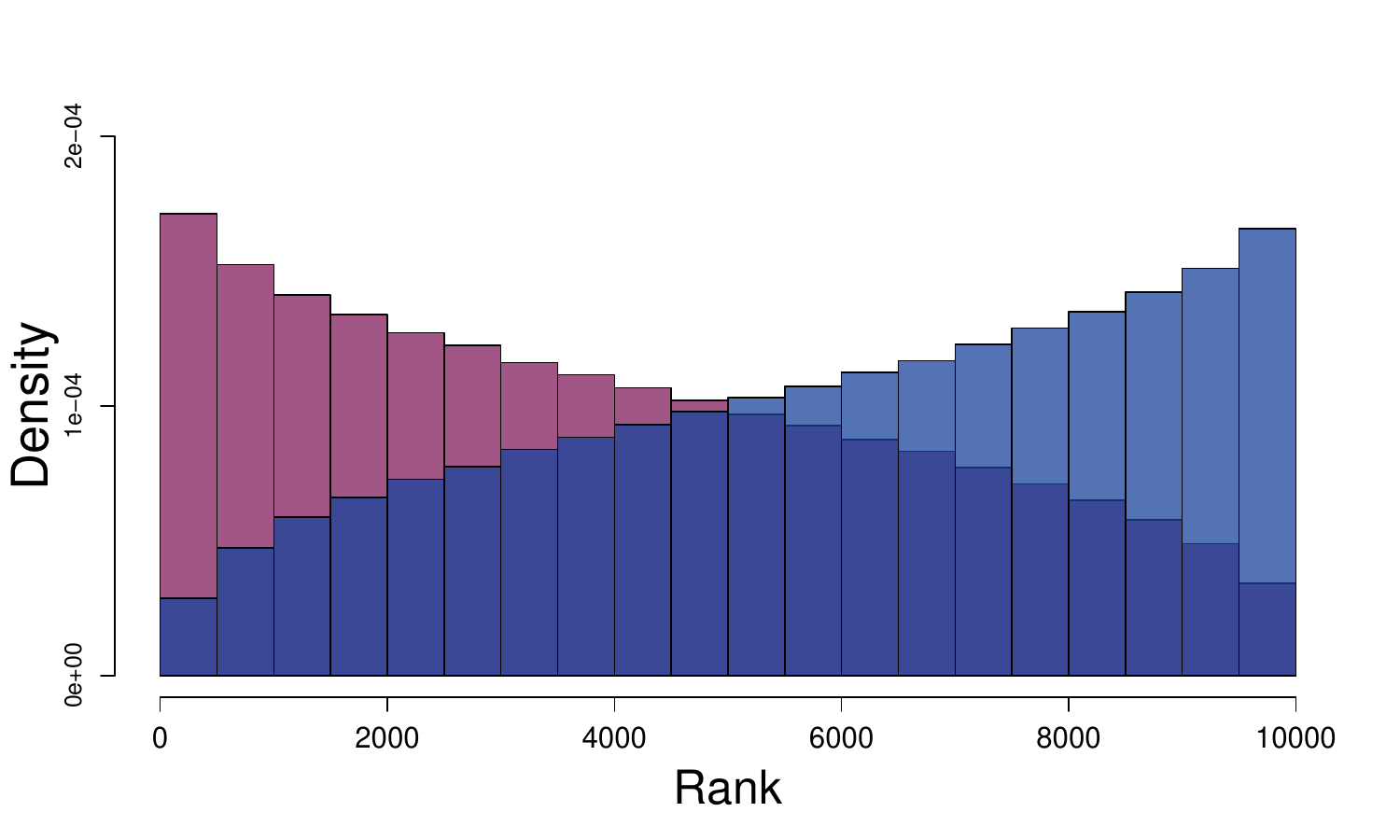}
\caption*{(c)}
\end{minipage}\hfill
\begin{minipage}[c]{.5\textwidth} 
\centering%
\includegraphics[width=\textwidth,trim={0 0 0 0},clip]{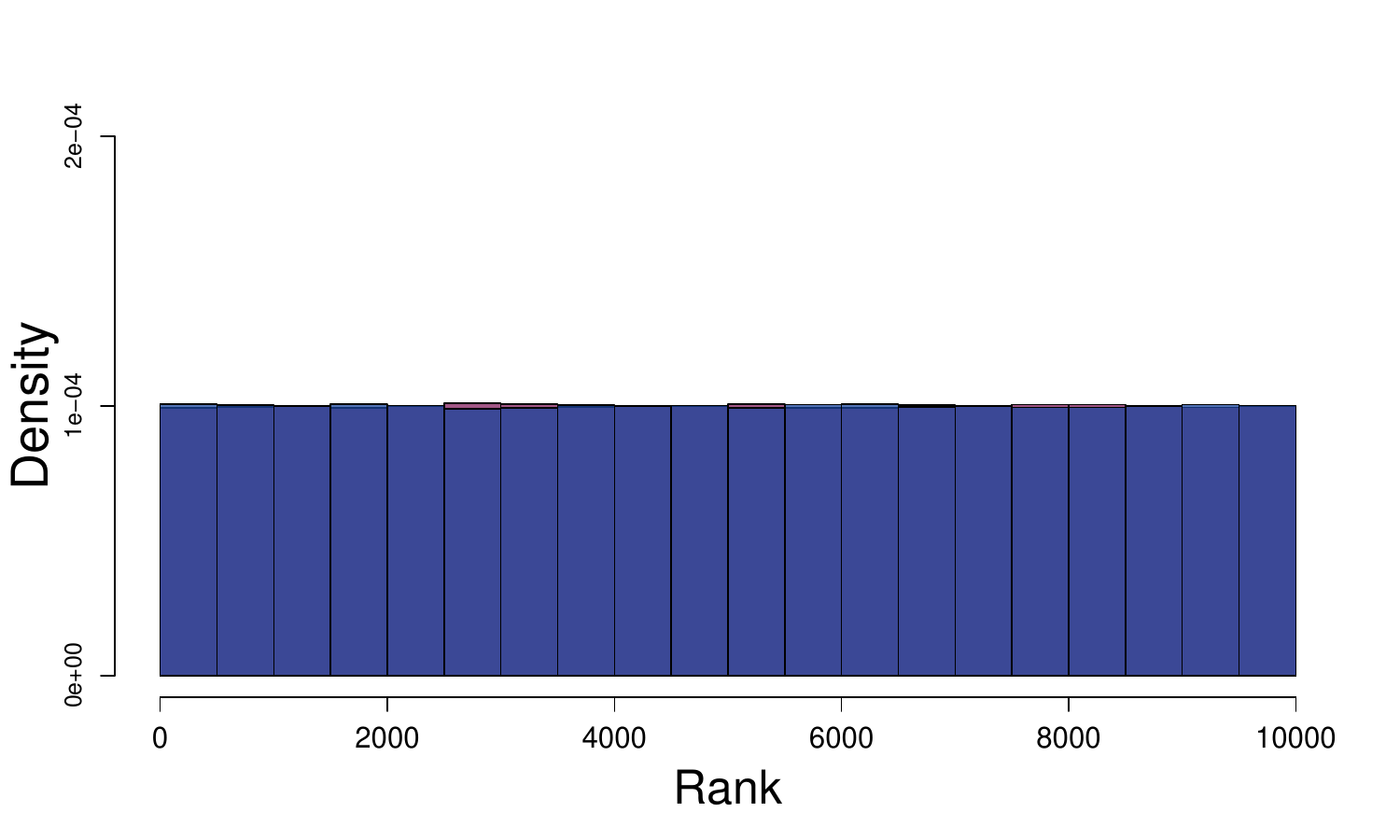}
\caption*{(d)}
\end{minipage}\hfill
\caption{Normalised histograms of the depth ranks of sample 1 (red) and sample 2 (blue) under a (a) submatrix on the diagonal change, (b) submatrix on the off diagonal change, (c) mixed change and (d) offsetting expansion and contraction.}%
\label{fig:scen}%
\end{figure}
We drew samples of size $N=5000$ from each population and computed the combined sample depth ranks. We then repeated this 100 times for each scenario. Figure \ref{fig:scen} shows histograms of each samples' depth ranks, one graph for each scenario. We see that expansions and contractions of submatrices correspond to changes in the rank distribution. 
Scenario three represents a mixture of these expansions and contractions (of different submatrices) and a change in the rank distributions is still exhibited. 
Scenario four shows that if we have two simultaneous contractions and expansions that `perfectly' offset each other, there won't be a change in the rank distributions. We note that if the offset is not perfect, (such as $\sigma_{4,4,2}=0.49\sigma_{4,4,1}$ instead) a change in the rank distribution will appear. 
This is fairly intuitive; since depth functions focus on the \textit{magnitude of outlyingness} and not necessarily the \textit{direction of outlyingness}.
We can summarize the results as follows:
\begin{itemize}
    \item Expansions/contractions in the submatrices produce a change in the rank distributions.
    \item The smaller the submatrix, the smaller the change in rank distribution.
    \item Certain combinations of expansions/contractions also admit changes in the rank distribution, provided the expansion(s) does not offset the contraction(s).
    \item Sign changes cannot be detected.
\end{itemize}

In conclusion, we aim to detect changes that can be expressed as contractions or expansions of submatrices. 
Additionally, we remark that many combinations of contractions and expansions can be detected, with the caveat that offsetting combinations of such changes make the change more difficult to detect, or in a special case, impossible. 
\section{Additional Simulation Results}\label{app::add_sim}
The second scenario is a set of expansions and contractions again, instead with $$\sigma_1^2=1,\ \sigma_2^2=3,\ \sigma_3^2=5,\ \sigma_4^2=3,\ \sigma_5^2=5,\ \sigma_6^2=1.$$ 
For the second scenario, we tested all depths paired with Algorithm \ref{alg:pelt} but only tested the Algorithm \ref{alg:wbs} paired with Mahalanobis depth.  
The results from the second scenario were very similar to the first and can be found in Appendix \ref{app::add_sim}. 
Additionally, due to computational limitations, for $N=2500$ and $N=5000$ Algorithm \ref{alg:wbs} was only ran with Mahalanobis depth and we did not run the WBSIP algorithm nor the BSOP algorithm. 
We did not run \texttt{ecp} methods with the second scenario, based on the other results being so similar. 

\begin{figure}[h!]
\begin{minipage}[c]{.32\textwidth} 
\centering%
\includegraphics[width=\textwidth]{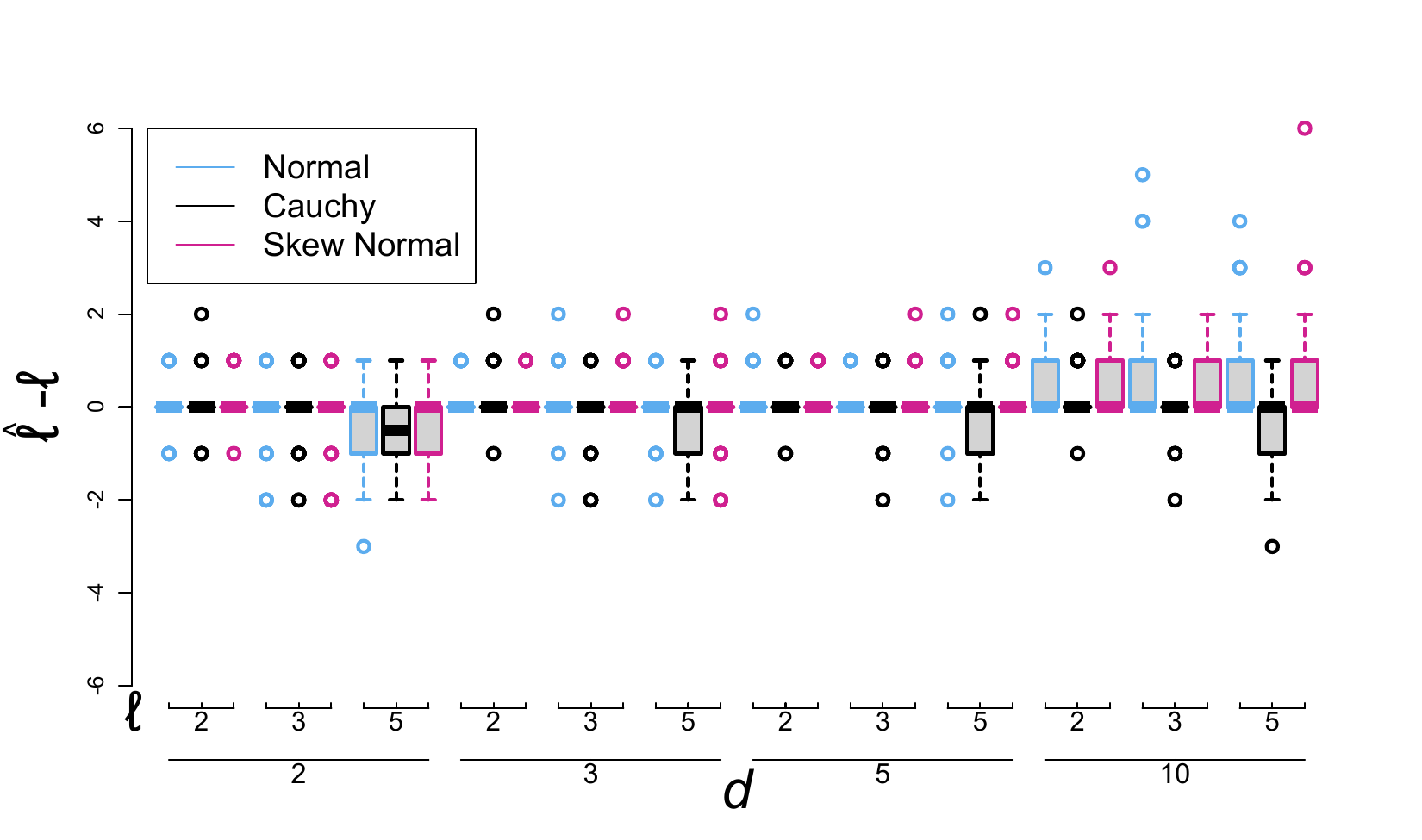}
\end{minipage}
\begin{minipage}[c]{.32\textwidth} 
\centering%
\includegraphics[width=\textwidth]{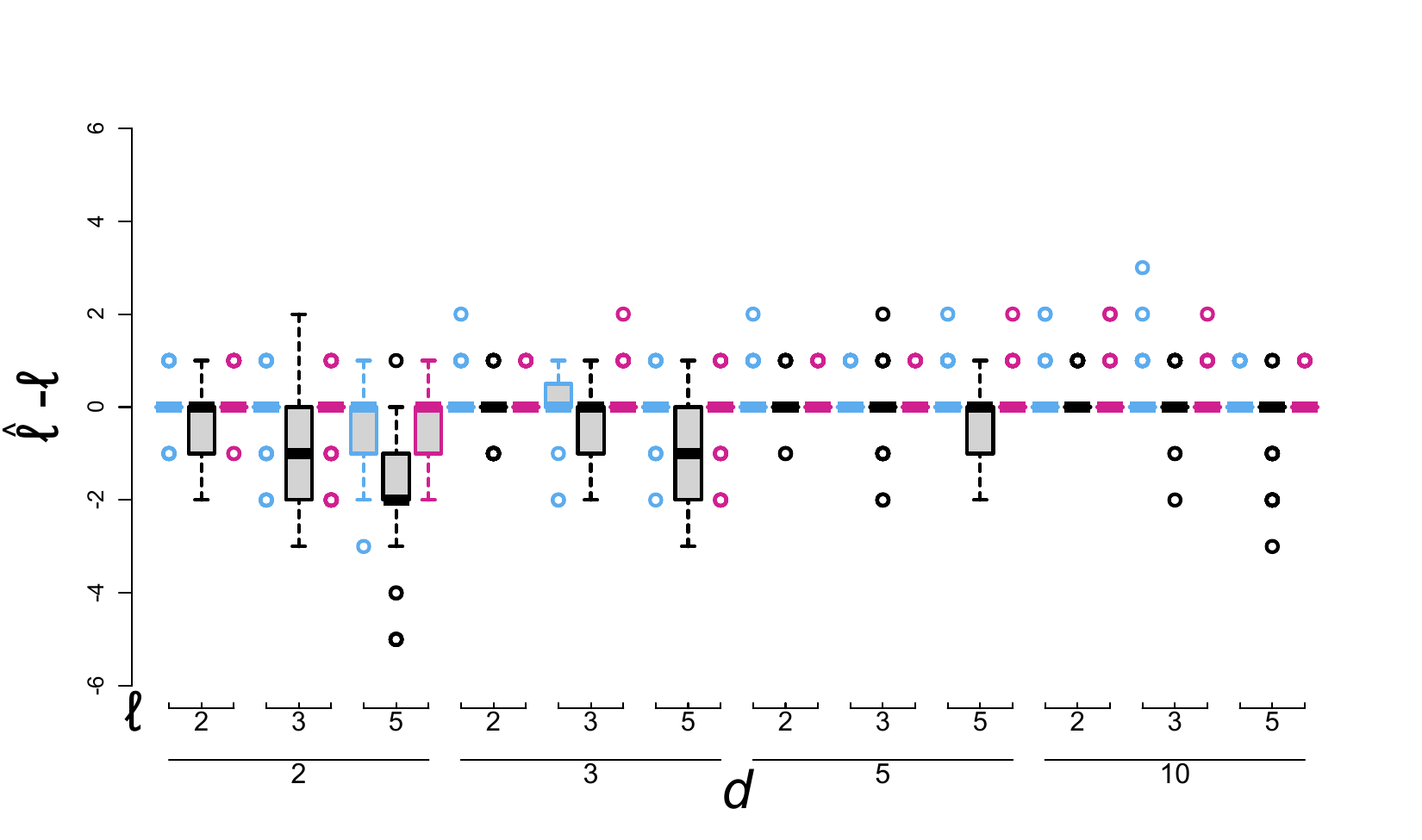}
\end{minipage}\hfill
\begin{minipage}[c]{.32\textwidth} 
\centering%
\includegraphics[width=\textwidth]{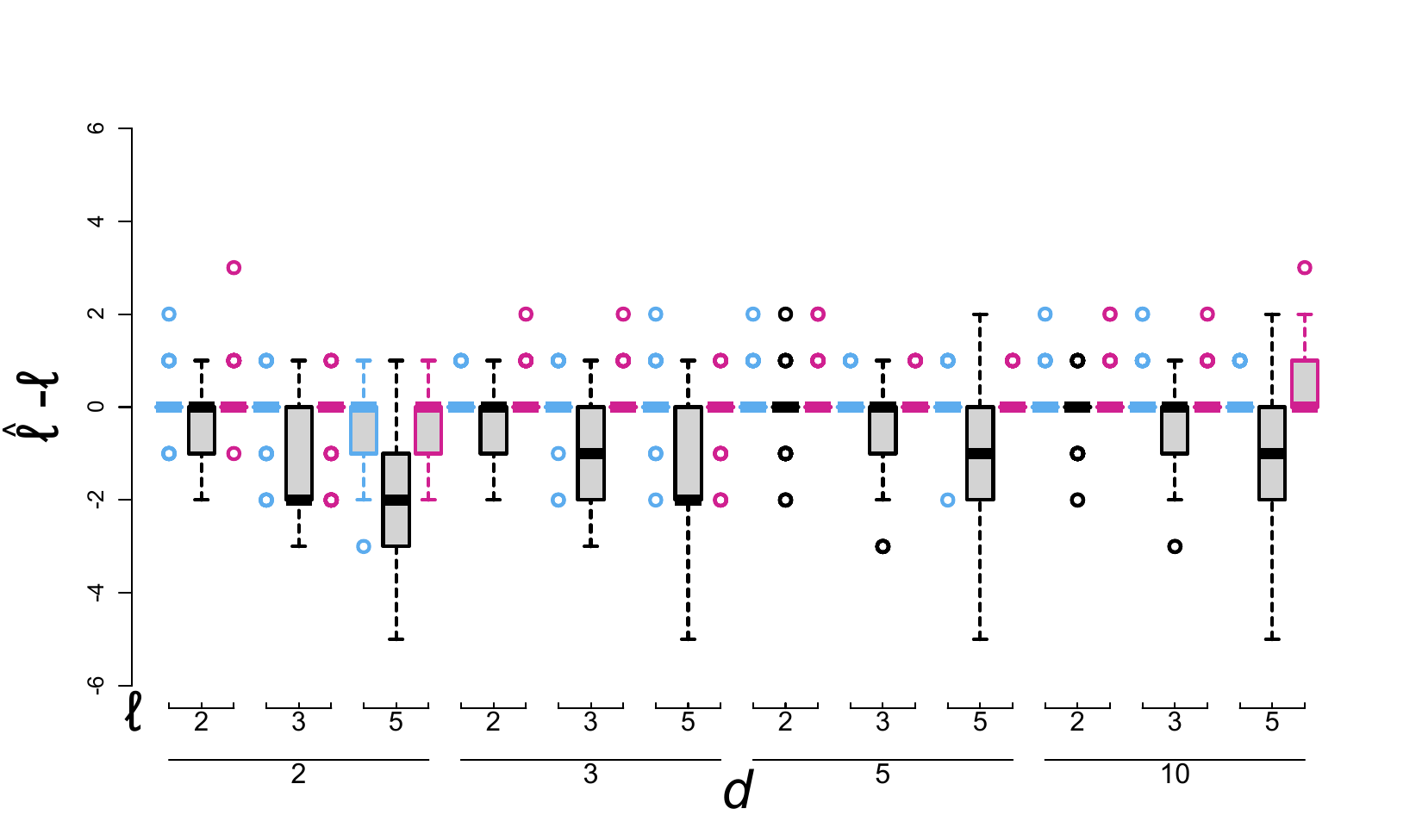}
\end{minipage}
\caption{Boxplots of $\hat{\ell}-\ell$ under the WBS algorithm in simulation scenario number 1. The parameters were $\alpha=0.9$, $N=1000$ and the depth functions used were half-space depth, Mahalanobis depth, and modified Mahalanobis depth, respectively.}%
\end{figure}
\begin{figure}[h!]
\begin{minipage}[c]{.49\textwidth} 
\centering%
\includegraphics[width=\textwidth]{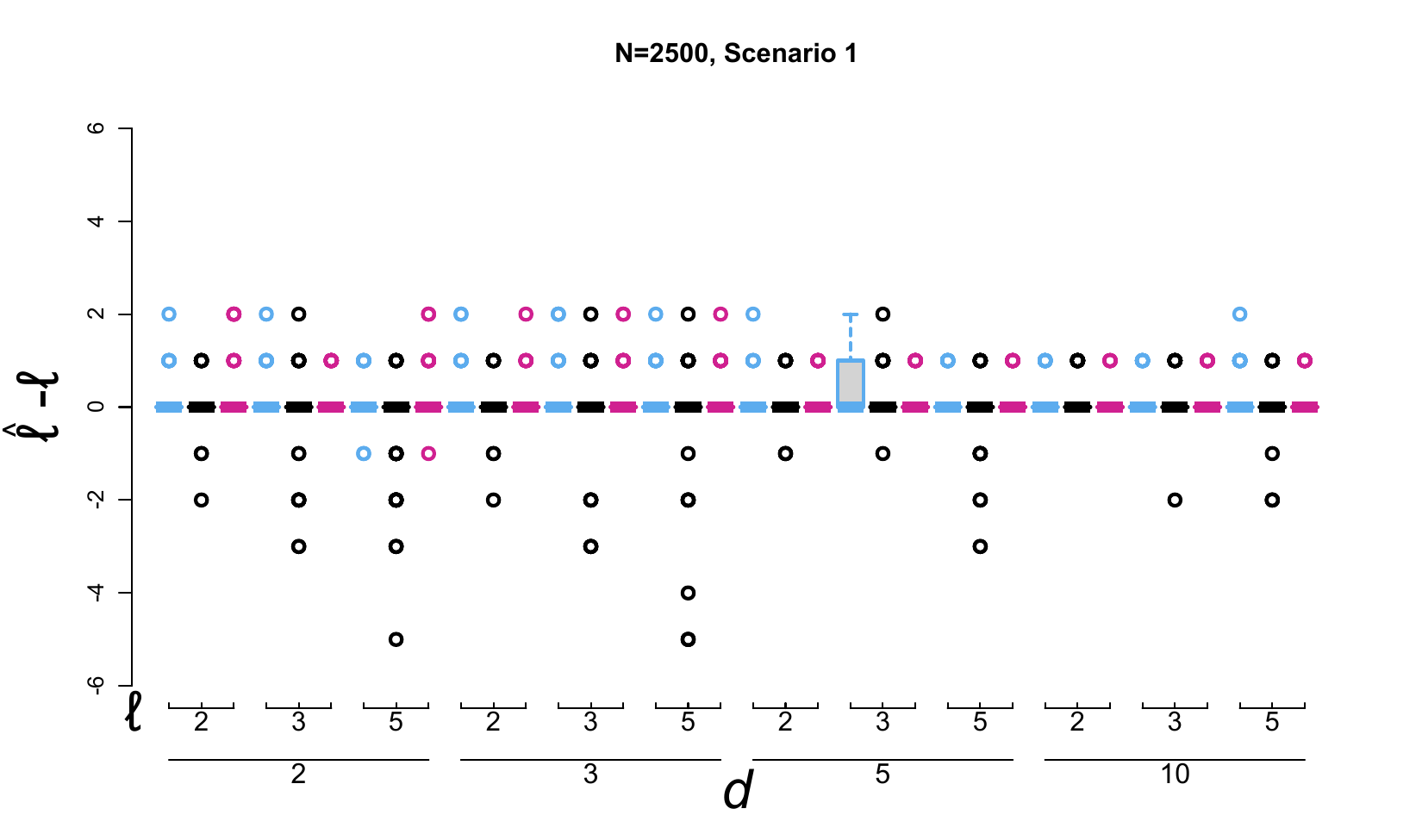}
\end{minipage}
\begin{minipage}[c]{.49\textwidth} 
\centering%
\includegraphics[width=\textwidth]{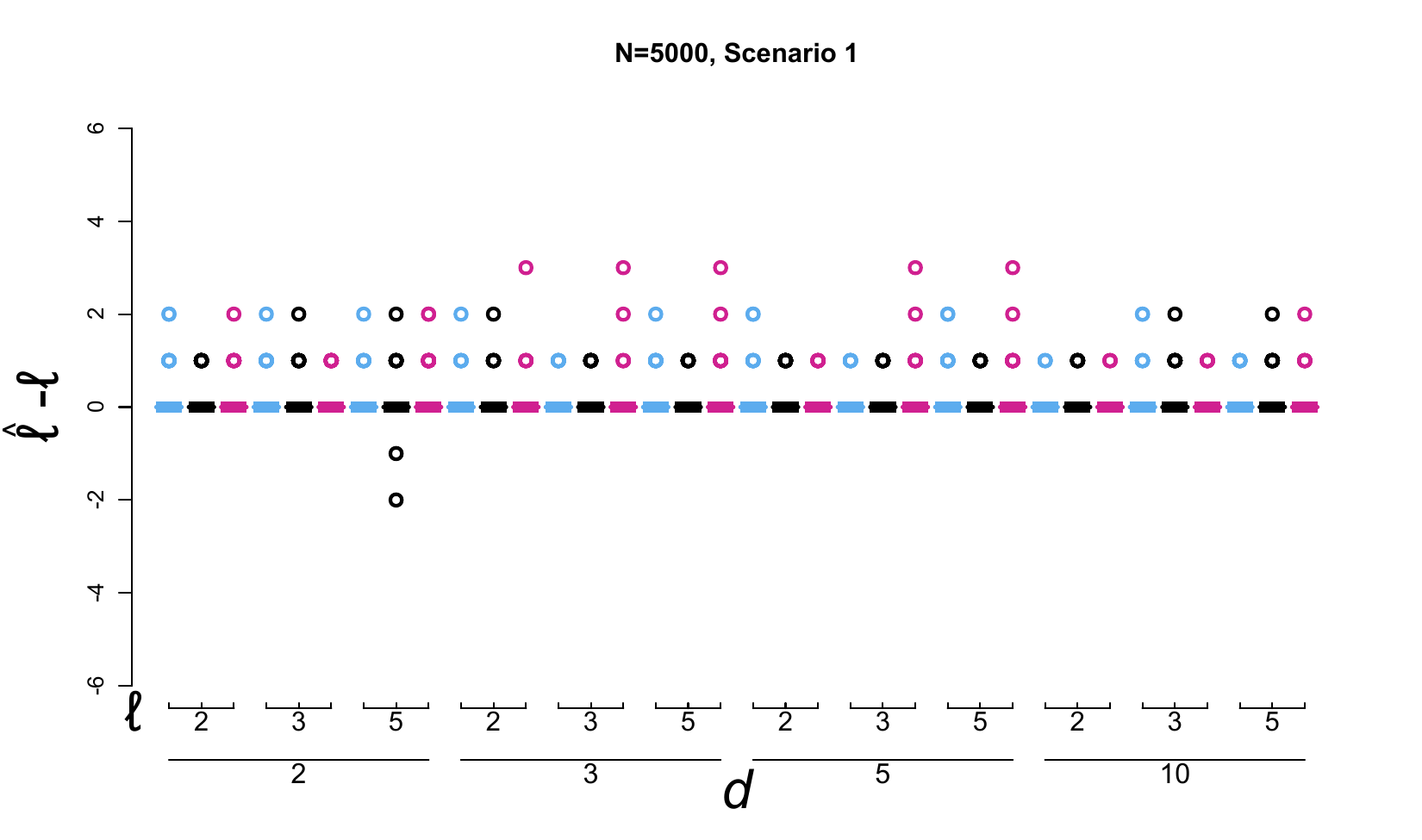}
\end{minipage}\hfill\newline
\begin{minipage}[c]{.49\textwidth} 
\centering%
\includegraphics[width=\textwidth]{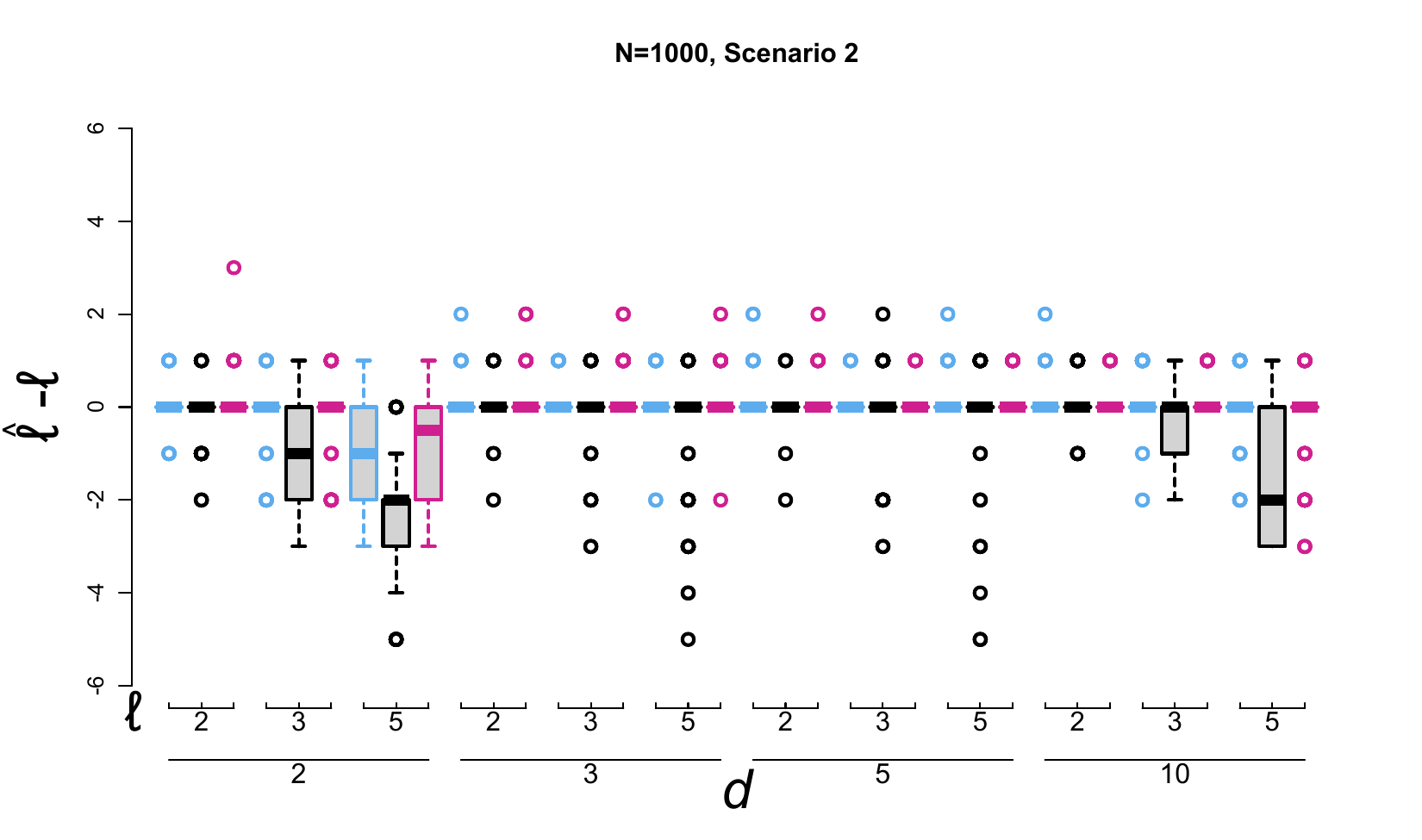}
\end{minipage}
\begin{minipage}[c]{.49\textwidth} 
\centering%
\includegraphics[width=\textwidth]{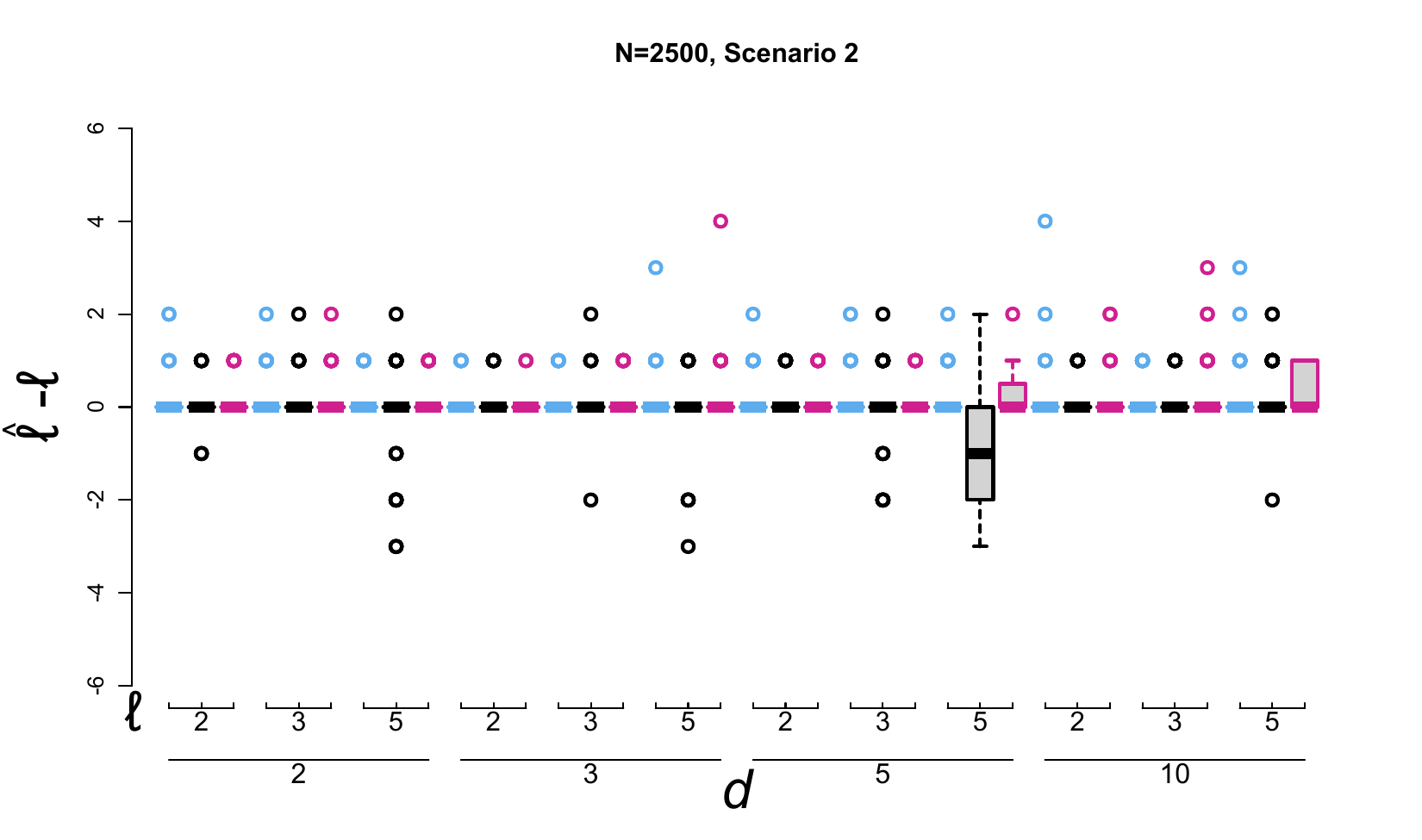}
\end{minipage}\hfill\newline
\begin{minipage}[c]{.49\textwidth} 
\centering%
\includegraphics[width=\textwidth]{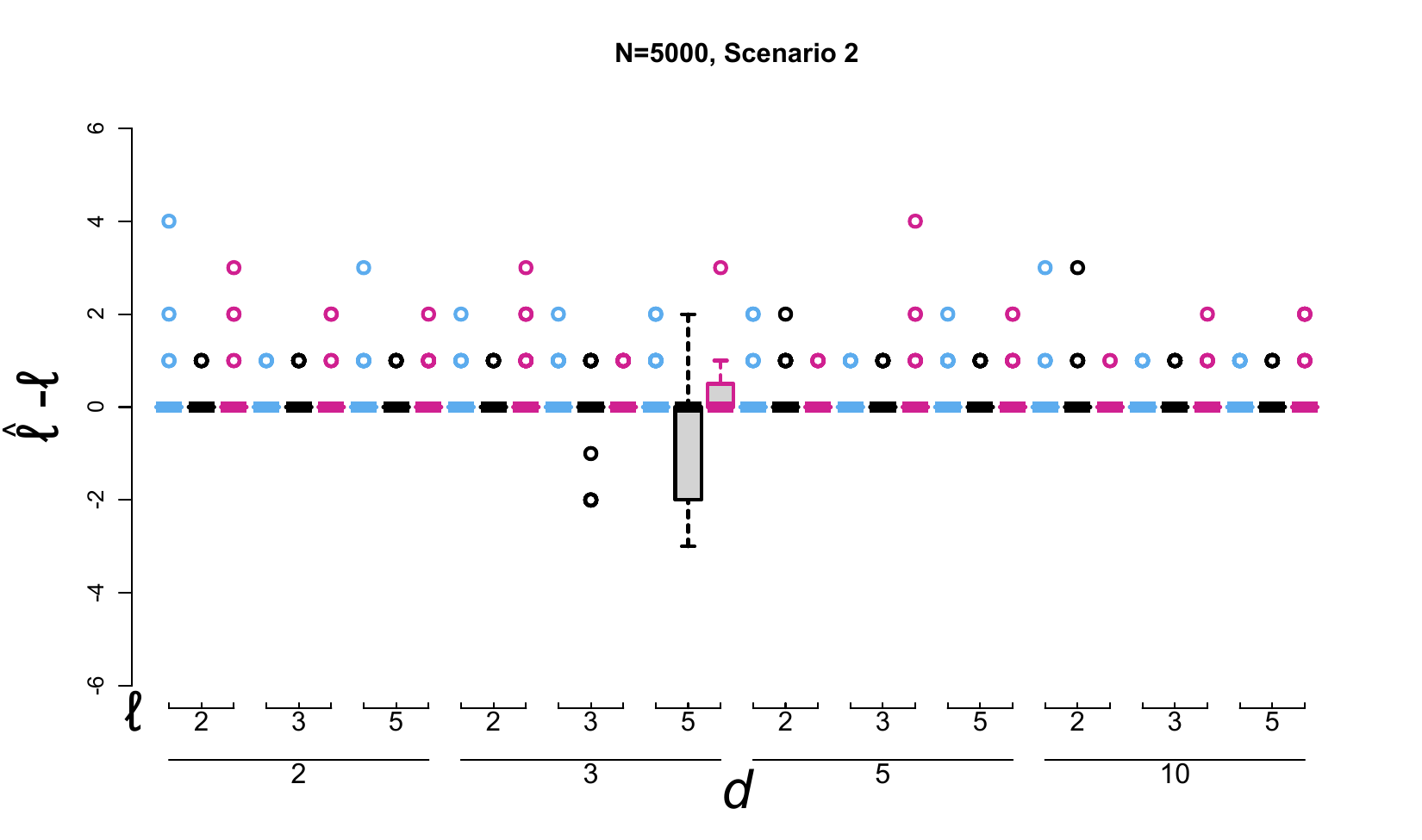}
\end{minipage}
\caption{Boxplots under Mahalanobis Depth of $\hat{\ell}-\ell$ under the WBS algorithm 
with $\alpha=0.9$.}%
\end{figure}

\begin{figure}
\begin{minipage}[c]{.32\textwidth} 
\centering%
\includegraphics[width=\textwidth]{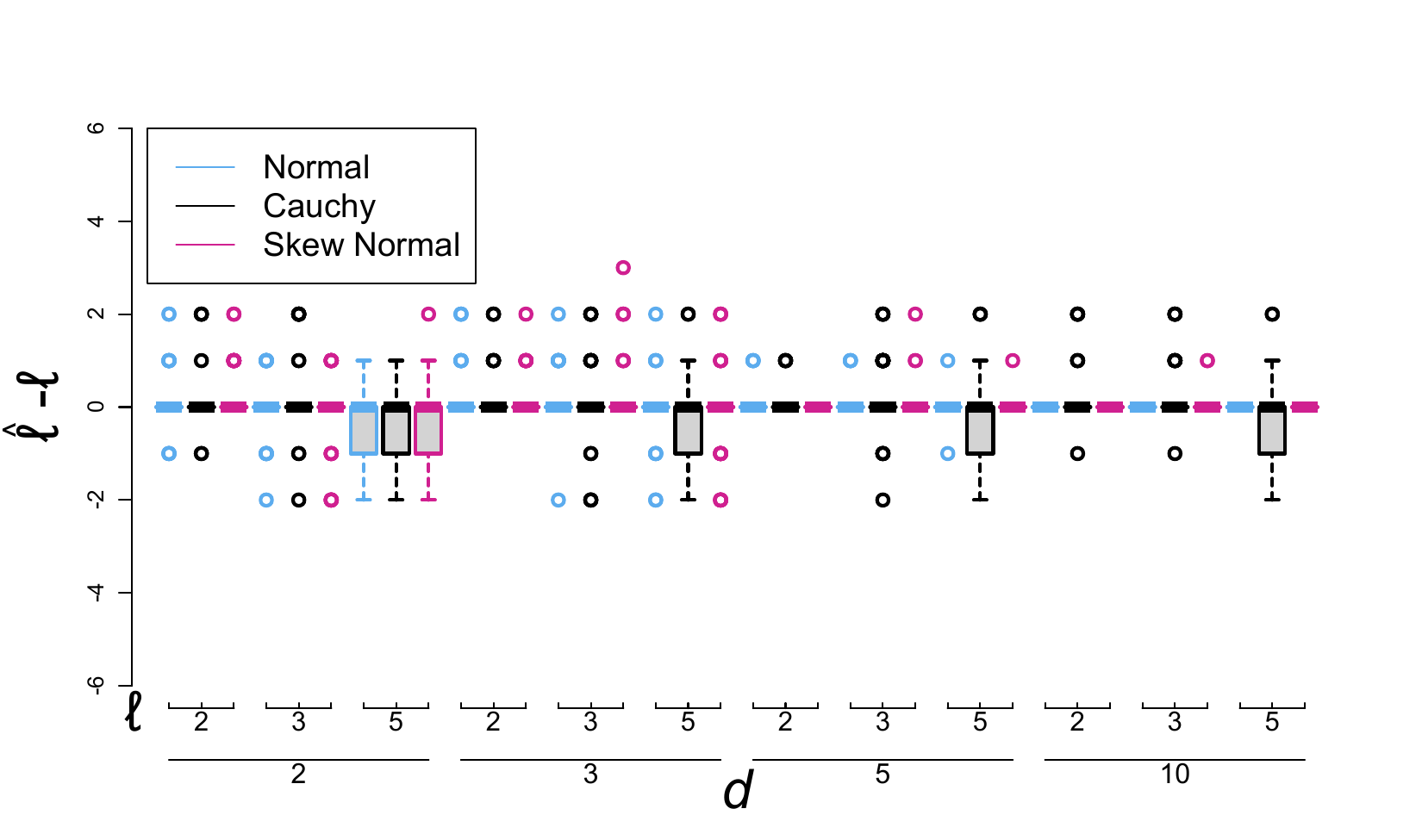}
\end{minipage}
\begin{minipage}[c]{.32\textwidth} 
\centering%
\includegraphics[width=\textwidth]{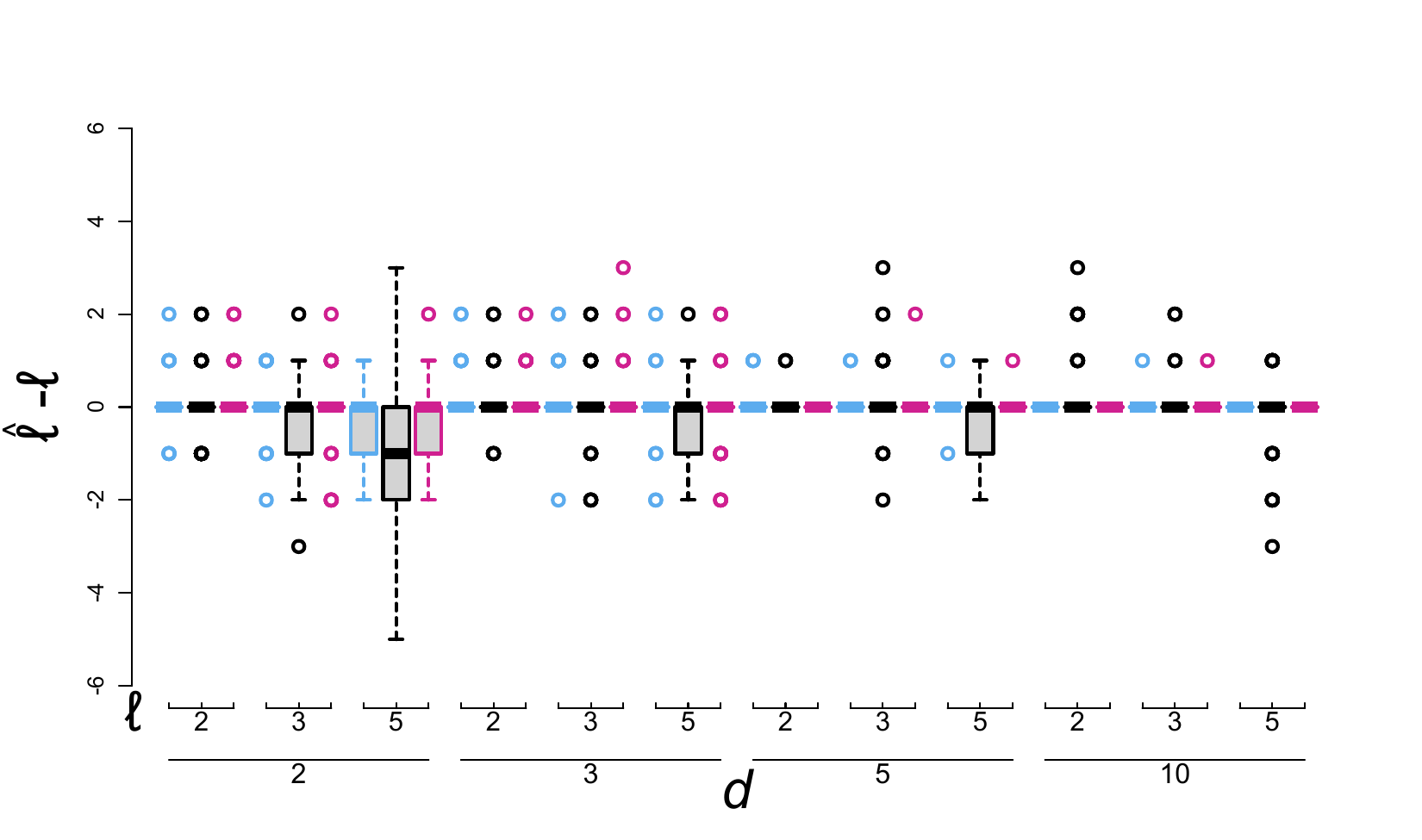}
\end{minipage}
\begin{minipage}[c]{.32\textwidth} 
\centering%
\includegraphics[width=\textwidth]{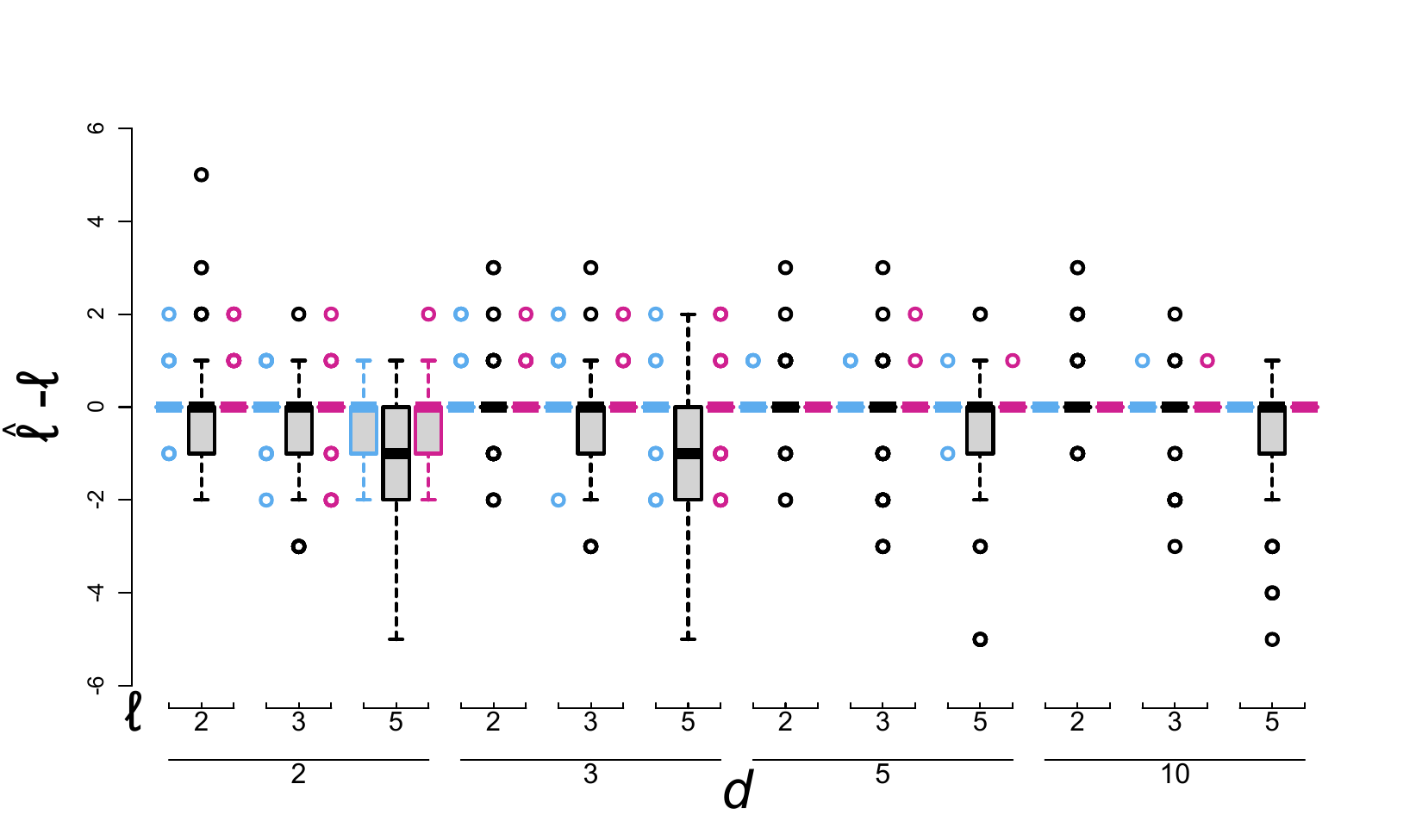}
\end{minipage}\hfill
\caption{Boxplots of $\hat{\ell}-\ell$ under the KW-PELT algorithm in simulation scenario number 1. The parameters were $C_1=0.18$ and $C_2=3.74$ for $N=1000$, with the depths being half-space depth, Mahalanobis depth, and modified Mahalanobis depth, respectively.}
\end{figure}
\begin{figure}
\begin{minipage}[c]{.49\textwidth} 
\centering%
\includegraphics[width=\textwidth]{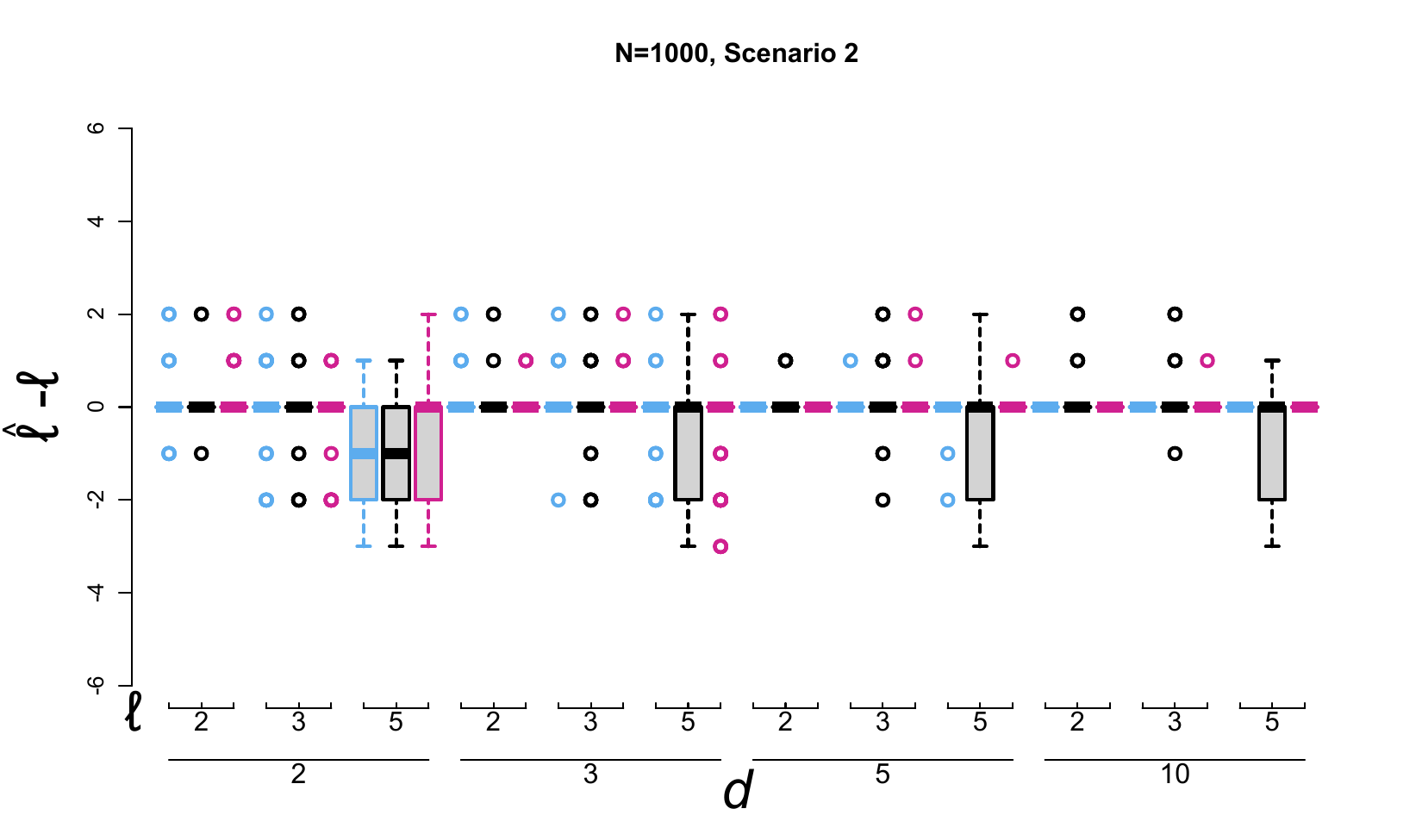}
\end{minipage}
\begin{minipage}[c]{.49\textwidth} 
\centering%
\includegraphics[width=\textwidth]{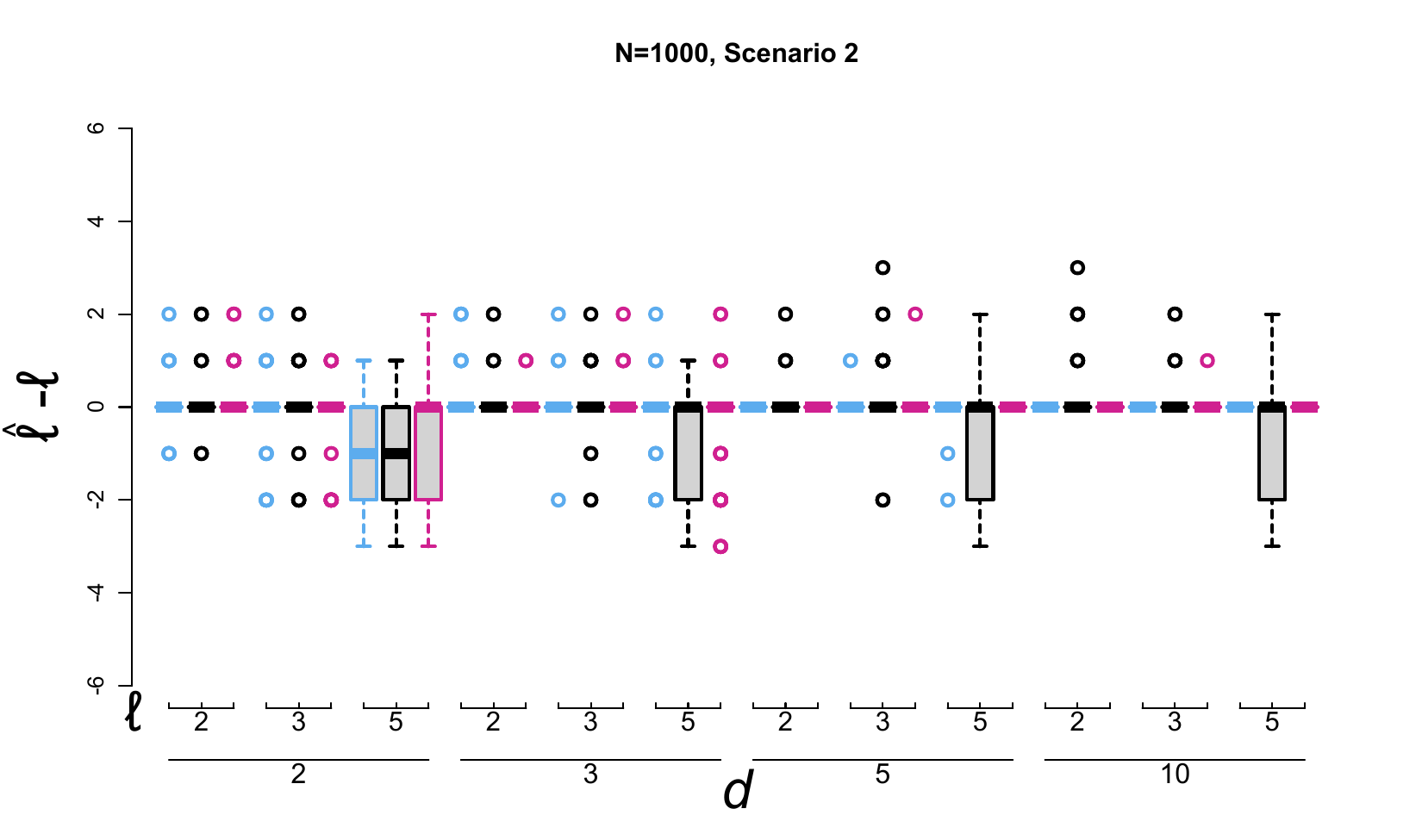}
\end{minipage}\hfill\newline
\begin{minipage}[c]{.49\textwidth} 
\centering%
\includegraphics[width=\textwidth]{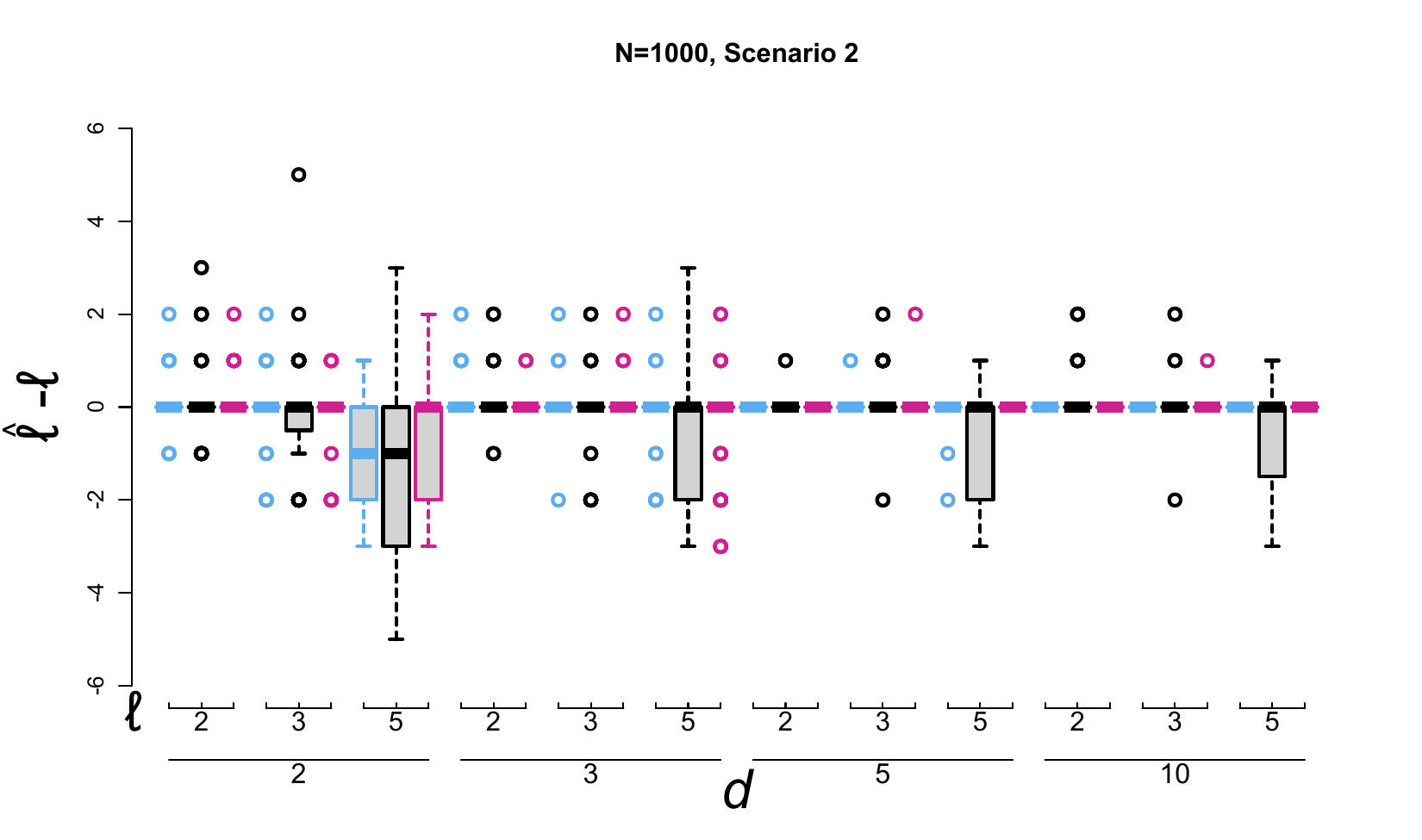}
\end{minipage}
\begin{minipage}[c]{.49\textwidth} 
\centering%
\includegraphics[width=\textwidth]{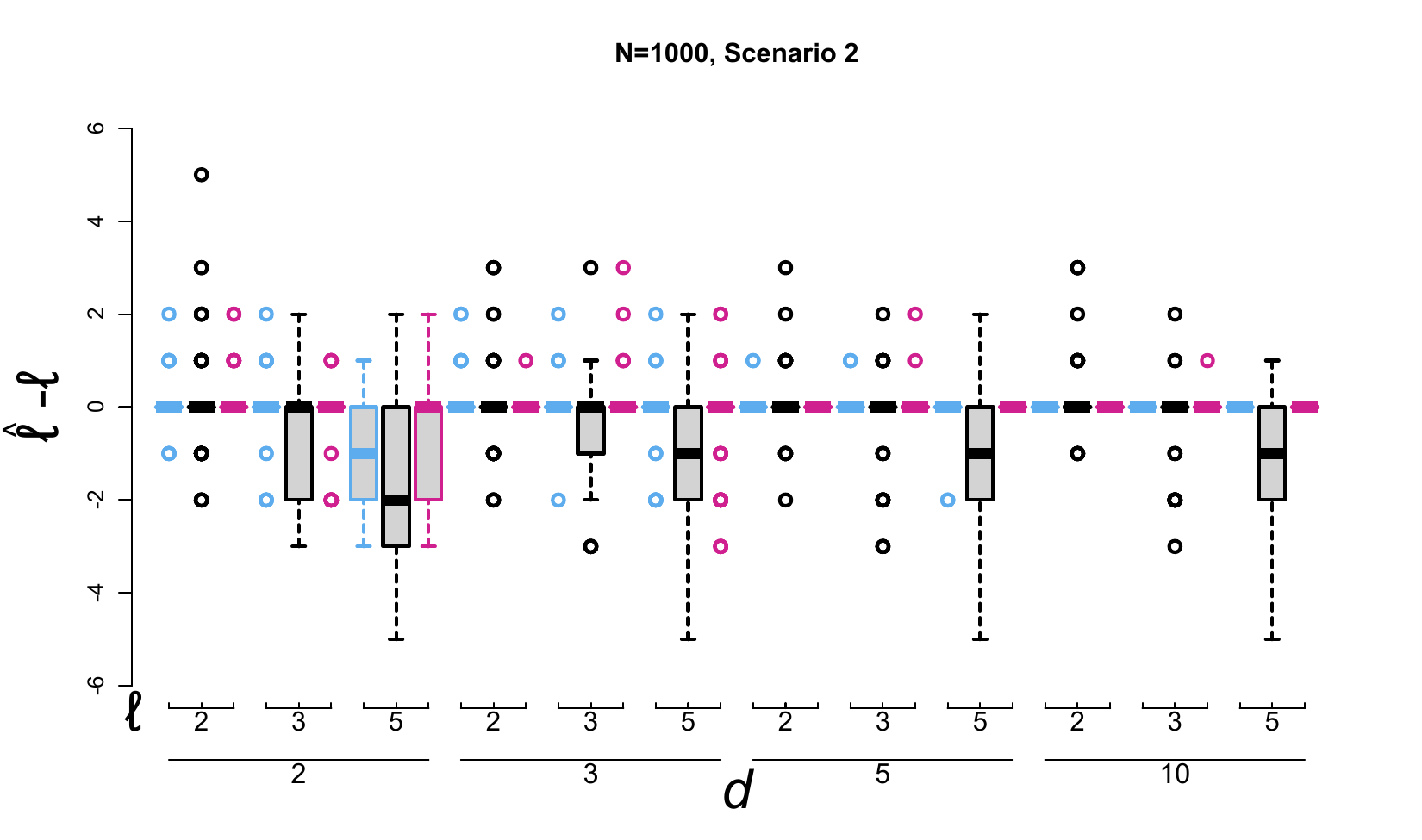}
\end{minipage}\hfill\newline
\caption{Boxplots of $\hat{\ell}-\ell$ under the KW-PELT algorithm with $C_1=0.18$ and $C_2=3.74$.}
\end{figure}
\begin{figure}
\begin{minipage}[c]{.49\textwidth} 
\centering%
\includegraphics[width=\textwidth]{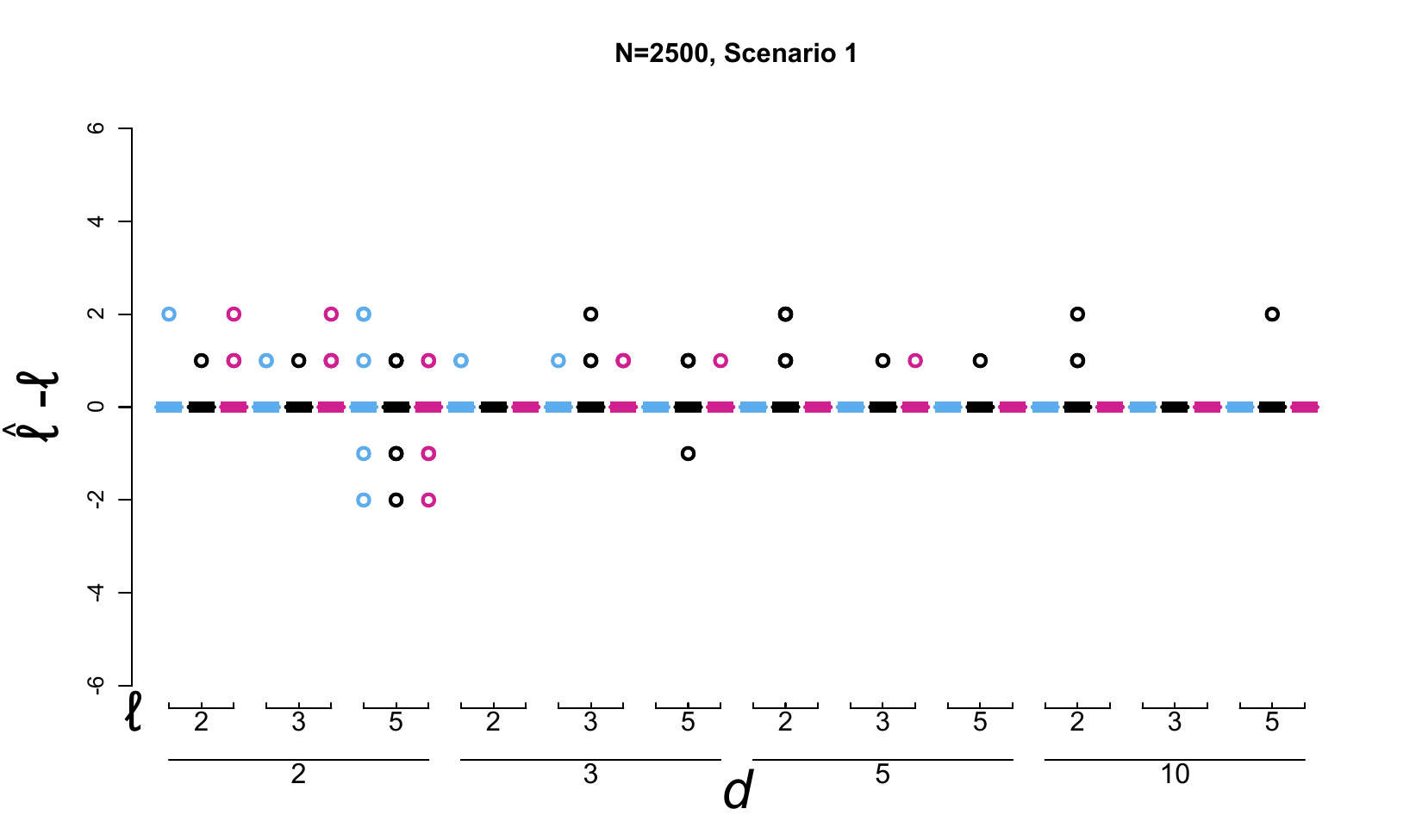}
\end{minipage}
\begin{minipage}[c]{.49\textwidth} 
\centering%
\includegraphics[width=\textwidth]{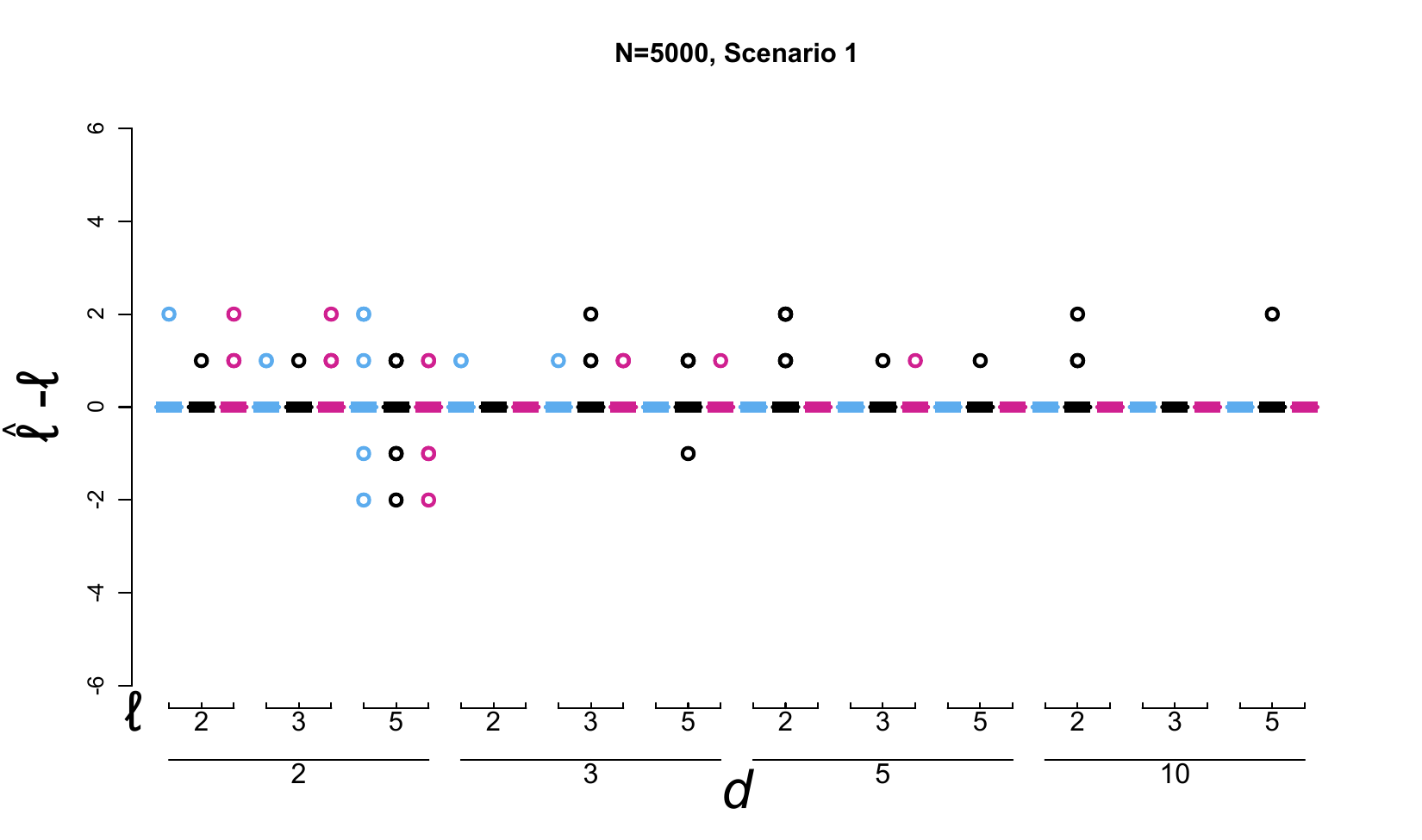}
\end{minipage}\hfill\newline
\begin{minipage}[c]{.49\textwidth} 
\centering%
\includegraphics[width=\textwidth]{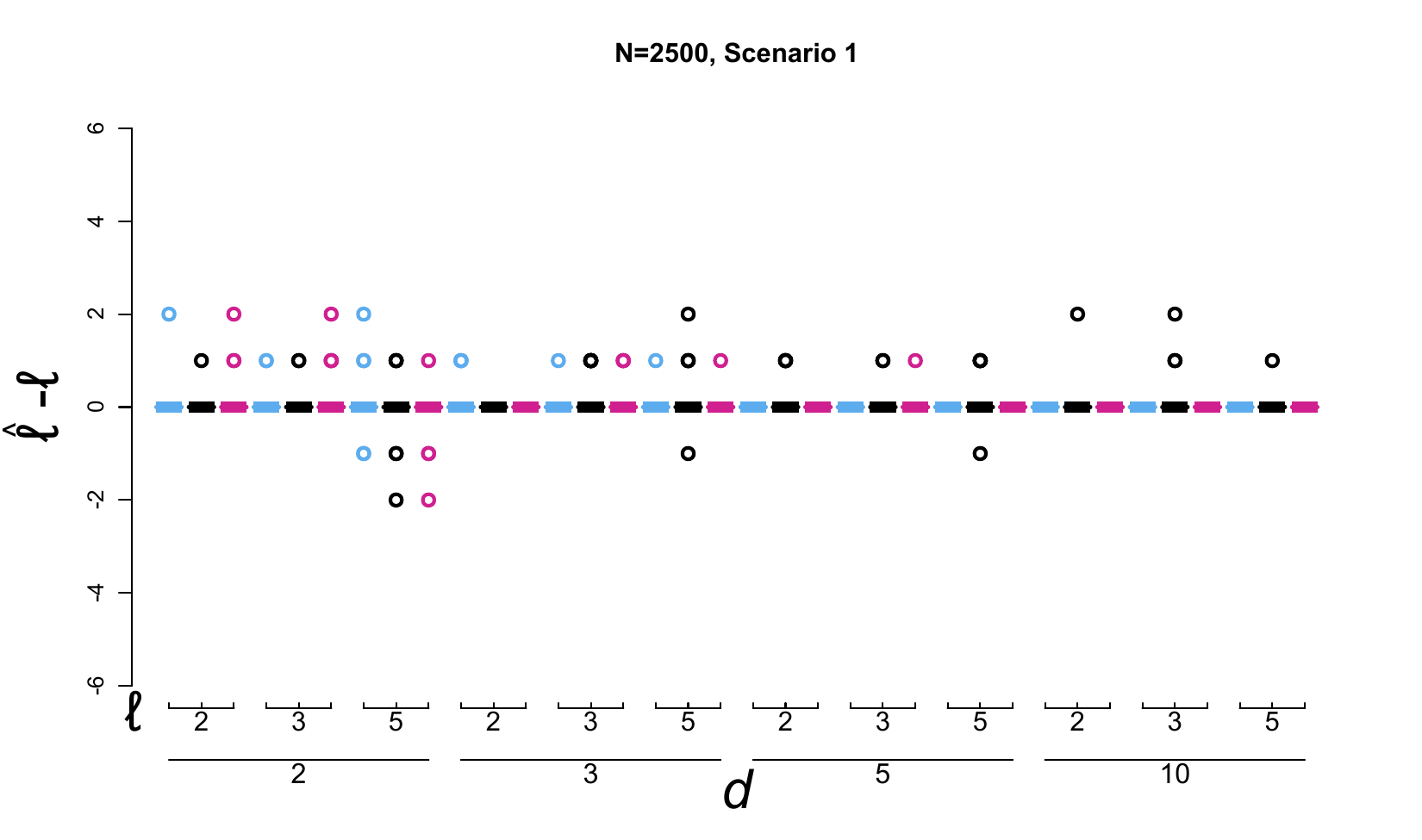}
\end{minipage}
\begin{minipage}[c]{.49\textwidth} 
\centering%
\includegraphics[width=\textwidth]{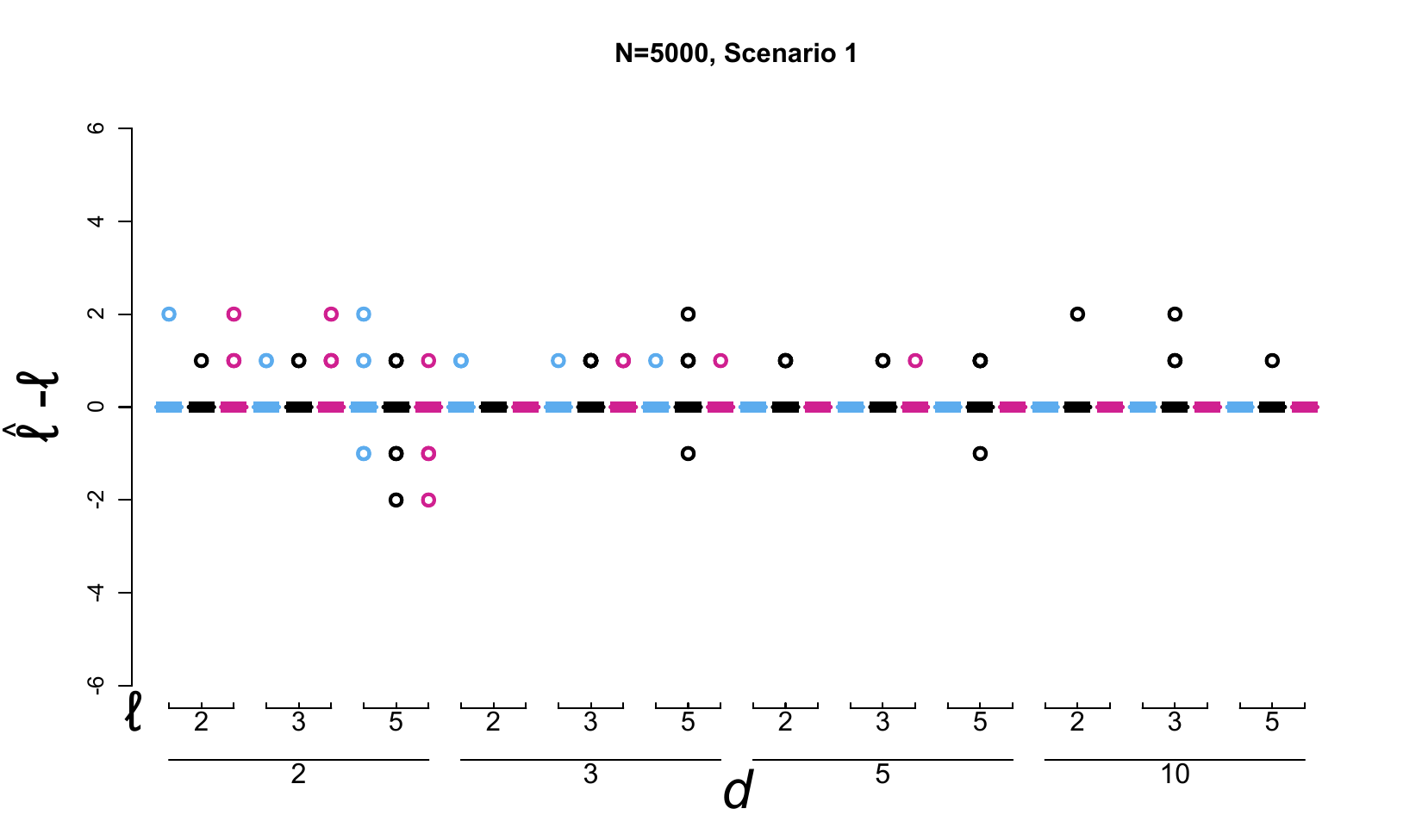}
\end{minipage}\hfill\newline
\begin{minipage}[c]{.49\textwidth} 
\centering%
\includegraphics[width=\textwidth]{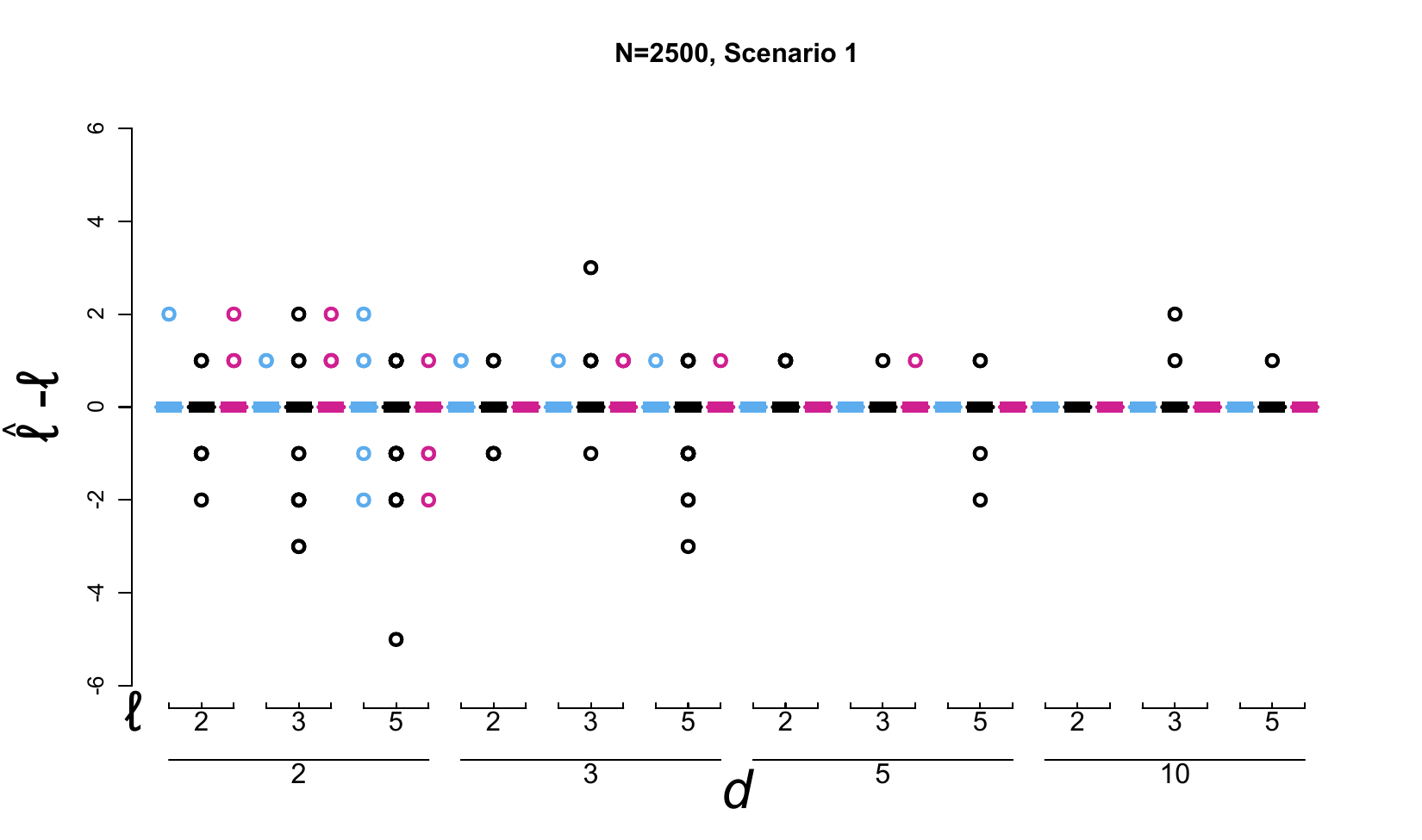}
\end{minipage}
\begin{minipage}[c]{.49\textwidth} 
\centering%
\includegraphics[width=\textwidth]{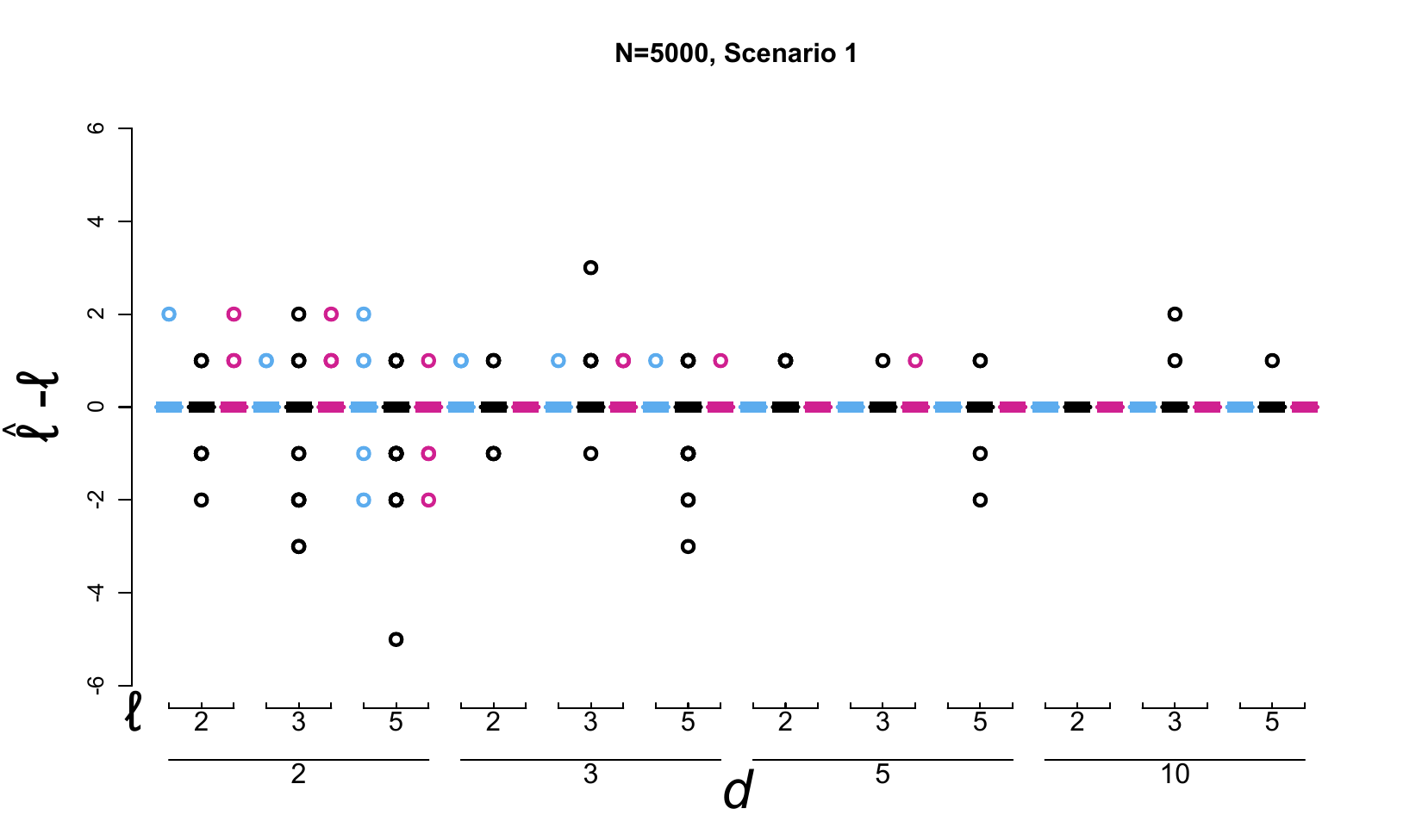}
\end{minipage}\hfill
\newline
\begin{minipage}[c]{.49\textwidth} 
\centering%
\includegraphics[width=\textwidth]{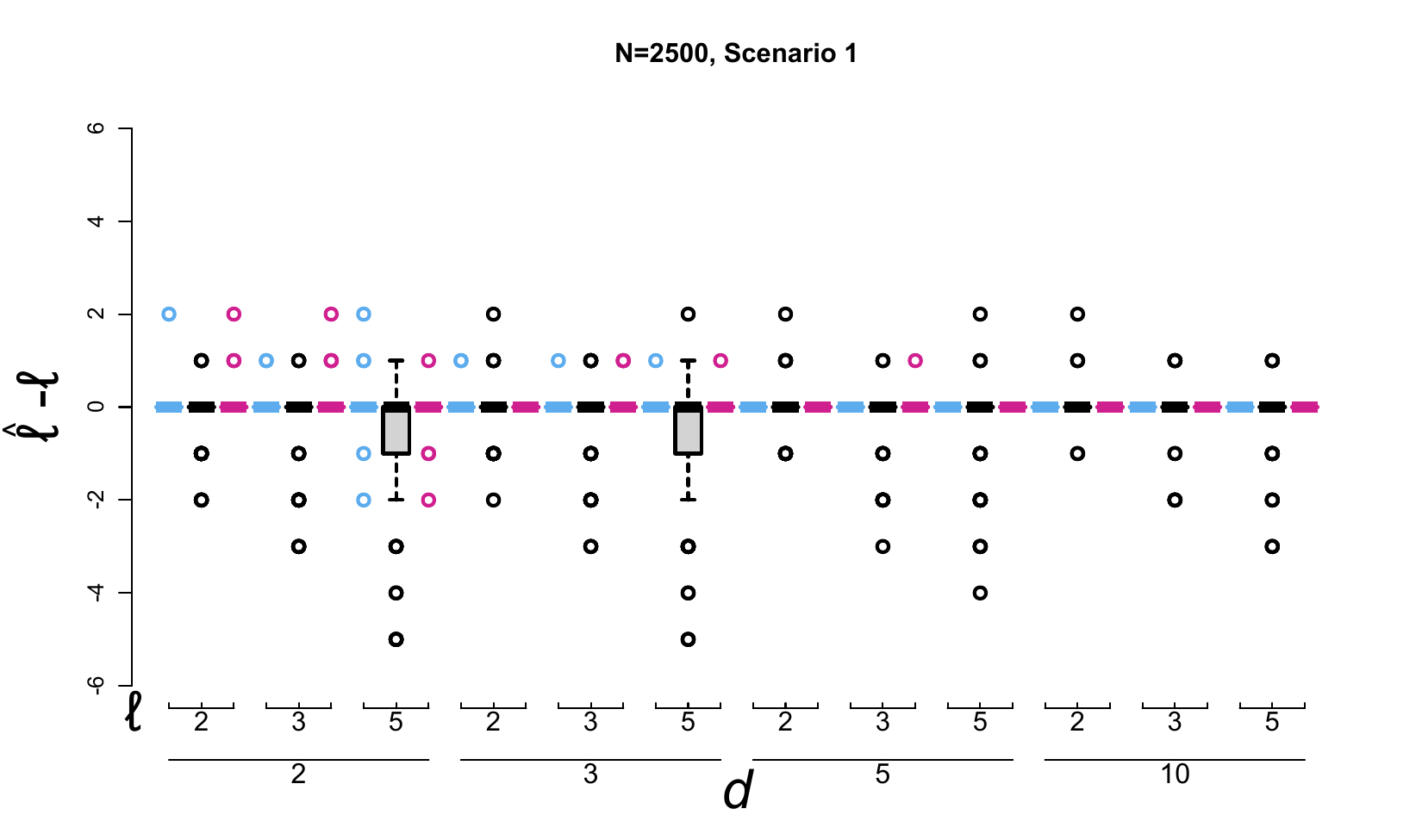}
\end{minipage}
\begin{minipage}[c]{.49\textwidth} 
\centering%
\includegraphics[width=\textwidth]{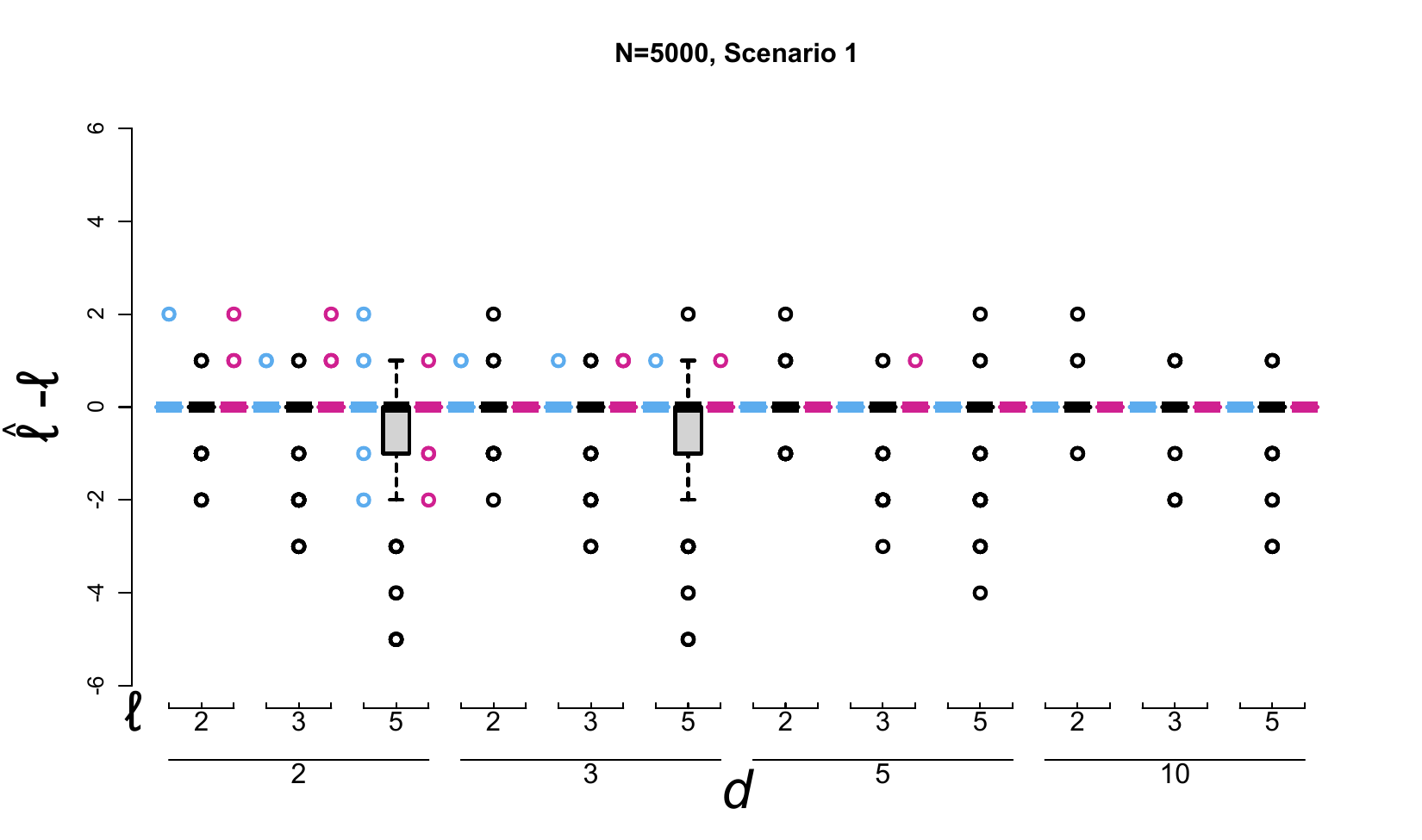}
\end{minipage}\hfill
\caption{Boxplots of $\hat{\ell}-\ell$ under the KW-PELT algorithm with $C_1=0.18$ and $C_2=3.74$ for (top row) halfspace depth, (second row) spatial depth, (third row) Mahalanobis depth, (last row) Modified Mahalanobis depth..}%
\end{figure}
\begin{figure}
\begin{minipage}[c]{.49\textwidth} 
\centering%
\includegraphics[width=\textwidth]{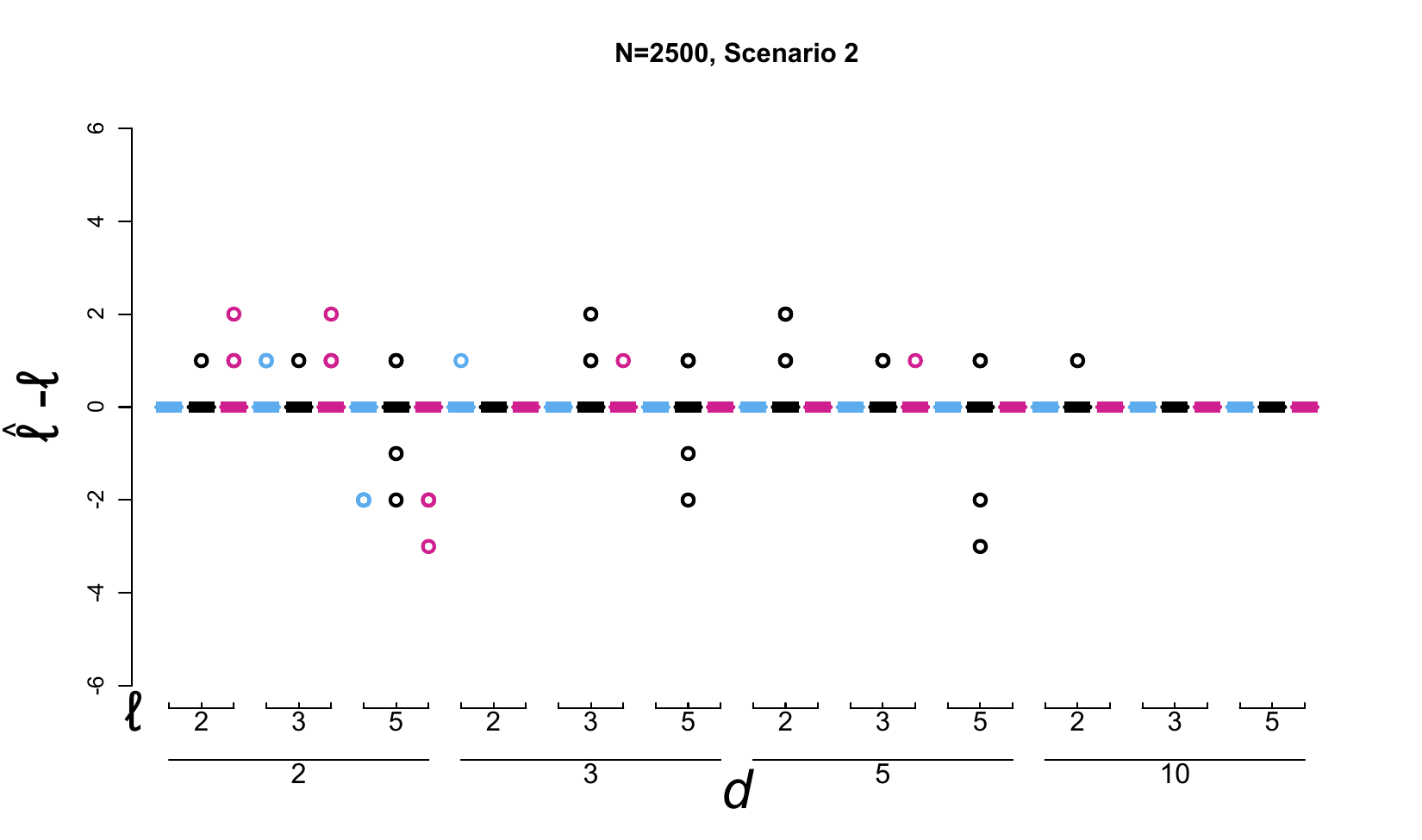}
\end{minipage}
\begin{minipage}[c]{.49\textwidth} 
\centering%
\includegraphics[width=\textwidth]{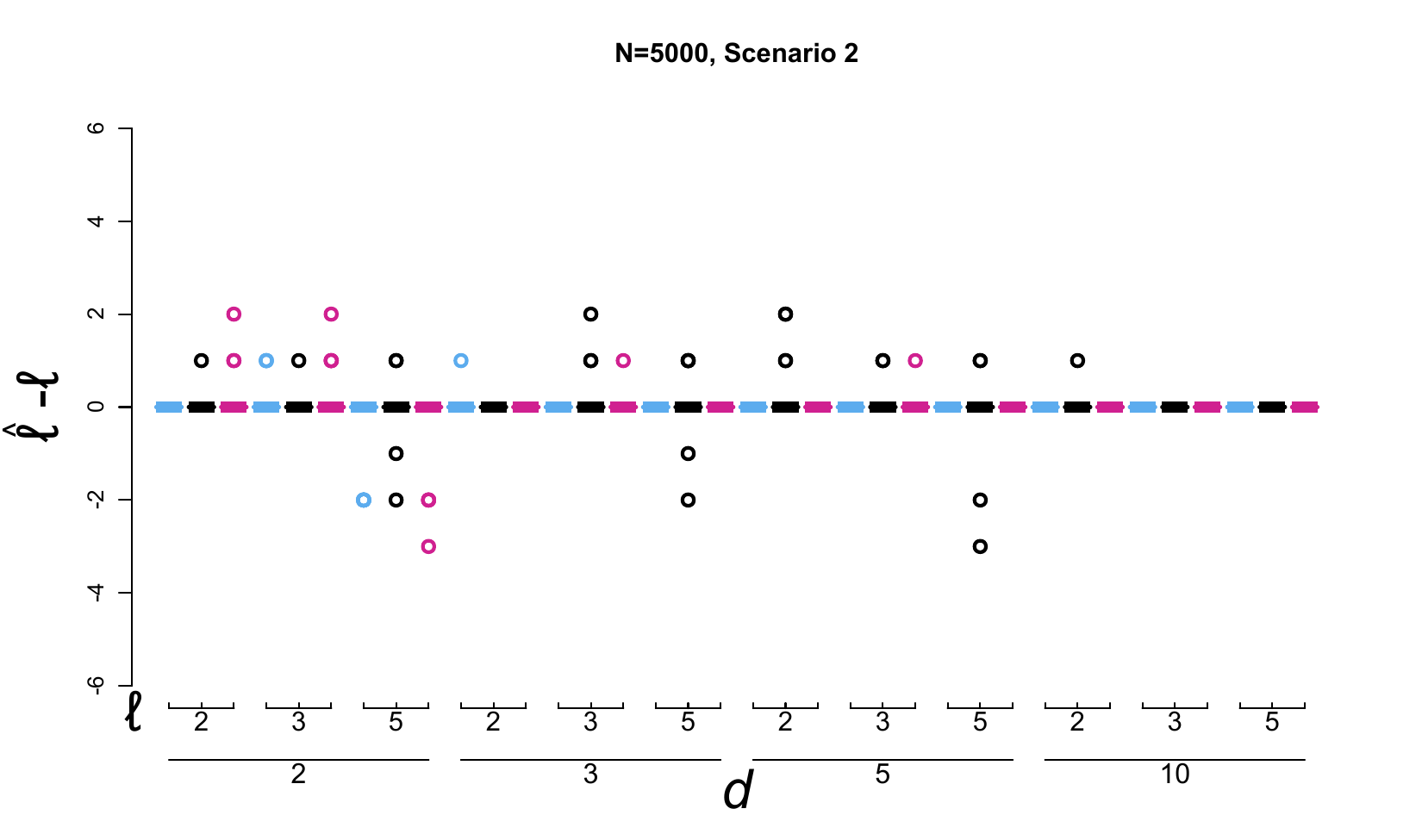}
\end{minipage}\hfill\newline
\begin{minipage}[c]{.49\textwidth} 
\centering%
\includegraphics[width=\textwidth]{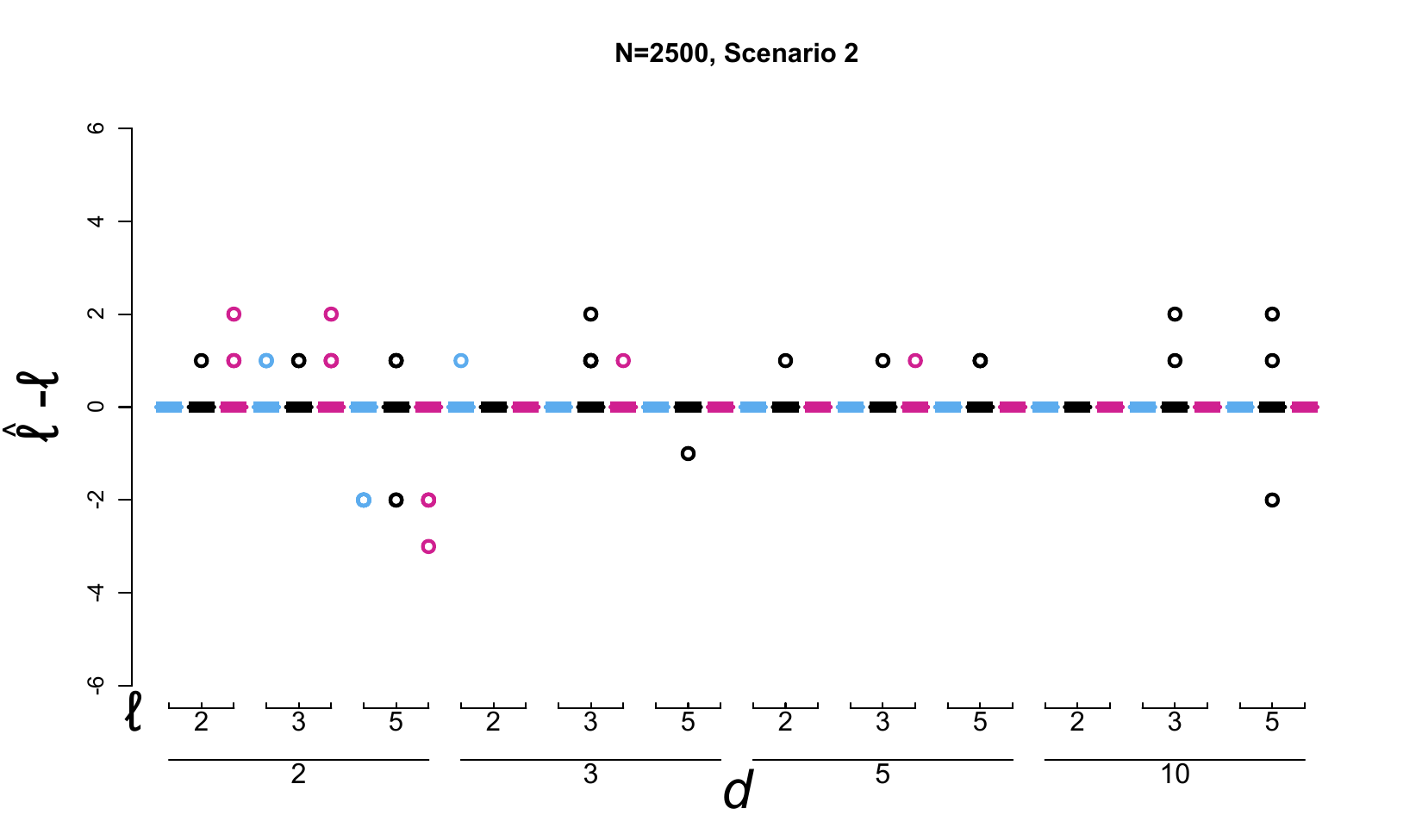}
\end{minipage}
\begin{minipage}[c]{.49\textwidth} 
\centering%
\includegraphics[width=\textwidth]{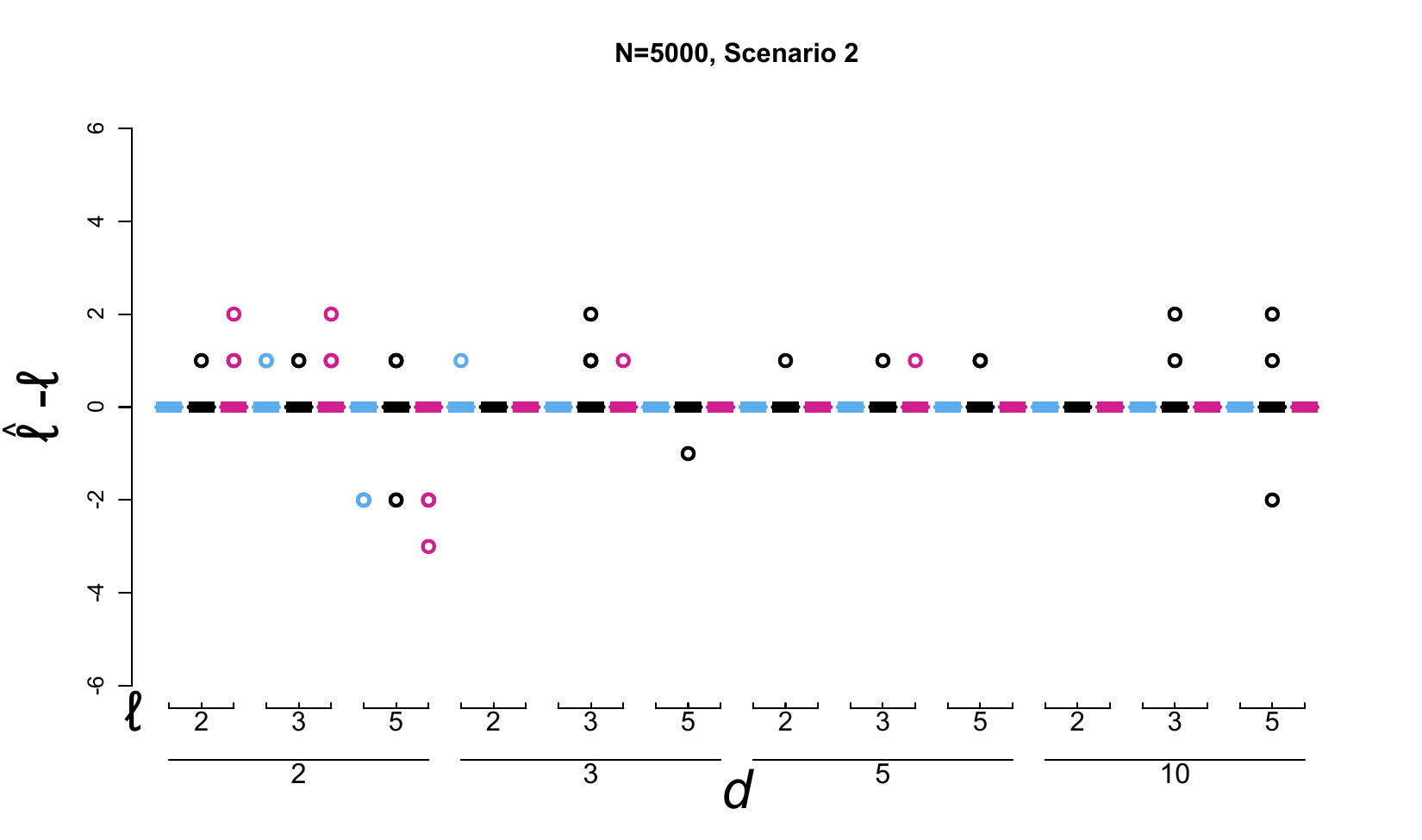}
\end{minipage}\hfill\newline
\begin{minipage}[c]{.49\textwidth} 
\centering%
\includegraphics[width=\textwidth]{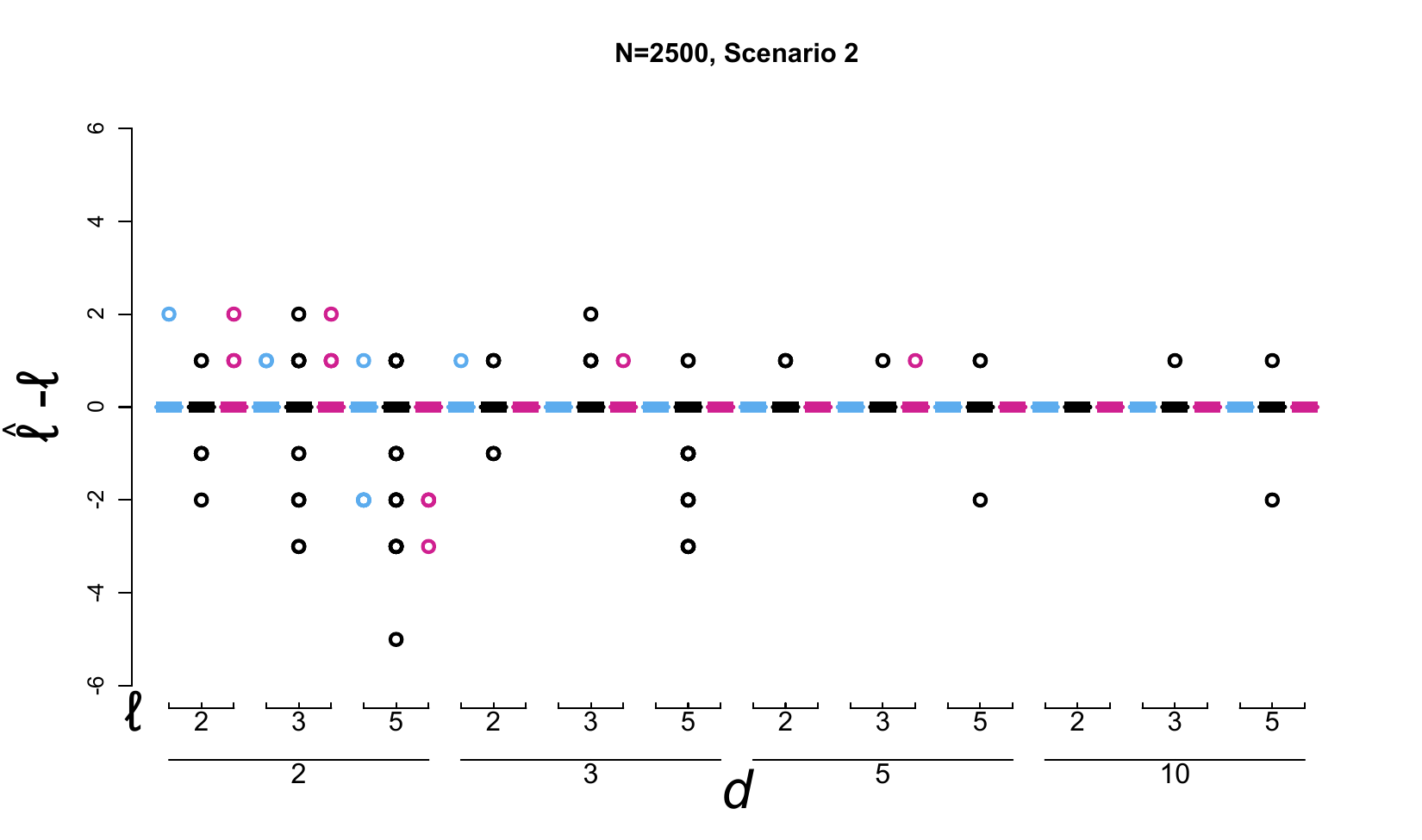}
\end{minipage}
\begin{minipage}[c]{.49\textwidth} 
\centering%
\includegraphics[width=\textwidth]{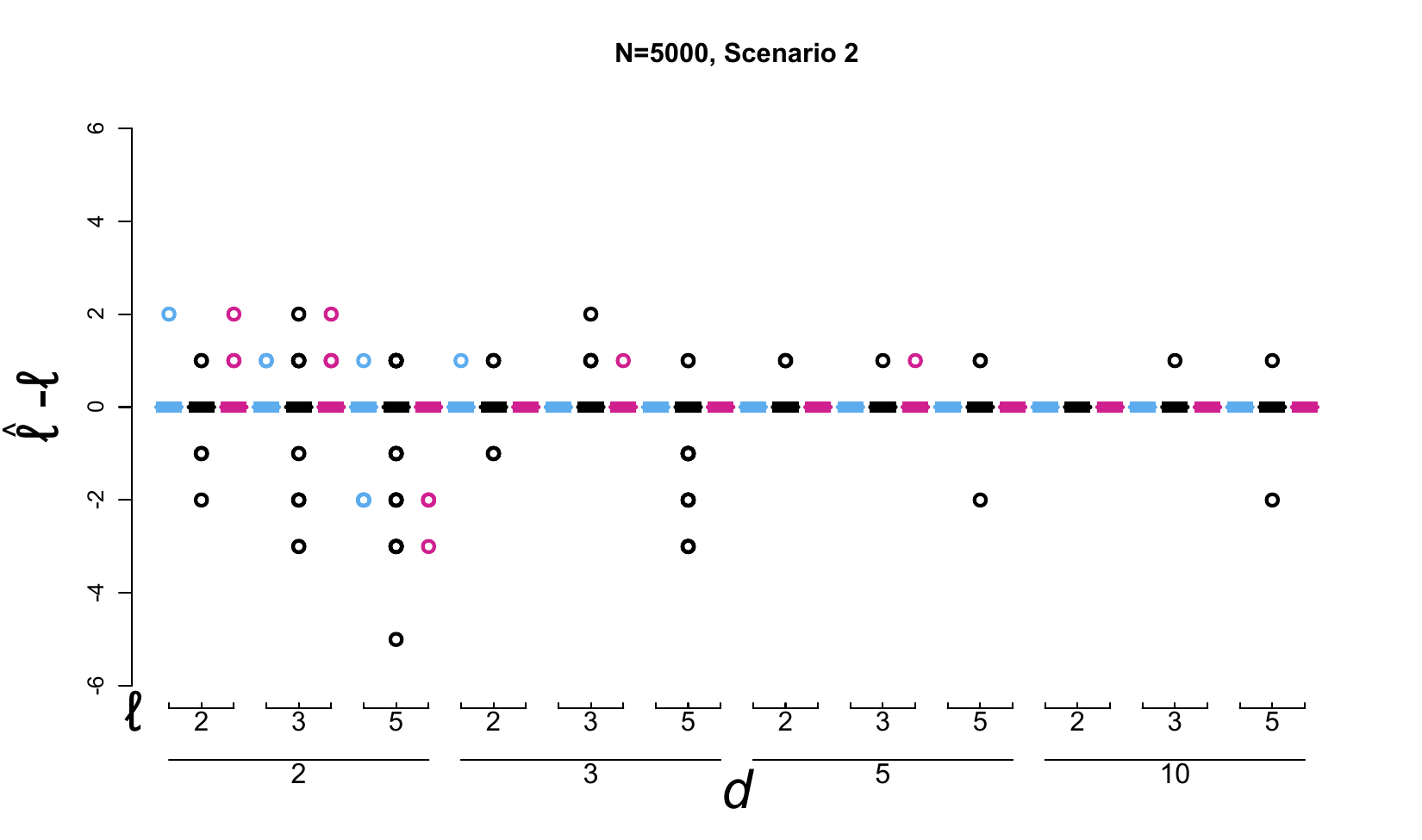}
\end{minipage}\hfill
\newline
\begin{minipage}[c]{.49\textwidth} 
\centering%
\includegraphics[width=\textwidth]{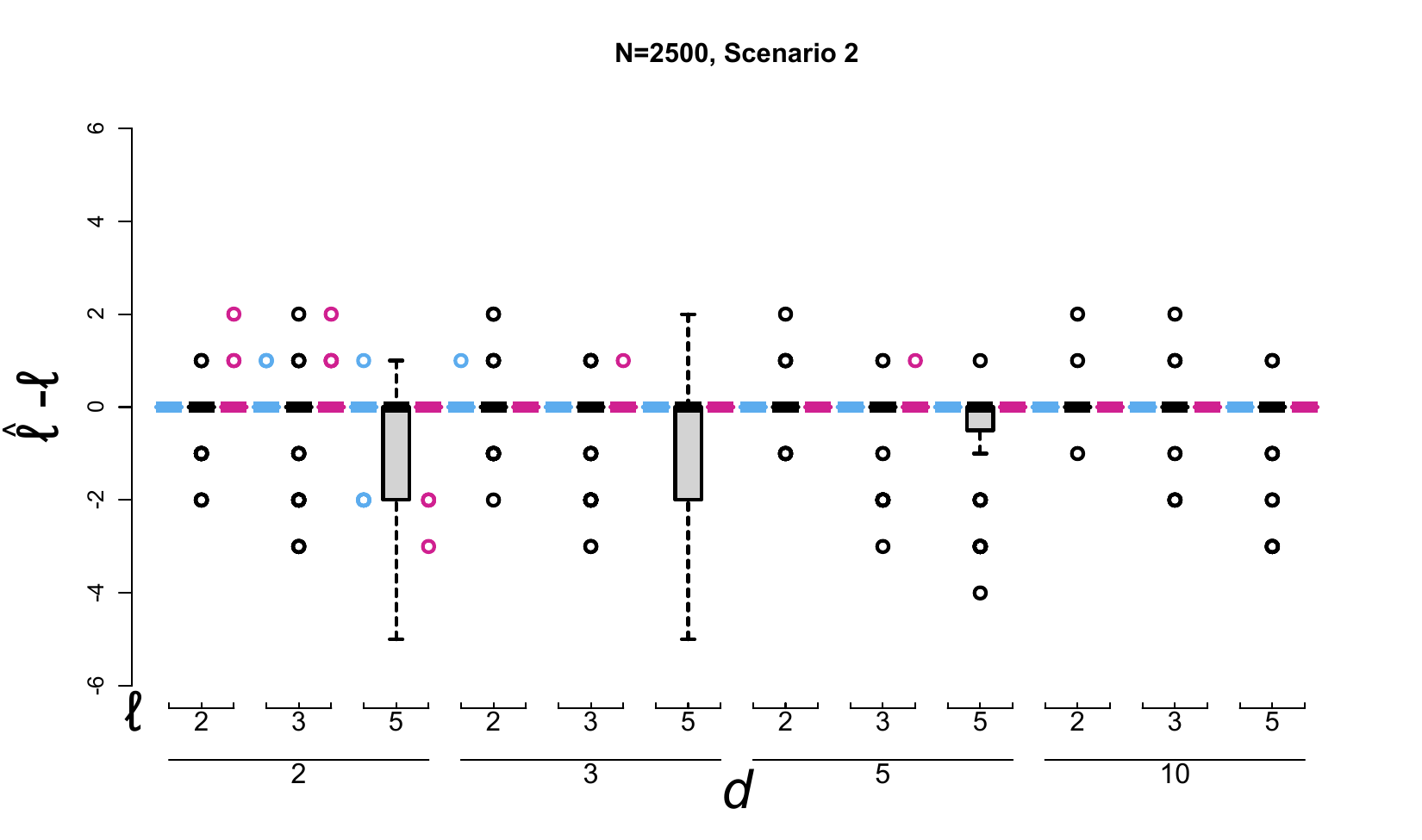}
\end{minipage}
\begin{minipage}[c]{.49\textwidth} 
\centering%
\includegraphics[width=\textwidth]{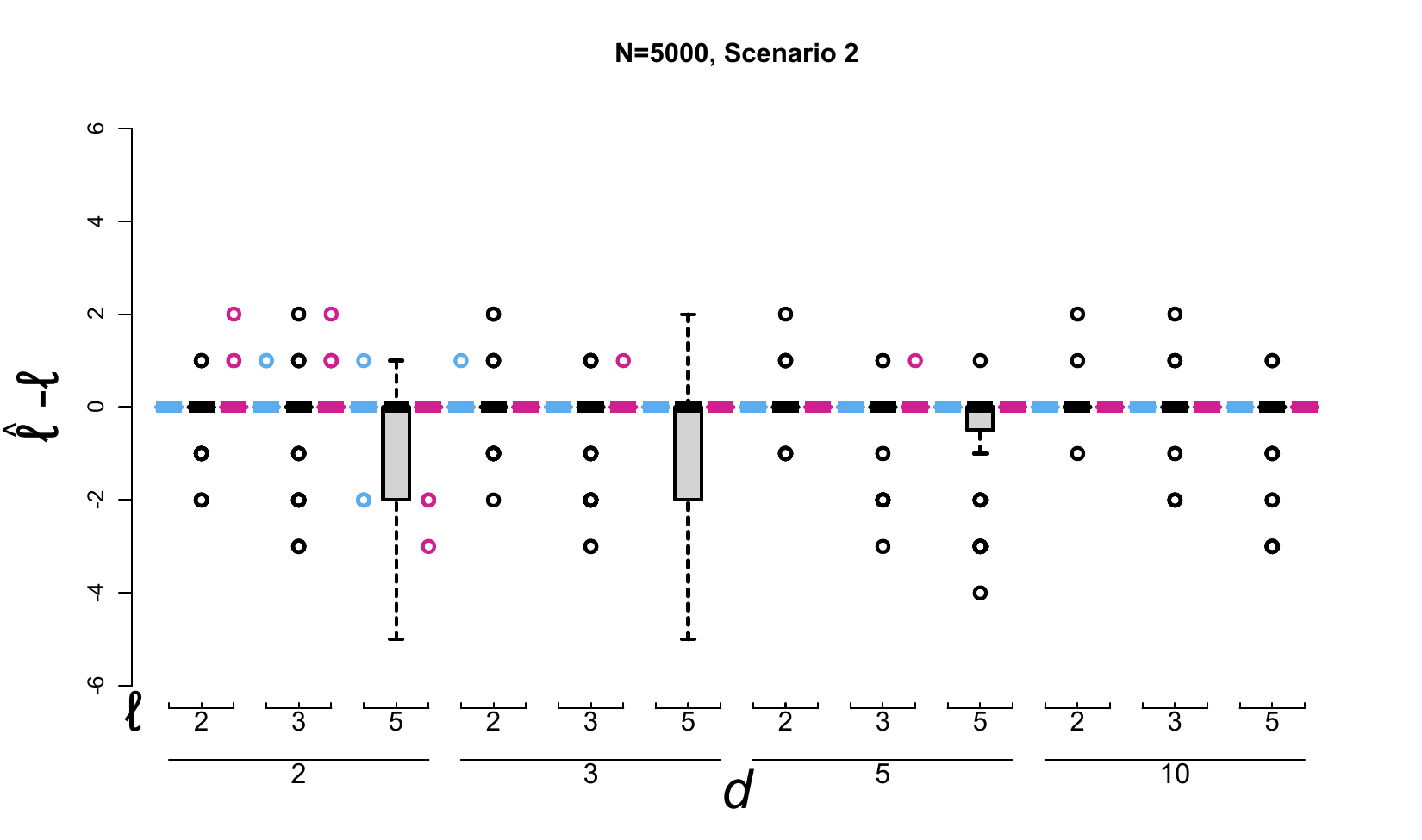}
\end{minipage}\hfill
\caption{Boxplots of $\hat{\ell}-\ell$ under the KW-PELT algorithm with $C_1=0.18$ and $C_2=3.74$ for (top row) halfspace depth, (second row) spatial depth, (third row) Mahalanobis depth, (last row) Modified Mahalanobis depth.}%
\end{figure}
\begin{figure}
\begin{minipage}[c]{.49\textwidth} 
\centering%
\includegraphics[width=\textwidth]{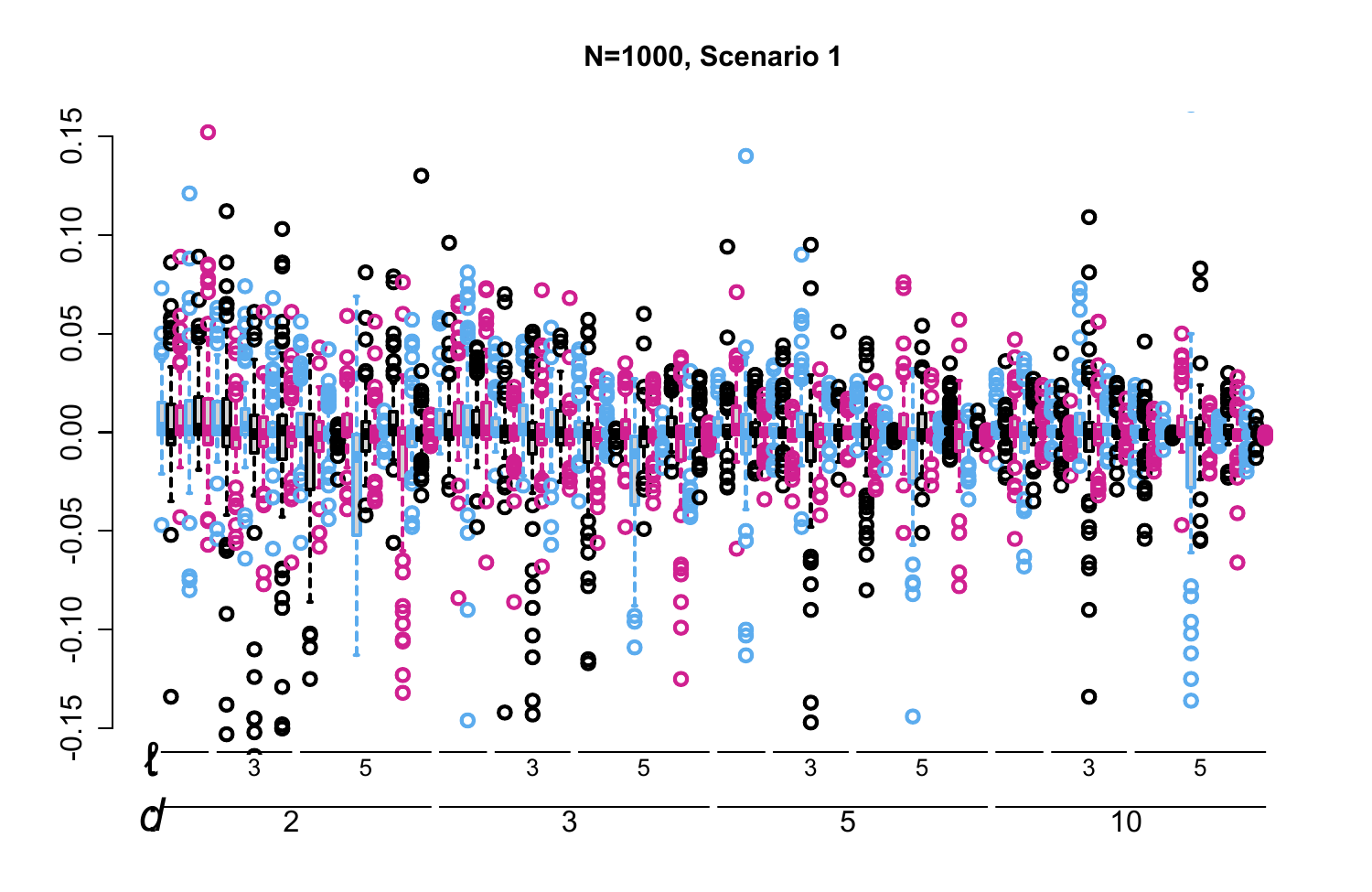}
\end{minipage}
\begin{minipage}[c]{.49\textwidth} 
\centering%
\includegraphics[width=\textwidth]{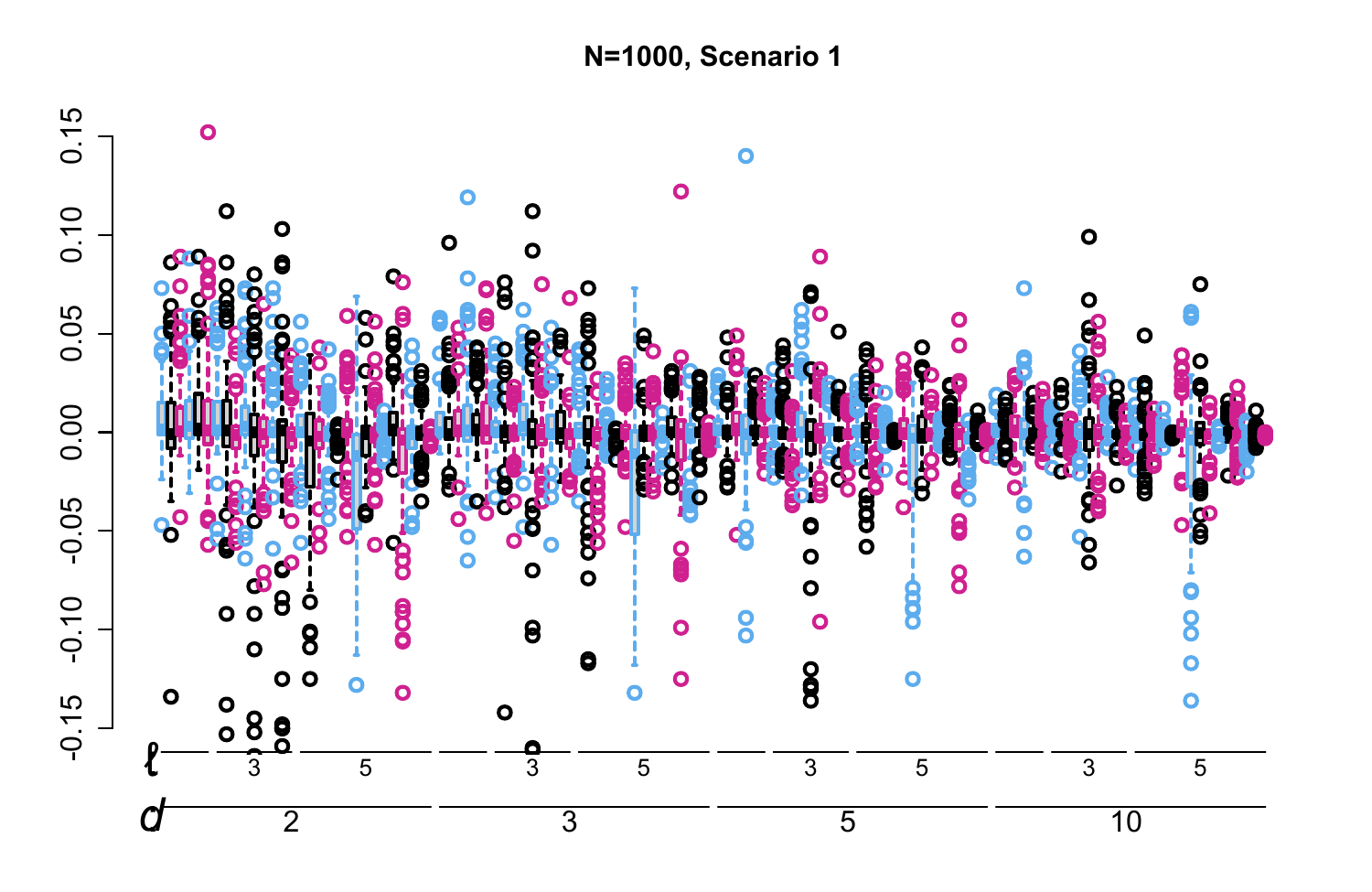}
\end{minipage}\hfill\newline
\begin{minipage}[c]{.49\textwidth} 
\centering%
\includegraphics[width=\textwidth]{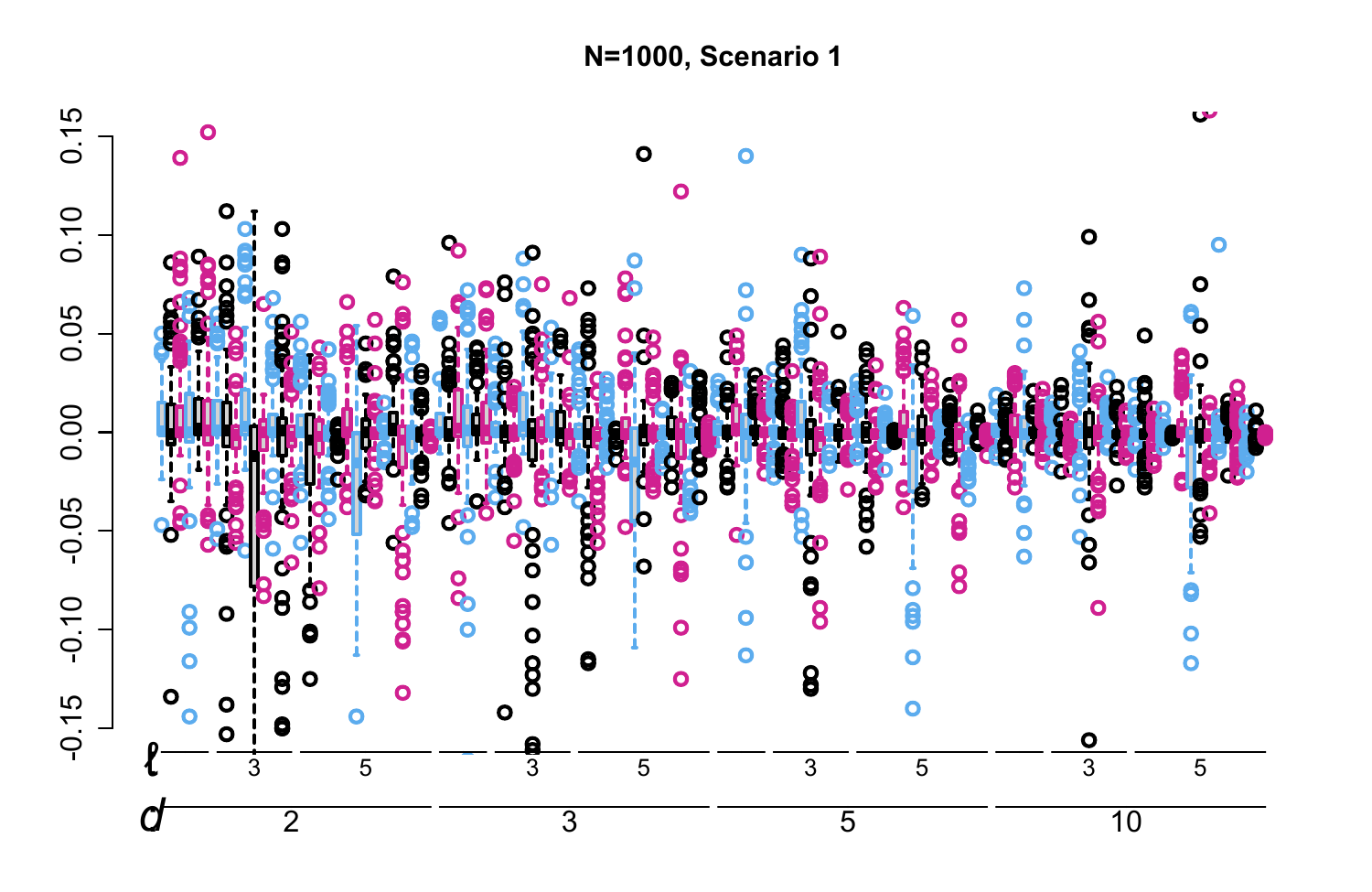}
\end{minipage}
\begin{minipage}[c]{.49\textwidth} 
\centering%
\includegraphics[width=\textwidth]{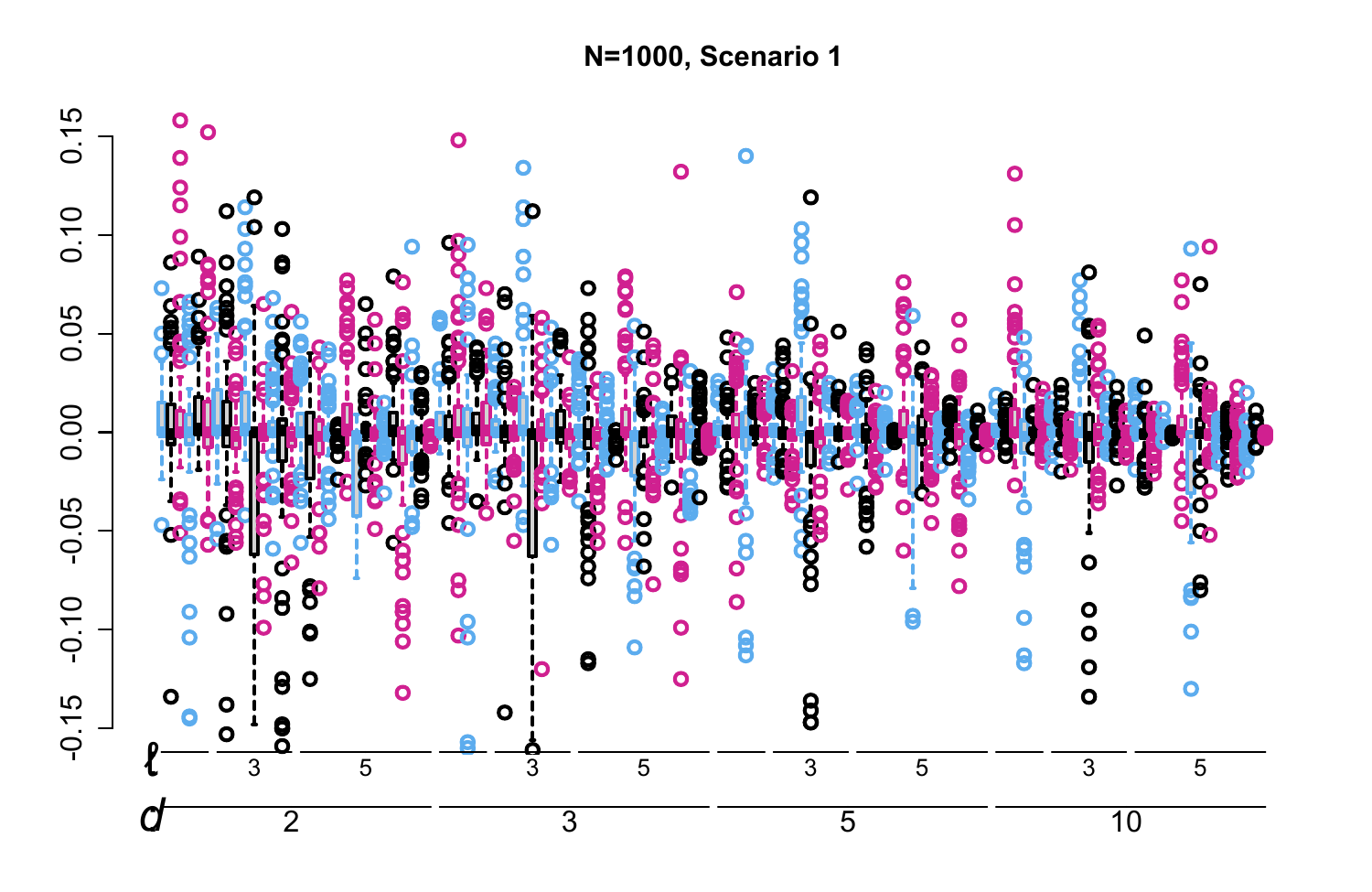}
\end{minipage}
\hfill\newline
\begin{minipage}[c]{.49\textwidth} 
\centering%
\includegraphics[width=\textwidth]{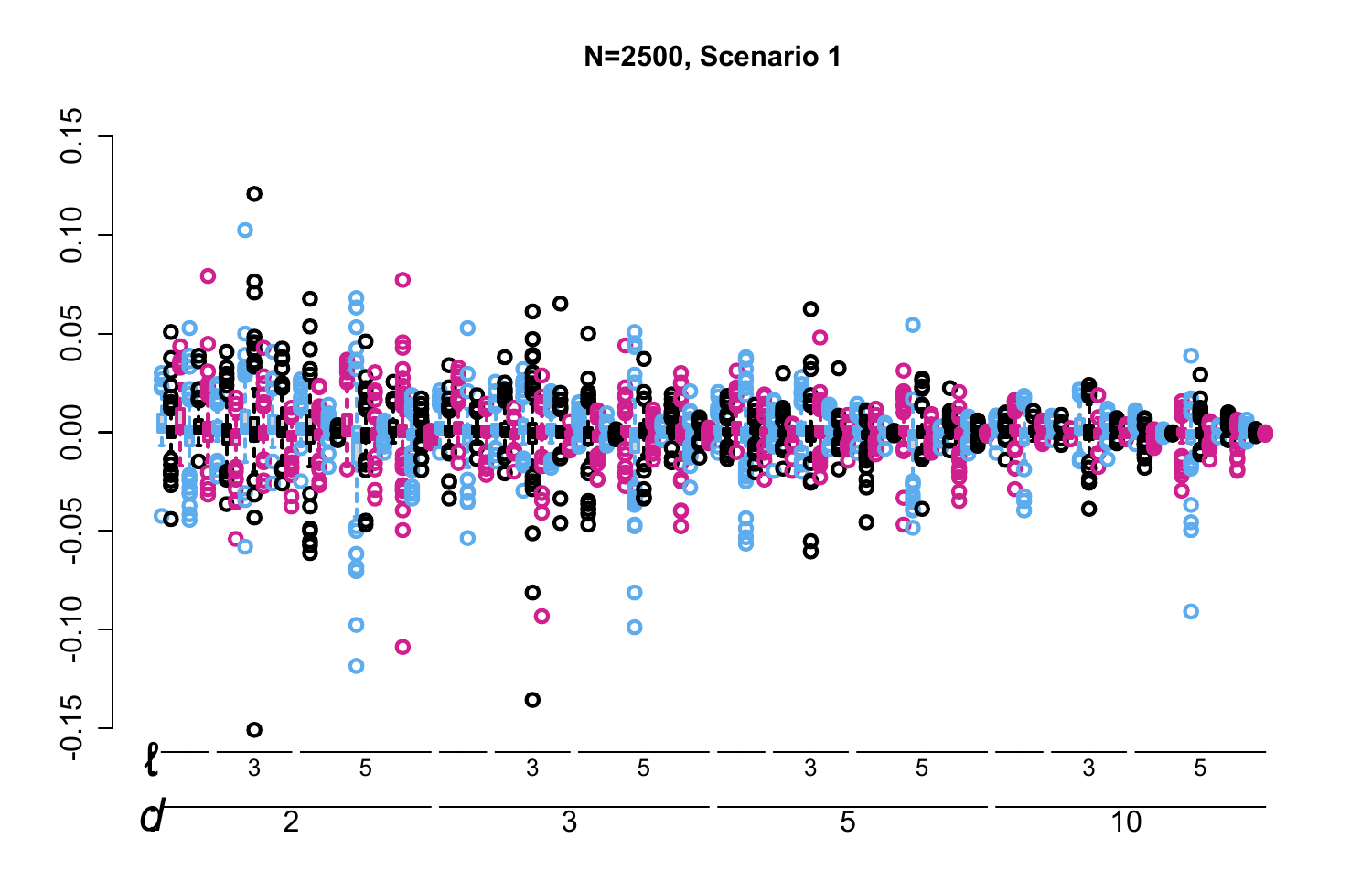}
\end{minipage}
\begin{minipage}[c]{.49\textwidth} 
\centering%
\includegraphics[width=\textwidth]{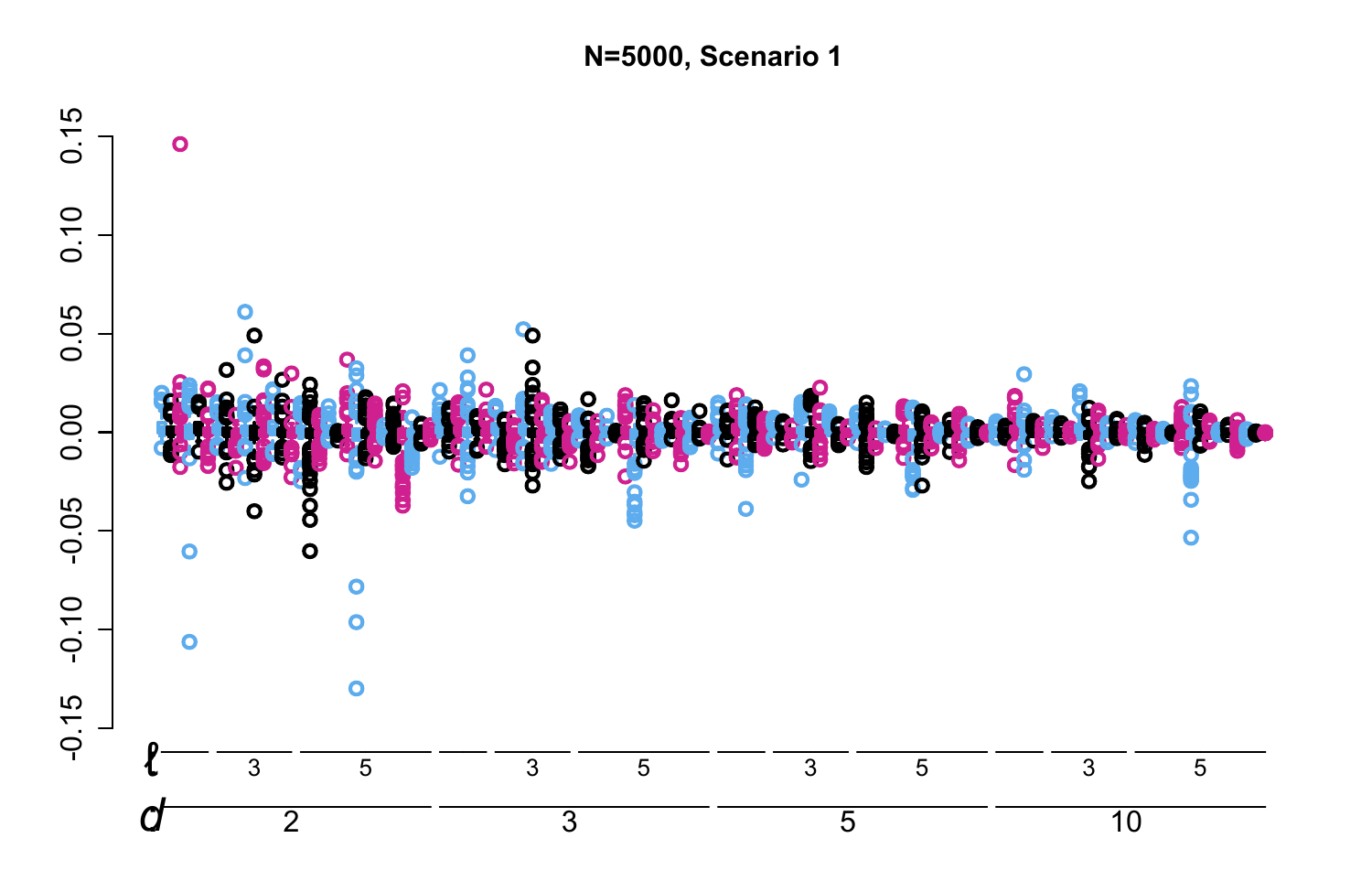}
\end{minipage}
\caption{Distribution of $\widehat{k}/N-\theta$ under the WBS algorithm for halfspace, spatial, Mahalanobis and modified Mahalanobis depth respectively. Followed by Mahalanobis depth for larger values of $N$.}%
\end{figure}
\begin{figure}
\begin{minipage}[c]{.49\textwidth} 
\centering%
\includegraphics[width=\textwidth]{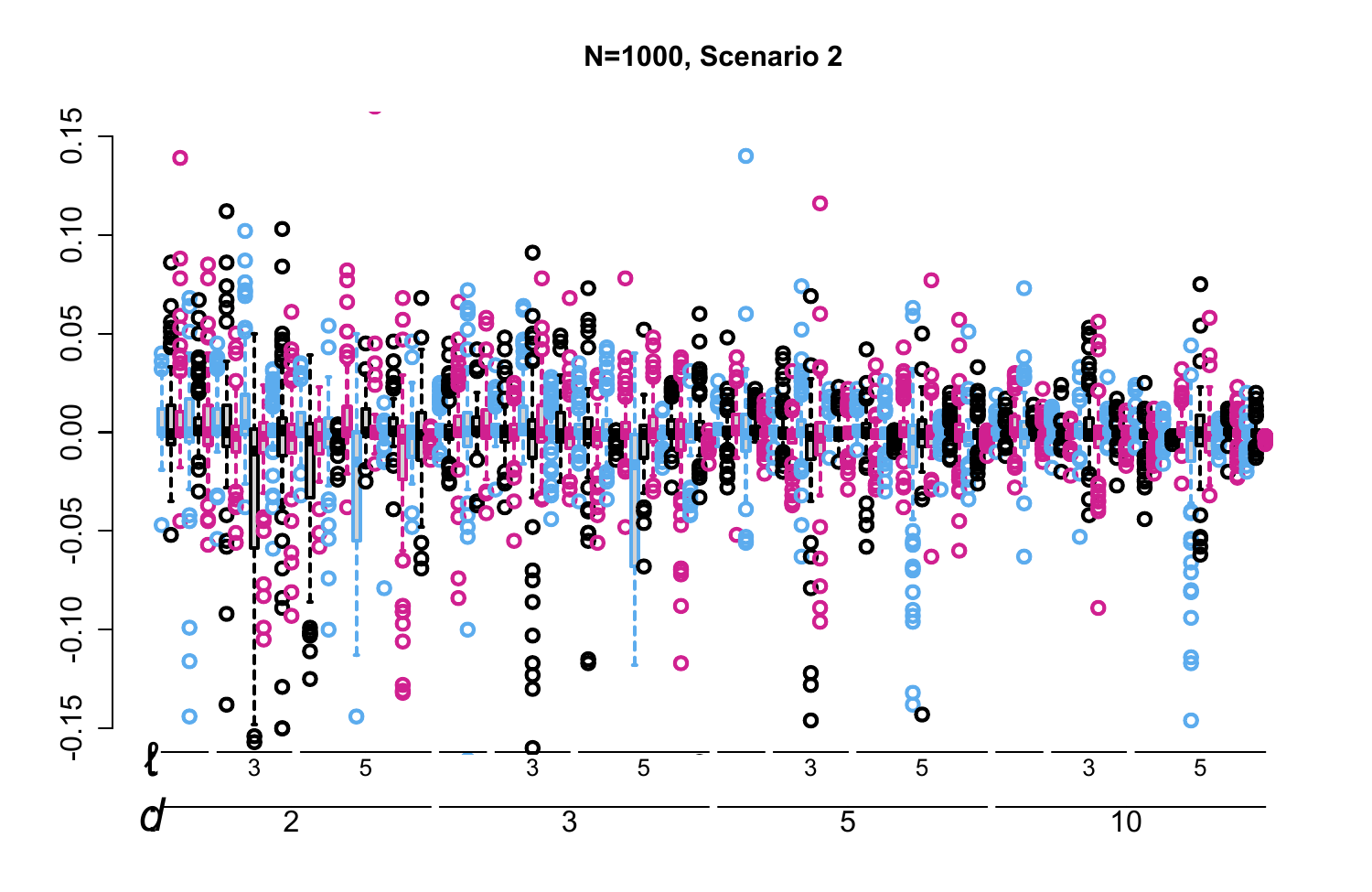}
\end{minipage}
\begin{minipage}[c]{.49\textwidth} 
\centering%
\includegraphics[width=\textwidth]{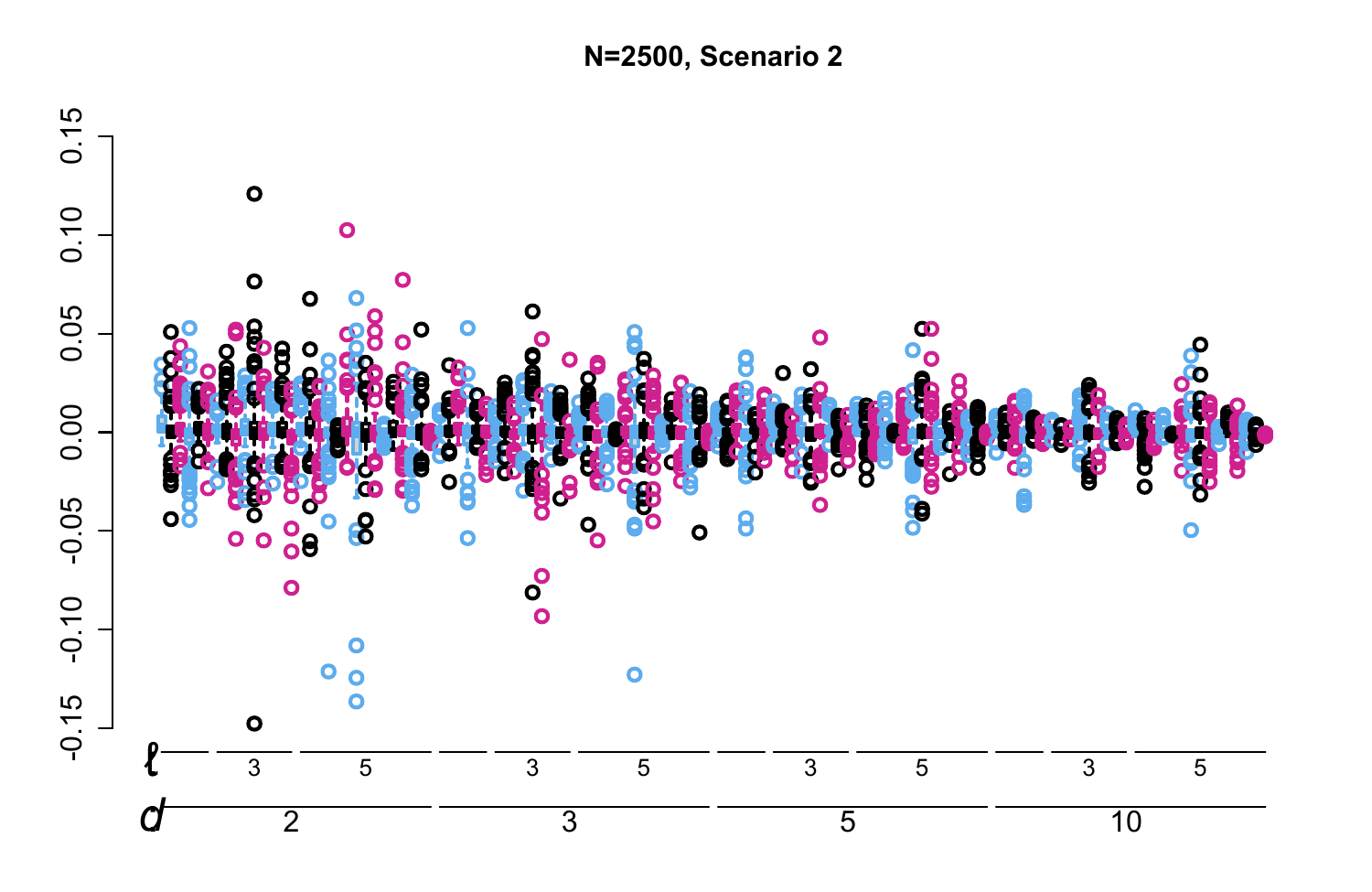}
\end{minipage}
\hfill\newline
\begin{minipage}[c]{.49\textwidth} 
\centering%
\includegraphics[width=\textwidth]{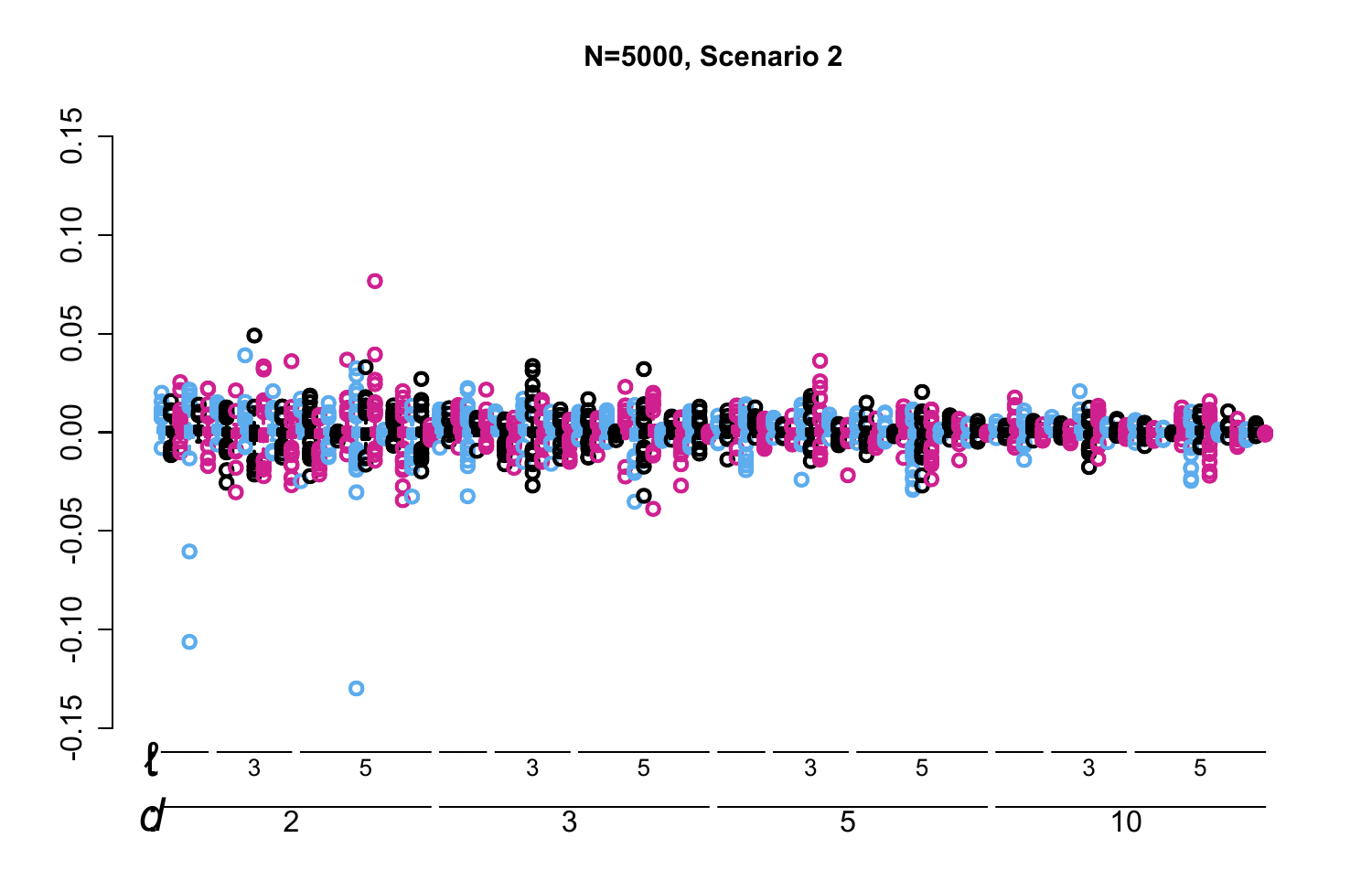}
\end{minipage}
\caption{Distribution of $\widehat{k}/N-\theta$ under the WBS algorithm for Mahalanobis depth.}%
\end{figure}

\begin{figure}
\begin{minipage}[c]{.49\textwidth} 
\centering%
\includegraphics[width=\textwidth]{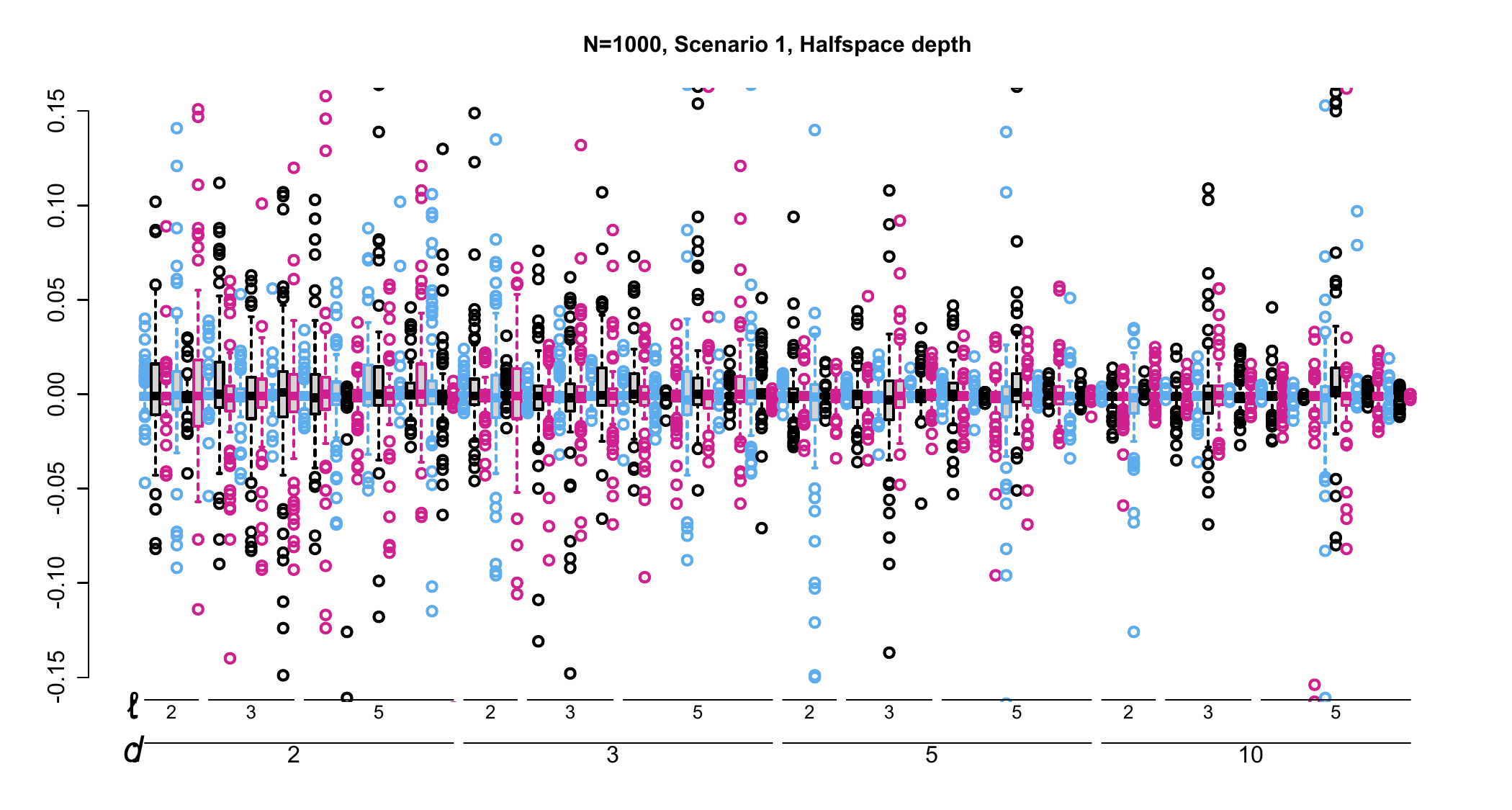}
\end{minipage}
\begin{minipage}[c]{.49\textwidth} 
\centering%
\includegraphics[width=\textwidth]{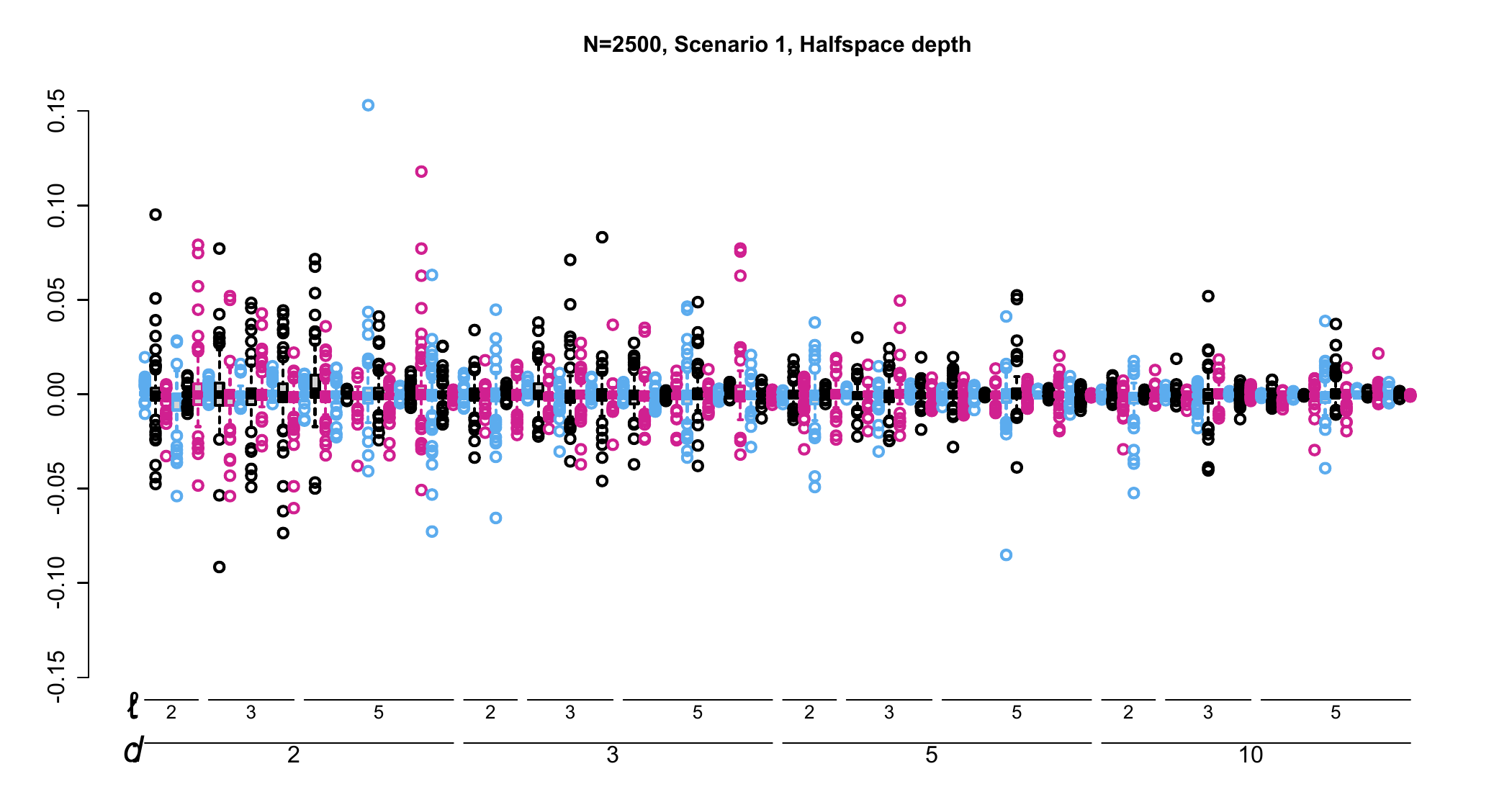}
\end{minipage}\hfill\newline
\begin{minipage}[c]{.49\textwidth} 
\centering%
\includegraphics[width=\textwidth]{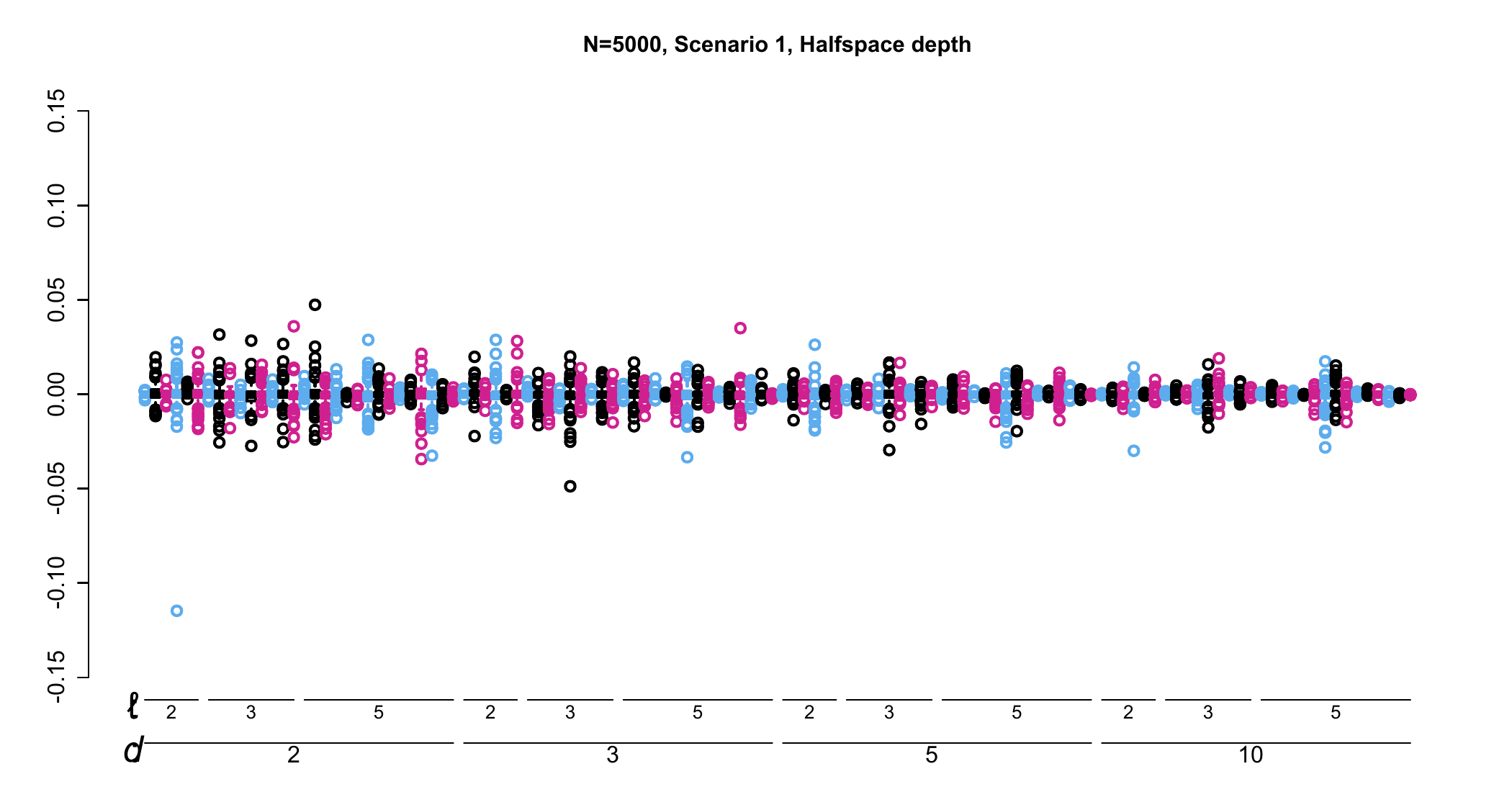}
\end{minipage}
\begin{minipage}[c]{.49\textwidth} 
\centering%
\includegraphics[width=\textwidth]{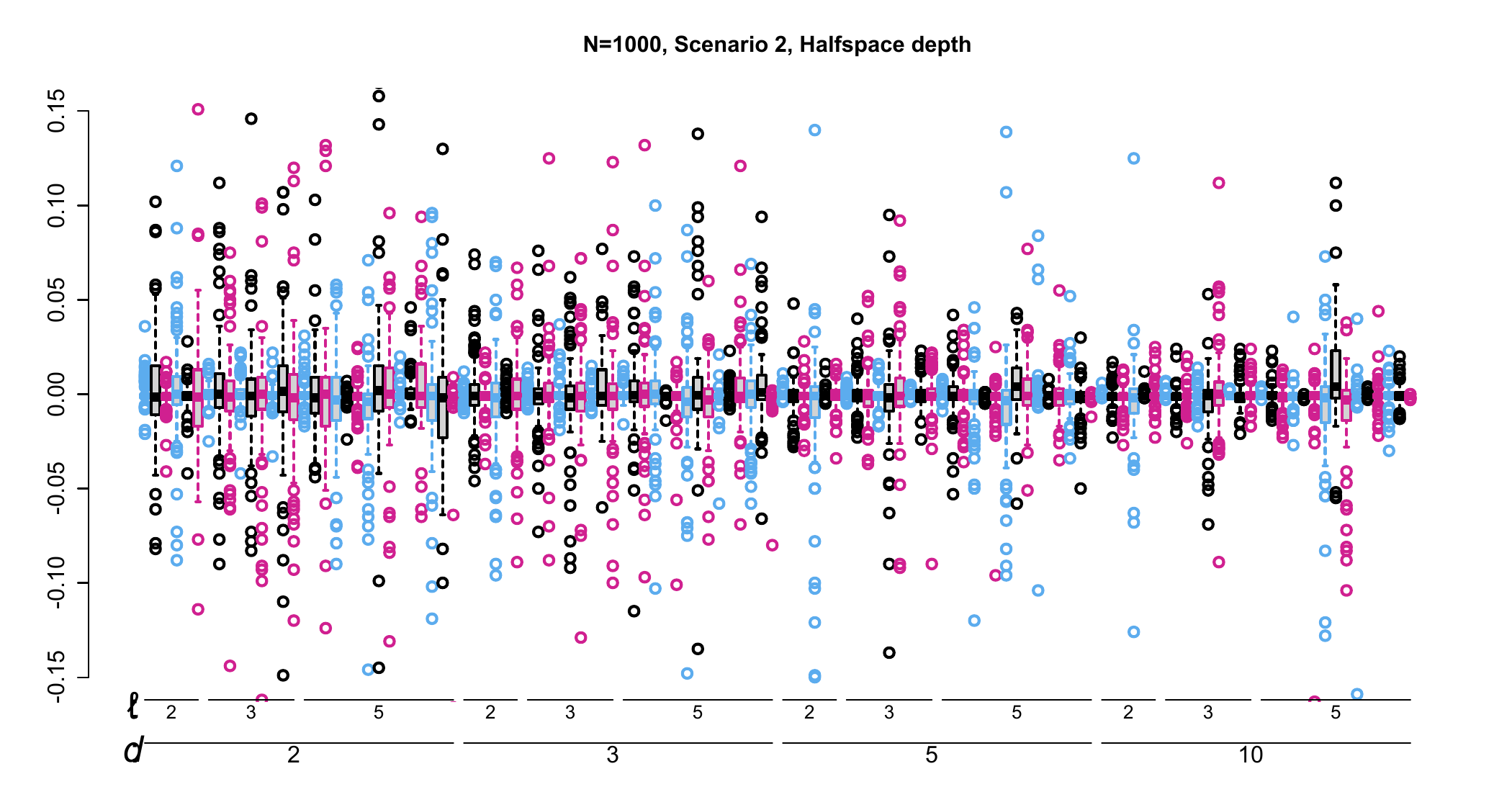}
\end{minipage}\hfill\newline
\begin{minipage}[c]{.49\textwidth} 
\centering%
\includegraphics[width=\textwidth]{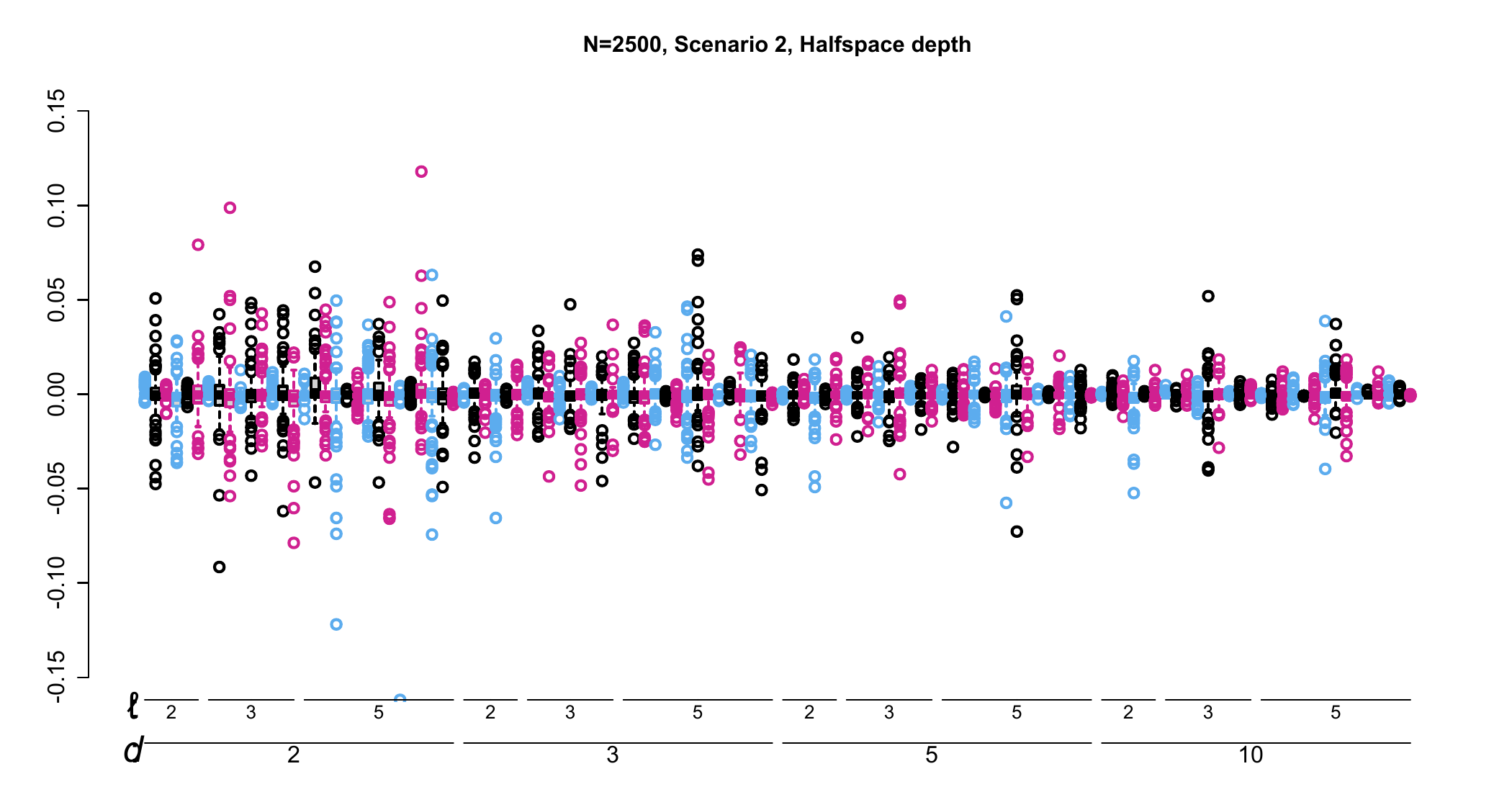}
\end{minipage}
\begin{minipage}[c]{.49\textwidth} 
\centering%
\includegraphics[width=\textwidth]{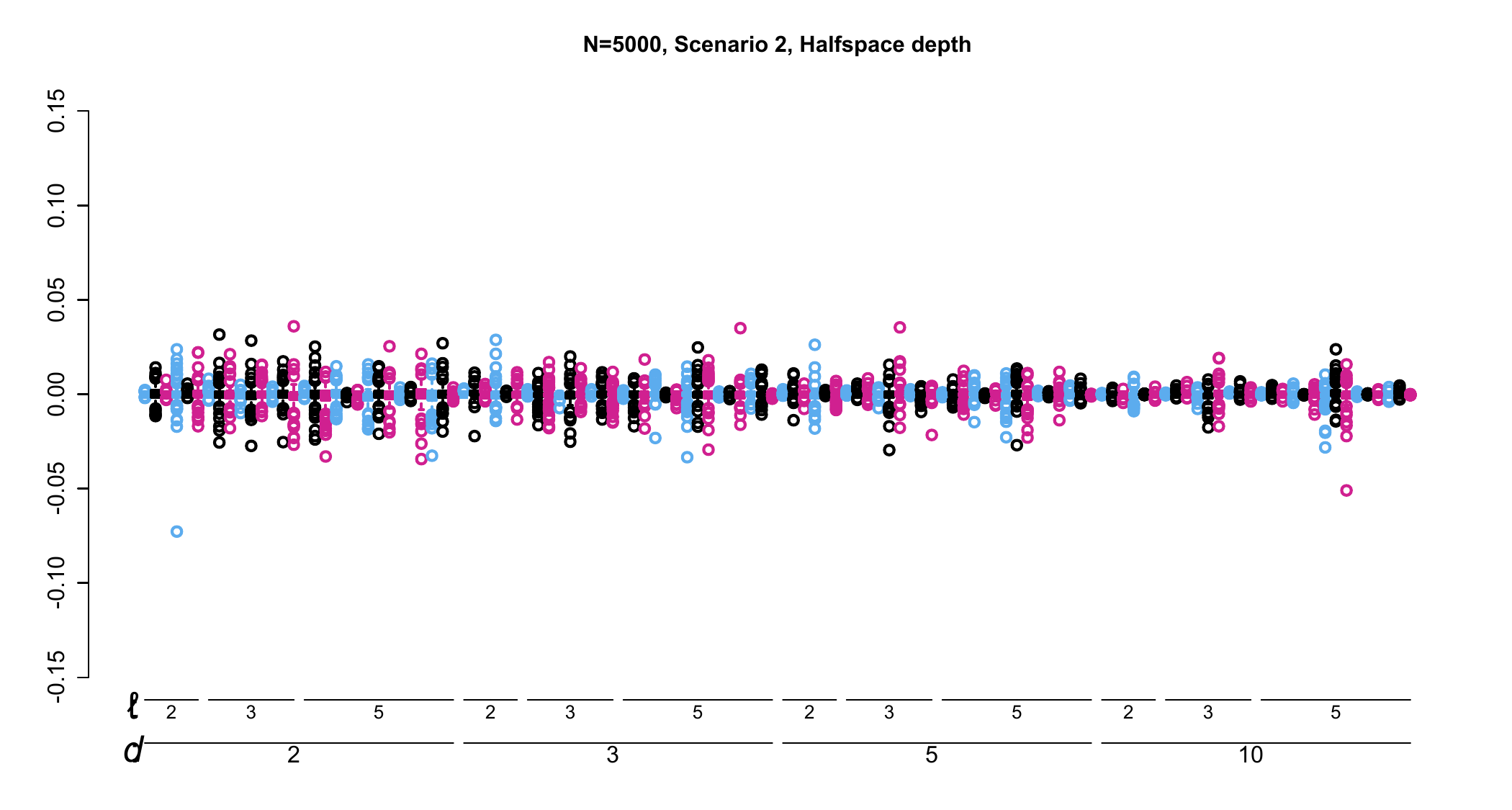}
\end{minipage}
\caption{Distribution of $\widehat{k}/N-\theta$ under the KW-PELT algorithm for halfspace depth.}%
\end{figure}
\begin{figure}
\begin{minipage}[c]{.49\textwidth} 
\centering%
\includegraphics[width=\textwidth]{Plots/K_BP_Rank/kbp_HS_KW_Big_N_1000_OC.pdf}
\end{minipage}
\begin{minipage}[c]{.49\textwidth} 
\centering%
\includegraphics[width=\textwidth]{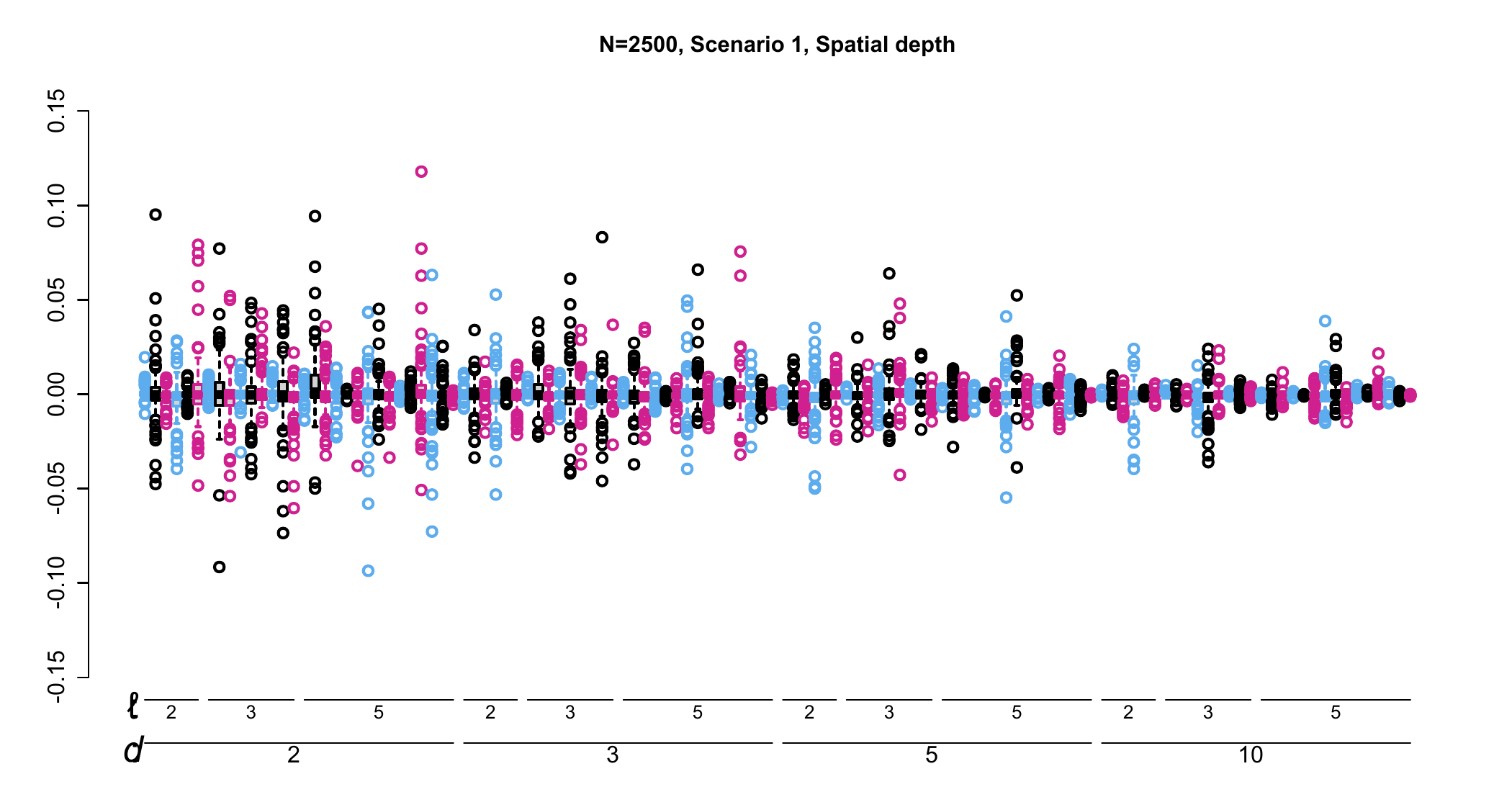}
\end{minipage}\hfill\newline
\begin{minipage}[c]{.49\textwidth} 
\centering%
\includegraphics[width=\textwidth]{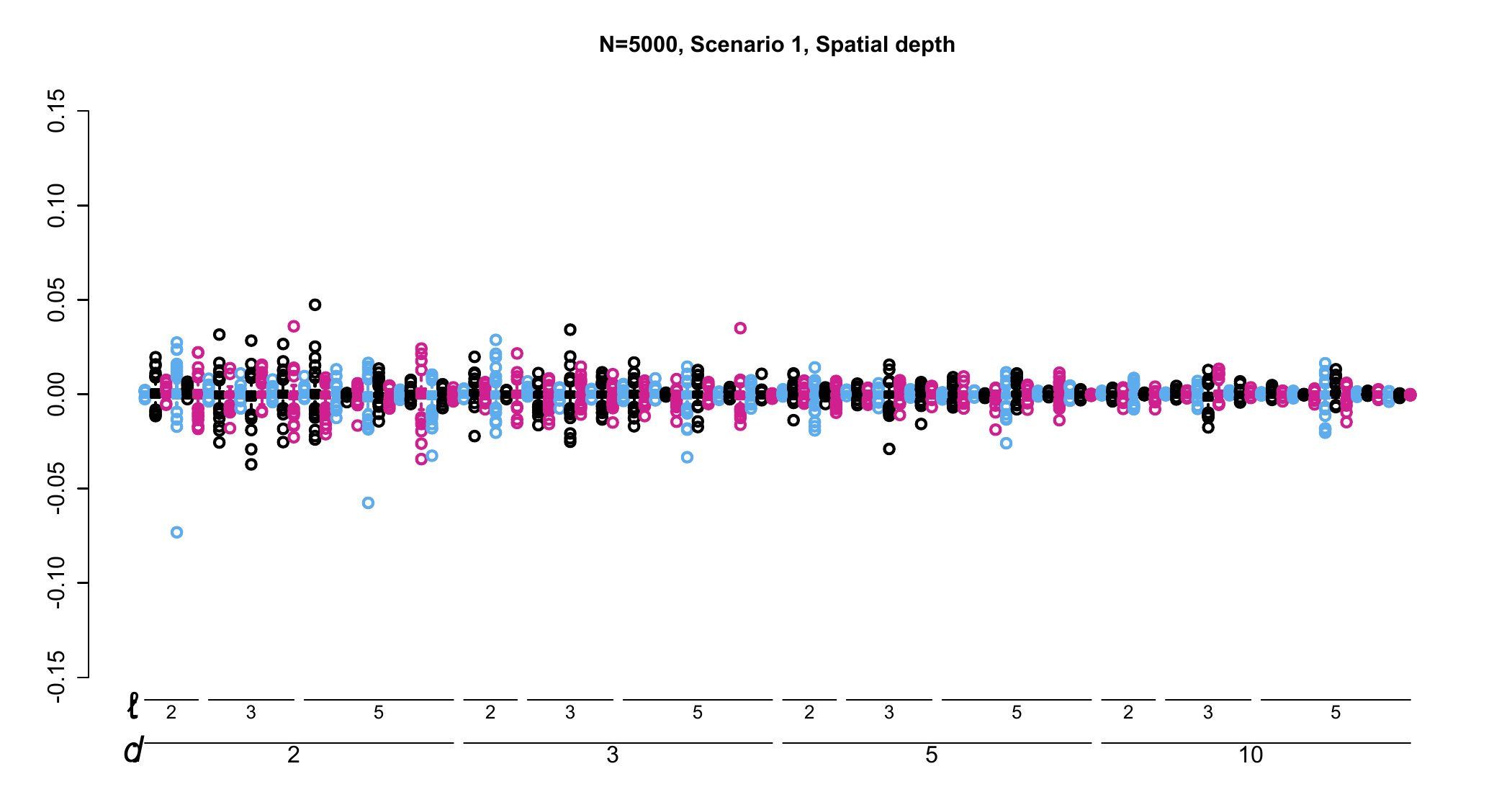}
\end{minipage}
\begin{minipage}[c]{.49\textwidth} 
\centering%
\includegraphics[width=\textwidth]{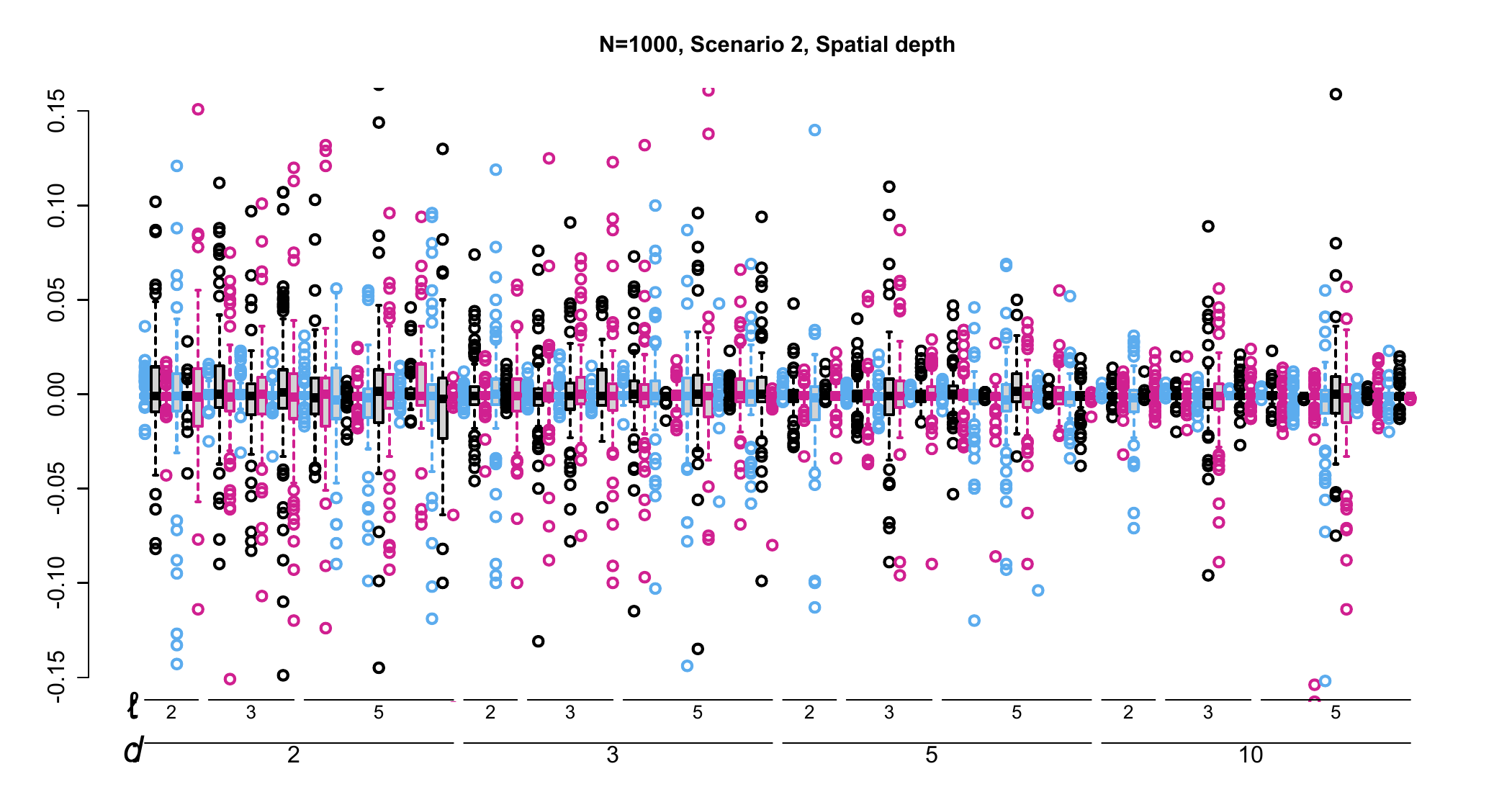}
\end{minipage}\hfill\newline
\begin{minipage}[c]{.49\textwidth} 
\centering%
\includegraphics[width=\textwidth]{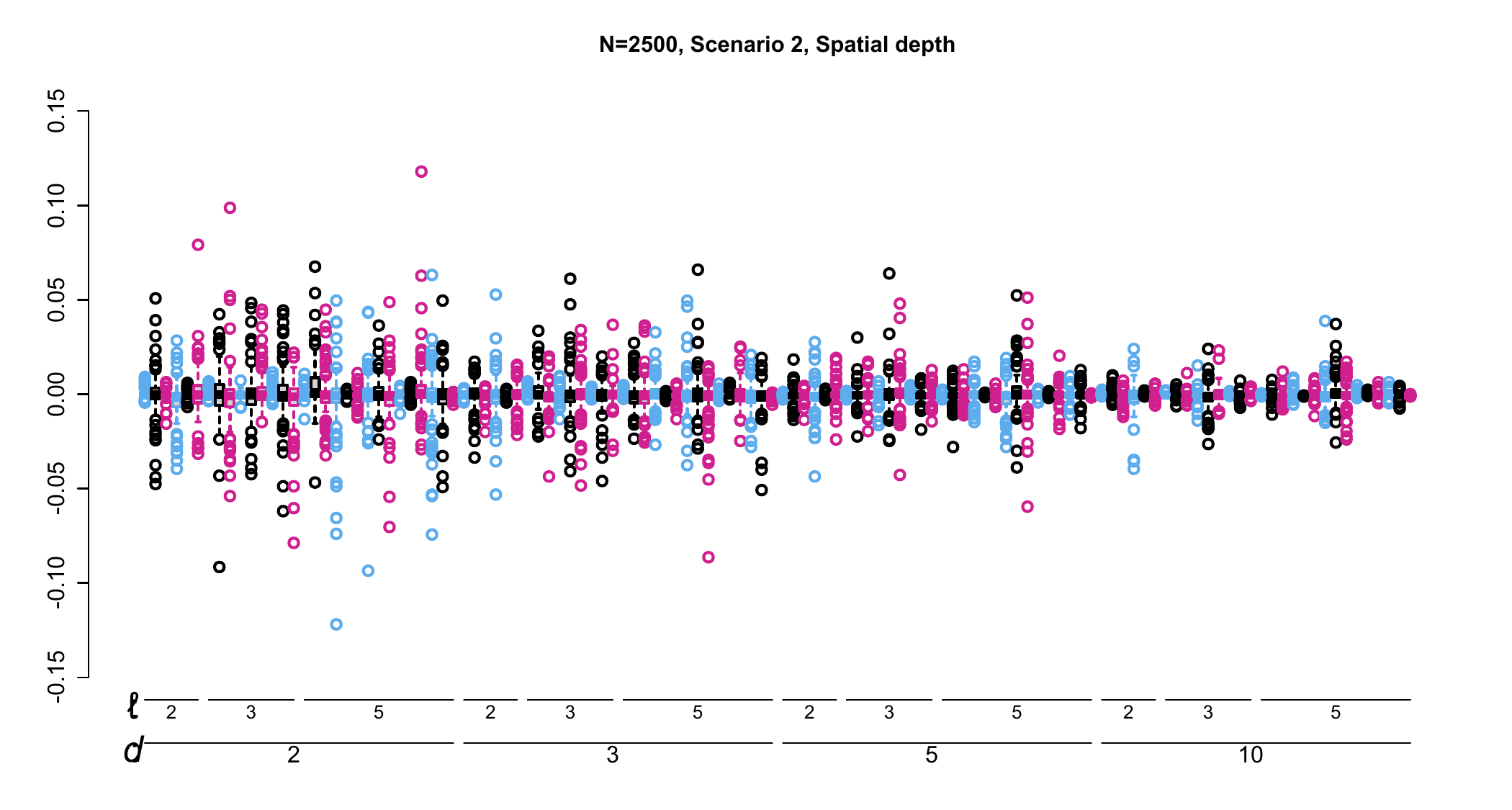}
\end{minipage}
\begin{minipage}[c]{.49\textwidth} 
\centering%
\includegraphics[width=\textwidth]{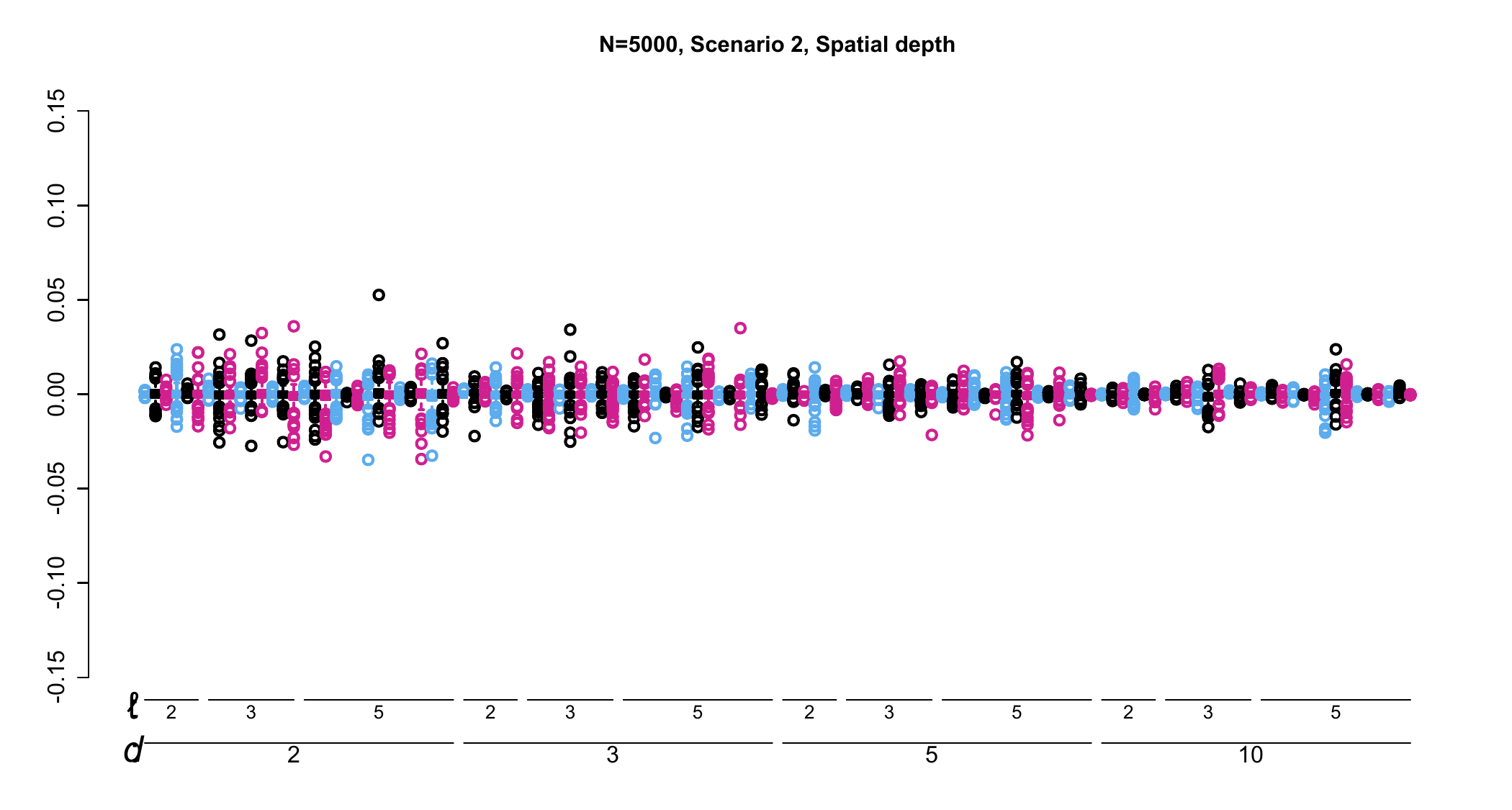}
\end{minipage}
\caption{Distribution of $\widehat{k}/N-\theta$ under the KW-PELT algorithm for spatial depth.}%
\end{figure}
\begin{figure}
\begin{minipage}[c]{.49\textwidth} 
\centering%
\includegraphics[width=\textwidth]{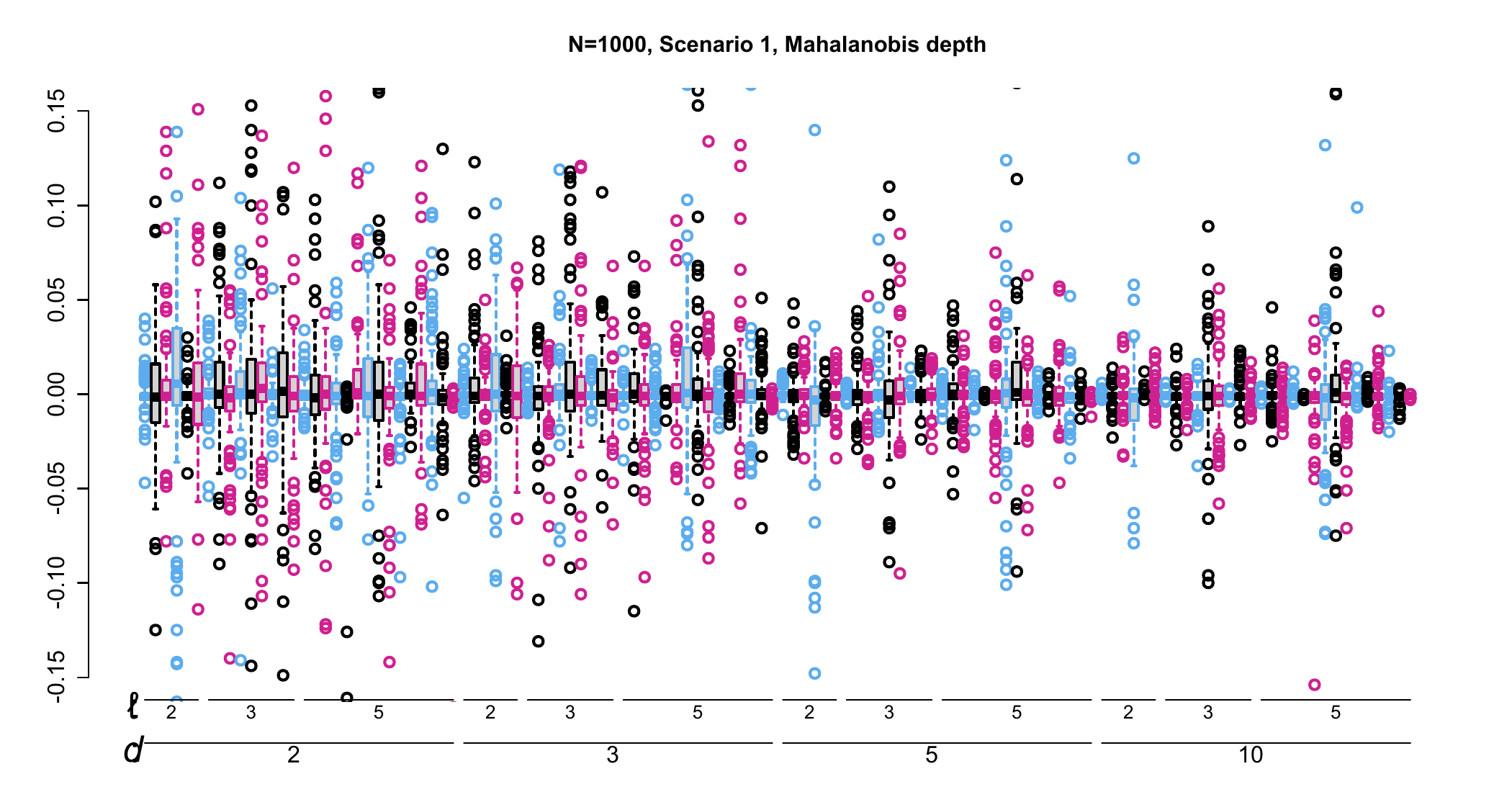}
\end{minipage}
\begin{minipage}[c]{.49\textwidth} 
\centering%
\includegraphics[width=\textwidth]{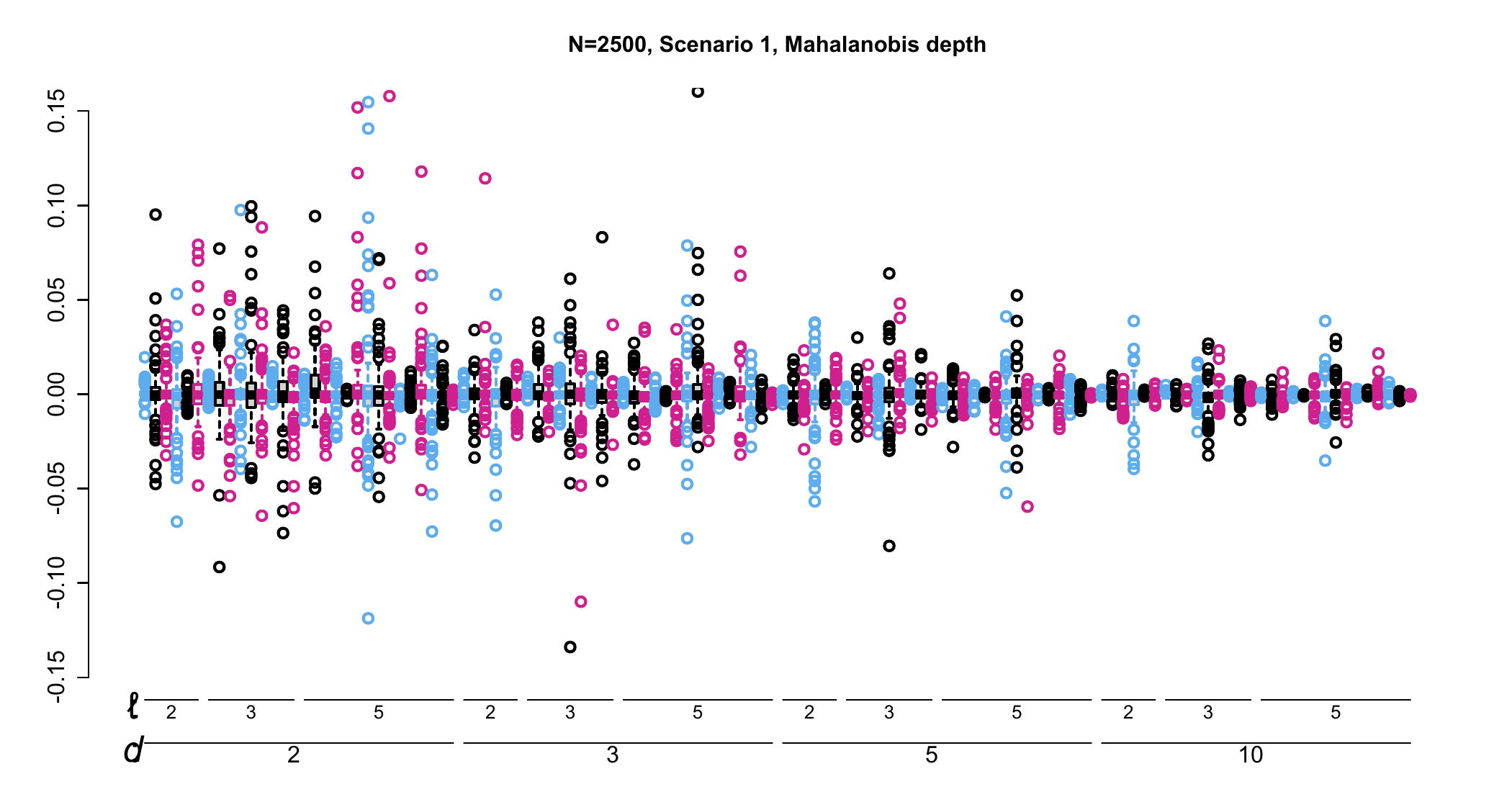}
\end{minipage}\hfill\newline
\begin{minipage}[c]{.49\textwidth} 
\centering%
\includegraphics[width=\textwidth]{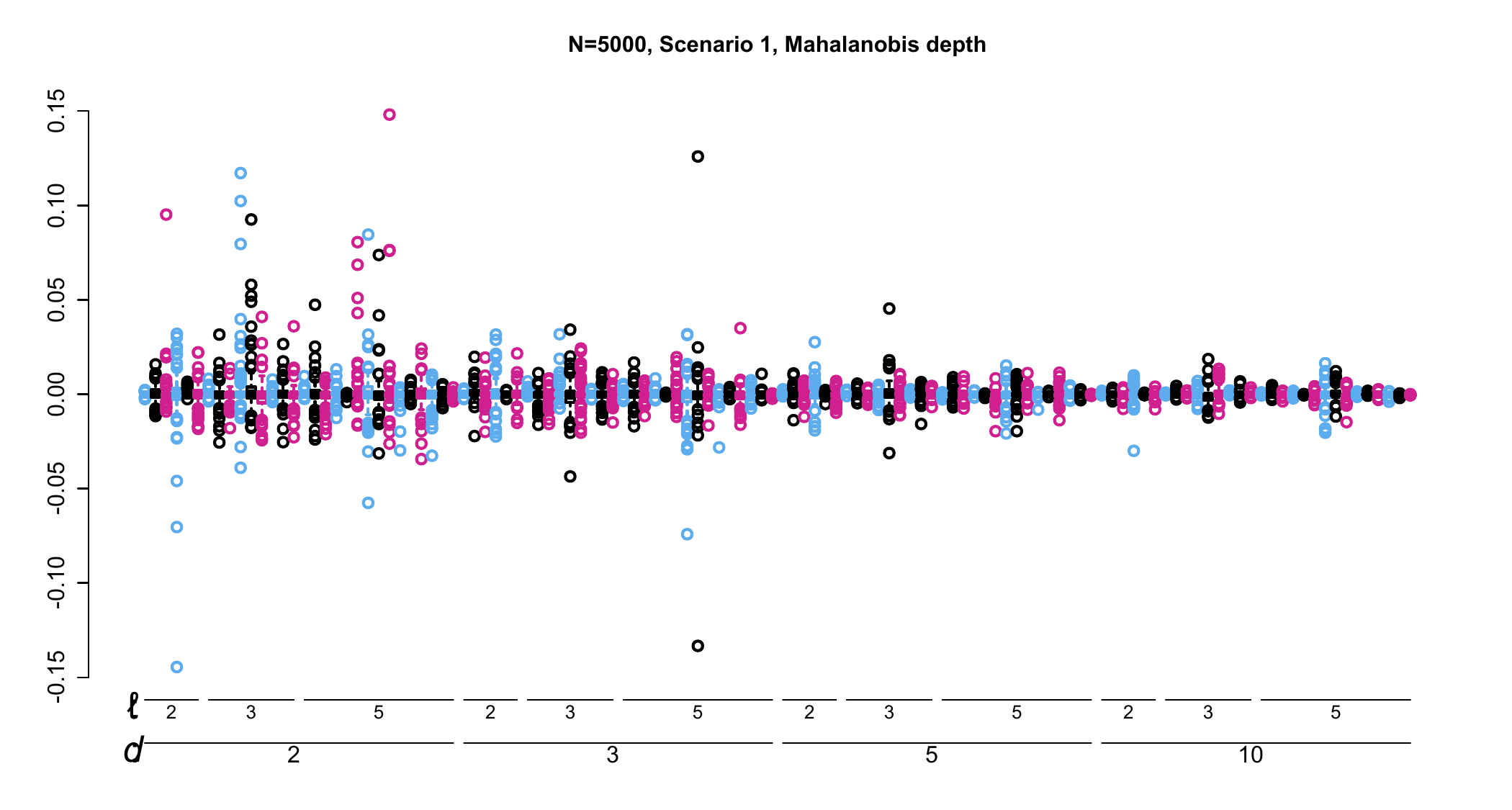}
\end{minipage}
\begin{minipage}[c]{.49\textwidth} 
\centering%
\includegraphics[width=\textwidth]{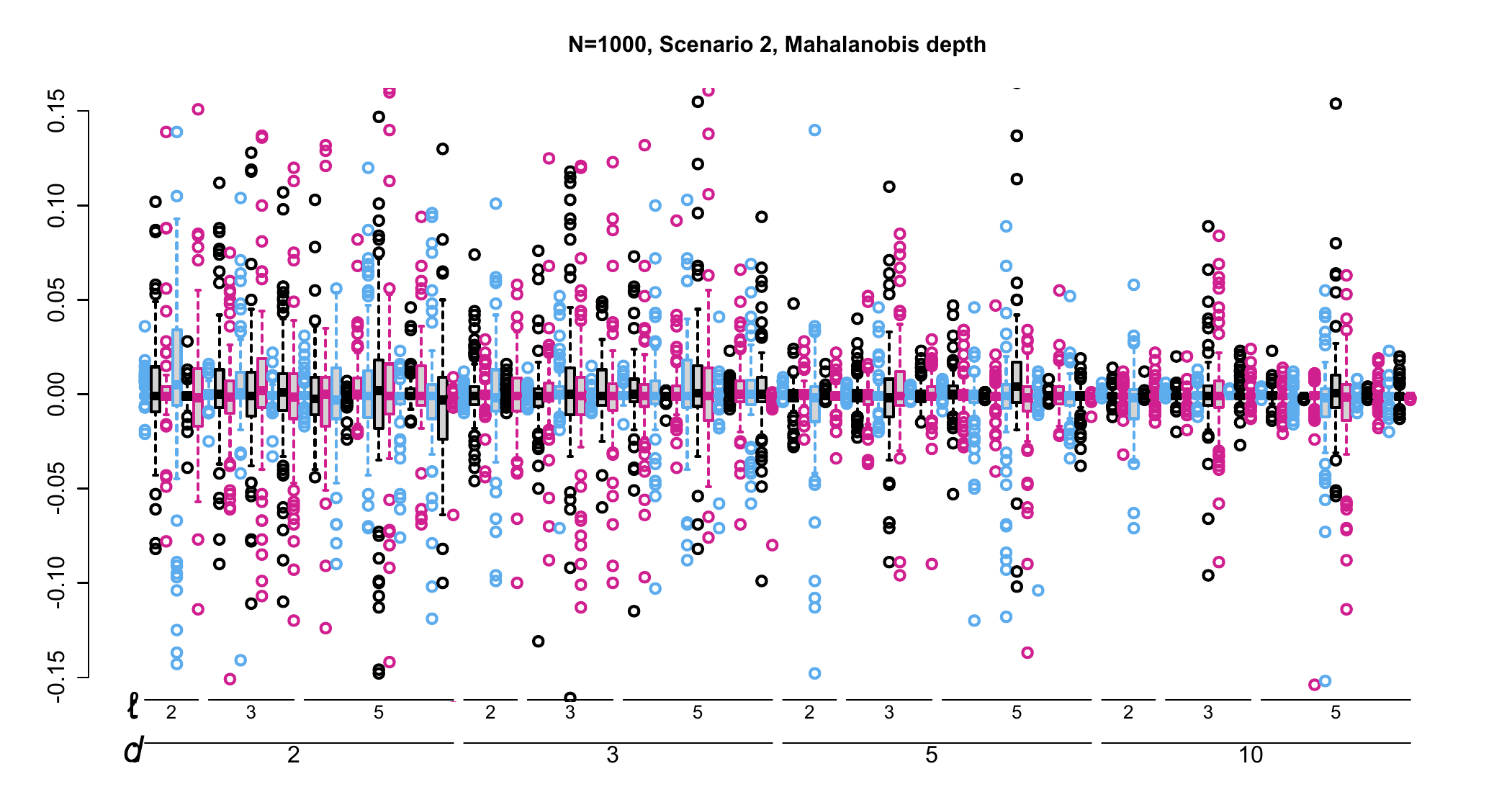}
\end{minipage}\hfill\newline
\begin{minipage}[c]{.49\textwidth} 
\centering%
\includegraphics[width=\textwidth]{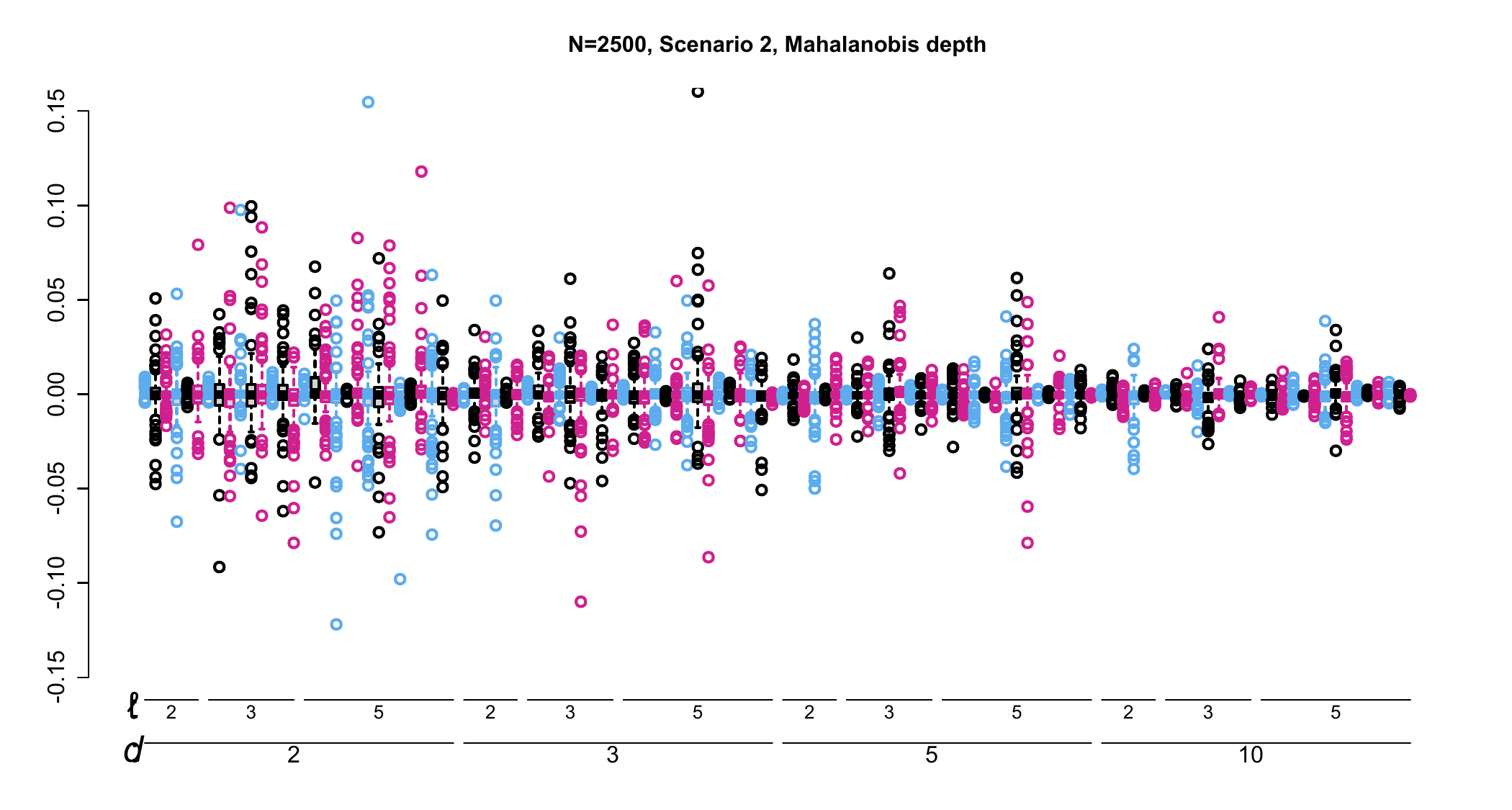}
\end{minipage}
\begin{minipage}[c]{.49\textwidth} 
\centering%
\includegraphics[width=\textwidth]{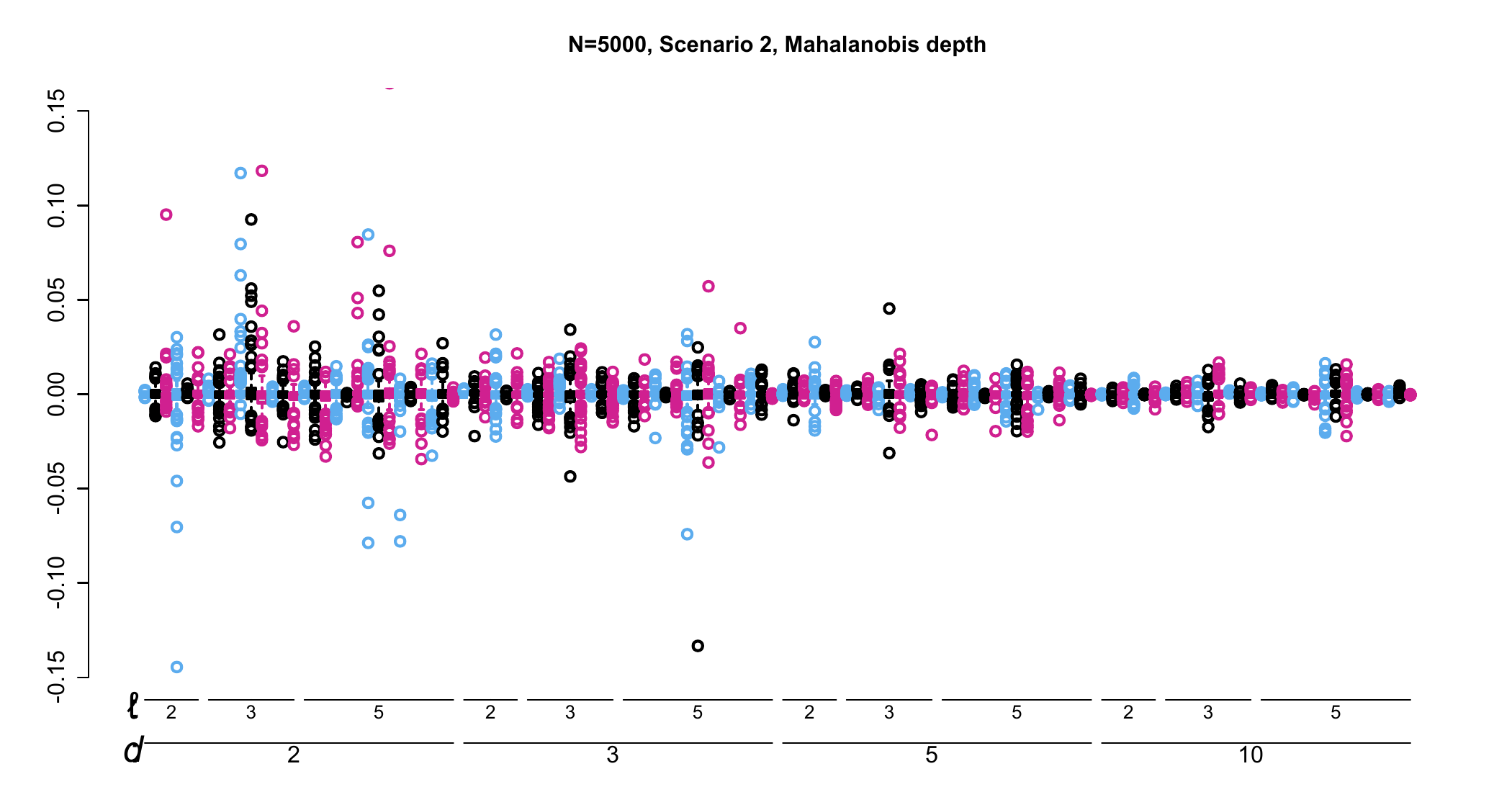}
\end{minipage}
\caption{Distribution of $\widehat{k}/N-\theta$ under the KW-PELT algorithm for Mahalanobis depth.}%
\end{figure}
\begin{figure}
\begin{minipage}[c]{.49\textwidth} 
\centering%
\includegraphics[width=\textwidth]{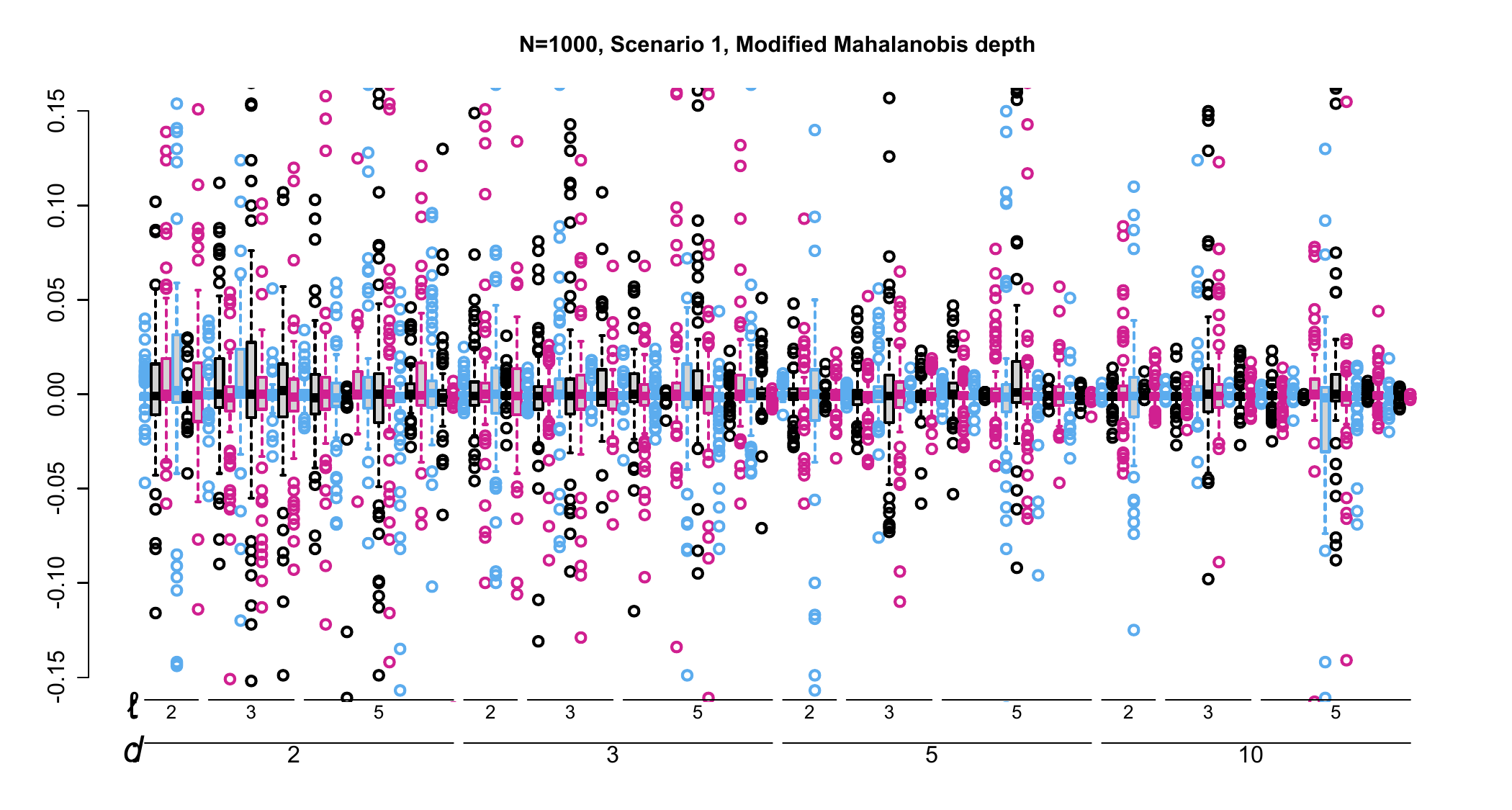}
\end{minipage}
\begin{minipage}[c]{.49\textwidth} 
\centering%
\includegraphics[width=\textwidth]{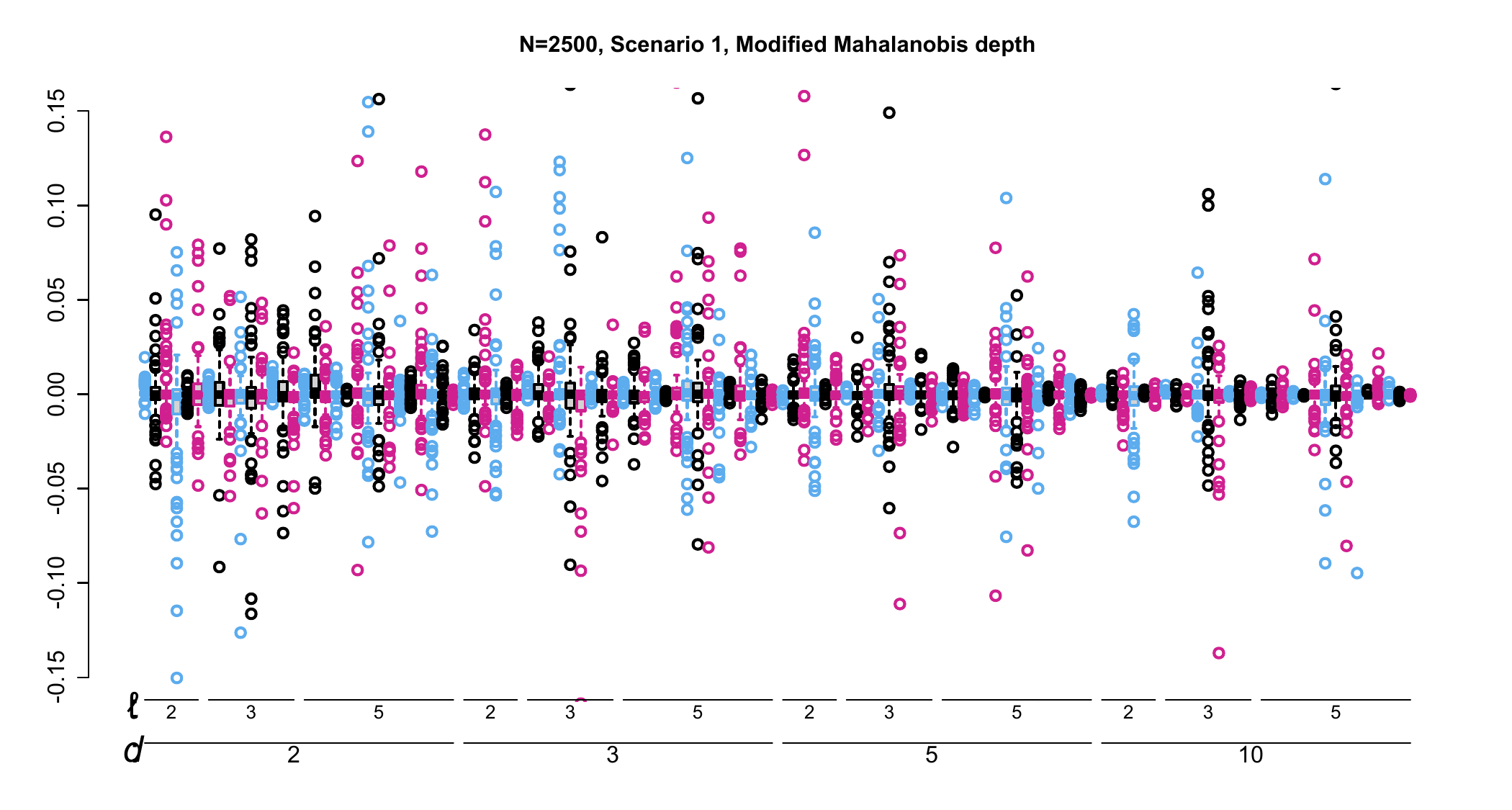}
\end{minipage}\hfill\newline
\begin{minipage}[c]{.49\textwidth} 
\centering%
\includegraphics[width=\textwidth]{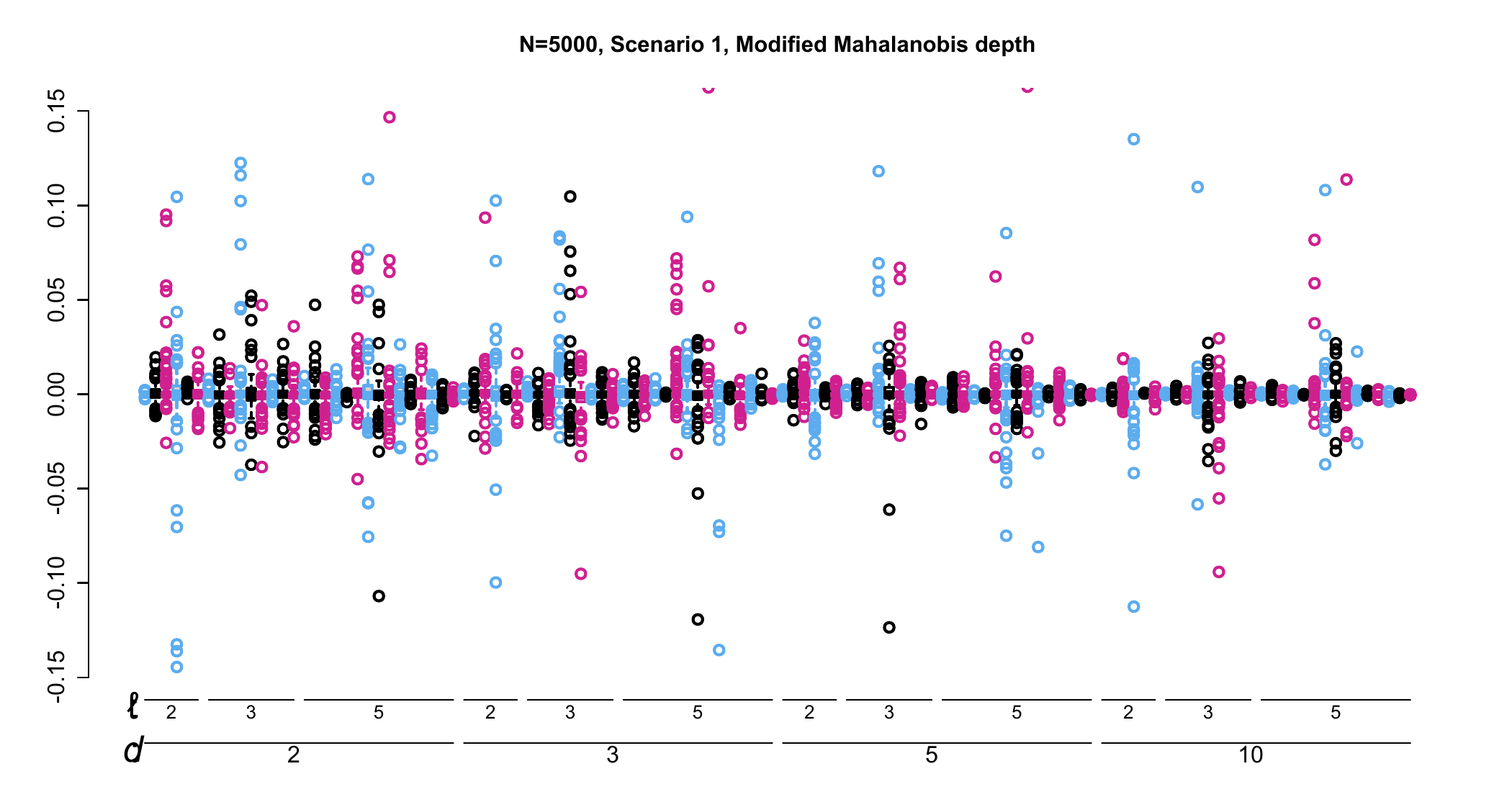}
\end{minipage}
\begin{minipage}[c]{.49\textwidth} 
\centering%
\includegraphics[width=\textwidth]{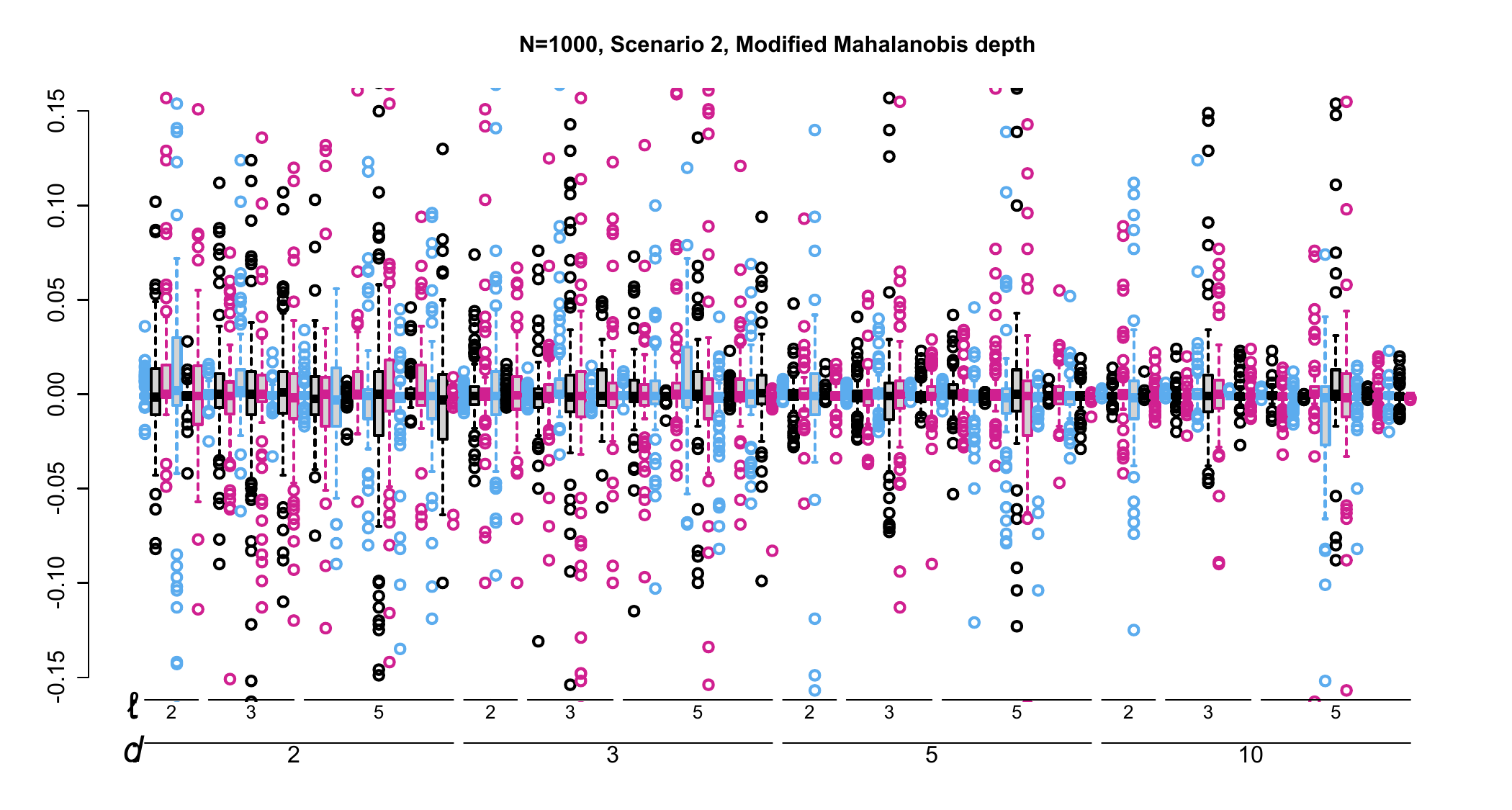}
\end{minipage}\hfill\newline
\begin{minipage}[c]{.49\textwidth} 
\centering%
\includegraphics[width=\textwidth]{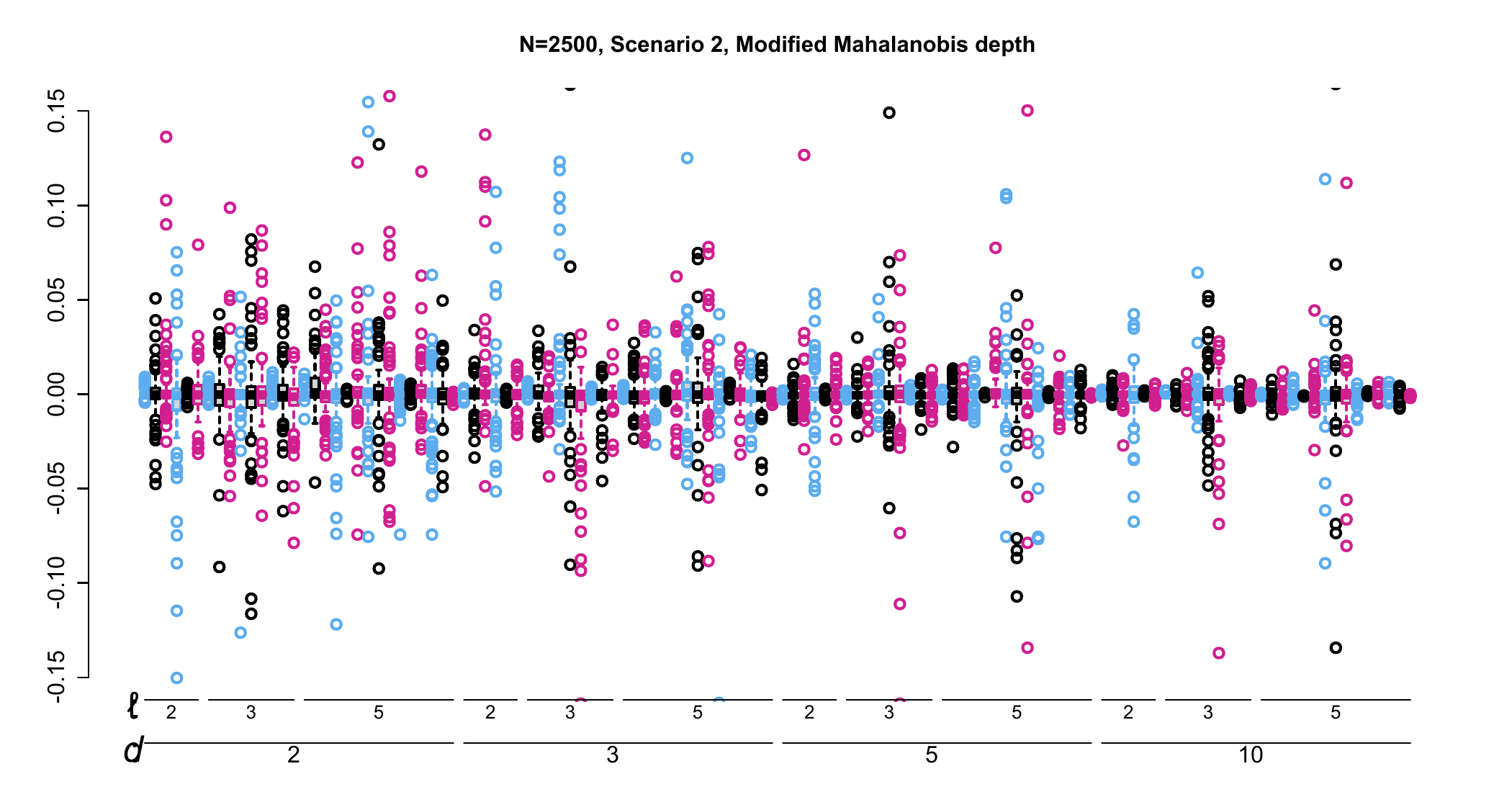}
\end{minipage}
\begin{minipage}[c]{.49\textwidth} 
\centering%
\includegraphics[width=\textwidth]{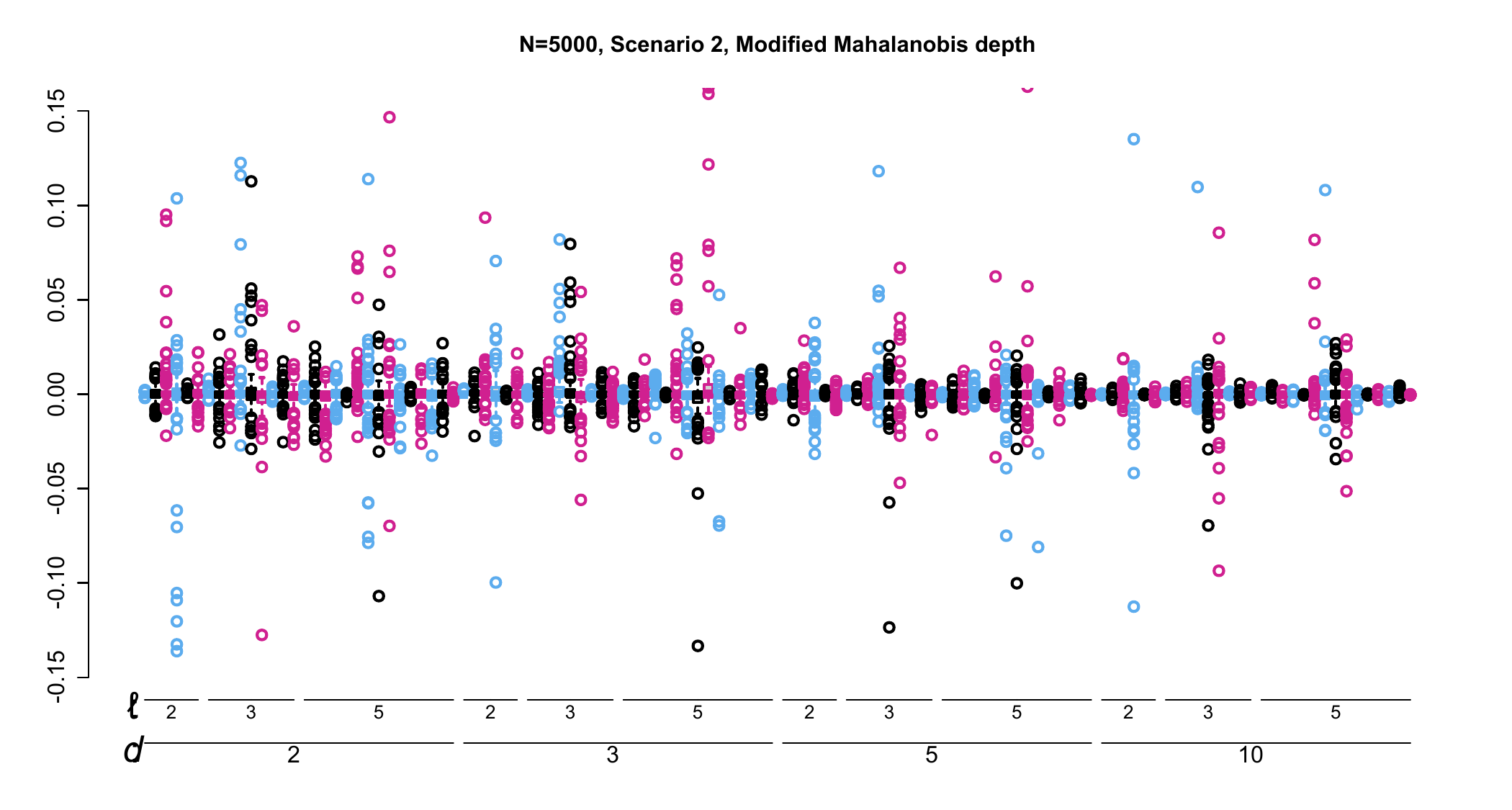}
\end{minipage}
\caption{Distribution of $\widehat{k}/N-\theta$ under the KW-PELT algorithm for modified Mahalanobis depth.}%
\end{figure}
\begin{figure}
\begin{minipage}[c]{.49\textwidth} 
\centering%
\includegraphics[width=\textwidth]{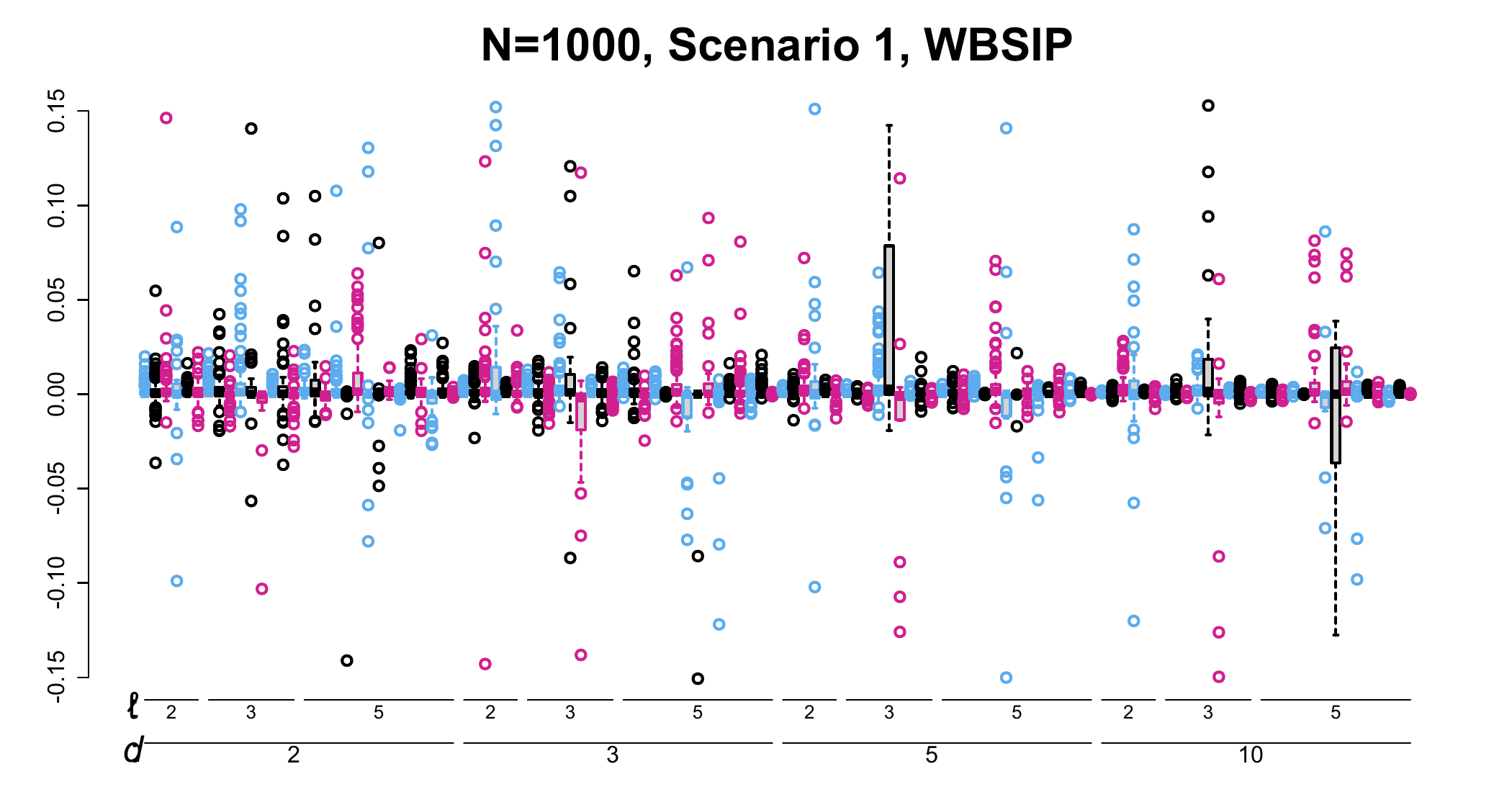}
\end{minipage}
\begin{minipage}[c]{.49\textwidth} 
\centering%
\includegraphics[width=\textwidth]{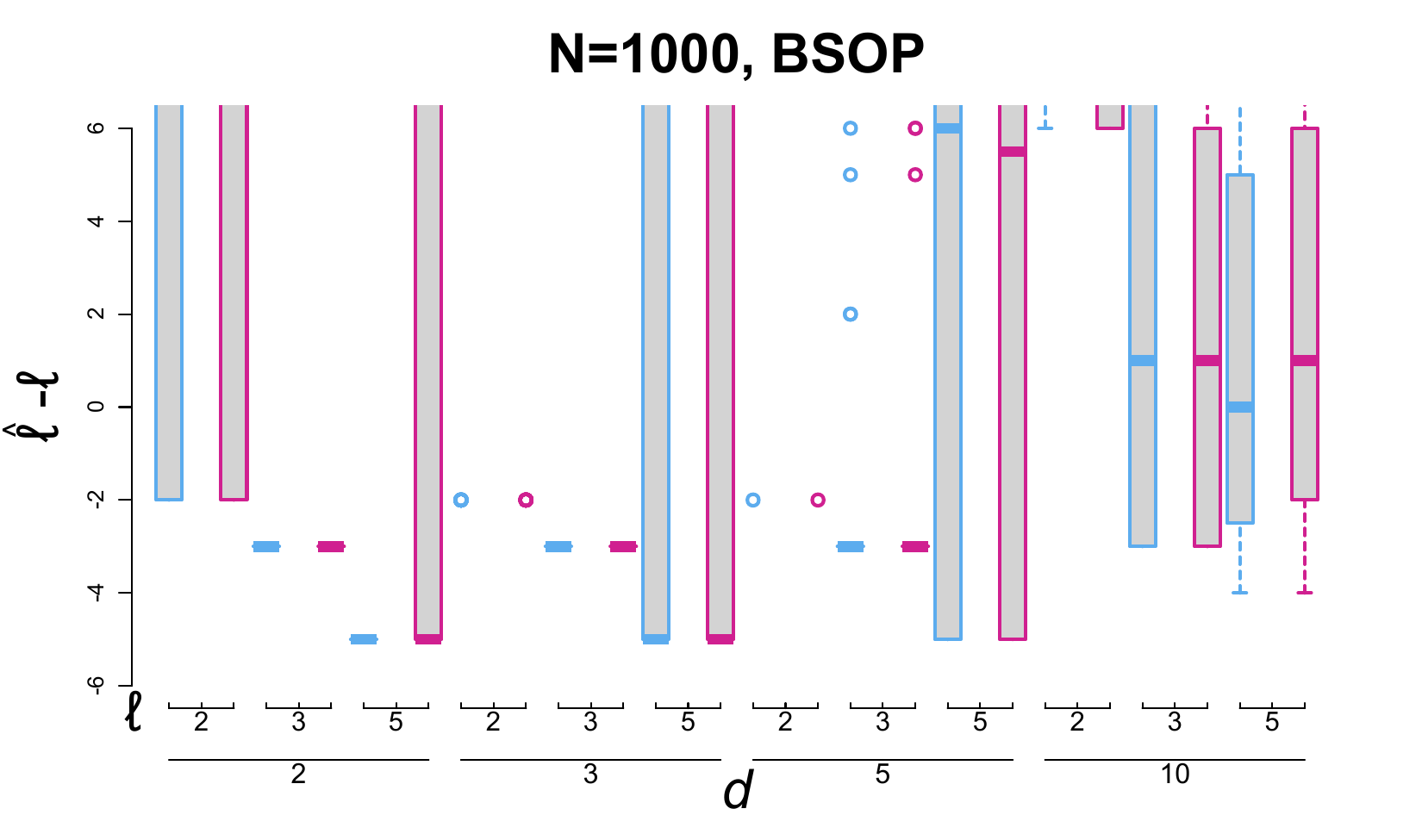}
\end{minipage}
\caption{(Left Pane) Distribution of $\widehat{k}/N-\theta$ under the WBSIP algorithm. (Right Pane) Distribution of $\widehat{\ell}-\ell$ under the BSOP algorithm.}
\end{figure}
\begin{figure}
\begin{minipage}[c]{.49\textwidth} 
\centering%
\includegraphics[width=\textwidth]{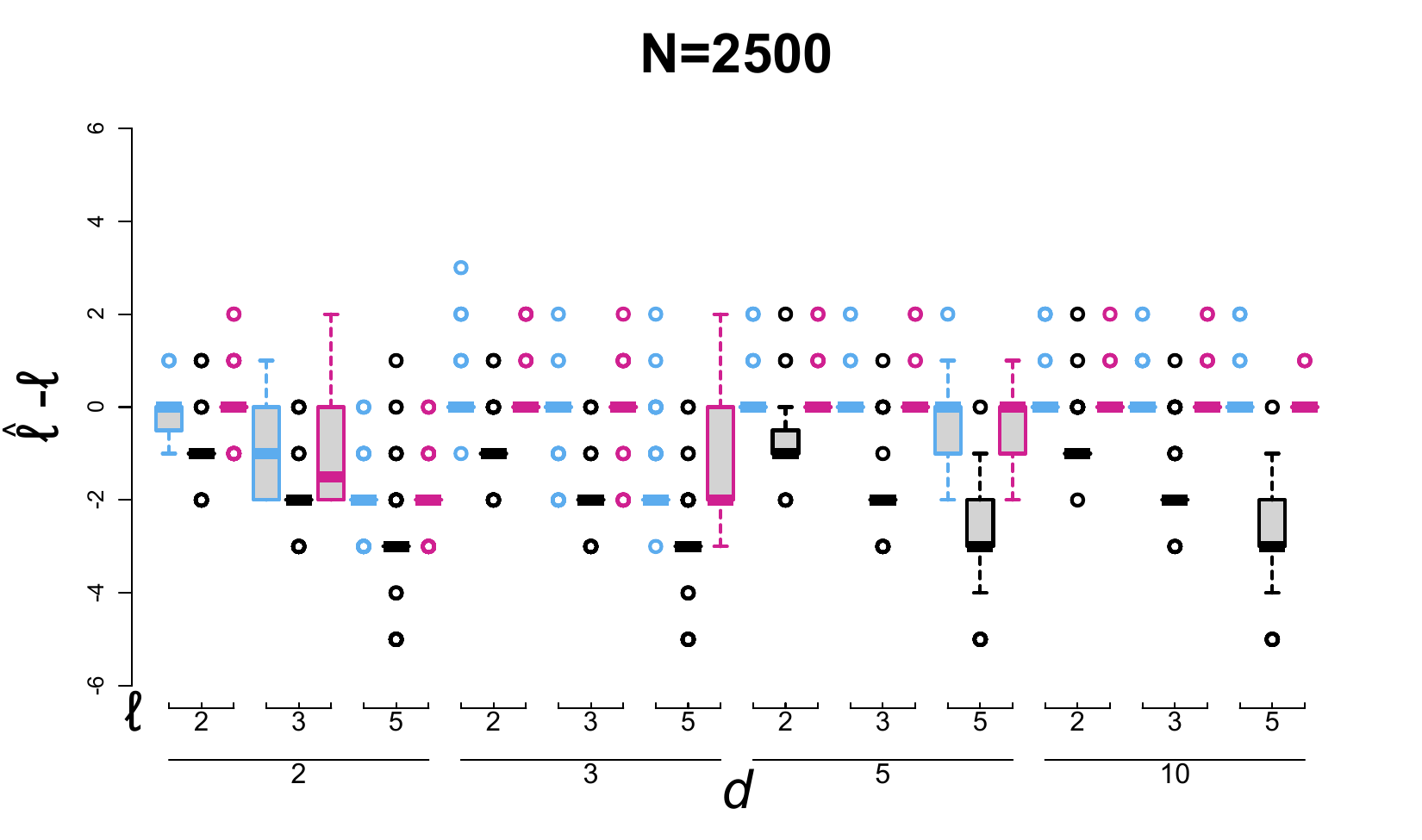}
\end{minipage}
\begin{minipage}[c]{.49\textwidth} 
\centering%
\includegraphics[width=\textwidth]{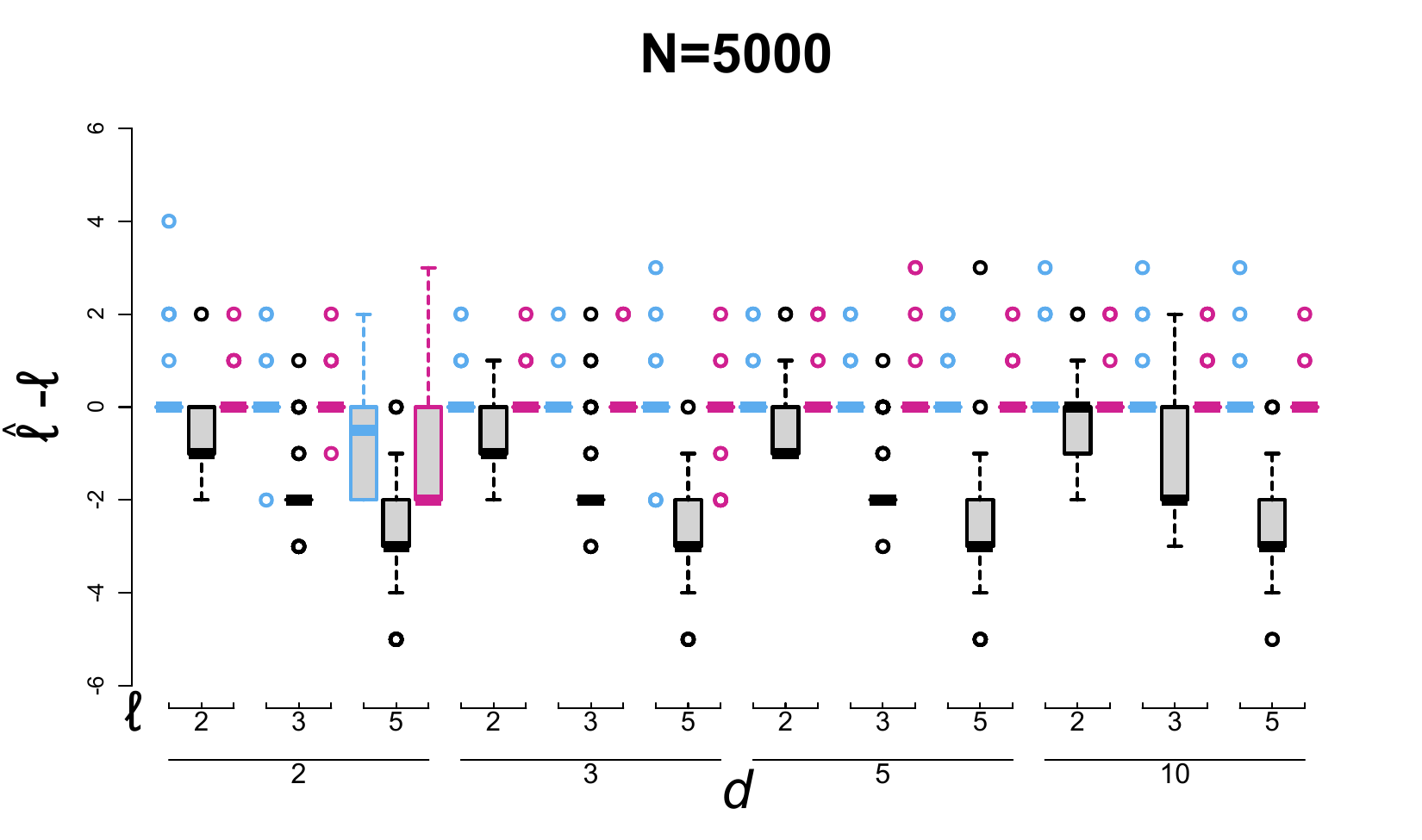}
\end{minipage}\hfill\newline
\begin{minipage}[c]{.49\textwidth} 
\centering%
\includegraphics[width=\textwidth]{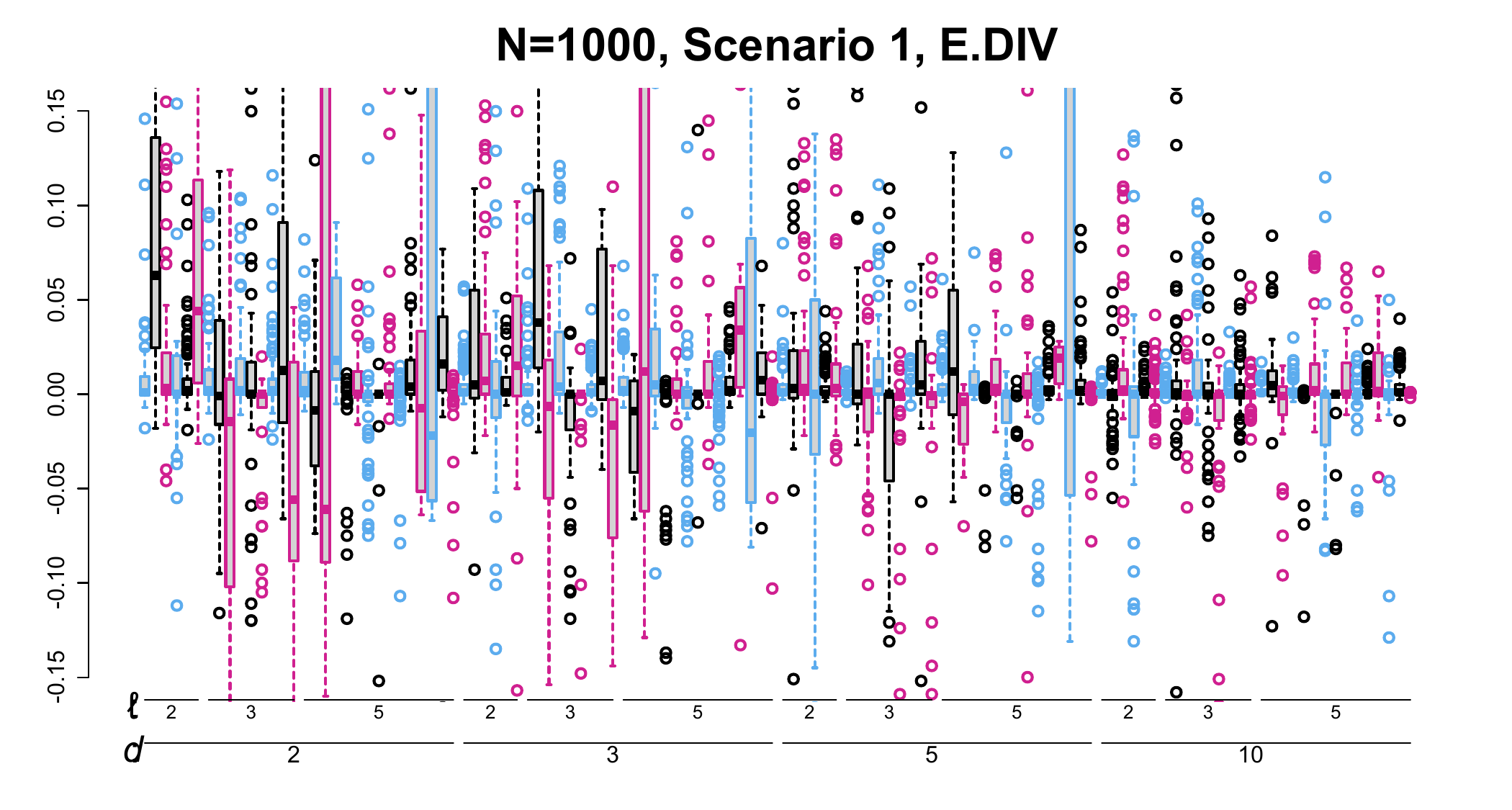}
\end{minipage}
\begin{minipage}[c]{.49\textwidth} 
\centering%
\includegraphics[width=\textwidth]{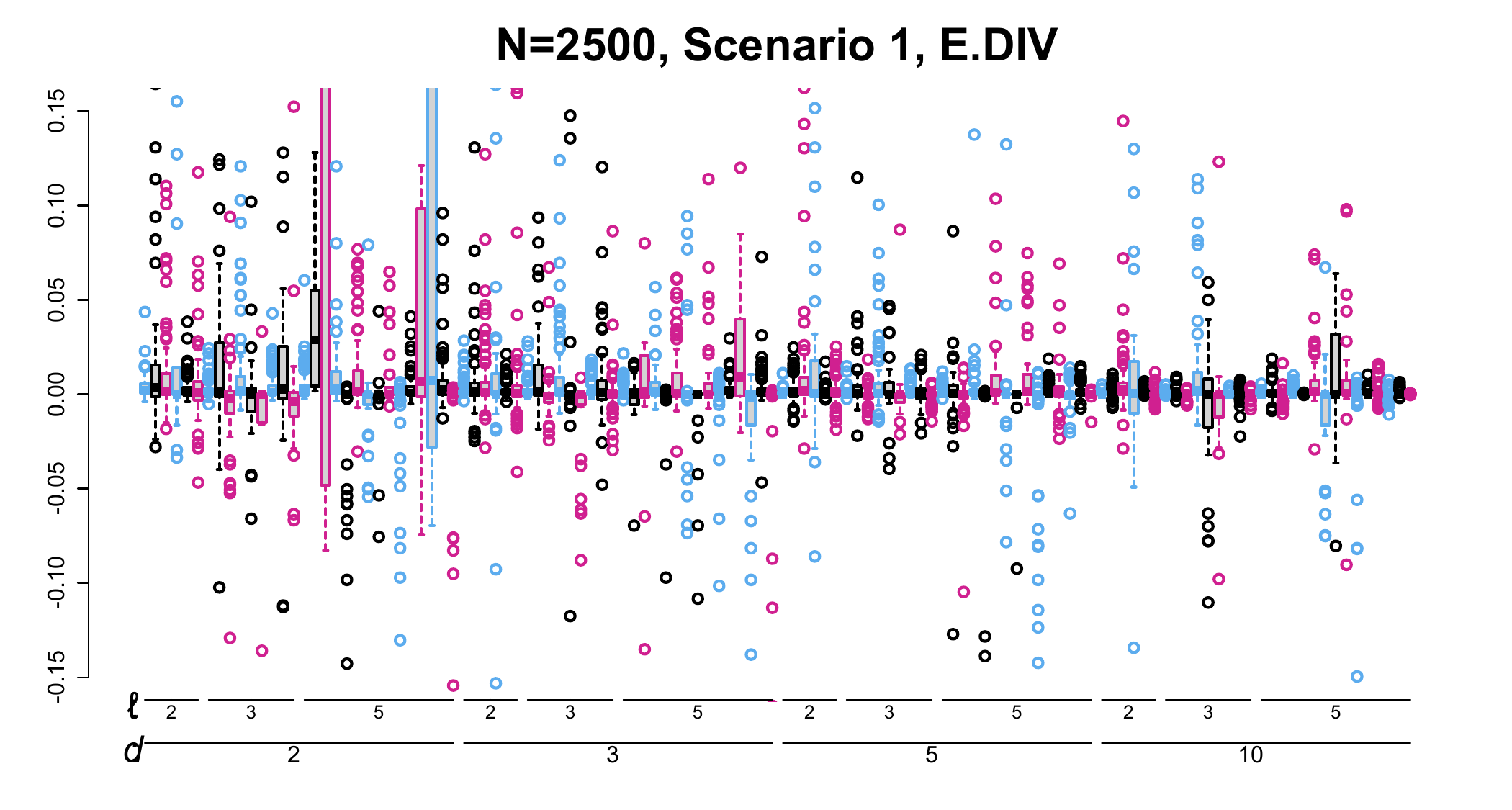}
\end{minipage}\hfill\newline
\begin{minipage}[c]{.49\textwidth} 
\centering%
\includegraphics[width=\textwidth]{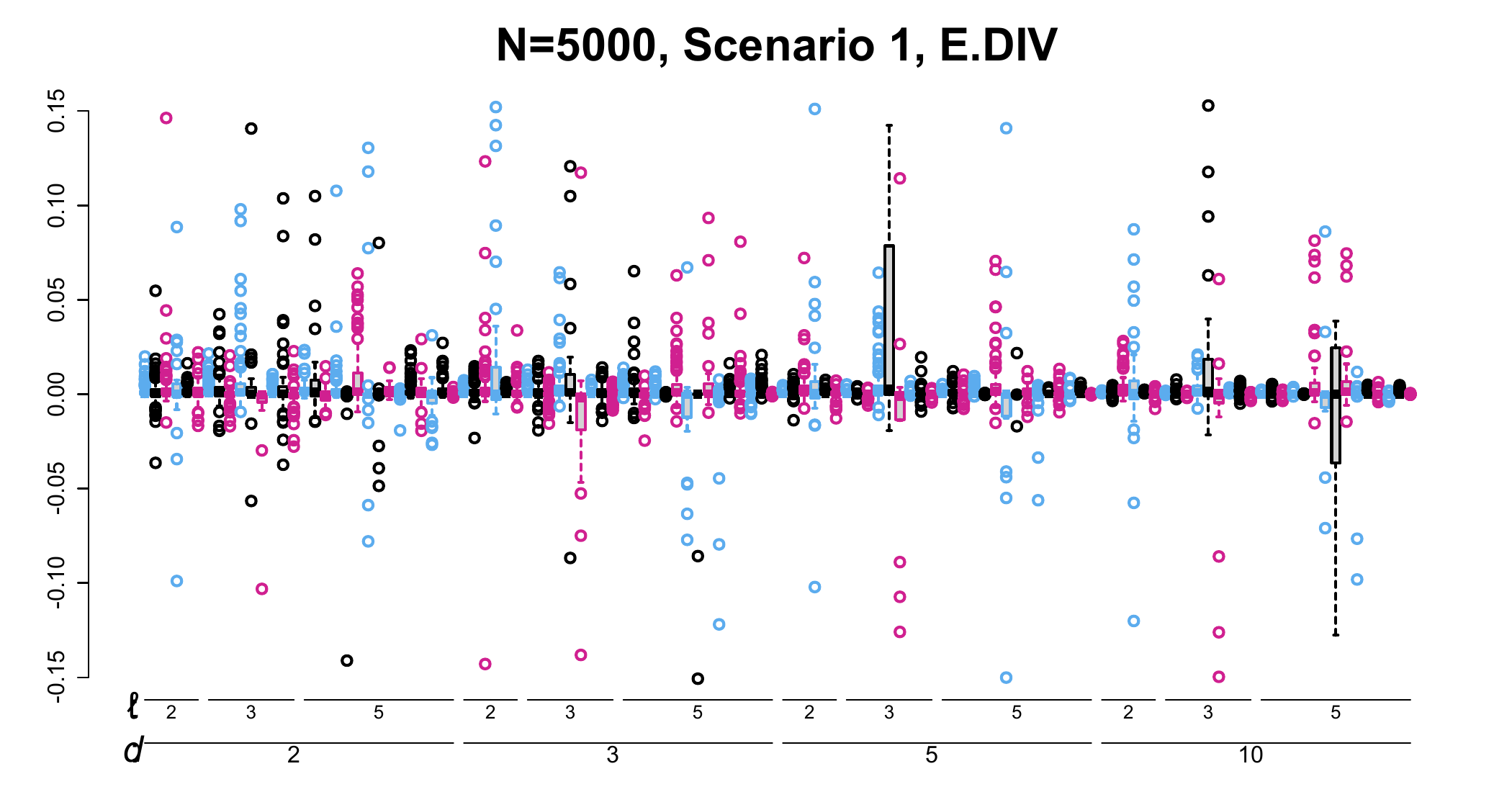}
\end{minipage}
\caption{Distribution of $\widehat{\ell}-\ell$ (first row) and $\widehat{k}/N-\theta$ (remaining rows) under simulation scenario 1 for the \texttt{e.divisive} method.}
\end{figure}
\begin{figure}
\begin{minipage}[c]{.32\textwidth} 
\centering%
\includegraphics[width=\textwidth]{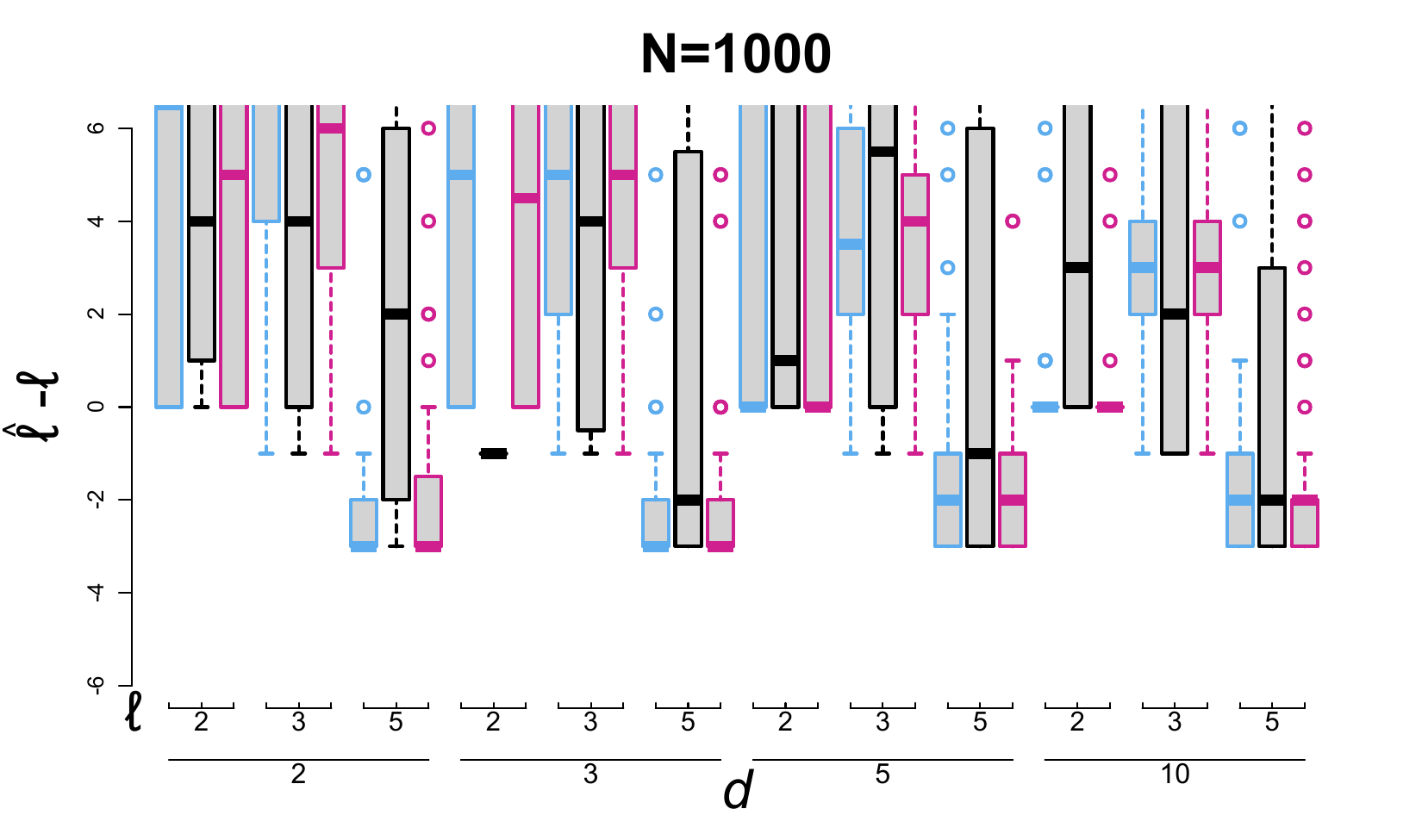}
\end{minipage}
\begin{minipage}[c]{.32\textwidth} 
\centering%
\includegraphics[width=\textwidth]{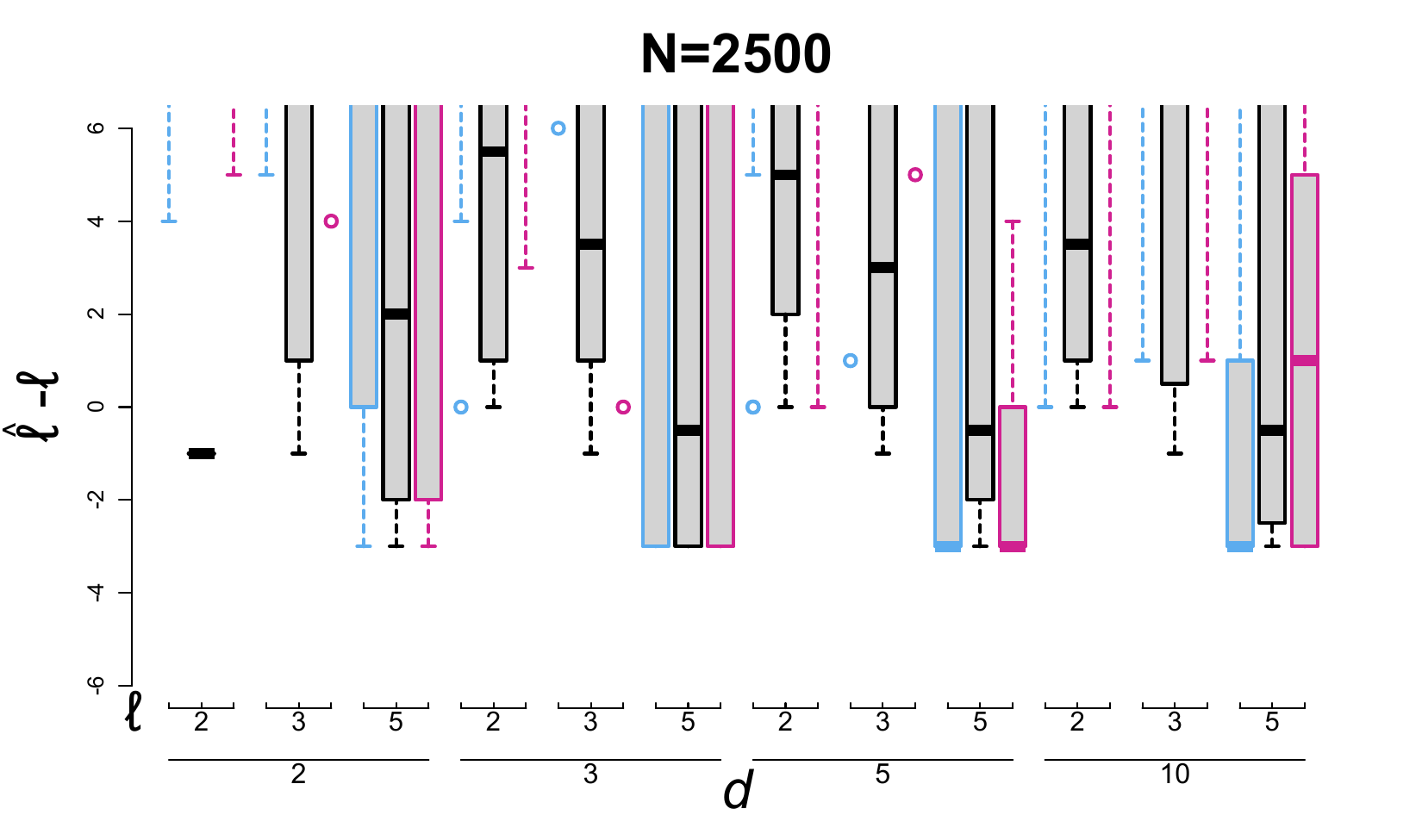}
\end{minipage}
\begin{minipage}[c]{.32\textwidth} 
\centering%
\includegraphics[width=\textwidth]{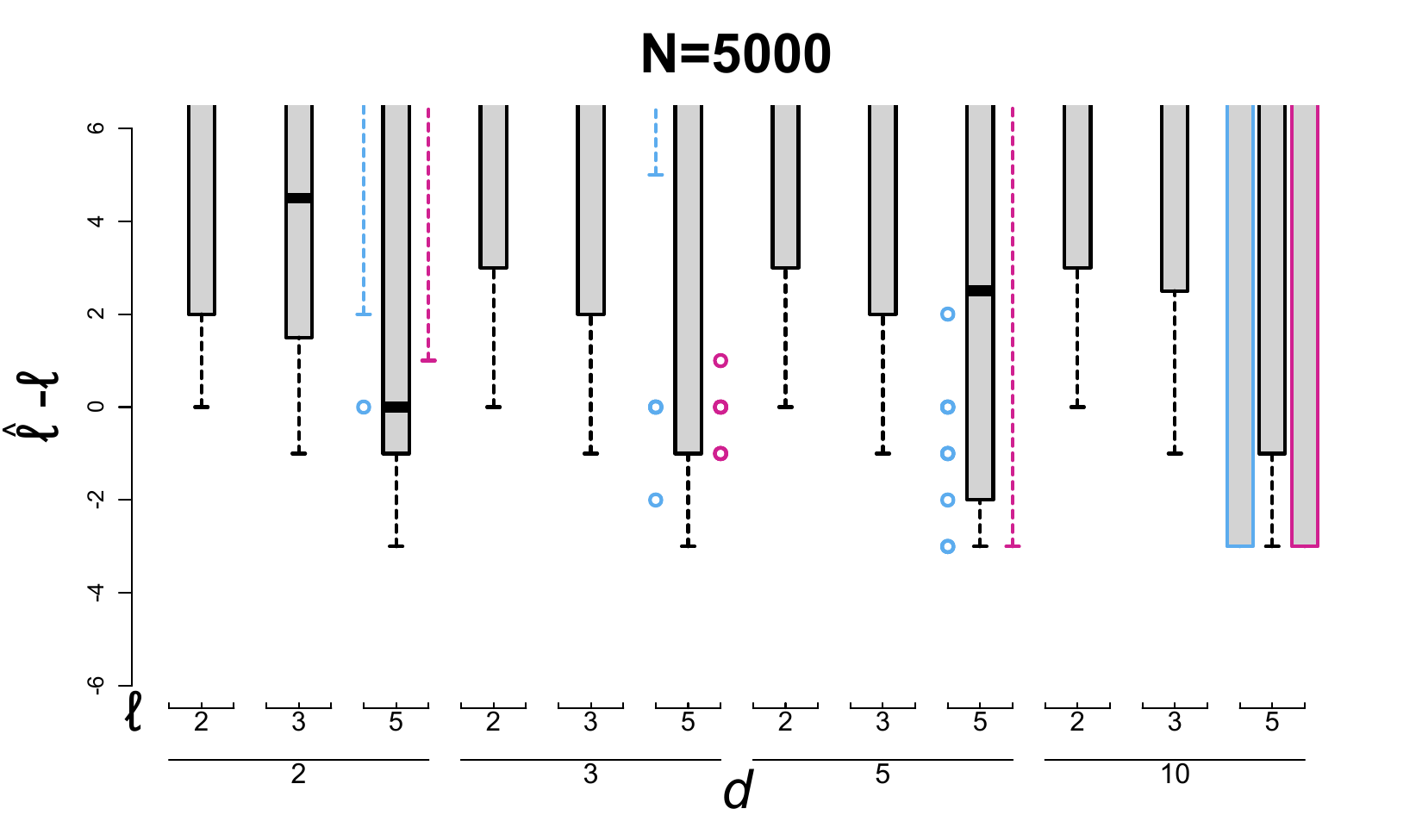}
\end{minipage}\hfill\newline
\begin{minipage}[c]{.32\textwidth} 
\centering%
\includegraphics[width=\textwidth]{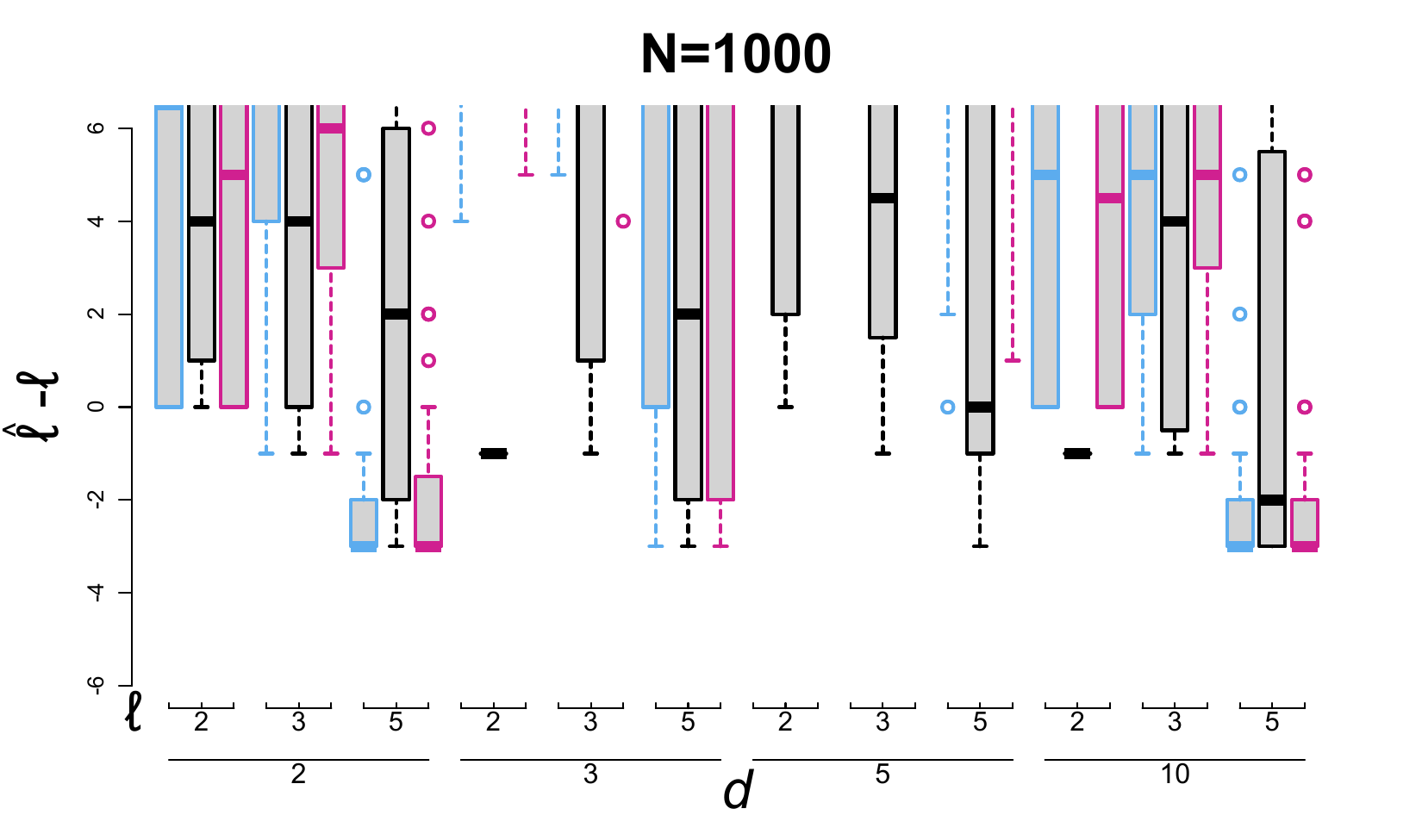}
\end{minipage}
\begin{minipage}[c]{.32\textwidth} 
\centering%
\includegraphics[width=\textwidth]{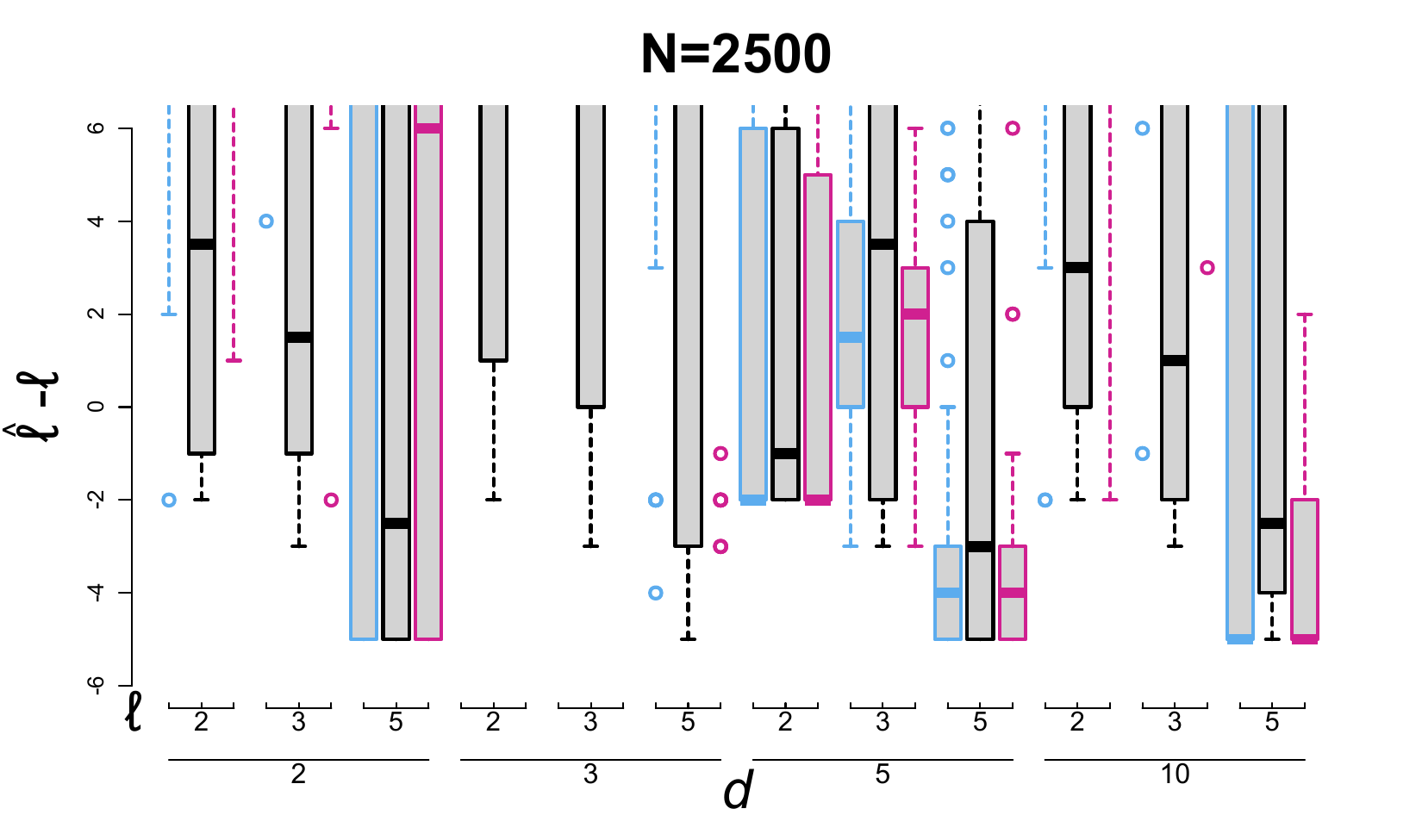}
\end{minipage}\hfill
\caption{Distribution of $\widehat{\ell}-\ell$ under simulation scenario 1 for the methods \texttt{e.cp3o\_delta} and \texttt{e.kcp3o}. Simulation was only ran for $N=1000$ and $N=2500$ for \texttt{e.kcp3o}, since the running time was very slow.}%
\end{figure}
\begin{figure}
\begin{minipage}[c]{.49\textwidth} 
\centering%
\includegraphics[width=.8\textwidth]{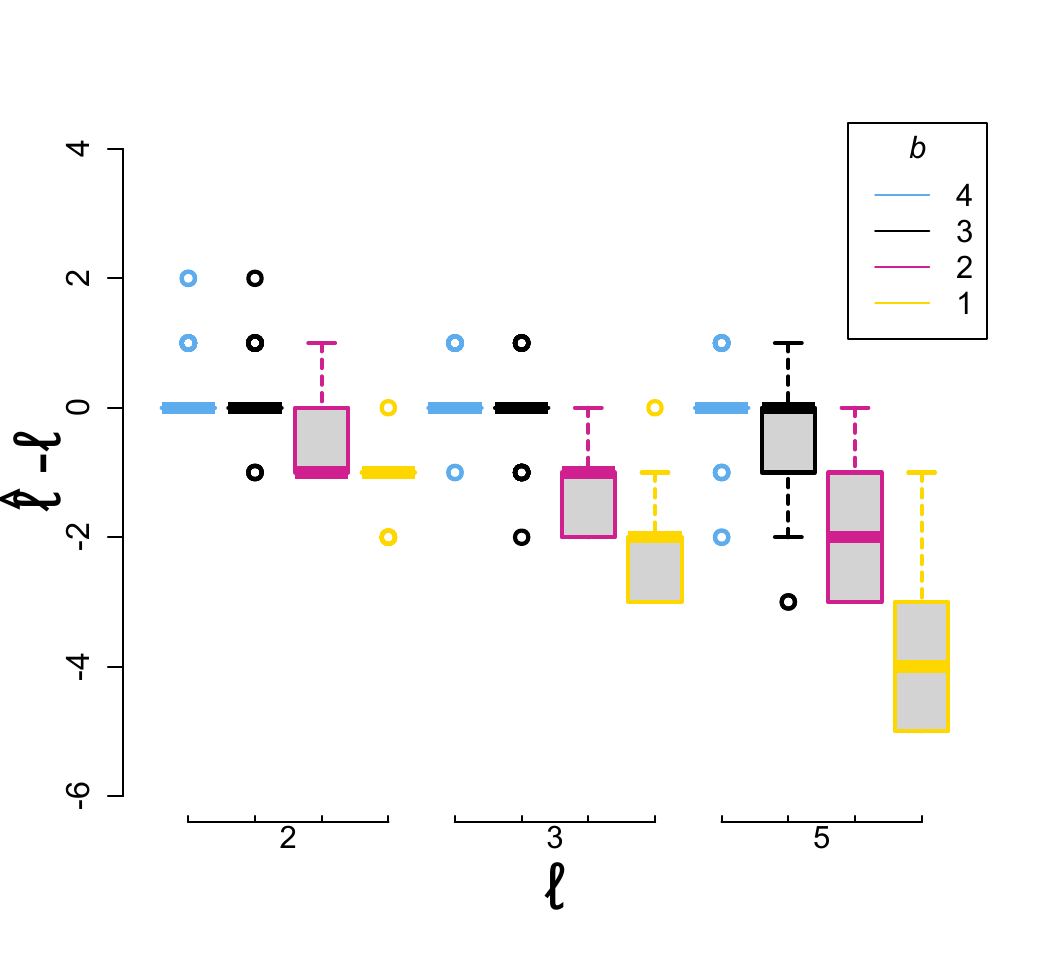}
\end{minipage}
\begin{minipage}[c]{.49\textwidth} 
\centering%
\includegraphics[width=.8\textwidth]{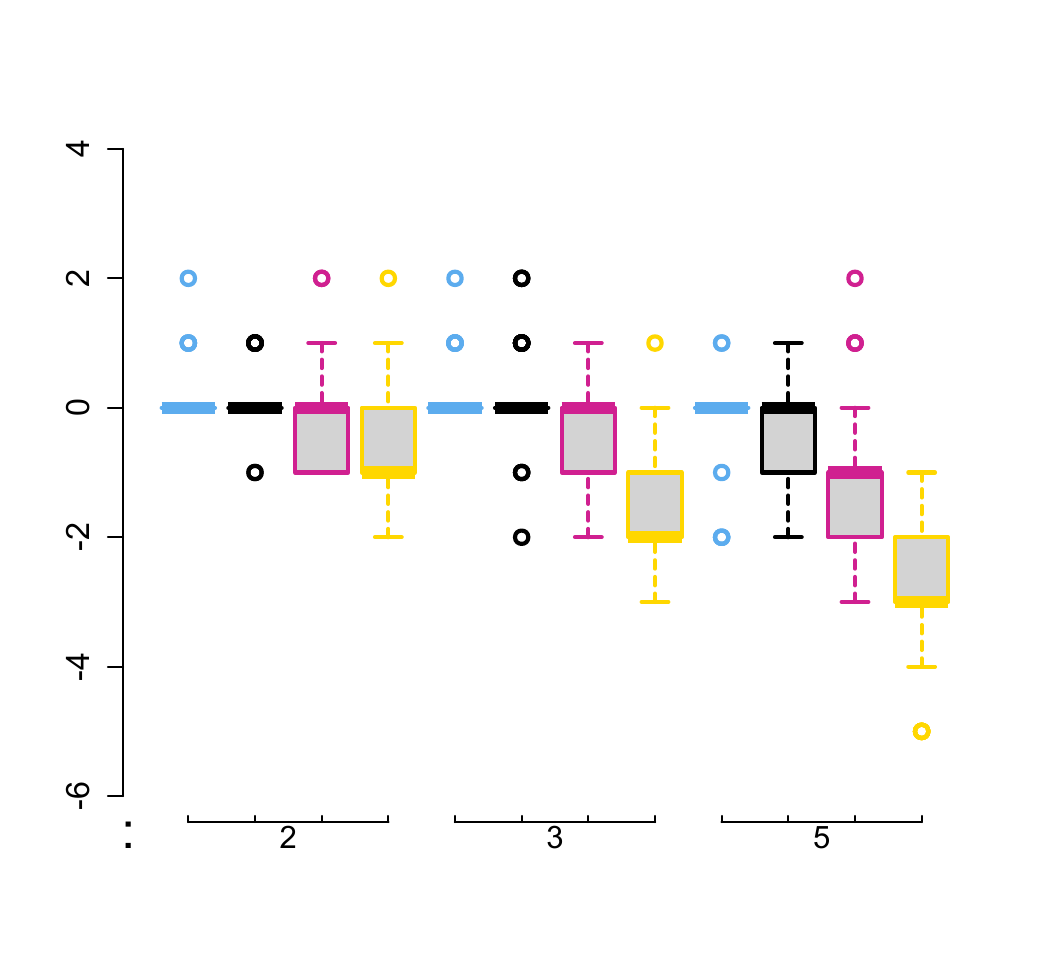}
\end{minipage}\hfill\newline
\begin{minipage}[c]{.49\textwidth} 
\centering%
\includegraphics[width=.8\textwidth]{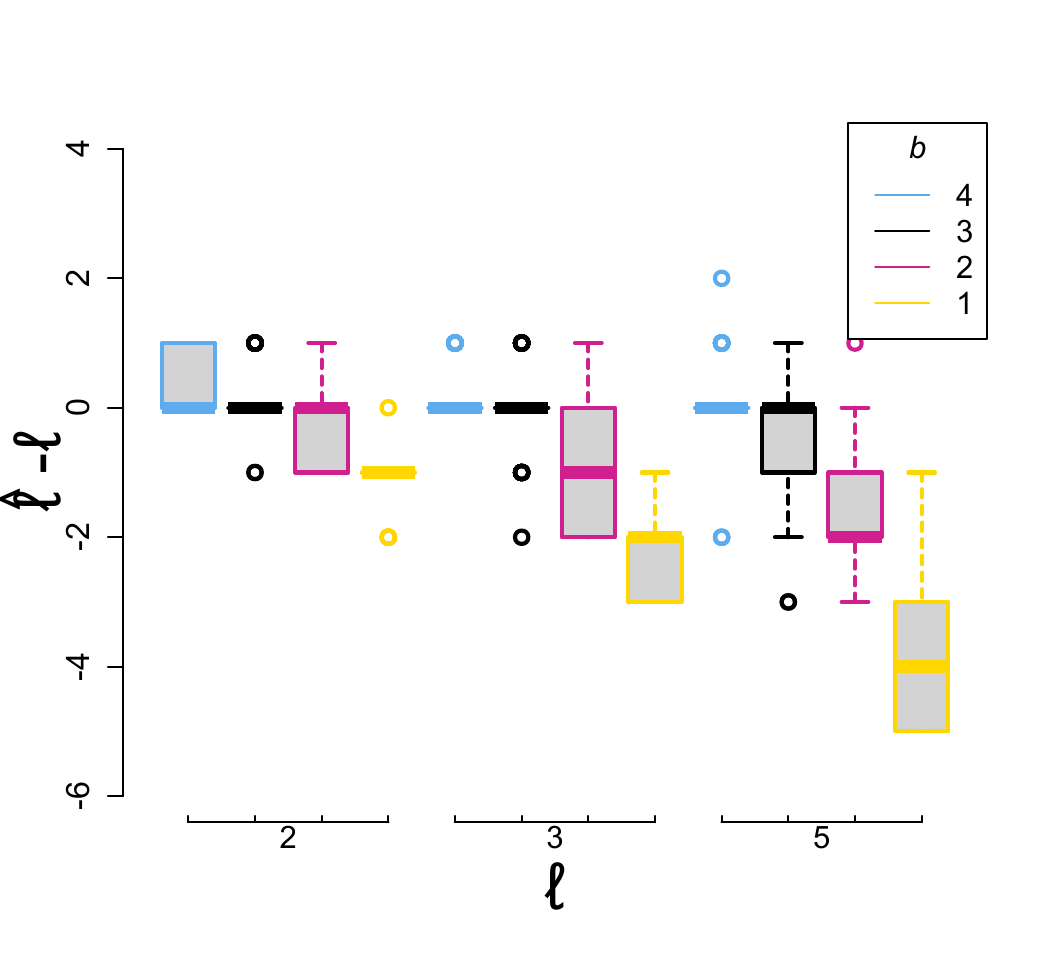}
\end{minipage}
\begin{minipage}[c]{.49\textwidth} 
\centering%
\includegraphics[width=.8\textwidth]{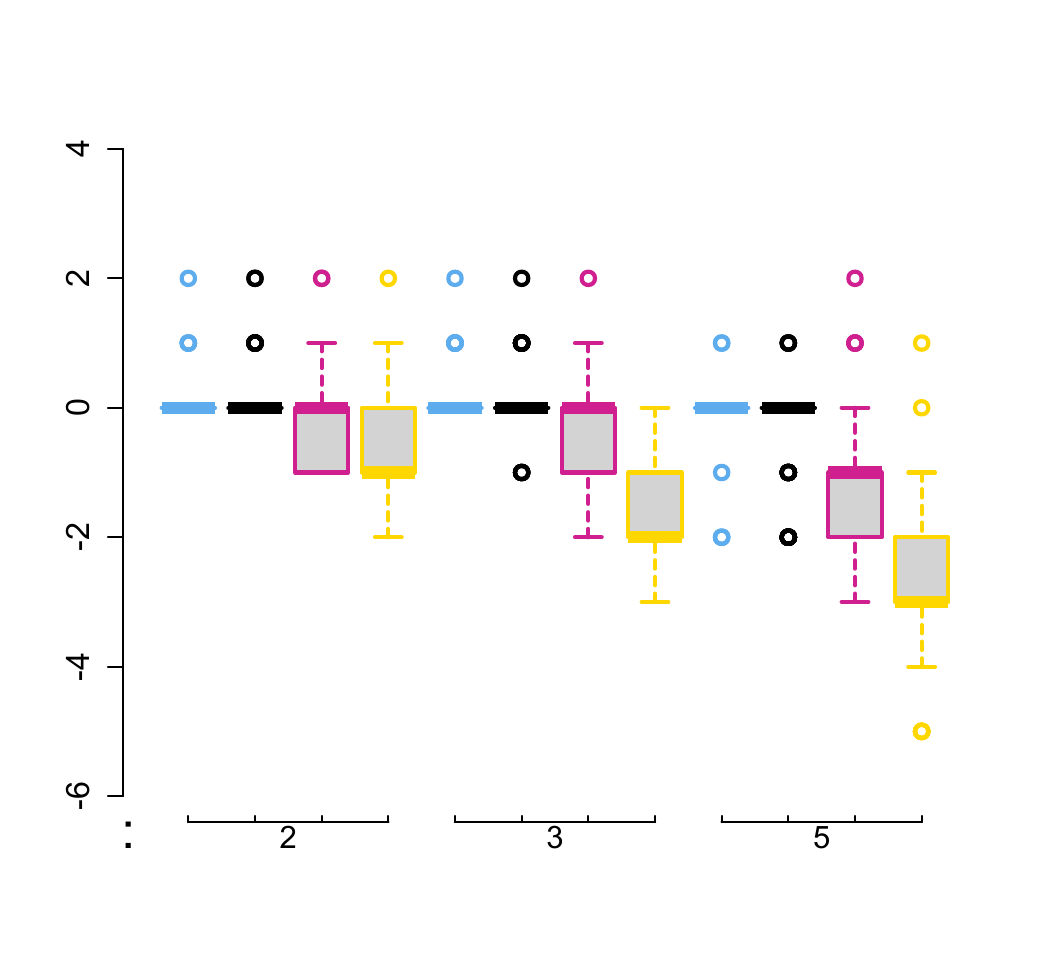}
\end{minipage}\hfill\newline
\begin{minipage}[c]{.49\textwidth} 
\centering%
\includegraphics[width=.8\textwidth]{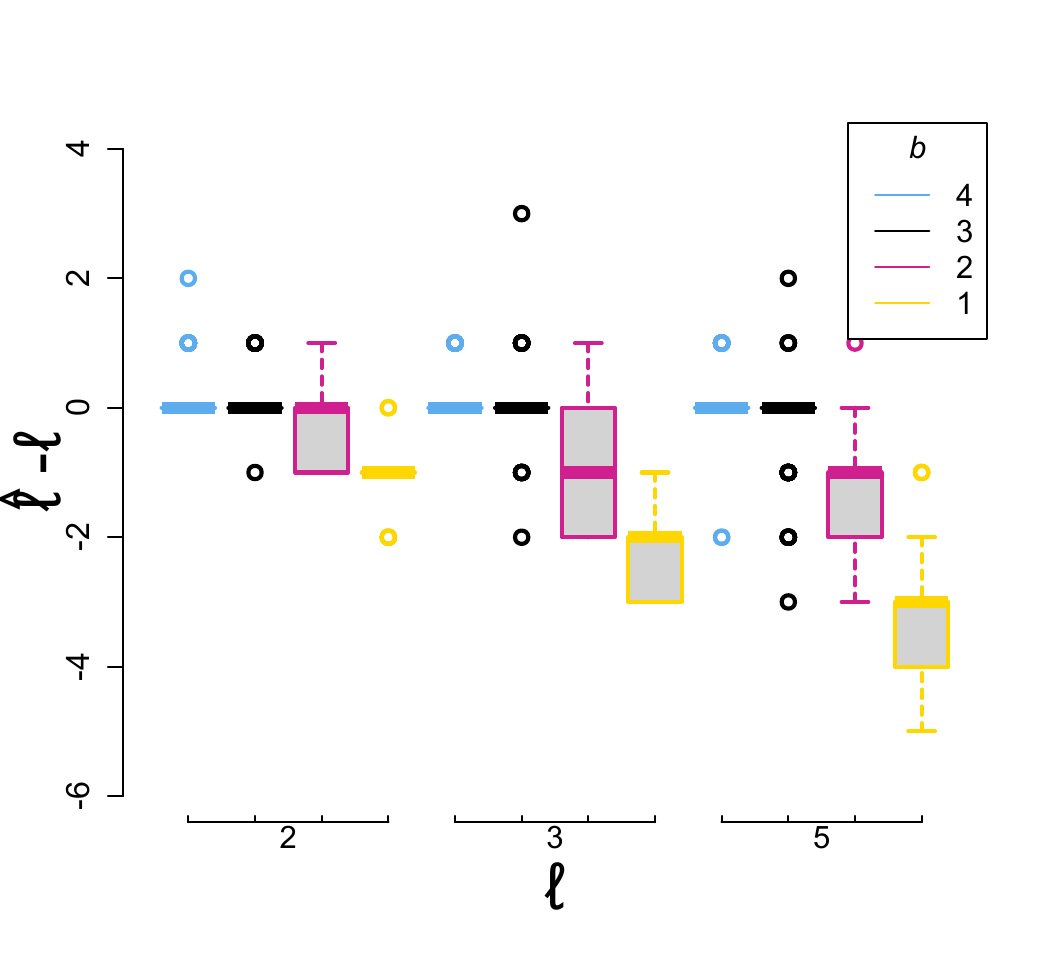}
\end{minipage}
\begin{minipage}[c]{.49\textwidth} 
\centering%
\includegraphics[width=.8\textwidth]{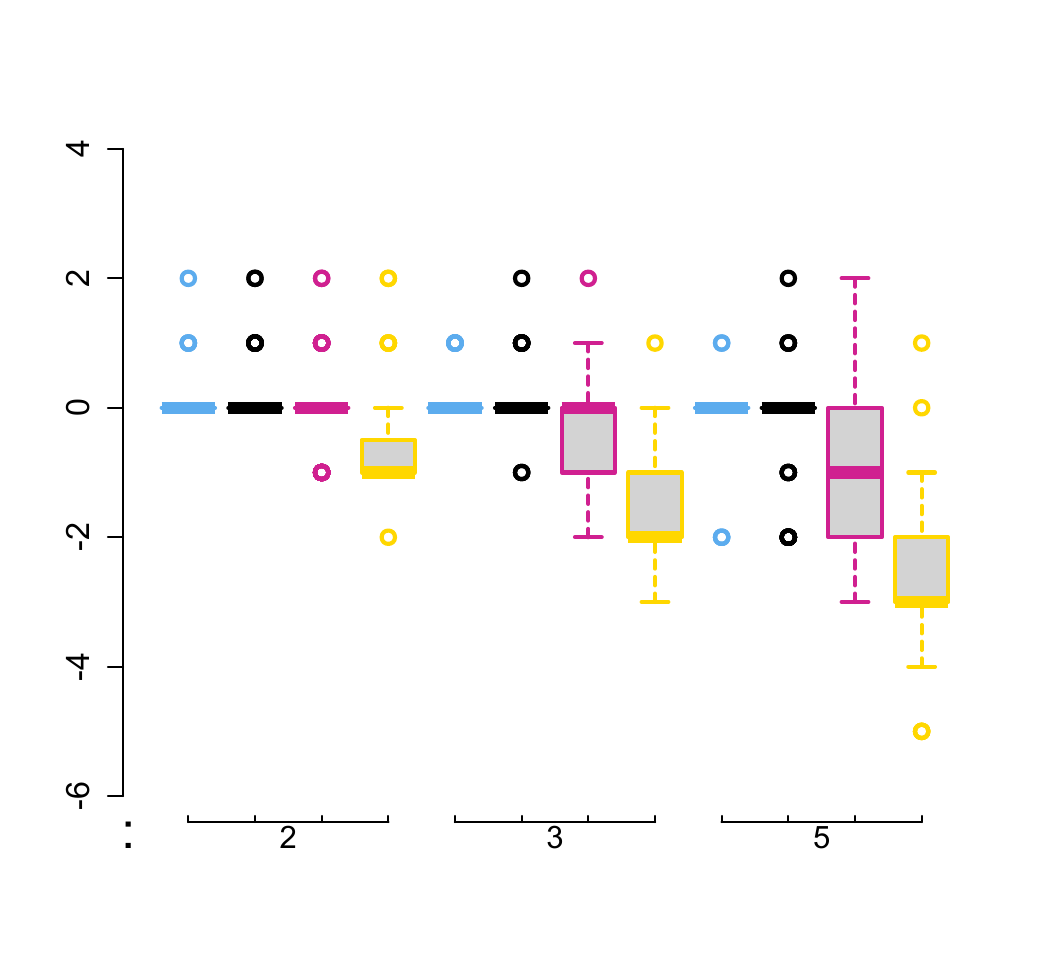}
\end{minipage}
\caption{Distribution of $\widehat{\ell}-\ell$ for simulation set-up 3, under the WBS algorithm (left column) and KW-PELT algorithm (right column) for half-space, Mahalanobis, and modified Mahalanobis depth, respectively.}%
\end{figure}
\begin{figure}
\begin{minipage}[c]{.31\textwidth} 
\centering%
\includegraphics[width=\textwidth]{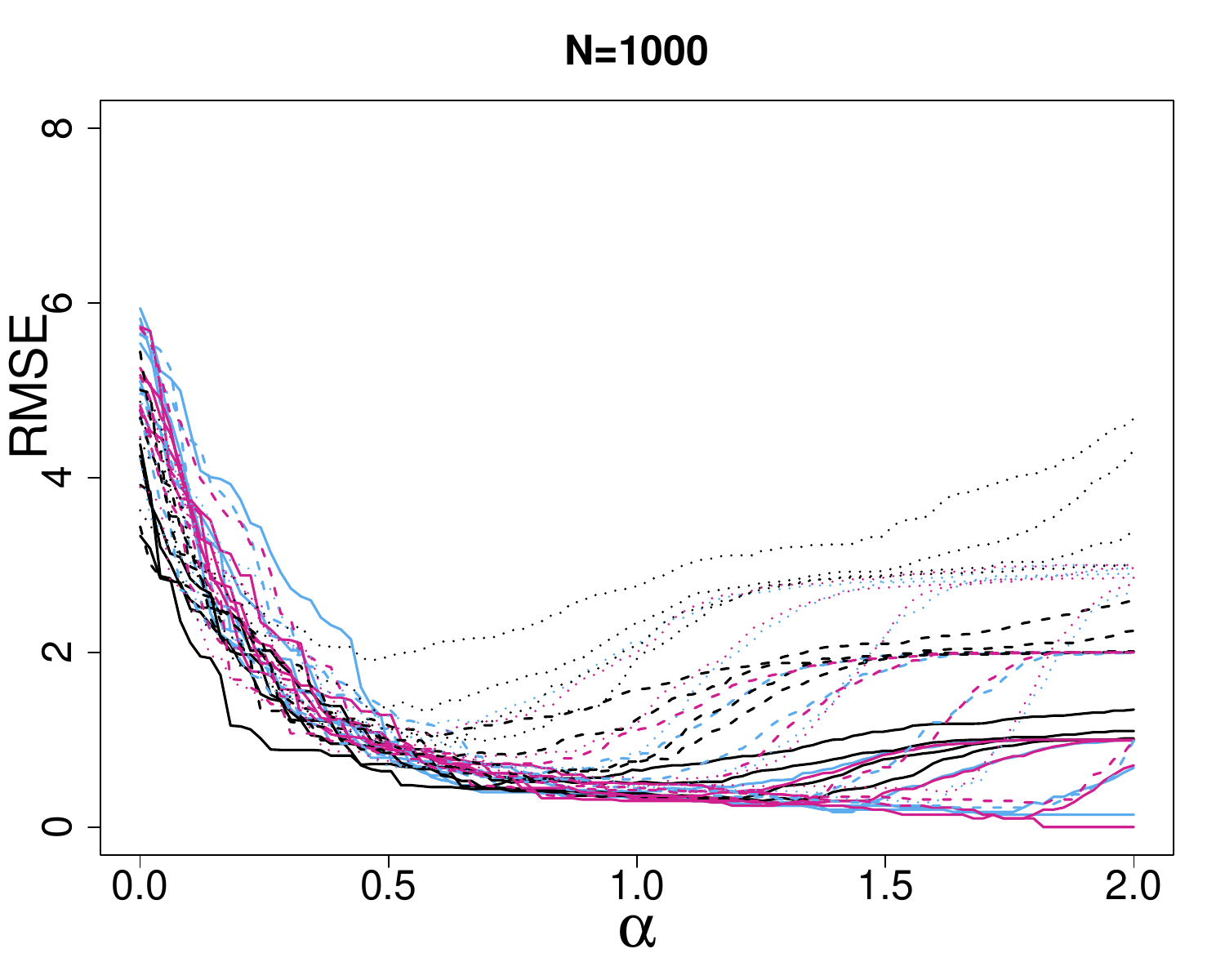}
\end{minipage}
\begin{minipage}[c]{.31\textwidth} 
\centering%
\includegraphics[width=\textwidth]{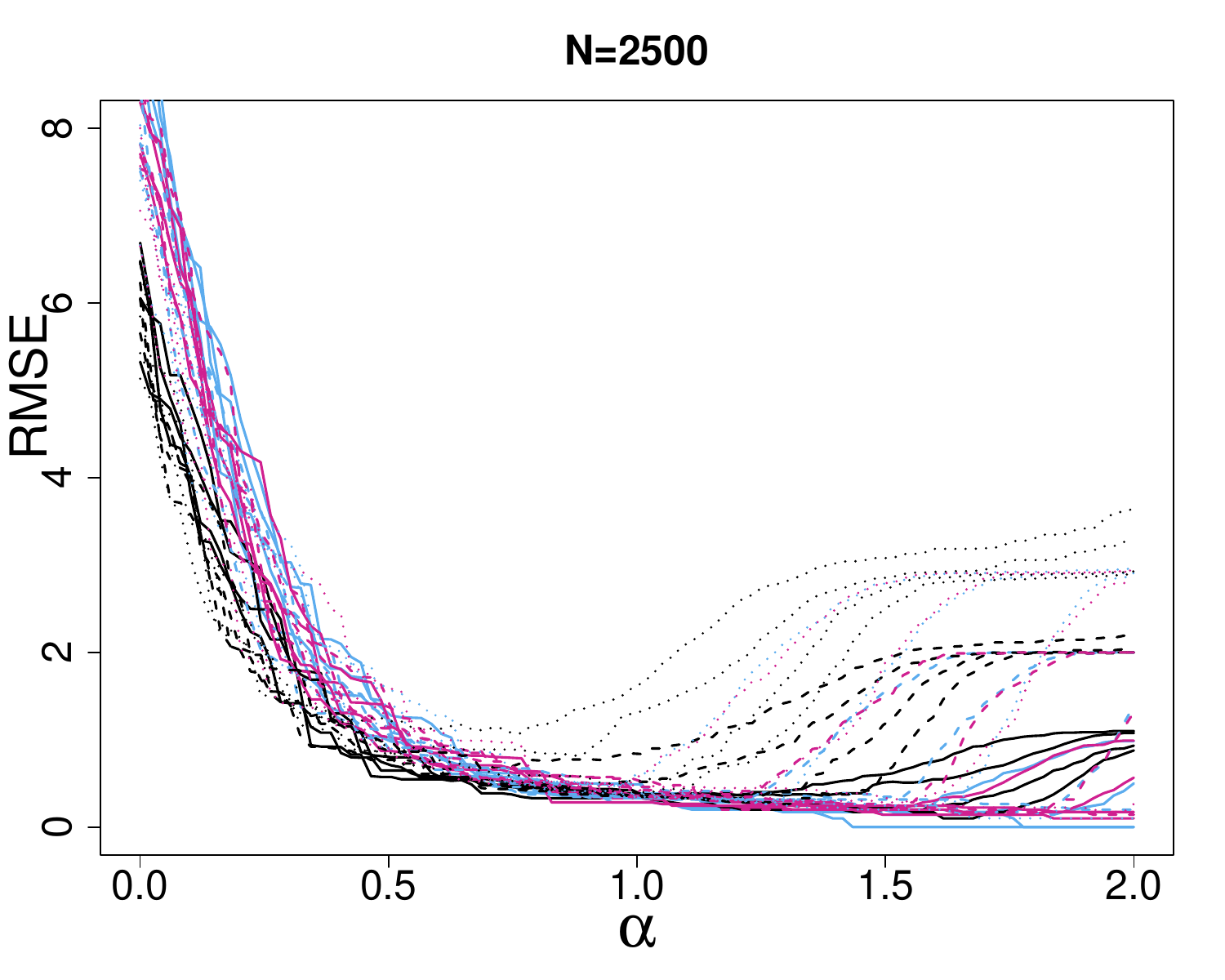}
\end{minipage}
\begin{minipage}[c]{.31\textwidth} 
\centering%
\includegraphics[width=\textwidth]{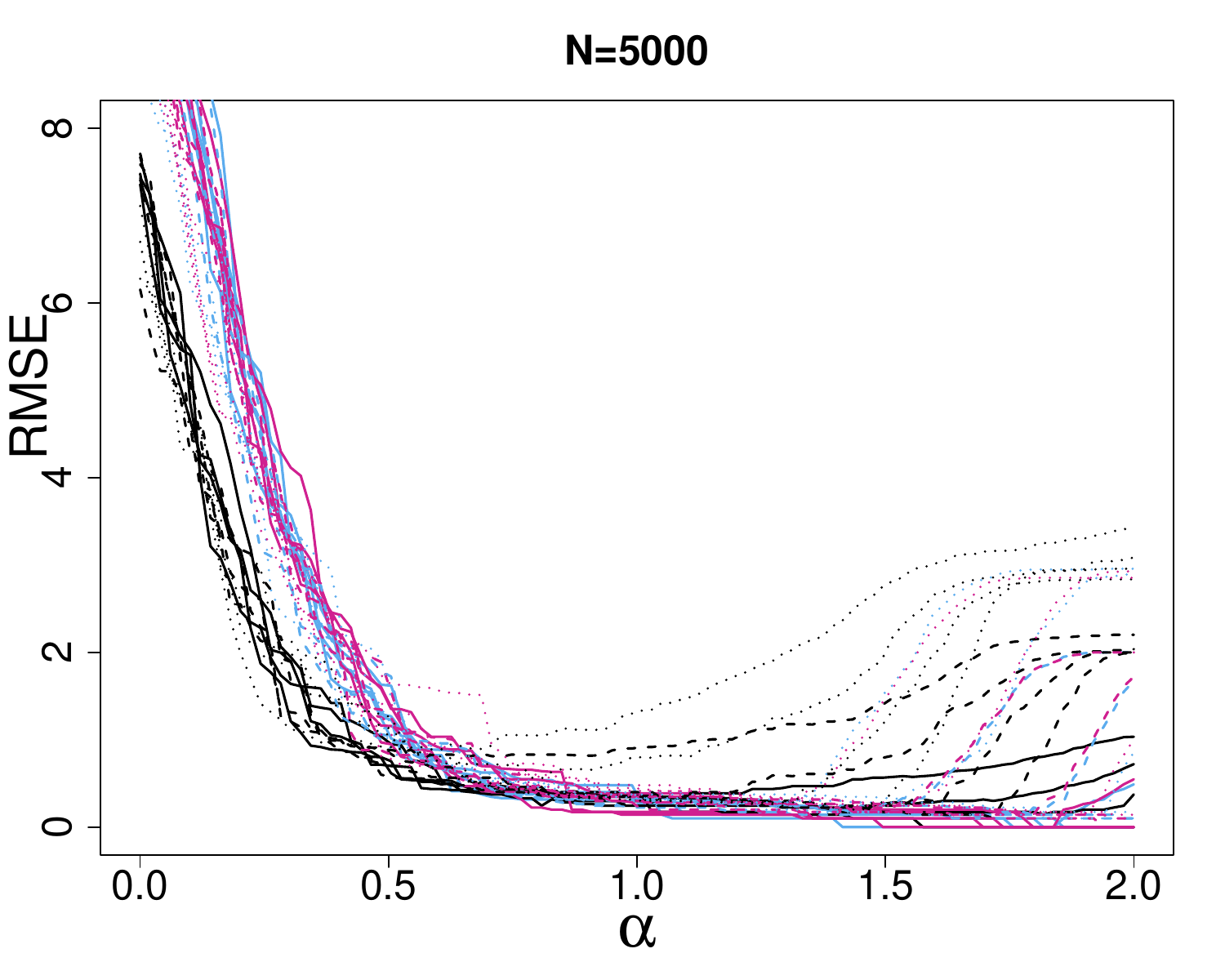}
\end{minipage}
\caption{RMSE for different values of $\alpha$ under Mahalanobis depth for the second simulation scenario.}%
\end{figure}
\begin{figure}
\begin{minipage}[c]{.31\textwidth} 
\centering%
\includegraphics[width=\textwidth]{Plots/Param_Plots/SPAT_C.pdf}
\end{minipage}
\begin{minipage}[c]{.31\textwidth} 
\centering%
\includegraphics[width=\textwidth]{Plots/Param_Plots/SPAT_C_OC_N_2500.pdf}
\end{minipage}
\begin{minipage}[c]{.31\textwidth} 
\centering%
\includegraphics[width=\textwidth]{Plots/Param_Plots/SPAT_C_OC_N_5000.pdf}
\end{minipage}\hfill\newline
\begin{minipage}[c]{.31\textwidth} 
\centering%
\includegraphics[width=\textwidth]{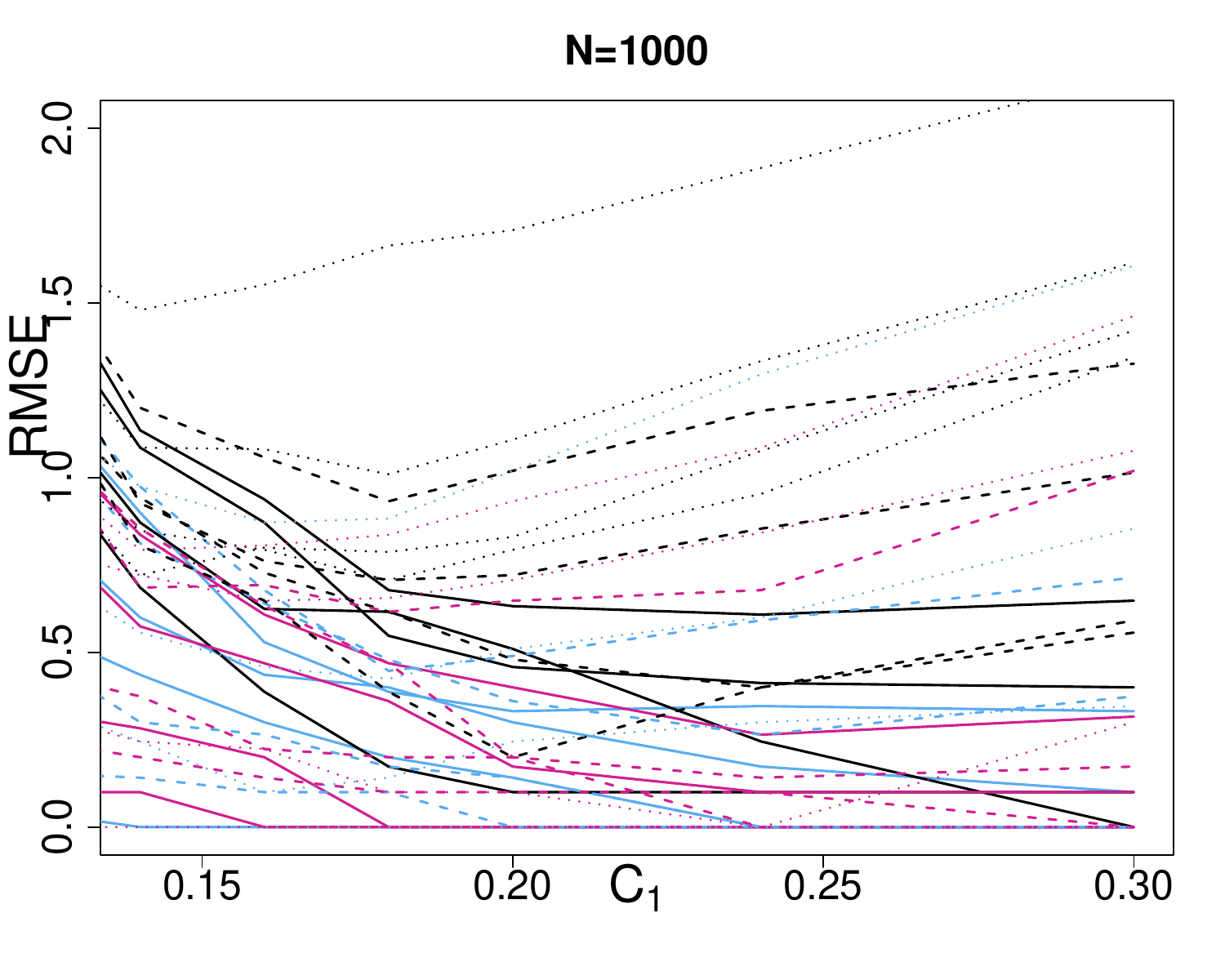}
\end{minipage}
\begin{minipage}[c]{.31\textwidth} 
\centering%
\includegraphics[width=\textwidth]{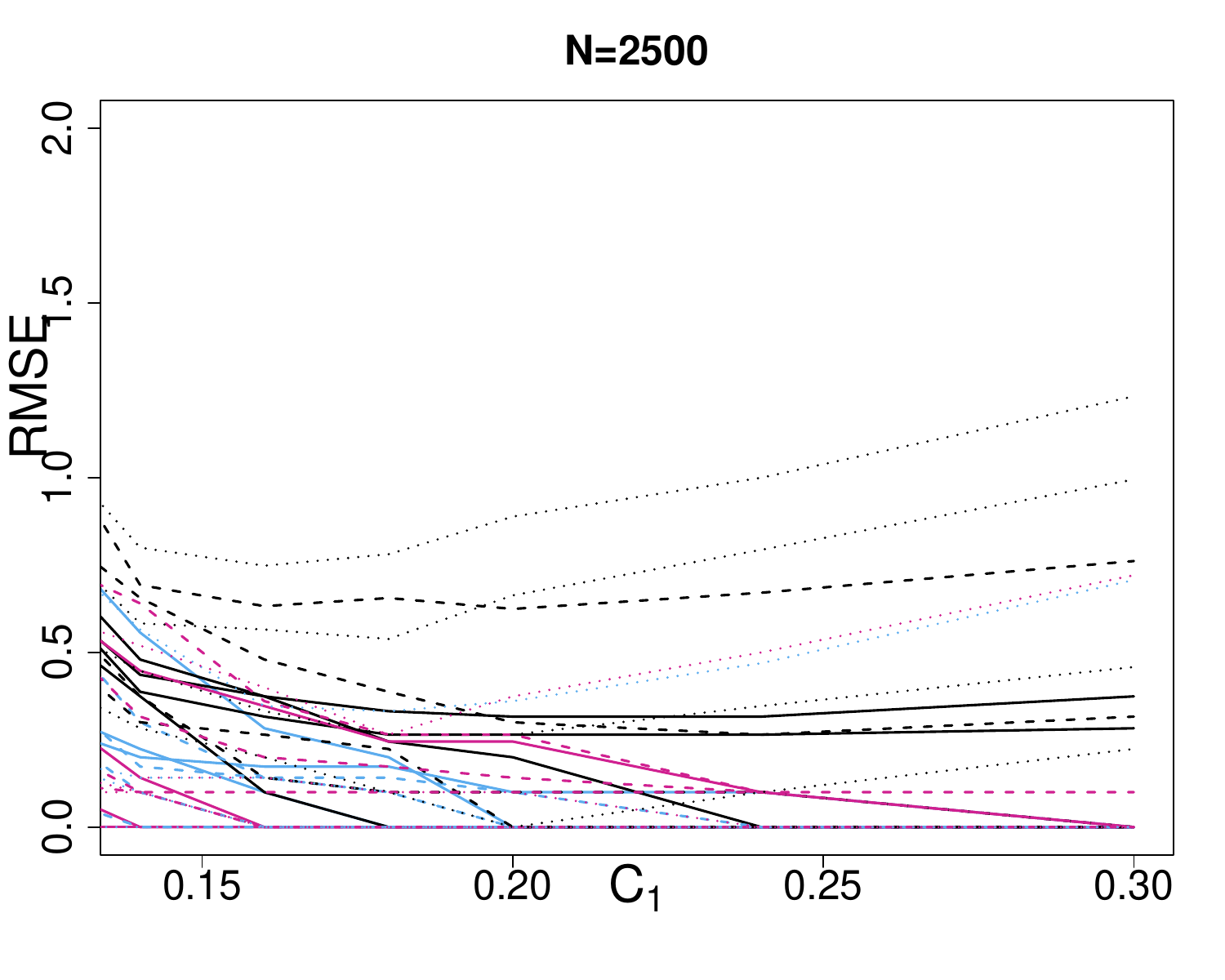}
\end{minipage}
\begin{minipage}[c]{.31\textwidth} 
\centering%
\includegraphics[width=\textwidth]{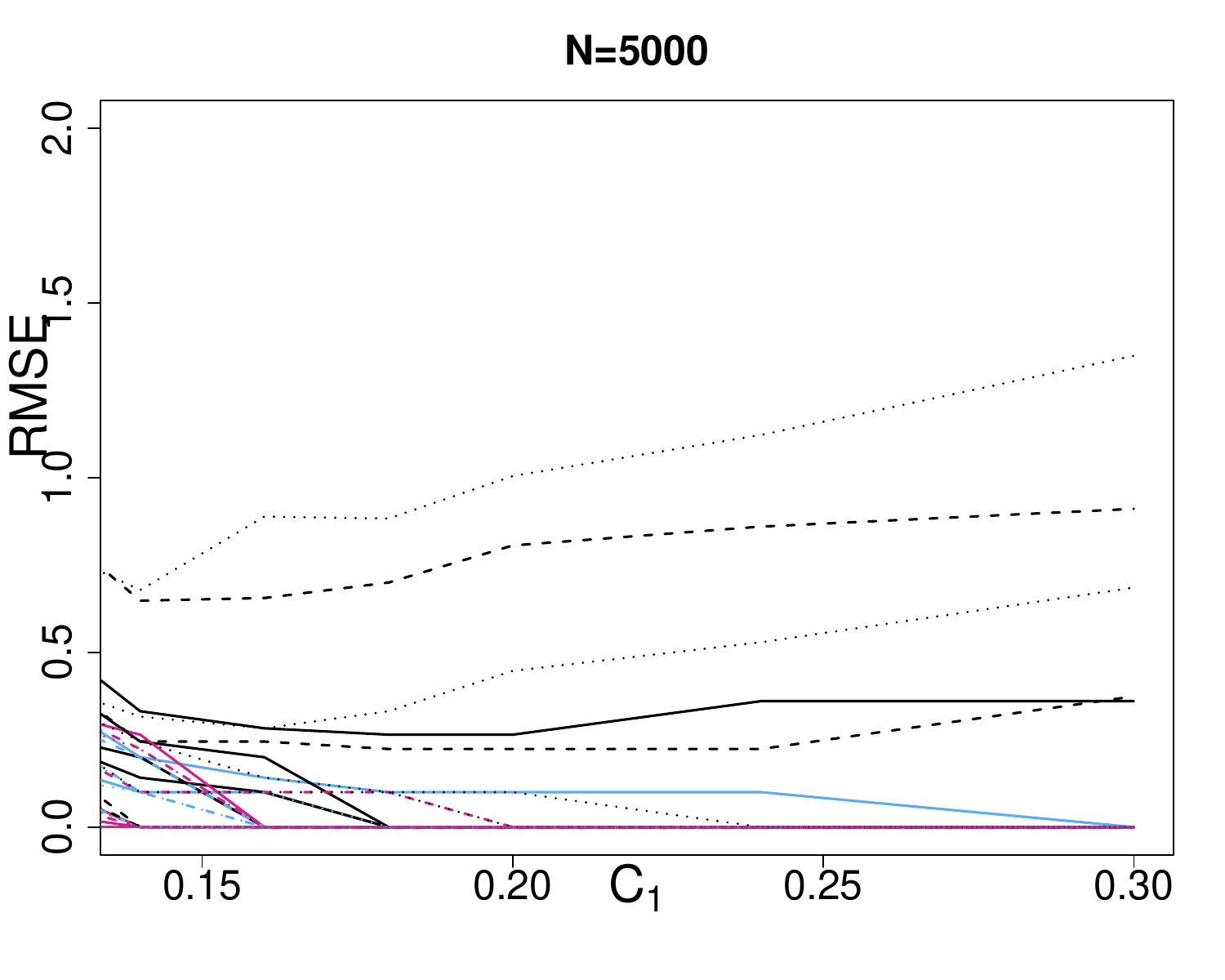}
\end{minipage}\hfill\newline
\begin{minipage}[c]{.31\textwidth} 
\centering%
\includegraphics[width=\textwidth]{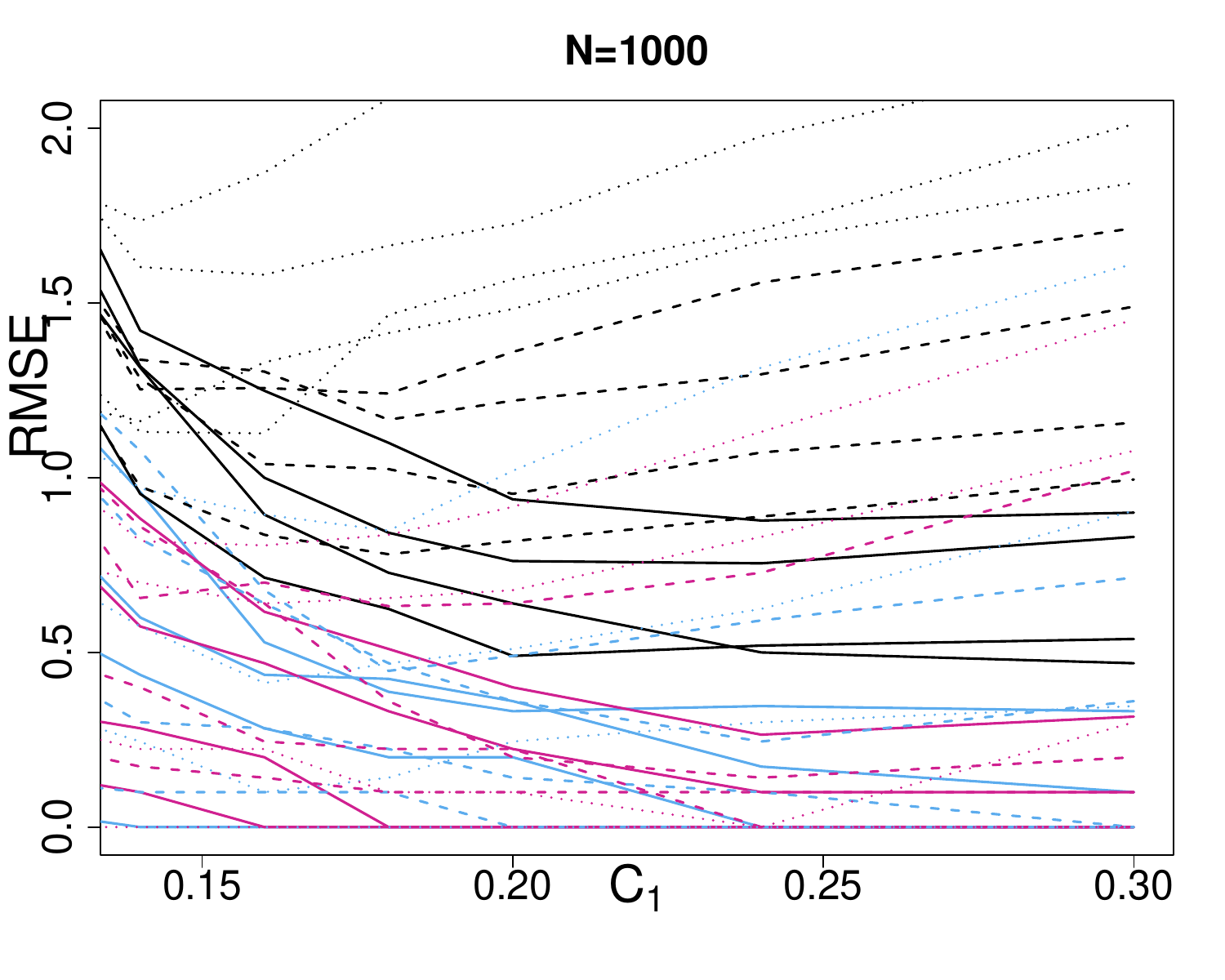}
\end{minipage}
\begin{minipage}[c]{.31\textwidth} 
\centering%
\includegraphics[width=\textwidth]{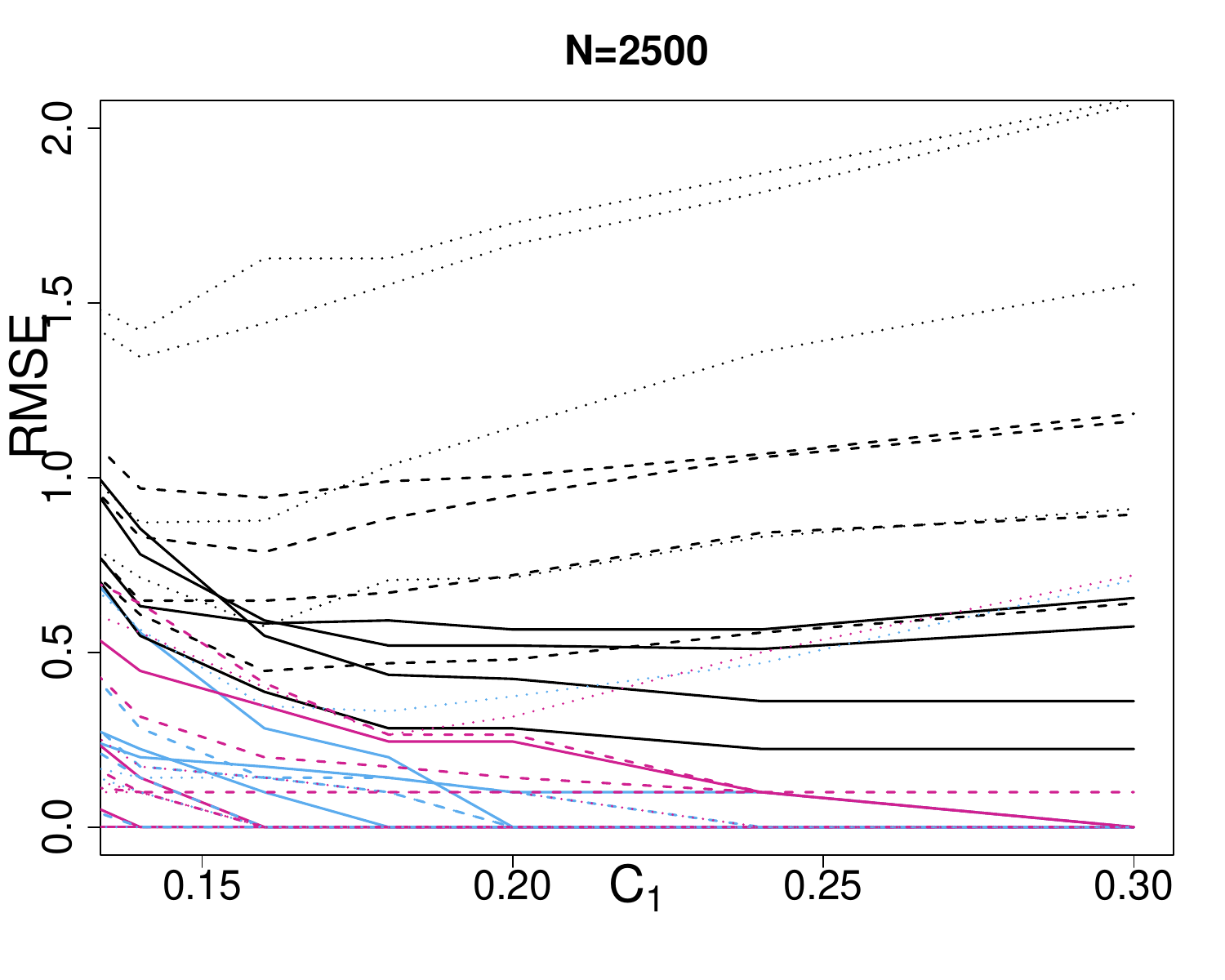}
\end{minipage}
\begin{minipage}[c]{.31\textwidth} 
\centering%
\includegraphics[width=\textwidth]{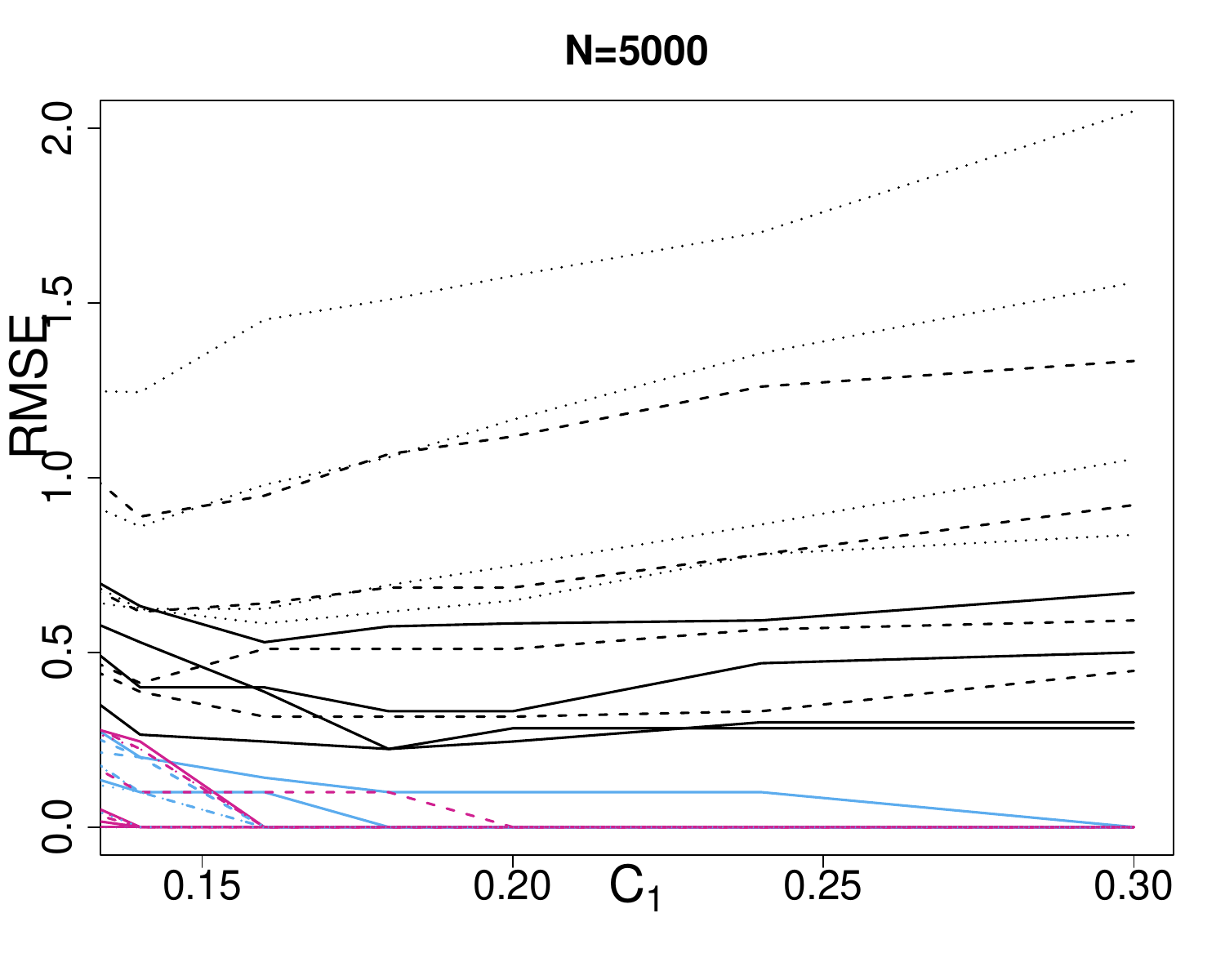}
\end{minipage}
\caption{RMSE Plots for the penalty constant $C_1$ under the first scenario for (top row) halfspace depth, (second row) spatial depth, (third row) Mahalanobis depth, (last row) Modified Mahalanobis depth.}%
\end{figure}
\begin{figure}
\begin{minipage}[c]{.31\textwidth} 
\centering%
\includegraphics[width=\textwidth]{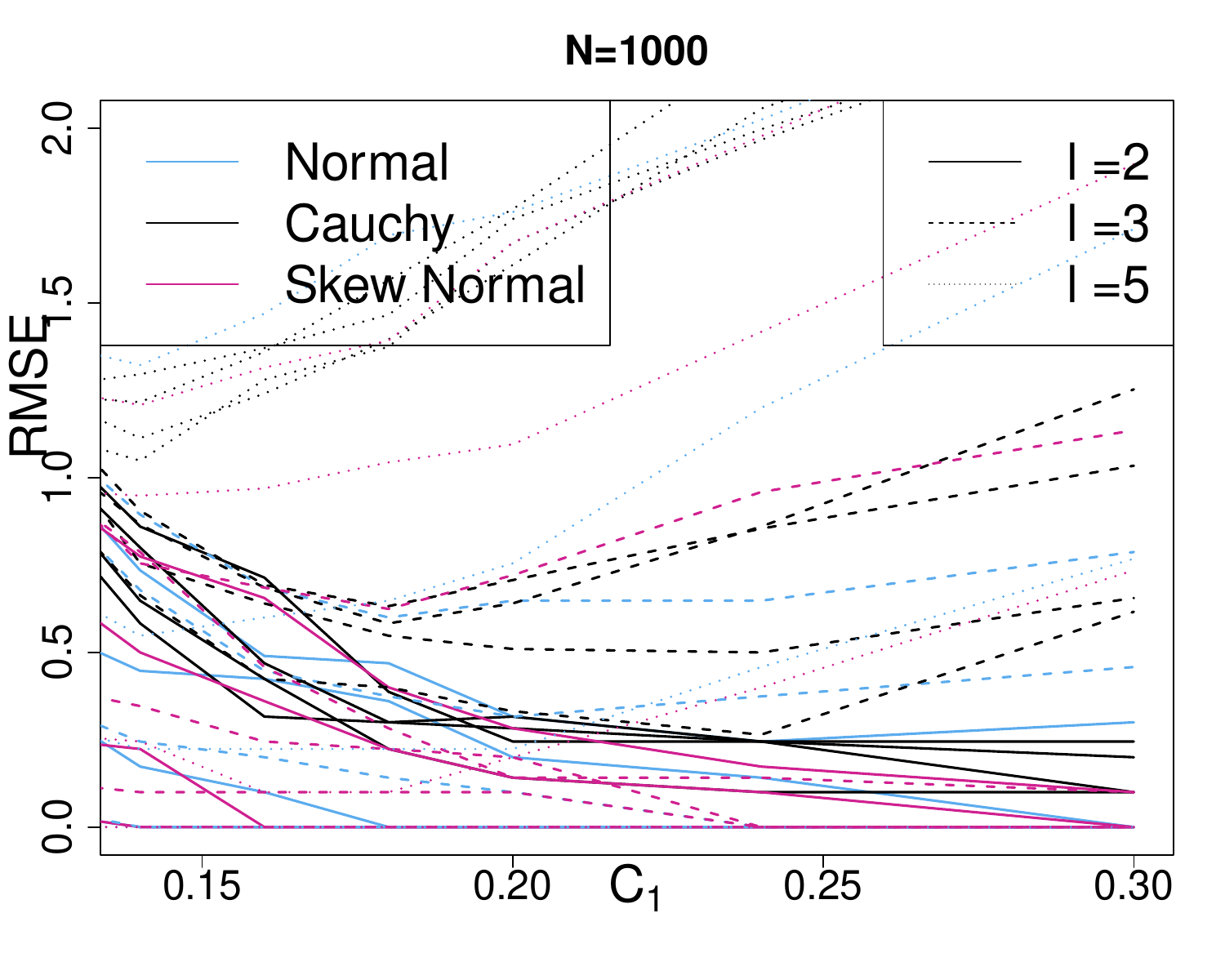}
\end{minipage}
\begin{minipage}[c]{.31\textwidth} 
\centering%
\includegraphics[width=\textwidth]{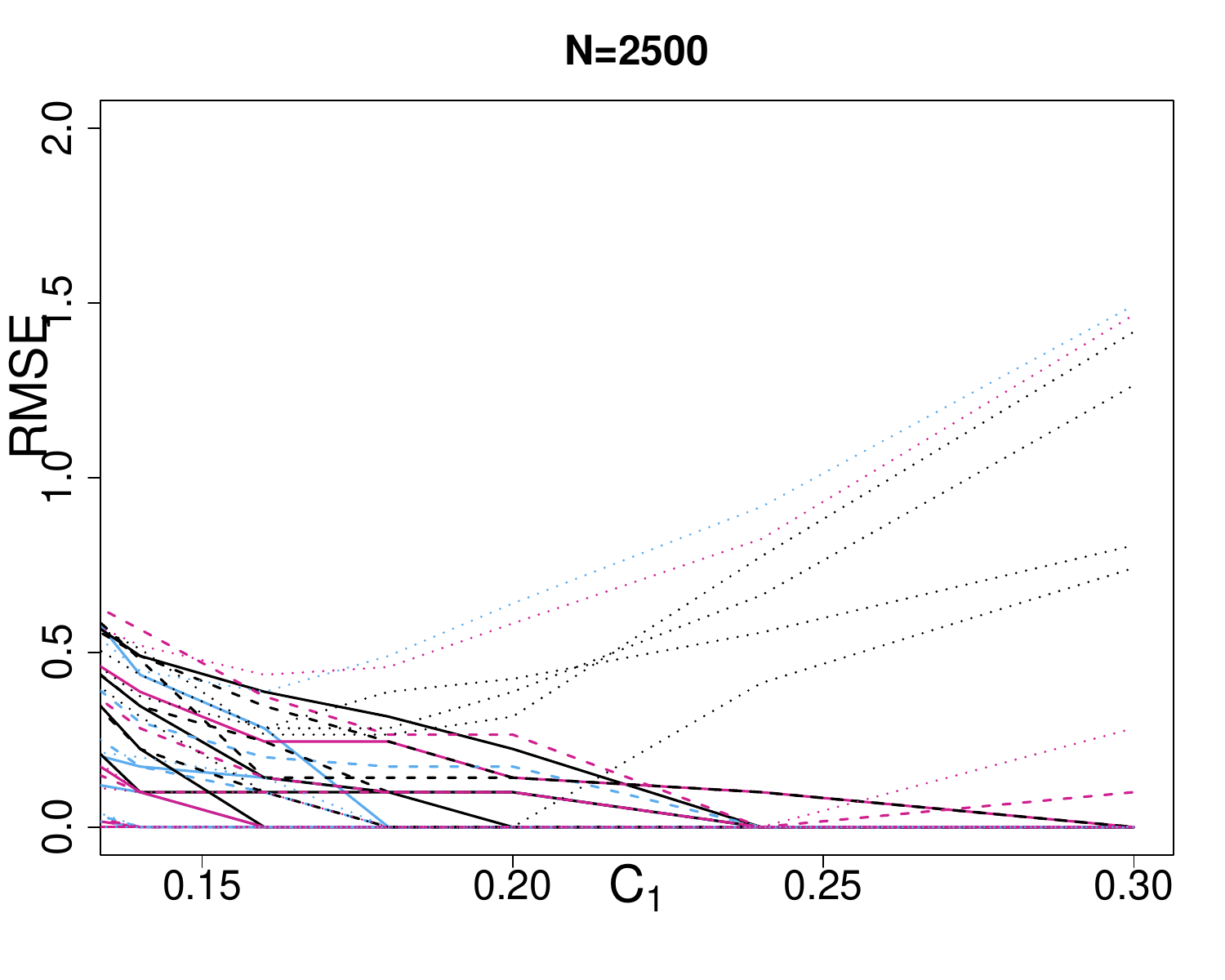}
\end{minipage}
\begin{minipage}[c]{.31\textwidth} 
\centering%
\includegraphics[width=\textwidth]{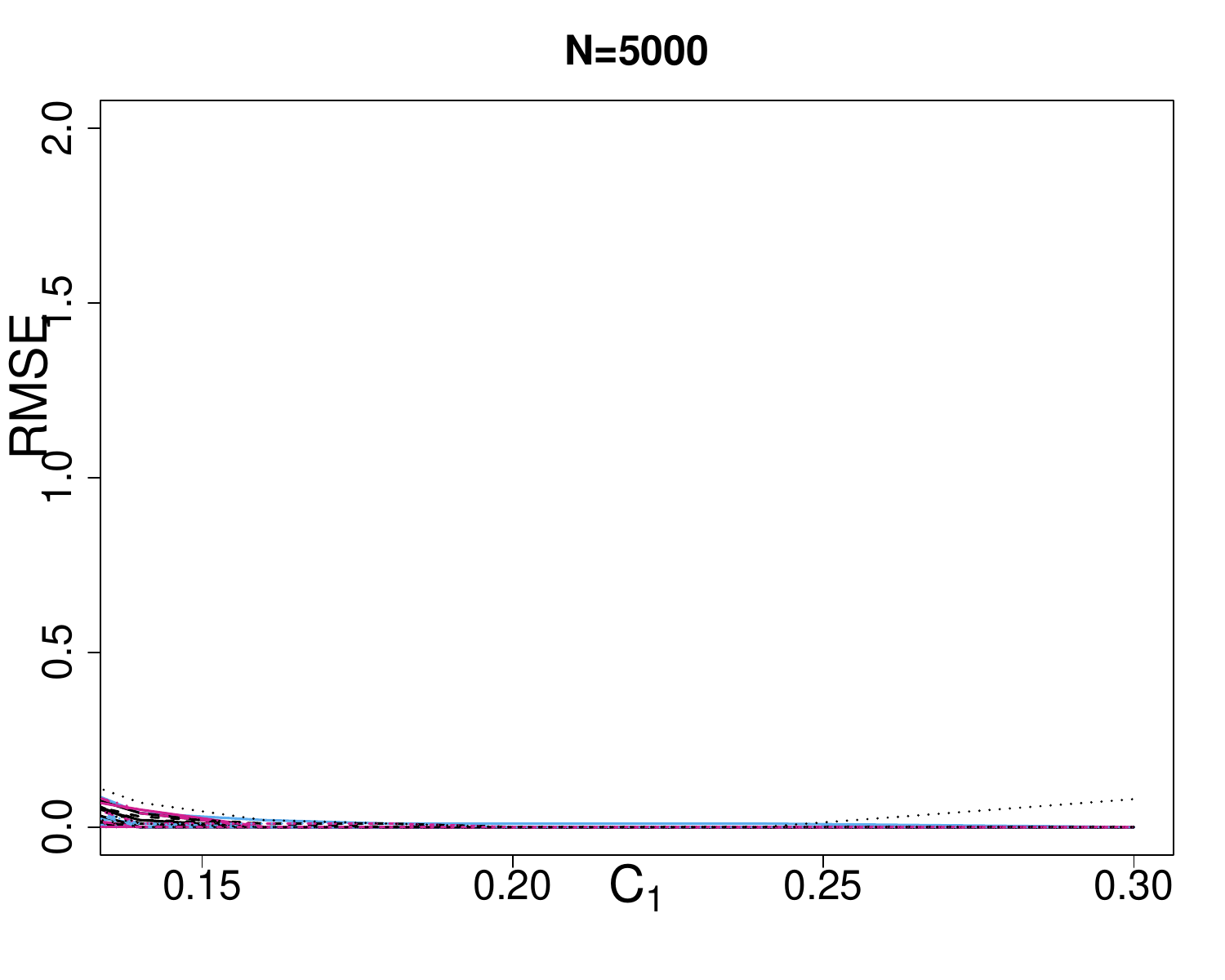}
\end{minipage}\hfill\newline
\begin{minipage}[c]{.31\textwidth} 
\centering%
\includegraphics[width=\textwidth]{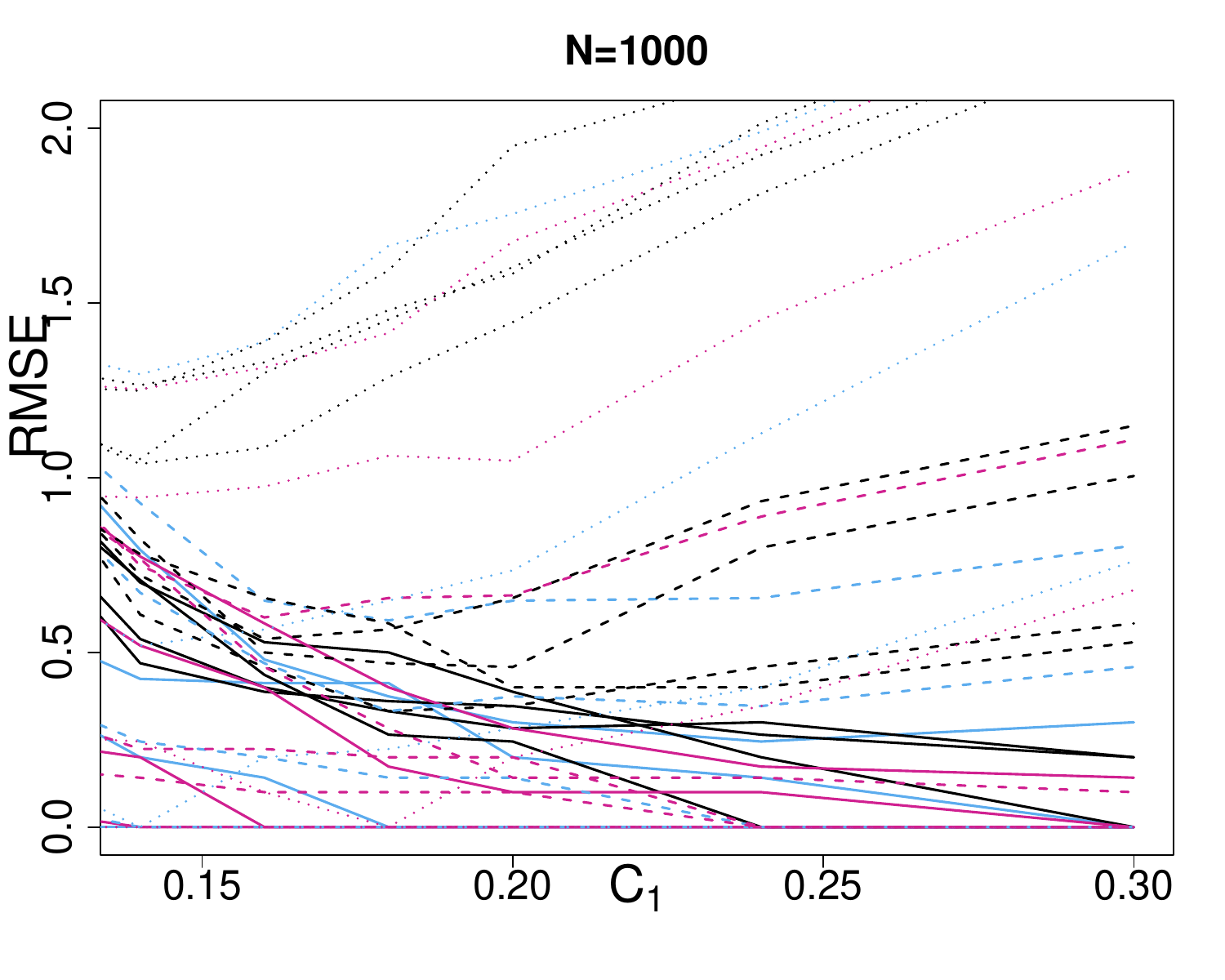}
\end{minipage}
\begin{minipage}[c]{.31\textwidth} 
\centering%
\includegraphics[width=\textwidth]{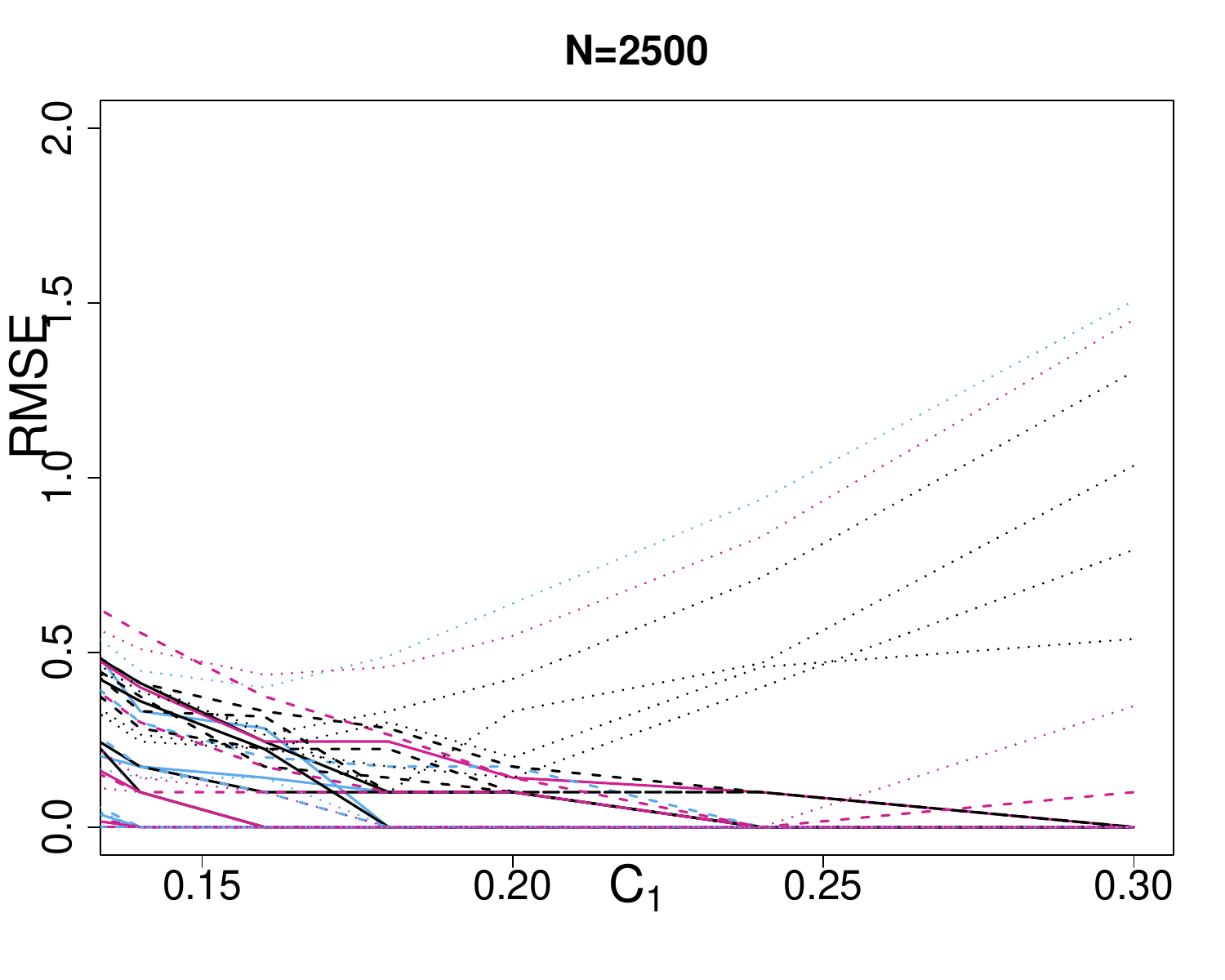}
\end{minipage}
\begin{minipage}[c]{.31\textwidth} 
\centering%
\includegraphics[width=\textwidth]{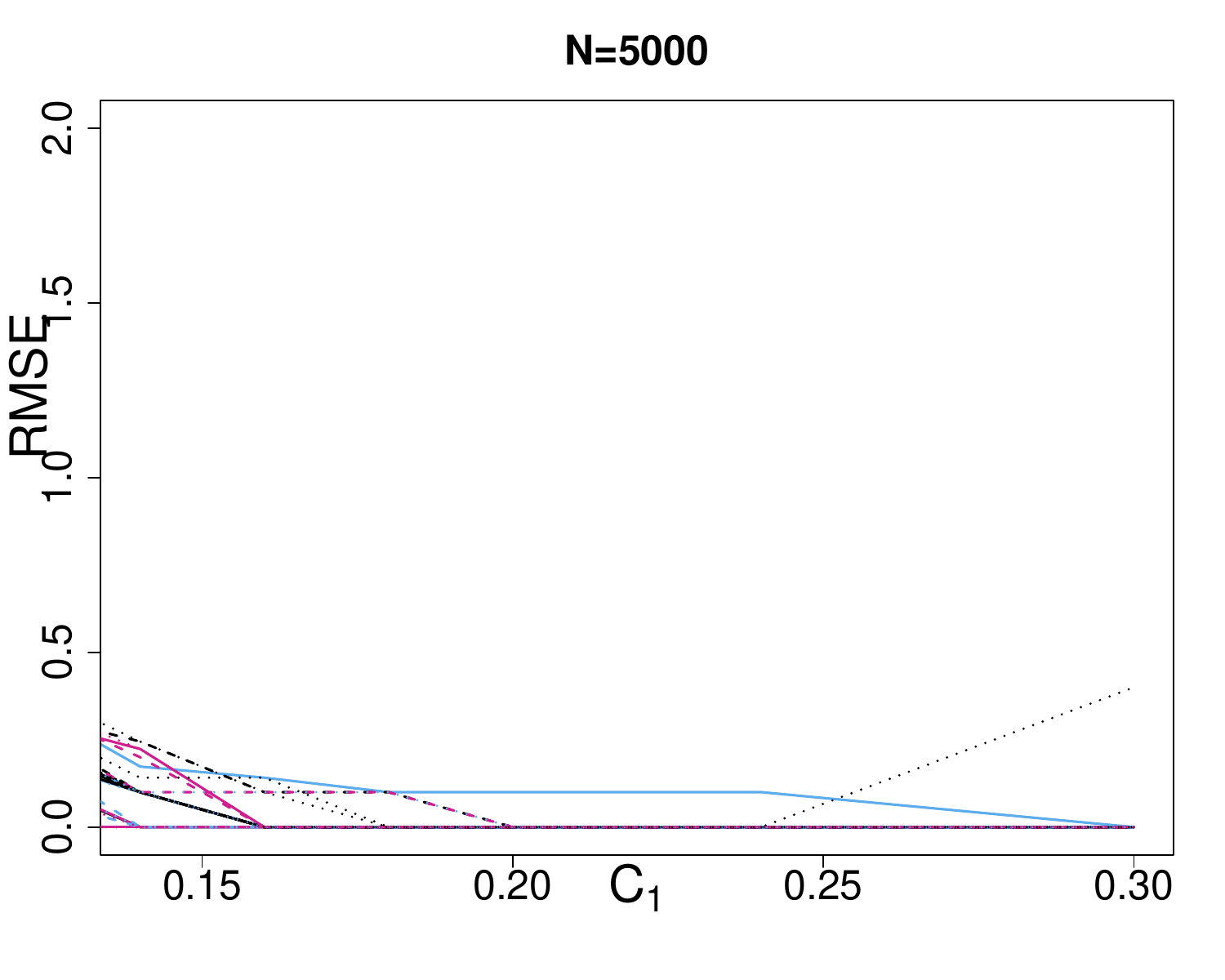}
\end{minipage}\hfill\newline
\begin{minipage}[c]{.31\textwidth} 
\centering%
\includegraphics[width=\textwidth]{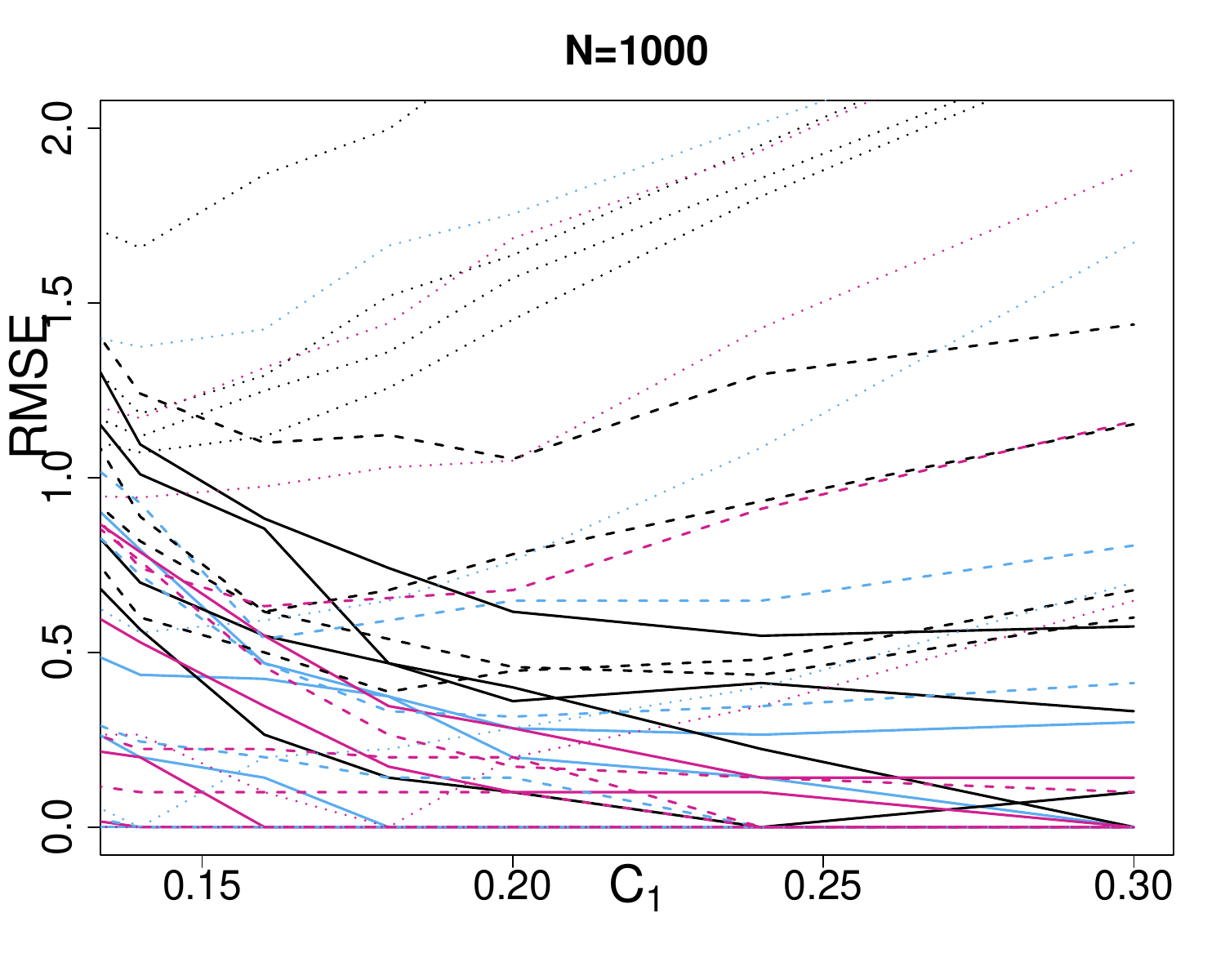}
\end{minipage}
\begin{minipage}[c]{.31\textwidth} 
\centering%
\includegraphics[width=\textwidth]{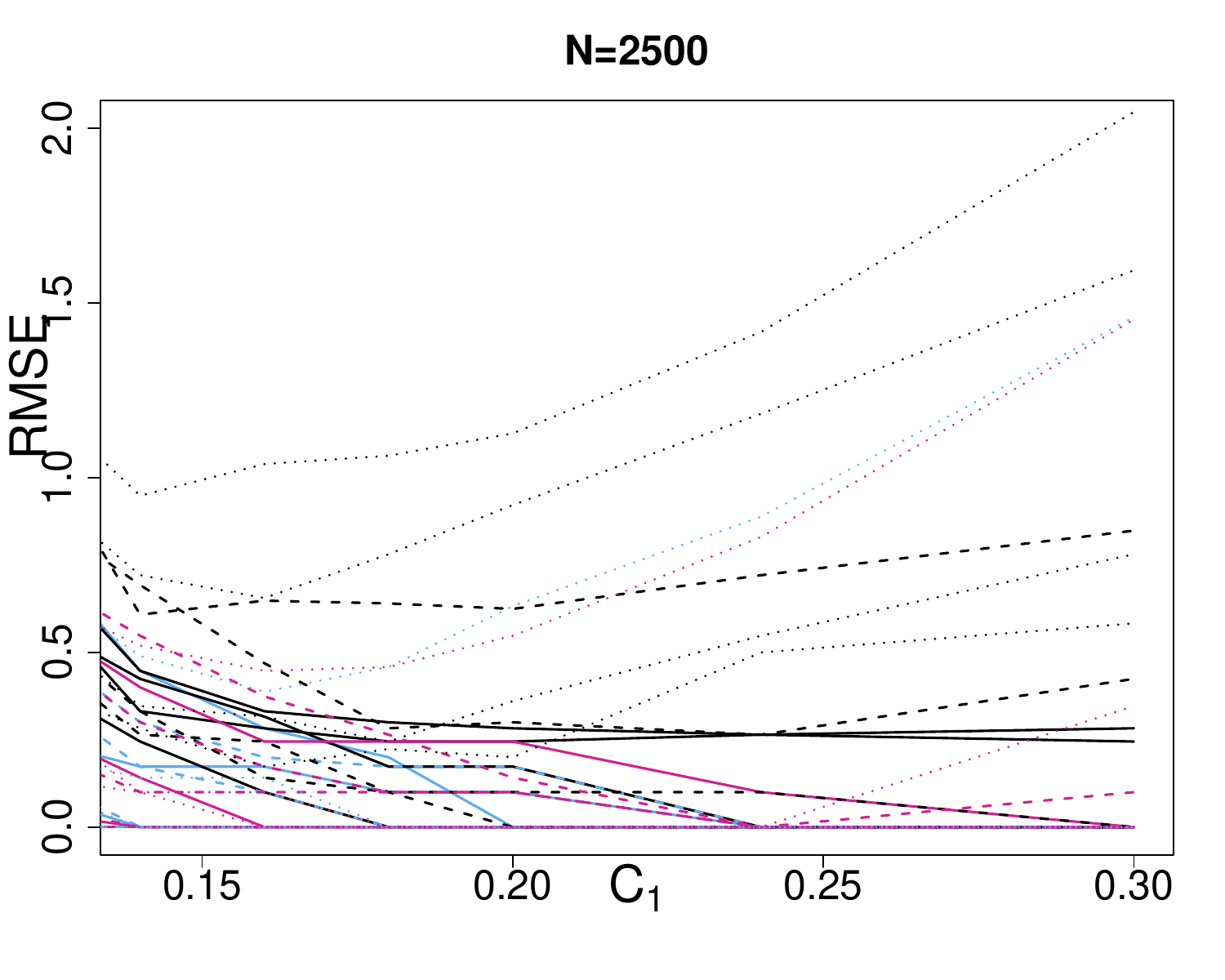}
\end{minipage}
\begin{minipage}[c]{.31\textwidth} 
\centering%
\includegraphics[width=\textwidth]{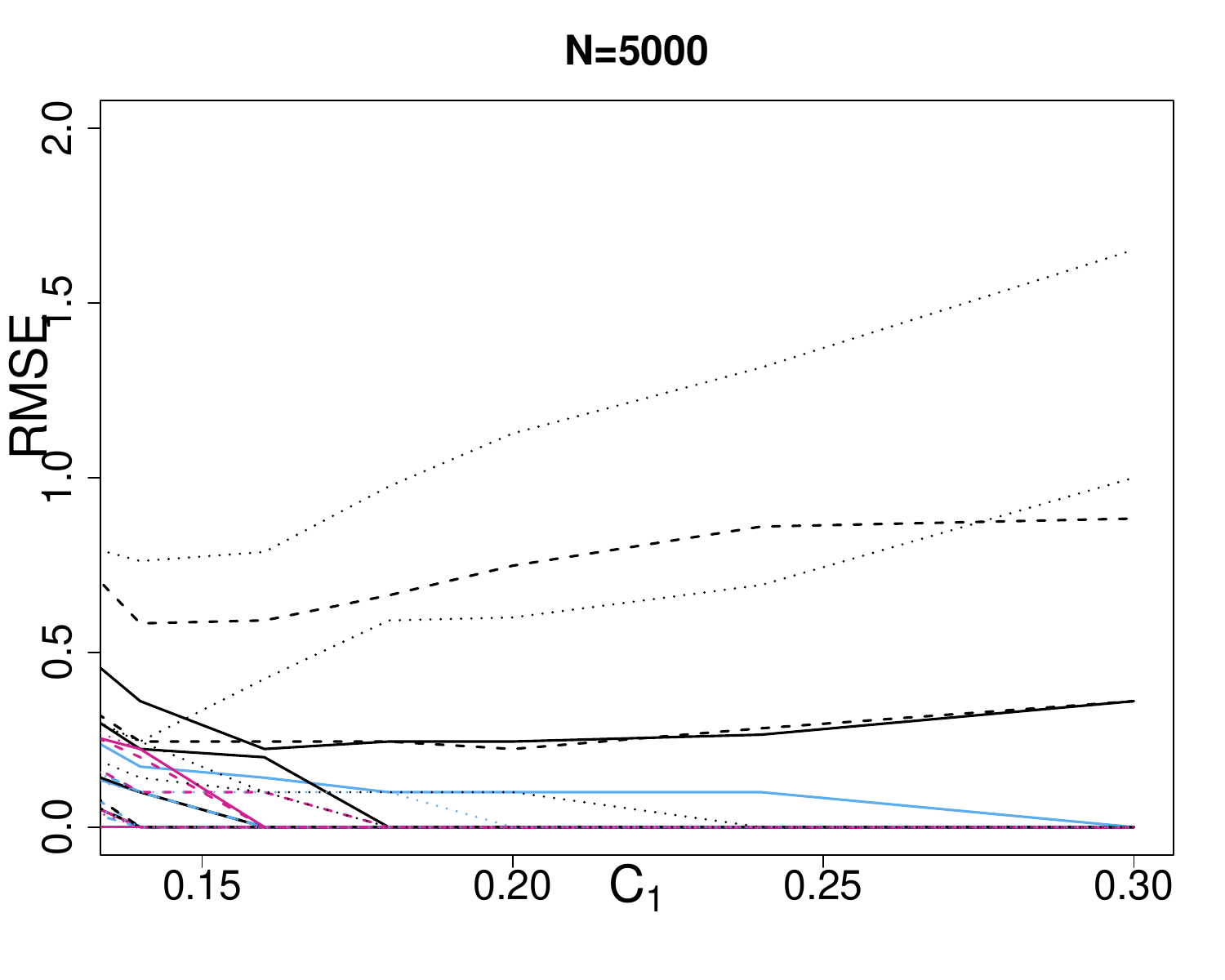}
\end{minipage}\hfill\newline
\begin{minipage}[c]{.31\textwidth} 
\centering%
\includegraphics[width=\textwidth]{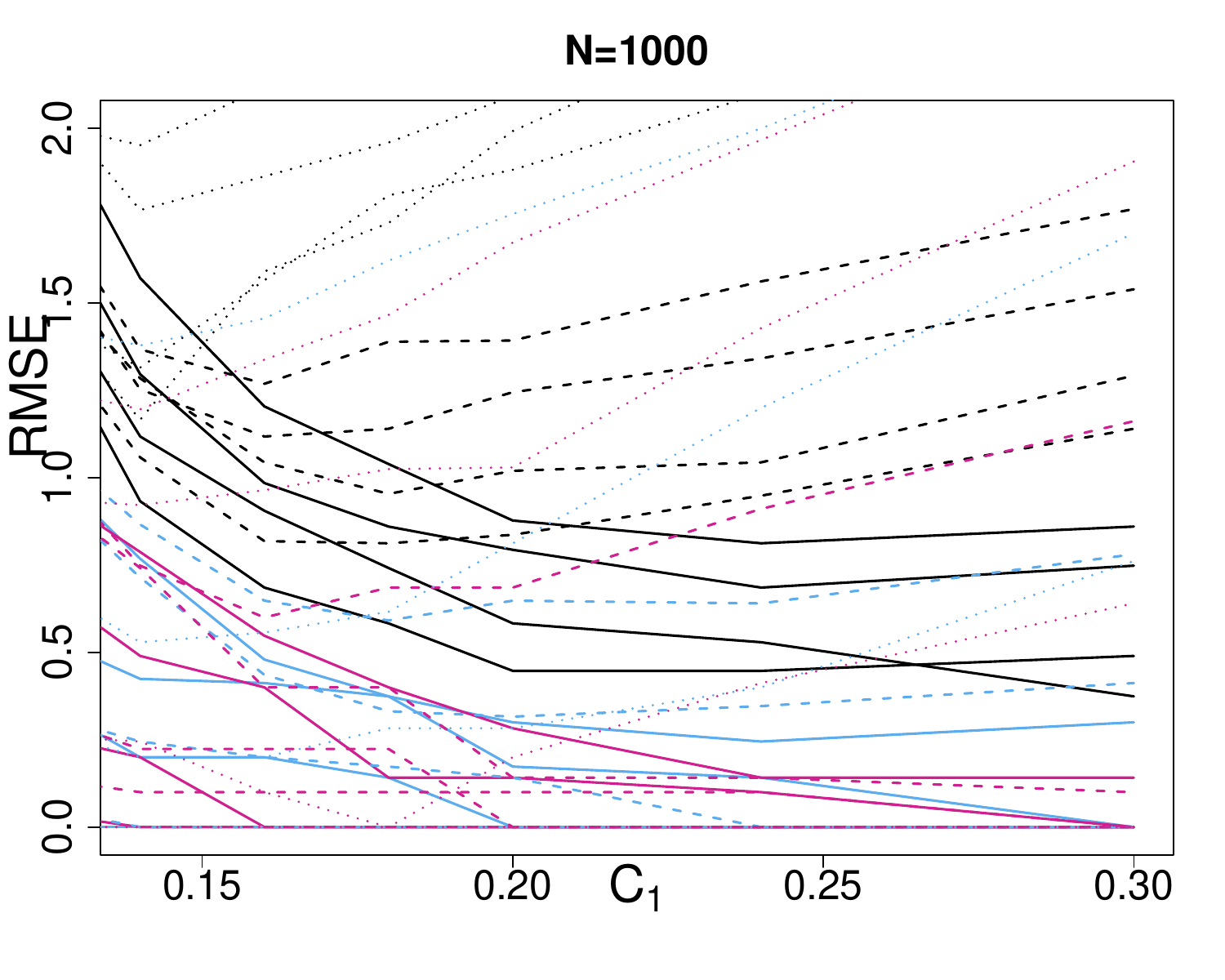}
\end{minipage}
\begin{minipage}[c]{.31\textwidth} 
\centering%
\includegraphics[width=\textwidth]{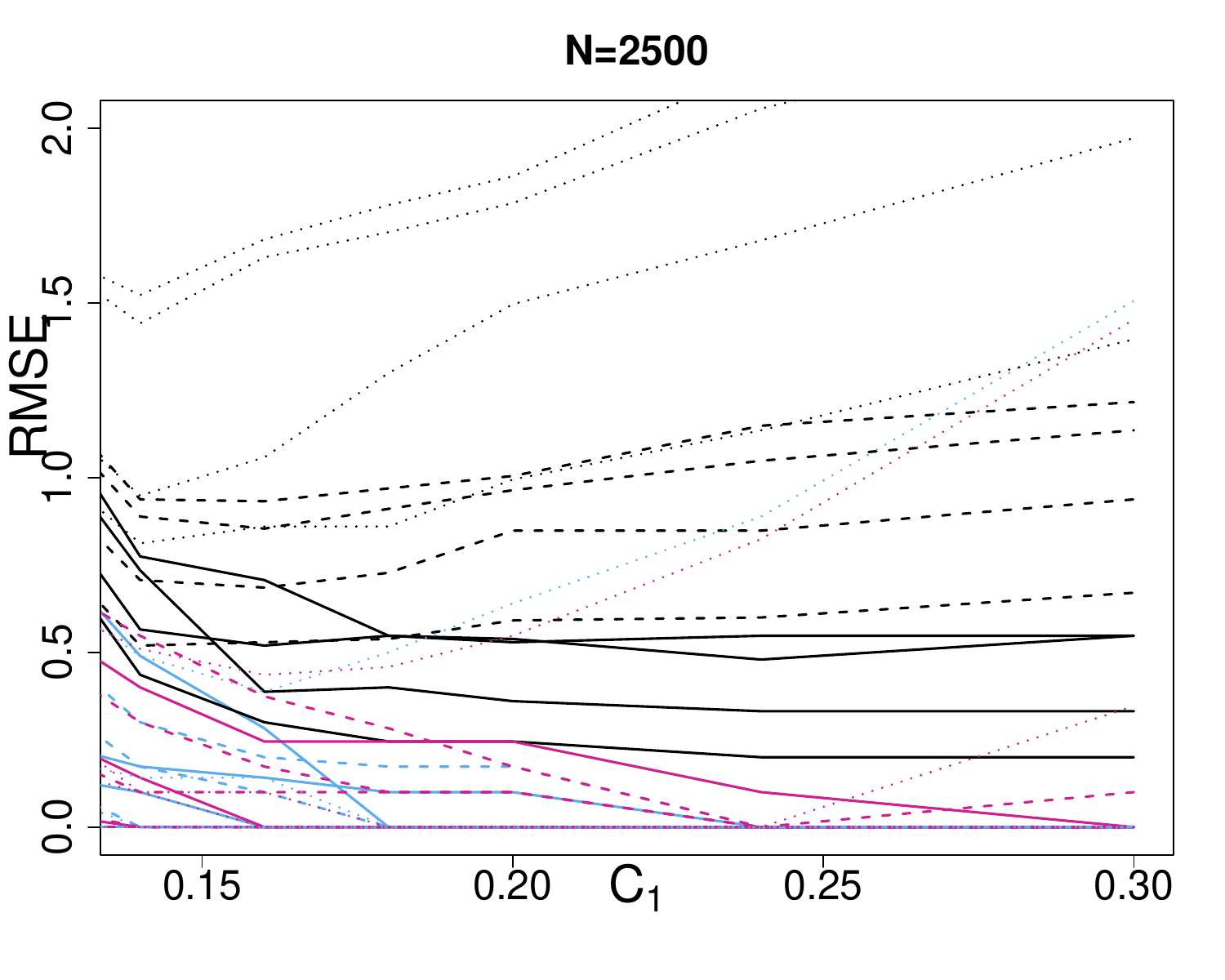}
\end{minipage}
\begin{minipage}[c]{.31\textwidth} 
\centering%
\includegraphics[width=\textwidth]{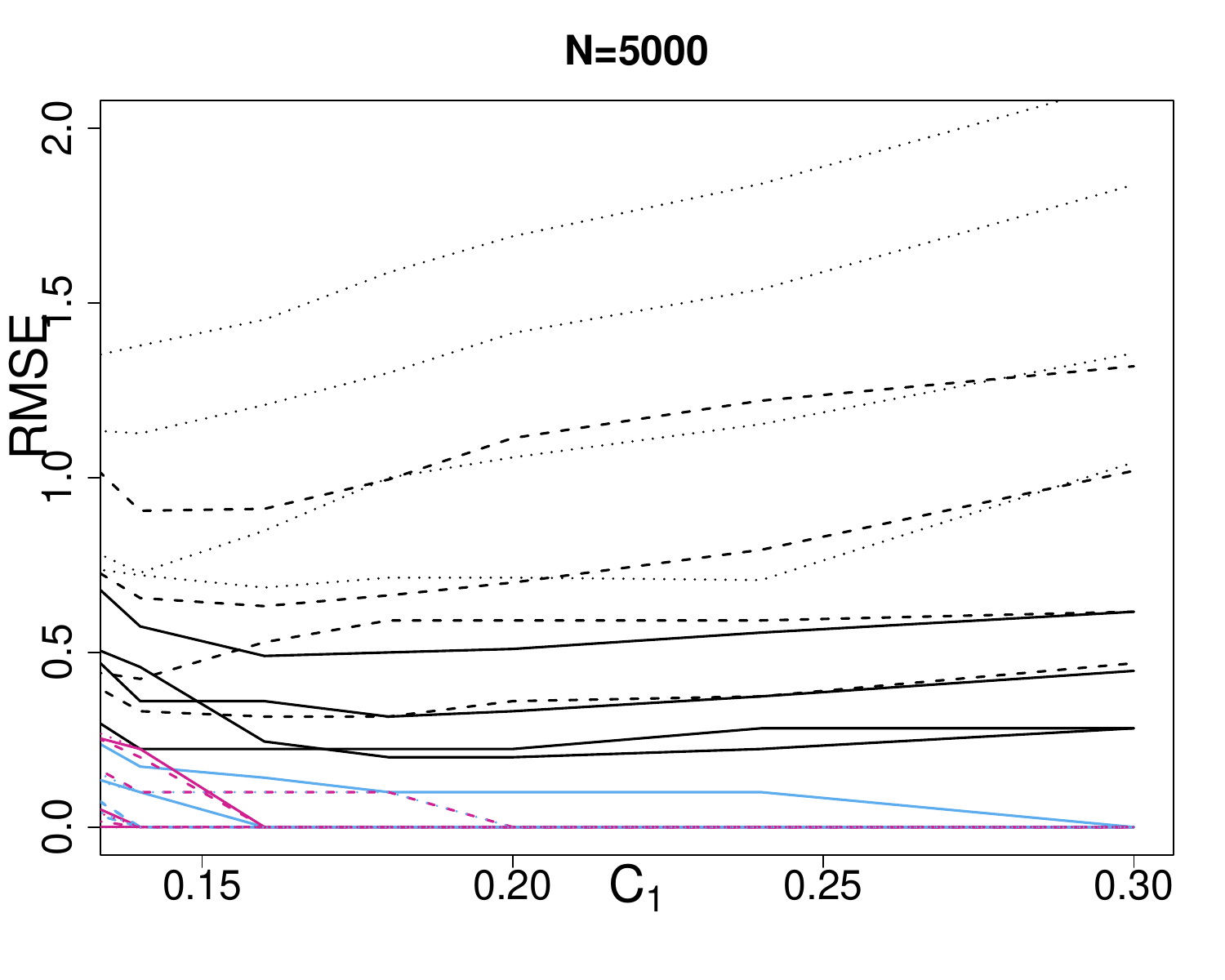}
\end{minipage}
\caption{RMSE Plots for the penalty constant $C_1$ under the second scenario for (top row) halfspace depth, (second row) spatial depth, (third row) Mahalanobis depth, (last row) Modified Mahalanobis depth.}%
\end{figure}
\end{document}